\documentclass[11pt,reqno]{amsart}
\usepackage{amssymb}

\usepackage{amsthm,amsfonts,amssymb,euscript,color,bbm}
\usepackage{comment}
\usepackage{mathtools}
\usepackage{stmaryrd}
\usepackage{mathrsfs}
\usepackage{pifont}
\usepackage{amsfonts,amsmath,amssymb,amsthm,amscd,cancel}
\usepackage{graphicx}
\usepackage[]{geometry}
\usepackage{verbatim}
\usepackage{mathrsfs}
\usepackage{fancyhdr}
\usepackage{footnpag,footmisc}
\usepackage[normalem]{ulem}

\usepackage{cite}
\usepackage{hyperref}
\usepackage{engord}
\usepackage[displaymath]{lineno}

\theoremstyle{definition}
\newtheorem{theorem}{Theorem}[section]
\newtheorem{lemma}{Lemma}[section]
\newtheorem{proposition}{Proposition}[section]
\newtheorem{corollary}{Corollary}[section]
\newtheorem{definition}{Definition}[section]

\newtheorem{remark}{Remark}[section]
\newtheorem*{conjecture}{Conjecture}
\newtheorem*{question}{Basic setup}

\allowdisplaybreaks[4]

\DeclareMathAlphabet{\mathsfsl}{OT1}{cmss}{m}{sl}
\numberwithin{equation}{section}

\newcommand{\D}{\mathrm{d}}

\newcommand{\tr}{\mathrm{tr}}

\renewcommand{\L}{\mathbb{L}}
\renewcommand{\H}{\mathbb{H}}

\def\alphab{\underline{\alpha}}
\def\betab{\underline{\beta}}
\def\chib{\underline{\chi}}
\def\chibh{\widehat{\underline{\chi}}}
\def\chih{\widehat{\chi}}
\def\etab{\underline{\eta}}

\def\Lb{\underline{L}}
\def\mub{\underline{\mu}}

\def\tr{\mathrm{tr}}
\def\omegab{\underline{\omega}}

\def\sigmac{\check{\sigma}}
\def\tensor{\widehat{\otimes}}

\def\ub{\underline{u}}
\def\Cb{\underline{C}}

\def\Lh{\widehat{L}}
\def\Lbh{\widehat{\underline{L}}}

\def\Xib{\underline{\Xi}}
\def\Ib{\underline{I}}
\def\Lambdab{\underline{\Lambda}}
\def\Nb{\underline{N}}
\def\Thetab{\underline{\Theta}}

\newcommand{\Db}{\underline{D}}
\newcommand{\Dh}{\widehat{D}}
\newcommand{\Dbh}{\widehat{\underline{D}}}

\newcommand{\logOmega}{\left|\log\frac{\Omega_0(u_1)}{\Omega_0(u_0)}\right|}

\newcommand{\tdelta}{\widetilde{\delta}}
\newcommand{\tOmega}{\widetilde{\Omega}}
\newcommand{\tu}{\widetilde{u}}

\def\nablas{\mbox{$\nabla \mkern -13mu /$ }}
\def\Deltas{\mbox{$\Delta \mkern -13mu /$ }}

\def\divs{\mbox{$\mathrm{div} \mkern -13mu /$ }}
\def\curls{\mbox{$\mathrm{curl} \mkern -13mu /$ }}

\def\omegabs{\mbox{$\omegab \mkern -13mu /$ }}
\def\ds{\mbox{$\nabla \mkern -13mu /$ }}
\def\gs{\mbox{$g \mkern -9mu /$}}
\def\epsilons{\mbox{$\epsilon \mkern -9mu /$}}

\def\Thetas{\mbox{$\Theta \mkern -12mu /$ }}

\def\gfs{\mbox{$\mathfrak{g} \mkern -9mu /$}}

\begin{document}

\title[Instability of spherical naked singularities]{Instability of spherical naked singularities of a scalar field under gravitational perturbations}

\author[Junbin Li]{Junbin Li}
\address{Department of Mathematics, Sun Yat-sen University\\ Guangzhou, China}
\email{lijunbin@mail.sysu.edu.cn}

\author[Jue Liu]{Jue Liu}
\address{Department of Mathematics, Sun Yat-sen University\\ Guangzhou, China}
\email{liuj337@mail2.sysu.edu.cn}

\date{}

\maketitle

\begin{abstract}
In this paper, we initiate the study of the instability of naked singularities without symmetries. In a series of papers, Christodoulou proved that naked singularities are not stable in the context of the spherically symmetric Einstein equations coupled with a massless scalar field. We study in this paper the next simplest case: a characteristic initial value problem of this coupled system with the initial data given on two intersecting null cones, the incoming one of which is assumed to be spherically symmetric and singular at its vertex, and the outgoing one of which has no symmetries. It is shown that, arbitrarily fixing the initial scalar field, the set of the initial conformal metrics on the outgoing null cone such that the maximal future development does not have any sequences of closed trapped surfaces approaching the singularity, is of first category in the whole space in which the shear tensors are continuous. Such a set can then be viewed as exceptional, although the exceptionality is weaker than the at least $1$ co-dimensionality in spherical symmetry. Almost equivalently, it is also proved that, arbitrarily fixing an incoming null cone $\Cb_\varepsilon$ to the future of the initial incoming null cone, the set of the initial conformal metrics such that the maximal future development has at least one closed trapped surface before $\Cb_\varepsilon$, contains an open and dense subset of the whole space. Since the initial scalar field can be chosen such that the singularity is naked if the initial shear is set to be zero, we may say that the spherical naked singularities of a self-gravitating scalar field are not stable under gravitational perturbations. This in particular gives new families of non-spherically symmetric gravitational perturbations different from the original spherically symmetric scalar perturbations given by Christodoulou.

\end{abstract}

\tableofcontents

\setcounter{tocdepth}{1}


\allowdisplaybreaks

\section{Introduction}
\subsection{Previous works}
In general relativity, one models the gravity in spacetimes using a Lorentzian manifold whose metric should satisfy the \emph{Einstein equations}
$$\mathbf{Ric}_{\alpha\beta}-\frac{1}{2}\mathbf{R}g_{\alpha\beta}=\mathbf{T}_{\alpha\beta}.$$

One fundamental question in the classical theory of general relativity is the weak cosmic censorship conjecture, which is proposed firstly by Penrose, and usually stated as follows: \emph{for suitable Einstein-matter field systems, the maximal future development of the generic asymptotically flat initial data possesses a complete future null infinity}. 

Perhaps the most interesting case is the \emph{vacuum Einstein equations}, i.e., setting $\mathbf{T}_{\alpha\beta}\equiv0$ and hence
\begin{align*}
\mathbf{Ric}_{\alpha\beta}=0,
\end{align*} 
This is because the singularities arise in this case only because of the effect of gravity itself. In order to simplify the equations, symmetries are usually imposed, for example, spherical symmetry. However, it is not appropriate to impose spherical symmetry in vacuum because the spacetime will reduce to a single family of the static Schwarzschild solutions due to the Birkhoff theorem. To gain insights to attack the original problem, one would rather investigate the spherically symmetric solution of the Einstein equations coupled with a simple material model that can be used to simulate the effect of gravity. A suitable choice is the massless scalar field $\phi$,  whose energy-momentum tensor reads
$$\mathbf{T}_{\alpha\beta}=\nabla_\alpha\phi\nabla_\beta\phi-\frac{1}{2}g_{\alpha\beta}\nabla_\mu\phi\nabla^\mu\phi.$$
The coupled system, which we call the \emph{Einstein-scalar field equations}, can be written in the following form:
$$\begin{cases}\mathbf{Ric}_{\alpha\beta}=2\nabla_\alpha\phi\nabla_\beta\phi\\
g^{\alpha\beta}\nabla_{\alpha}\nabla_{\beta}\phi=0
\end{cases}.$$

The study of this model in spherical symmetry is successful. Christodoulou studied in a series of papers, including \cite{Chr87, Chr91, Chr93, Chr94, Chr99}, the spherical symmetric solution of this model, and rigorously verified in the last paper \cite{Chr99} that the weak cosmic censorship conjecture holds in this category. Strictly speaking, he proved the following \emph{genericity} theorem.
\begin{theorem}[Christodoulou \cite{Chr87, Chr91, Chr93, Chr99}]\label{Christodoulougenericity}
Let $\mathcal{A}$ be the space of functions of bounded variation on the non-negative real line, which is the space of the initial data $\alpha_0=\frac{\partial}{\partial r}(r\phi)\Big|_{C_o}$ on the initial null cone $C_o$ issuing from a point $o$, a fixed point of the $SO(3)$ action, where $r$ is the area radius of the orbit spherical sections of $C_o$. Let $\mathcal{S}$ be the collection of $\alpha_0\in\mathcal{A}$ such that the maximal future development is not complete. Let $\mathcal{E}$ be the exceptional set which is the complement in $\mathcal{S}$ of the collection of initial data such that the maximal development has a complete future null infinity and terminates at a spacelike singular future boundary. Then if $\alpha_0\in\mathcal{E}$, there exists some $f\in\mathcal{A}$ depending on $\alpha_0$ such that the family $\alpha_0+tf$ lies in $\mathcal{S}\backslash\mathcal{E}$ for $t\ne0$. In addition, any two such families do not intersect, i.e., if $\alpha_1+t_1f_1=\alpha_2+t_2f_2$, then $\alpha_1\equiv\alpha_2$, $f_1\equiv f_2$, $t_1=t_2$.
\end{theorem}
The conclusion of the above theorem consists of two parts: the \emph{existence} of the families $(\alpha_0+tf)_{t\in\mathbb{R}}$ of perburbations, which implies that $\mathcal{E}$ is dense in $\mathcal{A}$, and the \emph{uniqueness}, which means that any functions in $\mathcal{A}$ belongs to no more than one of such families. The set $\mathcal{E}$ is then of co-dimension at least $1$\footnote{In fact, Christodoulou proved that indeed the co-dimension is at least $2$. Christodoulou's proof also suggested that the set $\mathcal{E}$ still has at least $1$ co-dimension in the space of all more regular absolutely continuous initial data.} and therefore the set $\mathcal{E}$ can be considered to be exceptional. Indeed, this theorem does not only establish the weak cosmic censorship but also the strong cosmic censorship, which says that \emph{the maximal future development of generic asymptotically flat initial data cannot extend as a Lorentzian manifold in suitable sense}. We remark that the exceptional set $\mathcal{E}$ is not empty, i.e., singularities not hidden inside a black hole, which we call \emph{naked singularities}, do occur in the spherically symmetric solutions of the Einstein-scalar field equations. Examples are provided by Christodoulou in \cite{Chr94} and this shows that the word ``\emph{generic}'' is needed in both the statements of the weak and strong cosmic censorship conjectures.

An essential step of proving the above theorem is to understand in which way a closed trapped surface, a 2-spacelike surface embedding in the spacetime such that the mean curvature relative to both future null normals are negetive forms surrounding a given singularity.  
This was accomplished by Christodoulou in \cite{Chr91}:
\begin{theorem}[Christodoulou, \cite{Chr91}]\label{Chr91}
Consider the spherically symmetric solution of the Einstein-scalar field equations with initial data given on a null cone $C_o$. Consider two spherical sections $S_1$ and $S_2$ with area radii $r_1$, $r_2$ and mass contents $m_1$, $m_2$, and $S_2$ is in the exterior to $S_1$. Denote
$$\delta=\frac{r_2}{r_1}-1.$$
Then there exists positive constants $c_0$, $c_1$ such that if $\delta\le c_0$ and
\begin{align}\label{Christodouloum2-m1}2(m_2-m_1)>c_1r_2\delta\log\left(\frac{1}{\delta}\right),\end{align}
then there is a closed trapped surface on the incoming null cone through $S_2$ in the maximal future development.
\end{theorem}
\begin{remark}\label{Chr91more}Theorem \ref{Chr91} is in fact part of the main theorem called collapse theorem in \cite{Chr91}, in which the behaviors of the singularity and the event horizon were also studied. First, it was proved that region of trapped surfaces terminates at a strictly spacelike singular boundary. Second, it was also proved that the event horizon will form and the completeness of the future null infinity follows from the analysis in \cite{Chr87}, that every causal curve $r=c$ has infinite length towards the future for $c$ larger than the final Bondi mass of the spacetime.\footnote{It was also discussed in \cite{Da05} by Dafermos how a single closed trapped surface implies the completeness of the future null infinity in spherically symmetric spacetimes. Beyond spherical symmetry, nothing is known about the relation between the existence of a closed trapped surface and the completeness of the future null infinity. We would like to mention however the works \cite{L-Z1, L-Z2} on the local existence theorems in retarded time, while the completeness of the future null infinity is equivalent to the global existence in retarded time.
}
\end{remark}

The proof of Theorem \ref{Christodoulougenericity} can then be sketched as follows. It was shown in \cite{Chr93} that the first singularity can only appear on the central line. Suppose that $e$ is this singular endpoint of the central line. We define $\Cb_e$ to be the boundary of the causal past of $e$, whose intersection with $C_o$ has area radius $r_e$. Then what was actually proved in the last paper \cite{Chr99} is the following \emph{instability} theorem:
\begin{theorem}[Christodoulou, \cite{Chr99}]\label{Christodoulouinstability} Suppose that the initial data $\alpha_0$ satisfies some \emph{generic} condition. Then there exist a sequence of $r_n\to r_e^+$ and a sequence of points $p_n\to e$ on the central line, such that the outgoing null cone $C_{p_n}$ issuing from $p_n$ satisfies the assumptions of Theorem \ref{Chr91} at two spherical sections: $S_{1,n}$, the intersection of $C_{p_n}$ and $\Cb_e$, and $S_{2,n}$, the intersection of $C_{p_n}$ and $\Cb_{r_n}$, where $\Cb_{r_n}$ is the incoming null cone through the spherical section $r=r_n$ on the initial null cone $C_o$.
\end{theorem}
\begin{remark}
When we say that a condition is \emph{generic}, it means that in a certain sense, almost all initial data satisfy this condition. It was also verified in \cite{Chr99} that the set of initial data not satisfying the \emph{generic} condition in Theorem \ref{Christodoulouinstability} is of co-dimension at least $1$ in the space of all initial data.
\end{remark}
As a consequence by applying Theorem \ref{Chr91}, there exists a sequence of closed trapped surfaces. They are the orbit spheres, approaching the singularity $e$ and their areas tend to zero. Then the apparent horizon issues from $e$ and Theorem \ref{Christodoulougenericity} follows immediately from the conclusions mentioned in Remark \ref{Chr91more}.

Inspired by the above argument, Christodoulou formulated a conjecture \cite{Chr99cqg} for the vacuum Einstein equations, which can be termed the \emph{trapped surface conjecture}, and be viewed as a local version of the weak cosmic censorship conjecture and part of the strong one. It can of course be viewed as a conjecture for the Einstein equations coupled with suitable matter field without any modifications.
\begin{conjecture}[Christodouou, \cite{Chr99cqg}]\label{conjecture}
The maximal future development $(\mathcal{M},g)$ of \emph{generic} asymptotically flat initial data $(\mathcal{H}, \bar{g}, k)$ has the following property. If $P$ is a TIP\footnote{TIP is short for \emph{the terminal indecomposable past set}, which is originally introduced in \cite{G-K-P}. Roughly speaking, a TIP is the past of a piece of the singular boundaries of the maximal future development.}  in $M$ whose trace on $\mathcal{H}$ has compact closure $\mathcal{K}$, then for any open domain $\mathcal{D}$ in $\mathcal{H}$ containing $\mathcal{K}$, the domain of dependence of $\mathcal{D}$ in $\mathcal{M}$ contains a closed trapped surface.
\end{conjecture}
The goal of this paper is to initiate the study of the above conjecture without symmetries, or precisely speaking, study how the arguments of proving Theorem \ref{Christodoulougenericity} can be generalized when no symmetries are imposed.
\subsection{Main results}
 In this paper, we study the next simpliest case that the TIP is the past spherical singularities of a massless scalar field studied in \cite{Chr99}. The meaning of a singularity being spherical in this paper is that the causal past of this singularity is a spherically symmetric spacetime and the boundary of this causal past is foliated by orbit spheres. In contrast to the spherical symmetry imposed on the causal past of the singularity, no symmetries are required in the future of the boundary of this causal past. We expect we can gain some insights to attack the problem when the singularity is not assumed to be spherical.

Recalling that in spherical symmetry, a sharp criterion of a point $e$ on the central line being singular is obtained in \cite{Chr93}. It says that if the ratio of the Hawking mass to the radius of the spheres does not tend to zero as we approach $e$ from the past, then the regular solution cannot extend across $e$. Then the question we study is formulated as follows.
\begin{question}
In this paper, we study the following characteristic initial value problem with a spherical singularity of a scalar field. The initial data is given on two intersecting null cone $C_{u_0}$ which is outgoing and $\Cb_0$ which is incoming. The data on $\Cb_0$ is assumed to be spherically symmetric. Denote $r$ to be the area radius of the spherical section on $\Cb_0$ and $m=m(r)$ be the Hawking mass. For $r>0$, it is assumed that $r>2m\ge0$ so no closed trapped surfaces exist on $\Cb_0$ and $\frac{2m}{r}\nrightarrow 0$ as $r\to0^+$. The outgoing null cone $C_{u_0}$ intersects $\Cb_0$ at $r=-u_0$. No symmetries are imposed on the data on $C_{u_0}$. We then solve the Einstein-scalar field equations of such data. This setup is depicted in Figure \ref{fig:setup}.
\end{question}
\begin{figure}
\centering
\includegraphics [width=4.5 in]{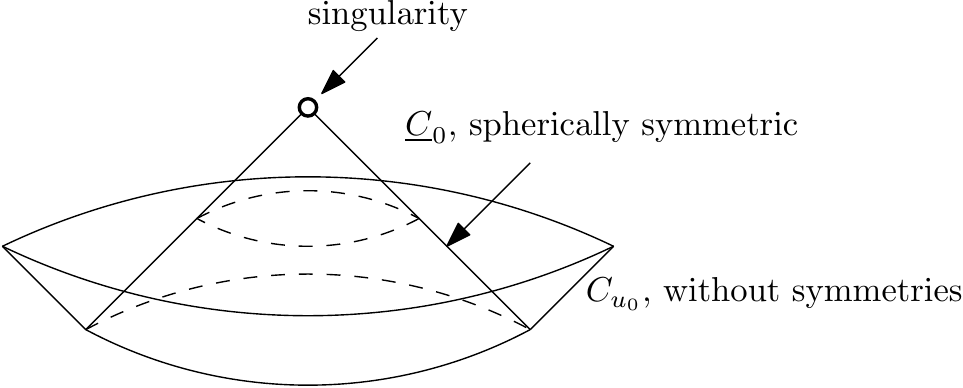}  
\caption{Basic setup}
 \label{fig:setup} 
\end{figure}

In this setup, we are able to prove some \emph{instability} theorems, generalizing Theorem \ref{Christodoulouinstability}, which will be discussed in Section \ref{instabilitysection}. Based on the \emph{instability} theorems, we are able to prove the \emph{genericity} theorem. Recall that the initial data set on $C_{u_0}$ consists of the conformal metric $\widehat{\gs}$ of $C_{u_0}$, the lapse $\Omega$ and the scalar field function $\phi$. For simplicity, we will fix \emph{smooth} the initial data on $\Cb_0$ and the \emph{smooth} lapse $\Omega$ and scalar field $\phi$ on $C_{u_0}$. One can certainly state and prove a complete version allowing perturbations both on the $\widehat{\gs}$ and $\phi$ without essential additional difficulties. We also let $(\ub,u,\vartheta)$ be the double null coordinate system introduced in Section \ref{doublenull}, and then the initial quantities are functions of three variables $(\ub,\vartheta)$. The initial data on $C_{u_0}$ is not assumed to be smooth. We present here a rough form of the main result of the present article, whose precise version is given in Theorem \ref{main1precise}.

\begin{theorem}[Main result]\label{main1}
Let the initial data on $\Cb_0$ with a singular vertex and initial $\Omega$ and $L\phi$ on $C_{u_0}$ be arbitrarily \emph{smoothly} fixed. Let $\mathcal{I}$ be all initial conformal metrics $\widehat{\gs}$ defined on $C_{u_0}$ for $0\le\ub\le 1$ such that $\widehat{\gs} \in C^1_{\ub}H^{10}_\vartheta$ up to the intersection $C_{u_0}\cap\Cb_0$\footnote{This means that $\widehat{\gs}$ at $\ub=0$ must be coincide in $C_{\ub}^1$ level with the data induced by the spherical symmetric initial data on $\Cb_0$. This in particular implies that $\widehat{\gs}$ is standard round metric and the shear $\chih=0$ at $\ub=0$.}. Let $\mathcal{E}$ be the set of $\widehat{\gs} \in \mathcal{I}$ such that the future maximal development does not have any sequences of closed trapped surfaces approaching the singularity. Then $\mathcal{E}$ is of first category in $\mathcal{I}$, i.e., $\mathcal{E}^c$, the complement of $\mathcal{E}$ in $\mathcal{I}$, contains a subset that is a countably intersection of open and dense subsets of $\mathcal{I}$.
\end{theorem}

\begin{remark} \label{remark-isotropic}
 Although in a weaker sense than the at least $1$ co-dimensionality proved in spherical symmetry, we can still say the set $\mathcal{E}$ is exceptional, or the initial conformal metrics lying in $\mathcal{I}\backslash\mathcal{E}$ are \emph{generic}. What we have proved is essentially the \emph{existence} of the perturbations. To prove the \emph{uniqueness} part, one will need some \emph{fully anisotropic} \emph{instability} theorems which is not proved in this paper. This is one of the main differences between spherical symmetry and non-spherical symmetry. We will try to study this in the future works. \end{remark}
 \begin{remark}
The topology of the space of the initial conformal metrics we work in is the most regular one in which the estimates are done based on the methods in this paper. For non-smooth initial data, it is not a priori clear what ``future maximal development'' means. We will discuss this in Section \ref{instabilitysection}.
\end{remark}
\begin{remark}
In the above theorem, the initial scalar field can be chosen such that the singularity is naked if no gravitational fields are present, i.e., the shear tensor is set to be zero. Such examples do exist by \cite{Chr94}.  We already know from \cite{Chr99} that the naked singularities are not stable under spherical scalar perturbations in the sense of being of at least $1$ co-dimension. In this paper, we find new perturbations contributed from the gravitational fields and establish the instability in the sense of being of first category. Therefore we may say that the spherical naked singularities of a self-gravitating scalar field are not stable under gravitational perturbations.
\end{remark}

In Theorem \ref{main1}, we investigate the exceptionality of the set $\mathcal{E}$ where the future maximal developments of the initial conformal metrics in $\mathcal{E}^c$ has a sequence of closed trapped surfaces approaching the singularity. But if one \emph{only} concerns about the weak cosmic censorship conjecture, then it seems that only one single closed trapped surface is sufficient. We then let $\Cb_\varepsilon$ be the incoming null cone where $\ub=\varepsilon$ and $\mathcal{E}_\varepsilon$ be the set of the initial conformal metrics in $\mathcal{I}$ such that the maximal future development before $\Cb_\varepsilon$ has no closed trapped surfaces. Note that $\mathcal{E}_\varepsilon\subset\mathcal{E}$ for any $\varepsilon>0$, and assuming the closed trapped surfaces are stable under small perturbations of the initial conformal metrics in the topology of $\mathcal{I}$, we will have the following theorem.
\begin{theorem}\label{main2}
$\mathcal{E}_\varepsilon^c$ contains a subset that is open and dense in $\mathcal{I}$ for all $\varepsilon>0$. 
\end{theorem}
This in particular implies that \emph{generically} in the sense of being open and dense, there is at least one closed trapped surface surrounding the singularity in the maximal future development. Moreover, assuming Theorem \ref{main2}, Theorem \ref{main1} is true by noting that $\mathcal{E}=\bigcup_{i}\mathcal{E}_{\varepsilon_i}$ for any sequence $\varepsilon_i\to0$. We should remark that the above argument on the equivalence of Theorem \ref{main1} and \ref{main2} is not rigorous since we are not working in smooth solutions and the stability of closed trapped surfaces, which is of course true in more regular solutions by Cauchy stability, will not be proved in this paper. But this argument still illustrates the close relation between these two statements, which will be proved simultaneously in their precise forms in Theorem \ref{main1precise} and \ref{main2precise}.

\subsection{The incoming cone $\uline{C}_0$ and the singularity} In this subsection, we derive some basic properties of various quantities on $\Cb_0$ which will be used throughout the whole paper. Readers who are not familiar with the double null coordinate system should refer to the next section first, where all notations are introduced in detail.

Let $(\ub,u,\vartheta)$ be the double null coordinate system, and the optical function $u$ is defined such that $u=-r$ on $\Cb_0$. Then the vector field $\Lb$, tangent to $\Cb_{\ub}$, preserving the double null foliation, will satisfy $\Lb=\frac{\partial}{\partial u}$ and $\Lb u=1$ on $\Cb_0$. Therefore, the null expansion $\Omega\tr\chib$ along $\Lb$ on $\Cb_0$, simply equals to $-\frac{2}{|u|}$. Let $L'$ be the null vector field tangent to $C_{u}$, with $g(\Lb,L')=-2$, and $\tr\chi'$ be the null expansion relative to $L'$. Denote
\begin{align*}
h=h(u)=\frac{|u|}{2}\cdot \tr\chi'\Big|_{\Cb_0},
\end{align*}
then by the definition of Hawking mass $m=m(u)$, we have $1-\frac{2m}{r}=h$. The condition $r>2m\ge0$ implies that $0<h\le1$ for $|u|>0$. Denote also 
\begin{align*}
\Omega_0=\Omega_0(u)=\Omega\Big|_{\Cb_0}
\end{align*}
where $\Omega$ is the lapse function. $\Omega_0$ can also defined intrinsicly on $\Cb_0$ by $\nabla_{\Lb}\Lb=2(\Lb\log\Omega_0)\Lb$. Then from the null structure equation for $\Db(\Omega\tr\chi)$, we can derive 
\begin{align}\label{equ-DbOmega02h}
\frac{\partial}{\partial u}\left(\log(\Omega_0^2h)\right)=\frac{1}{|u|}\left(1-\frac{1}{h}\right).
\end{align}
This implies that $\Omega_0^2h$ is monotonically decreasing, even when the vertex of $\Cb_0$ is assumed to be regular. According to Lemma 2 in \cite{Chr99}, that $\frac{2m}{r}\nrightarrow0$ implies that 
\begin{align*}
\lim_{u\to0^-}\int_{u_0}^u\frac{1}{|u'|}\left(\frac{1}{h}-1\right)\D u'=+\infty.
\end{align*}
Consequently, $\Omega_0^2h\to0$ as $u\to0^-$. It is mentioned in \cite{Chr99cqg} that the quantity $-\log(\Omega_0^2h)$ measures the blue shift of light received by $e$. Similar to \cite{Chr99}, this fact is crucial in proving the instability of naked singularities.

Now we consider the scalar field on $\Cb_0$. Denote
\begin{align*}
\psi=\psi(u)=|u|\Lb\phi\Big|_{\Cb_0},\ \varphi=\varphi(u)=|u|L\phi\Big|_{\Cb_0}
\end{align*}
where $L$ tangent to $C_u$ preserves the double null foliation, and satisfies $g(\Lb,L)=-2\Omega^2$, or $L=\Omega^2L'$. Consider the Raychaudhuri equation on $\Cb_0$:
$$\Db(\Omega\tr\chib)=-\frac{1}{2}(\Omega\tr\chib)^2-|\Omega\chibh|^2-2(\Lb\phi)^2+2\omegab\Omega\tr\chib.$$
Plugging in $\Omega\tr\chib=-\frac{2}{|u|}$,
we have
$$\frac{\partial}{\partial u}\log\Omega=\omegab=-\frac{1}{2}\frac{\psi^2}{|u|}.$$
Integrating the above equation, we have
\begin{align}\label{Lbphi}-\log\frac{\Omega_0^2(u)}{\Omega_0^2(u_0)}=\int_{u_0}^u\frac{\psi^2}{|u'|}\D u'.\end{align}
This implies that $\Omega_0$ is monotonically decreasing no matter whether the vertex of $\Cb_0$ is singular. This is an important formula which is used to estimate $\Lb\phi$ in the a priori estimates. We also claim that $\Omega_0^2h\to0$ implies $\Omega_0\to0$ as $u\to 0^-$. If not, then $\Omega_0$ has a lower bound because it is decreasing, and then $\Omega_0^2h$ tending to zero implies that $h$ should tend to zero. However, from \eqref{equ-DbOmega02h},  we have $|u|\partial_u(\Omega_0^2h)=\Omega_0^2(h-1)$, which would imply $\Omega_0^2h\to-\infty$, a contradiction. The fact that $\Omega_0$ tends to zero monotonically is used throughout this paper. In summarize, 
\begin{lemma}\label{Omega_0to0}
On $\Cb_0$, $\Omega_0^2h$ is monotonically decreasing approaching the vertex, and if in addition the vertex of $\Cb_0$ is singular, then $\Omega_0^2h\to0$. The same conclusions hold for $\Omega_0$.
\end{lemma}
\begin{remark}
The choice of $u=-r$ on $\Cb_0$ determines $\Omega_0$ up to a constant multiple. Therefore, throughout the paper, we will assume that $\Omega_0(u_0)\le1$ for simplicity. It follows from the monotonicity of $\Omega_0$ that $\Omega_0(u)\le1$ for all $u\in[u_0,0)$.
\end{remark}

At last, we investigate the function $\varphi$. In terms of $\psi$ and $\varphi$, the wave equation restricted on $\Cb_0$ can be written in the following form:
$$\frac{\partial}{\partial u}\varphi=-\frac{\Omega_0^2h}{|u|}\psi,$$
then we have
\begin{align*}
\varphi(u)=\varphi(u_0)-\int_{u_0}^{u}\frac{\Omega_0^2h\psi}{|u'|}\D u'\end{align*}
The integral on the right hand side can be estimated by
\begin{equation}\label{estimate-varphi}
\begin{split}
\int_{u_0}^{u}\frac{\Omega_0^2h\psi}{|u'|}\D u'&\leq\left(\int_{u_0}^{u}\frac{\Omega_0^2}{|u'|}\D u'\right)^{\frac{1}{2}}\left(\int_{u_0}^{u}\frac{\Omega_0^2|\psi|^2}{|u'|}\D u'\right)^{\frac{1}{2}}\\
&\leq\Omega_0^2(u_0)\left|\log{\frac{|u|}{|u_0|}}\right|^{\frac{1}{2}}\\
\end{split}
\end{equation}
where we use \eqref{Lbphi} to estimate $\int_{u_0}^{u}\frac{\Omega_0^2|\psi|^2}{|u'|}\D u'\le\Omega_0^2(u_0)$. This estimate is related to the fact that the lower bound assumption of $m_2-m_1$ in Theorem \ref{Chr91} is sharp in general situations.

\subsection{Formation of trapped surfaces}

It is not difficult to imagine that in order to prove the \emph{genericity} theorems, a theorem similar to Theorem \ref{Chr91} without symmetries should be proved. This is Theorem \ref{formationoftrappedsurfaces1} in this paper. It is a corollary of  Theorem \ref{formationoftrappedsurfaces}, a more general form of the theorem on the formation of trapped surface. To illustrate the main ideas and difficulties 
we only present Theorem \ref{formationoftrappedsurfaces1} here, which can also be stated as follows.
 \begin{theorem}\label{formationrough}
We assume the smooth data on $\Cb_0$ is spherically symmetric but not necessarily singular at its vertex. Suppose also that the smooth initial data on $C_{u_0}$ between two sections $S_{1}=S_{0,u_0}$ and $S_2=S_{\delta,u_0}$ where $\delta$ is a small parameter, satisfies an initial estimate\footnote{The scale invariant $\H^n$ norm is defined in Section \ref{statementexistencetheorem}.}:
 \begin{equation*}
\max\left\{\sup_{u_0\le u\le u_1}|\varphi(u)|,|u_0|\sup_{0\le\ub\le\delta}\left(\|\Omega\chih\|_{\H^7(\ub,u_0)}+\|\omega, L\phi\|_{\H^5(\ub,u_0)}\right)\right\}\le\Omega_0^2(u_0)a\left|\log\frac{|u_1|}{|u_0|}\right|
 \end{equation*}
for some $a\ge1$, where $u_1\in(u_0,0)$ is defined by
\begin{equation}\label{def-u1introduction}
\Omega^2_0(u_1)|u_1|=C^2\Omega_0^4(u_0)\delta a\left|\log\frac{|u_1|}{|u_0|}\right|
\end{equation}
and such that $\Omega_0^2(u_0)\left|\log\frac{|u_1|}{|u_0|}\right|\ge1$. Then there exists some universal large constant $C_1$ such that the smooth solution of the Einstein-scalar field equations exists for $0\le \ub\le\delta, u_0\le u\le u_1$. If in addition
\begin{equation}\label{ourm2-m1}
\inf_{\vartheta\in S^2}\int_0^{\delta}|u_0|^2(|\Omega\chih|^2+2|L\phi|^2)(\ub',u_0,\vartheta)\D\ub'\ge  17C^2\Omega_0^4(u_0)\delta a\left|\log\frac{|u_1|}{|u_0|}\right|,
\end{equation}
then $S_{\delta,u_1}$ is a closed trapped surface.
\end{theorem}
The main difference of this theorem from Theorem \ref{Chr91} is that the initial data should satisfy some a priori bounds and this bound should appear in the lower bound assumption of the initial energy, and a claim on the existence of the solution is needed. This is because in spherically symmetric case, Theorem \ref{Chr91} is proved by contradiction, using monotonicity properties of the Einstein equations in spherical symmetry. These properties break down even when we consider a small perturbation of spherical symmetry. The main difficulty is now the existence of the solution deep into the place where a closed trapped surface has chance to form eventually. In order to overcome this difficulty, we should find a correct form of the a priori estimates which can only be derived by $L^2$ based energy. 
\begin{remark}
As mentioned above, without symmetries, all theorems should be proved under some a priori bounds assumed on the initial data. However, no a priori bounds are assumed on the initial data on $\Cb_0$. We will explain this in Section \ref{existencesection}.
\end{remark}

It is worth mentioning that Theorem \ref{formationrough} should be very carefully written down because only the sharp form of such a theorem can be used to prove the Theorem \ref{main1}. In particular, the dependence of the assumptions of the theorem on $u_0$ should be made explicit, because this theorem is applied for $u_0$ replaced by a sequence of $u_{0,n}\to0^-$. So we compare it with Theorem \ref{Chr91}, which is already stated in its sharp form. First, when $\delta$ is sufficiently small,
\begin{align*}
\inf_{\vartheta\in S^2}\int_0^{\delta}|u_0|^2(|\Omega\chih|^2+2|L\phi|^2)(\ub',u_0,\vartheta)\D\ub'\approx 4\Omega_0^2(u_0)(m_2-m_1)
\end{align*}
where $m_i$ is the Hawking mass of $S_i$.  Second, we have the equation
\begin{align*}
Dr=\frac{r}{2}\overline{\Omega\tr\chi}
\end{align*}
where $\overline{\Omega\tr\chi}$ is the average of $\Omega\tr\chi$ in $S_{\ub,u}$. Recalling $\frac{r}{2}\Omega\tr\chi$ takes value $\Omega_0^2h$ on $\Cb_0$, then
\begin{align*}
r_2-r_1\approx\Omega_0^2(u_0)h(u_0)\delta,
\end{align*}
where $r_i$ is the area radius of $S_i$. Therefore the lower bound conditions \eqref{Christodouloum2-m1} and \eqref{ourm2-m1} only differ by a factor $h(u_0)$. This difference is acceptable. 
At last, $r_2-r_1$ (or $\Omega_0^2(u_0)\delta$ by the above analysis) in the logarithm factor in \eqref{Christodouloum2-m1} is replaced in  \eqref{ourm2-m1} by $|u_1|$, which is the place where a closed trapped surface will form as predicted in Theorem \ref{formationrough}. 

A simple discussion may help to understand why a closed trapped surface will form at $u=u_1$ where $u_1$ is defined by \eqref{def-u1introduction}. Assume the regular solution exists to $u=u_1$ and some correct a priori estimates are obtained. From the Raychaudhuri equation along outgoing null direction, $D\tr\chi'=-\frac{1}{2}(\tr\chi)^2-|\chih|^2-2|\widehat{L}\phi|^2,$
where $\widehat{L}=\Omega^{-1}L$, we have, on $C_u$,
$$\tr\chi'-\frac{2h}{|u|}\approx-\Omega_0^{-2}(u)\int_0^\delta(|\Omega\chih|^2+2|L\phi|^2)\D\ub.$$
In order that to ensure that $\tr\chi'<0$ at $\ub=\delta$, $u=u_1$, we require
\begin{align*}
&\Omega_0^{-2}(u_1)\int_0^\delta|u_0|^2(|\Omega\chih|^2+2|L\phi|^2)(\ub,u_0,\vartheta)\D\ub\\
\approx&\Omega_0^{-2}(u_1)\int_0^\delta|u|^2(|\Omega\chih|^2+2|L\phi|^2)(\ub,u_1,\theta)\D\ub>2|u_1|.
\end{align*}
So the definition \eqref{def-u1introduction} is obtained from the lower bound \eqref{ourm2-m1}. 
\begin{remark}\label{aboutomegab-intro}We make an important remark that it is only the relation written in the vector field $L$
\begin{align*}
\Omega\chih\big|_{C_{u_1}}\approx\Omega\chih\big|_{C_{u_0}},\  L\phi\big|_{C_{u_1}}\approx L\phi\big|_{C_{u_0}}
\end{align*}
can be proven but not the relation written in $\widehat{L}=\Omega^{-1}L$. This is because only written in $L$ the equations for $\Db(\Omega\chih)$ and $\Db L\phi$ do not involve $\omegab$, whose initial value on $\Cb_0$ has no a priori bounds.
\end{remark}

The investigation of the problem of the formation of trapped surfaces when no symmetries are imposed was started in the celebrated work of Christodoulou \cite{Chr}, which also opened the path to the study of large data problem without symmetries. His theorem can be stated in a form similar to Theorem \ref{formationrough} as follows.
\begin{theorem}[Christodoulou, \cite{Chr}]\label{Chr08}
Suppose that the initial incoming null cone $\Cb_0$ is a null cone in Minkowski space, and the smooth initial data on $C_{u_0}$ satisfies
\begin{align*}
|u_0|\sup_{0\le\ub\le\delta}\sum_{i=0}^3\delta^i\|D^i\chih\|_{\H^5(\ub,u_0)}\le \delta^{-\frac{1}{2}}F
\end{align*}
for some constants $\delta$, $F$, and $\Omega\equiv1$ on $C_{u_0}\bigcup\Cb_0$. Then there exists a function $M$ of two variable such that if $\delta\le M(F,u_1)$, the smooth solution of the vacuum Einstein equations exists for $0\le\ub\le\delta$, $u_0\le u\le u_1$. If also
\begin{align}\label{Christodoulou08-m2-m1}
\inf_{\vartheta\in S^2}\int_0^{\delta}|u_0|^2|\chih|^2(\ub,u_0,\vartheta)\D\ub> 2|u_1|,
\end{align}
then $S_{\delta,u_1}$ is a closed trapped surface after modifying $M$ to a smaller function $\widetilde{M}$.
\end{theorem}
\begin{remark}
Christodoulou only proved this theorem when $u_1=-1$. But it is trivial to extend the proof when $u_1\in[u_0,0)$ is arbitrary. In addition, Christodoulou constructed a class of initial data satisfying the assumptions of this theorem in his short pulse ansatz. He confirmed that the closed trapped surfaces are evolutionary.
\end{remark}
\begin{remark}
The choice of $\delta$ in the above theorem is in particular independent of $u_0$ when $|u_0|\ge1$. This fact allowed Christodoulou to pull the initial data back to the past null infinity, and then the condition \eqref{Christodoulou08-m2-m1} has a clear physical meaning. In contrast, Theorem \ref{formationrough}, which we prove in this paper, allows us to push the initial data deep into the vertex.
\end{remark}

The main difference between Theorem \ref{Chr08} and Theorem \ref{formationrough} is the the difference between two lower bound conditions \eqref{Christodoulou08-m2-m1} and \eqref{ourm2-m1}. When $\delta$ is suffciently small, the lower bound in \eqref{Christodoulou08-m2-m1} is much large than \eqref{ourm2-m1}, and we have the following consequences: compared to the closed trapped surface forming in Theorem \ref{Chr08},  the closed trapped surface forming in Theorem \ref{formationrough}  is much smaller, and is located much closer to the vertex. Therefore, more refined a priori estimates are need to capture the growth properties of the solution approaching the vertex in Theorem \ref{formationrough}. The first extension of  Theorem \ref{Chr08} by relaxing the lower bound condition \eqref{Christodoulou08-m2-m1} to a condition similar to  \eqref{ourm2-m1} was given by An-Luk \cite{An-Luk}. They proved
\begin{theorem}[An-Luk, \cite{An-Luk}]\label{An-Luk}
Suppose that the initial incoming null cone $\Cb_0$ is a null cone in Minkowski space, and the smooth initial data on $C_{u_0}$ with $u_0=-1$ satisfies
\begin{align*}
\sup_{0\le\ub\le\delta}\|\chih\|_{H^7(\ub,u_0)}\le A^{\frac{1}{2}}
\end{align*}
and $\Omega\equiv1$ on $C_{u_0}\bigcup\Cb_0$. Then there exists a universal large constant $b_0$ such that if $b_0\le b\le A$ and $\delta A^{\frac{1}{2}}b<1$, then the smooth solution to the vacuum Einstein equations exists for $0\le\ub\le\delta, -1\le u\le-\delta A^{\frac{1}{2}}b$. Moreover, if the initial data also verify the lower bound 
\begin{align}\label{anlukm2-m1}
\inf_{\vartheta}\int_0^\delta|\chih|^2(\ub,-1,\vartheta)\D\ub\ge 4b\delta A^{\frac{1}{2}},
\end{align}
then $S_{\delta, -b\delta A^{\frac{1}{2}}}$ is a closed trapped surface.
 \end{theorem}
Setting $\phi\equiv0$, Theorem \ref{formationrough} recovers a special case of the above theorem if we set $A^{\frac{1}{2}}=a\left|\log\frac{|u_1|}{|u_0|}\right|$, $b_0=C_1^2$ and $b=C^2$. A similar condition to $b\le A$ is not required in Theorem \ref{formationrough} but holds when $u_1$ is chosen sufficiently close to zero. However, the lower bound condition \eqref{anlukm2-m1} is even better than \eqref{ourm2-m1}, since we can choose $A$ to be a constant in Theorem \ref{An-Luk}. We should remark that the more general Theorem \ref{formationoftrappedsurfaces} can completely recover Theorem \ref{An-Luk} in vacuum. Nevertheless, when the scalar field is presented, it is not obvious that we can choose $A$, or the corresponding quantity in Theorem \ref{formationoftrappedsurfaces} to be a constant in general situations. There are at least two reasons for this. The first reason is due to a loss on the a priori estimates, which we call $\mathscr{E}$, a quantity involving derivatives of $L\phi$ and growing like $\left|\log\frac{|u_1|}{|u_0|}\right|^{\frac{1}{2}}$. We will explain this quantity in Section \ref{existencesection}. This quantity prevents us from solving the solution even deeper into the future than $u=u_1$ defined in \eqref{def-u1introduction}. Another reason is that $\varphi=|u|L\phi|_{\Cb_0}$ may be unbounded, which prevents  the energy integral
$$\int_0^\delta|u|^2|L\phi|^2(\ub,u,\vartheta)\D\ub$$ 
from being lower bounded before a closed trapped surface forms eventually. On the other hand, both the quantity $\mathscr{E}$ and the estimate \eqref{estimate-varphi} suggest that the lower bound \eqref{ourm2-m1} is sharp up to the factor $h(u_0)$ in general.

 
There are also many other extensions to Theorem \ref{Chr08}. The first extension and simplification was given by Klainerman-Rodnianski \cite{K-R}, by proving the geometric quantities obey some scale-invariant estimates relative to the parabolic scaling on null hypersurfaces. Although they only considered a finite problem, An \cite{An} extended it to an initial value problem on the past null infinity. Luk-Rodnianski proved in \cite{L-R2} among many other things that the formation of trapped surfaces theorem also holds when $\Cb_0$ is not necessarily Minkowskian, while the other existed works should assume $\Cb_0$ to be Minkowskian. Klainerman-Luk-Rodnianski proved in \cite{K-L-R} the lower bound condition \eqref{Christodoulou08-m2-m1} can be relaxed, by replacing $\inf$ by $\sup$. Similar extensions are expected to hold for Theorems \ref{formationrough} and \ref{An-Luk}. We also point out the first paper on the formation of trapped surfaces for non-vacuum equations by Yu \cite{Yu}, which considered the Einstein equations coupled with the electric-magnetic field, and a paper by the first author and Yu \cite{Li-Yu}, in which a class of asymptotically flat Cauchy data whose future development has a closed trapped surface was constructed, while the other existed works were about characteristic initial value problem.

\subsection{Instability theorems}\label{instabilitysection}

We will discuss the \emph{instability theorems} in this subsection. The \emph{instability} theorem proved in \cite{Chr99}, Theorem \ref{Christodoulouinstability} should also be generalized beyond spherical symmetry. The original proof of Theorem \ref{Christodoulouinstability} is by contradiction. By introducing the dimensionless coordinates $(t,s)$ instead of $(u,r)$, the Bondi coordinate, by
\begin{align*}
u=u_0\mathrm{e}^{-t},\ -2r=u_0\mathrm{e}^{s-t},
\end{align*}
that the conclusion of Theorem \ref{Christodoulougenericity} is false implies that there exists some $\varepsilon>0$ such that 
\begin{align}\label{contradictionassumption}
2(m(s,t)-m(0,t))\le c_1r(s,t)s\log\left(\frac{1}{s}\right),
\end{align}
for all $(s,t)\in\{s\le c_0\}\bigcap\{0\le\ub\le\varepsilon\}$. Then \eqref{contradictionassumption} is used to conclude that the opposite of \eqref{contradictionassumption} , or equivalently \eqref{Christodouloum2-m1} holds, however for some $(s_{\varepsilon},t_{\varepsilon})\in\{s\le c_0\}\bigcap\{0\le\ub\le\varepsilon\}$, allowing a perturbation on the initial data. This argument also implicitly use the fact that the solution exists in the region $\{s\le c_0\}\bigcap\{0\le\ub\le\varepsilon\}$.

If we go beyond spherical symmetry, the above arguments may break down for at least two reasons. First, a similar assumption to \eqref{contradictionassumption} is far from enough to control the whole system of the Einstein equations, in contrast to the spherically symmetric case, where the Hawking mass $m$ governs the whole system. Second, the existence region of the solution relies on some suitable a priori estimates of the solution, which is also different from the spherically symmetric case. For these two reasons, instead of \eqref{contradictionassumption}, an appropriate a priori estimate is needed, to guarantee that the solution exists deep enough into the future such that the assumptions of Theorem \ref{formationrough} eventually hold. 
\begin{remark}
It is not difficult to imagine that one should first find another robust argument of proving Theorem \ref{Christodoulougenericity}, or more precisely, Theorem \ref{Christodoulouinstability}, and then generalize it to the problem we study in this paper. We are able to find this argument and this is given in detail in a separated paper \cite{Li-Liu}. 
\end{remark}

\begin{figure}
\centering
\includegraphics [width=4 in]{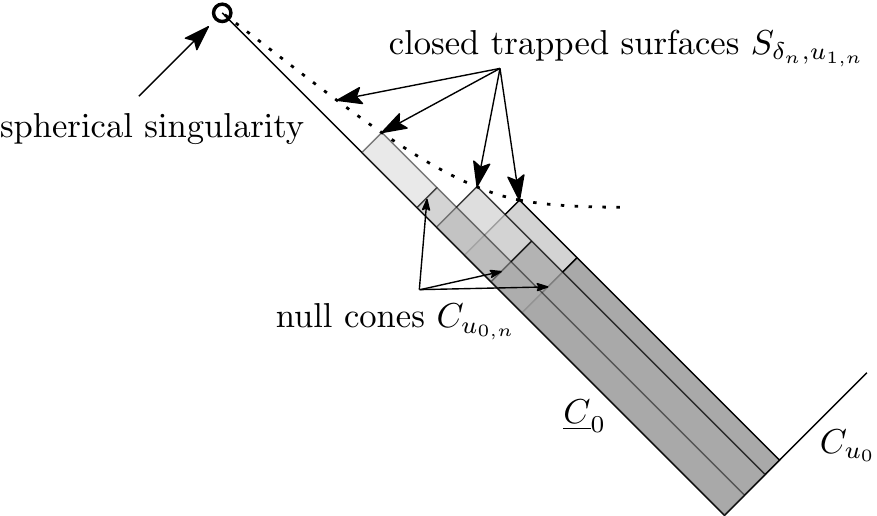}  
\caption{The conclusions of Theorem \ref{instabilityrough} and \ref{instabilityrough2}}

 \label{fig:instability} 
\end{figure}
The instability theorems we prove in this paper is Theorem \ref{instabilitytheorem} and Corollary \ref{instabilitycorollary}. We state a rough version here and depict them in Figure \ref{fig:instability}. Recall that the proof of Theorem \ref{Christodoulouinstability} in \cite{Chr99} is divided into several cases according to the asymptotic behavior of $\varphi$ as $u\to0^-$. Similarly, we divide the proof of the instability theorems in two separated cases in this paper. The first case is the following.
\begin{theorem}[Case 1 of Theorem \ref{instabilitytheorem}]\label{instabilityrough}
If $\varphi(u)$ is unbounded as $u\to0^-$ and the data on $C_{u_0}$ is smooth, then there exists two sequences $\delta_n\to0^+$, $u_{0,n}\to0^-$ such that the smooth solution of the Einstein-scalar field equaitons exists for $0\le\ub\le \delta_n$, $u_0\le u\le u_{0,n}$ and the assumptions of Theorem \ref{formationrough} hold for $\delta=\delta_n$, $u_0=u_{0,n}$ for all $n$. Consequently,  there exists a sequence $u_{1,n}\to0^-$ such that the solution remains smooth for $0\le\ub\le\delta_n$, $u_{0,n}\le u_0\le u_{1,n}$ and $S_{\delta_n,u_{1,n}}$ are closed trapped surfaces for all $n$.
\end{theorem}

In the first case, nothing essential are required on the initial data on $C_{u_0}$. In the second case when $\varphi(u)$ is bounded, it is more complicated. We introduce a monotonically decreasing positive function $\tdelta=\tdelta(\tu)$ defined for $\tu\in[u_0,0)$ where $\tdelta\to0^+$ as $\tu\to0^-$. A crucial fact is that the definition of this function is independent of the initial data on $C_{u_0}$. We will not write down the explicit expression here but it is designed based on the following idea: it is expected that the assumptions of Theorem \ref{formationrough} hold for $\delta=\tdelta$, $u_0=\tu$. Now assume we have such a function $\tdelta=\tdelta(\tu)$. We define another function
\begin{align}\label{def-f-intro}
f(\tu)=\frac{1}{\tdelta(\tu)}\inf_{\vartheta\in S^2}\int_0^{\tdelta(\tu)}(|u_0|^2|\Omega\chih|^2(\ub,u_0,\vartheta)+||u_0|L\phi(\ub,u_0,\vartheta)+(\varphi(\tu)-\varphi(u_0))|^2)\D\ub.
\end{align}
Then we have the following theorem about the second case.
\begin{theorem}[Case 2 of Theorem \ref{instabilitytheorem} and Corollary \ref{instabilitycorollary}]\label{instabilityrough2}
Suppose that $\varphi(u)$ is bounded and the data on $C_{u_0}$ is smooth. We also fix some $\gamma\in(0,2)$ on which the definition of the function $\tdelta=\tdelta(\tu)$ also depends. Then there exists some $\varepsilon>0$ such that if we can find some $|\tu|<\varepsilon$ such that
\begin{align*}
\Omega_0^{\gamma-2}(\tu)f(\tu)\ge2,
\end{align*}
then the smooth solution of the Einstein-scalar field equations exists for $0\le\delta\le\tdelta$, $u_0\le u\le \tu$ and the assumptions of Theorem \ref{formationrough} hold for $\delta=\tdelta$, $u_0=\tu$. Consequently, there exists some $\tu_*$ such that the solution remains smooth for $0\le\ub\le\tdelta$, $\tu\le u\le \tu_*$ and $S_{\tdelta,\tu_*}$ is a closed trapped surface. If in addition \begin{align}\label{instabilitycondition-chih-intro}
\limsup_{\tu\to0^-}\Omega_0^{\gamma-2}(\tu)f(\tu)>2,
\end{align}
then the above conclusions hold for three sequences $\tdelta=\tdelta_n\to0^+$, $\tu=\tu_{0,n}\to0^-$, $\tu_*=\tu_{1,n}\to0^-$.
\end{theorem} 

\begin{remark}
In Remark \ref{remark-isotropic}, we have mentioned that \emph{fully anisotropic instability} theorems are needed to prove the positive co-dimensionality of the exceptional set in $\mathcal{I}$. This is to say we will try to replace the $\inf$ in the definition \eqref{def-f-intro} by a $\sup$ in future works. This can be achieved by proving a correct anisotropic version of Theorem \ref{formationrough} based on the techniques introduced in the work \cite{K-L-R}.
\end{remark}

Note that in general we cannot choose $\mathcal{I}$ to be the space of all smooth initial conformal metrics from the form of condition \eqref{instabilitycondition-chih-intro}. Therefore we should study a more general class of the initial data on $C_{u_0}$ including non-smooth data. For the Cauchy problem of the Einstein equations with non-smooth initial data, a best result is a recent development which is the resolution of the bounded $L^2$ conjecture by Klainerman-Rodnianski-Szeftel (see \cite{K-R-S}), allowing the curvature of the initial metric lying in $L^2$. For the characteristic problem, the classical theory (see \cite{Ren}) only provides us a way to solve the equation with sufficiently smooth initial data. Recently, Luk-Rodnianski (see \cite{L-R1}) developed a theory to solve  by limiting argument the Einstein equations with initial $\alpha$ not even in $L^2$ relative to $\ub$ (but the other curvature components should have more regularity). A key ingredient is, they developed a renormalization technique of the Bianchi equations such that the a priori estimates can be derived for initial shear $\chih$ only lying in $L^\infty$ relative to $\ub$ (but $\chih$ should have several order of the angular derivatives) and the convergence and uniqueness can also be proved. In this paper, we also choose to work in the space of initial data such that $\chih\in L^\infty$ relative to $\ub$. It is natural in view of the work \cite{L-R1, An-Luk}, since the a priori estimates are done in terms of that norm. Moreover, we require in addition $\chih\in C^0$ relative to $\ub$ up to the intersection $C_{u_0}\cap \Cb_0$, which is the completion of $C^\infty$ in the norm $L^\infty$. It is better to prove the \emph{genericity} theorems in more regular space because this would mean the singularity is more unstable. But $L^\infty$ relative to $\ub$, or $C^0$, is the strongest norm on which the a priori estimates are based, and in which we can prove the \emph{genericity} theorems using the techniques in this paper.

In principle, one should develop a theory to solve the equations with non-smooth initial data, or apply the theory developed in \cite{L-R1} with some modifications.  However, in order not to obscure our main idea, we are not going to write down the full detail of developing such a theory. Instead, we state and prove Theorem \ref{instabilitycorollaryweak}, which can be stated roughly as follows.
\begin{theorem}\label{instabilityweak-intro}
Fix some singular initial data on $\Cb_0$ as above. Suppose that the initial data $(\widehat{\gs},\Omega,\phi)$ on $C_{u_0}$ satisfies $\chih,\omega,L\phi\in L^\infty_{\ub}H^N(S_{\ub,u_0})$ and $L\phi-\overline{L\phi}\in H^1_{\ub}H^N(S_{\ub,u_0})$ for sufficiently large integer $N$. If $\varphi(u)$ is bounded, we assume in addition
\begin{align}\label{instabilitycondition-chih-introweak}
\limsup_{\tu\to0^-}\Omega_0^{\gamma-2}(\tu)f(\tu)>32.
\end{align} 
Suppose also that we have a sequence of smooth initial data $(\widehat{\gs}_n,\Omega_n,\phi_n)$ such that $\chih_n$ and $(L\phi)_n$ converge to $\chih$ and $L\phi$ in $L^\infty_\vartheta L^2_{\ub}$. Then there exist two sequences $\delta_k\to0^+$, $u_{1,k}\to0^-$ such that, for every $k$, there exists some $N_k$ such that for all $n>N_k$, $S_{\delta_k,u_{1,k}}$ is a closed strictly trapped surface in the maximal future development of $(\widehat{\gs}_n,\Omega_n,\phi_n)$.
\end{theorem}

Here the strictness means that both null expansions of $S_{\delta_k,u_{1,k}}$ in the maximal future development of $(\widehat{\gs}_n,\Omega_n,\phi_n)$ are less than a negative number, which depends on $k$ but is independent of $n>N_k$. This theorem has the following implication: If the corresponding sequence of smooth solutions has a limiting spacetime such that both outgoing and incoming null expansions converge pointwisely, then the limiting spacetime has a sequence of closed trapped surfaces approaching the singularity. The limiting spacetime is clearly not unique, and even not necessarily satisfies the Einstein equations. However, any theories of solving the Einstein-scalar field equations with non-smooth initial data by limiting argument, satisfying the requirements in Theorem \ref{instabilityweak-intro} are included. Therefore, we may say that the maximal future development of $(\widehat{\gs},\Omega,\phi)$ satisfying the assumptions of Theorem \ref{instabilityweak-intro} has a sequence of closed trapped surfaces approaching the singularity. Then finally, the proof of Theorem \ref{main1} is to show that
\begin{align*}
\{\widehat{\gs}\in\mathcal{I}|\limsup_{\tu\to0^-}\Omega_0^{\gamma-2}(\tu)f(\tu)>32\}\subset\mathcal{E}
\end{align*}
contains a countably intersection of open and dense subsets of $\mathcal{I}$. This can be achieved by showing that
\begin{align*}
\{\widehat{\gs}\in\mathcal{I}|\text{ there exists some $\tu$ with $\tdelta(\tu)<\varepsilon$ such that }\Omega_0^{\gamma-2}(\tu)f(\tu)>33\}
\end{align*}
is open and dense for all $\varepsilon>0$. Indeed, the above set is a subset of $\mathcal{E}_\varepsilon$ for sufficiently $\varepsilon>0$ and this also proves Theorem \ref{main2} in view of Theorem \ref{instabilityrough2}.


\subsection{A priori estimates and the existence theorem}\label{existencesection}

To establish the formation of trapped surface theorem and the instability theorems, the most important part is prove the existence of the solution. For the initial value problem of the Einstein equations without symmetries, this is usually the most lengthy and technical part. The basic strategy is to do a priori estimates with suitable weight functions through $L^2$ based energy method, which can be traced back to the celebrated work \cite{Ch-K} of the stability of Minkowski space by Christodoulou and Klainerman. While the stability of Minkowski space is essentially a small data problem, the study of large data problem was pioneered by Christodoulou in the work \cite{Chr} on the formation of trapped surface.

There are statements on the existence of the solutions in Theorem \ref{formationrough}, \ref{instabilityrough} and \ref{instabilityrough2}. We need to prove two parts of the existence. The first part is to solve the solution from $u=u_0$ to $u=\tu$ where the assumptions of Theorem \ref{formationrough} hold, which is included in the statement of Theorem \ref{instabilityrough} and \ref{instabilityrough2}. The second part is to solve the solution from $u=\tu$ to $u=\tu_*$ where a closed trapped surface will form, which is included in the statement of Theorem \ref{formationrough}. It turns out that these two parts of the existence can be included in a general form of existence theorem, which is Theorem \ref{existencetheorem} in this paper. Let us state it without a statement on the quantitative behavior of the solution as follows.
\begin{theorem}\label{existencerough}
There exists a universal constant $C_0\ge1$ such that the following statement is true. Let $C$, $u_0$ and $u_1$ be three number such that $C\ge C_0$ and $u_0<u_1<0$. The smooth initial data on $\Cb_0$ is spherically symmetric (not necessarily singular) and on $C_{u_0}$, satisfies
 \begin{equation}\label{def-A-intro}
 \begin{split}
\mathcal{A}= \mathcal{A}(\delta,u_0,u_1):=&\max\left\{1,\sup_{u_0\le u\le u_1}\left(\mathscr{F}^{-1}|\varphi(u)|\right),\right.\\
 &\left.\mathscr{F}^{-1}|u_0|\sup_{0\le\ub\le\delta}\left(\|\Omega\chih\|_{\H^7(\ub,u_0)}+\|\omega\|_{\H^5(\ub,u_0)}+\|L\phi\|_{\H^5(\ub,u_0)}\right)\right\}<+\infty.
 \end{split}
 \end{equation}
for some function $\mathscr{F}=\mathscr{F}(\delta,u_0,u_1)\ge1$, and $\Omega_0(u_0)\le1$. Denote
\begin{align}
\label{def-E-intro}\mathscr{E}=\mathscr{E}(\delta,u_0,u_1):=&\max\left\{1, \mathscr{F}^{-1}\mathcal{A}^{-1}\||u_0|(|u_0|\nablas)L\phi\|_{\L^2_{[0,\delta]}\H^4(u_0)}\left|\log\frac{|u_1|}{|u_0|}\right|^{\frac{1}{2}}\right\}\ge1,\\
\label{def-W-intro}\mathscr{W}=\mathscr{W}(u_0,u_1):=&\max\left\{1,\logOmega\right\}\ge1.
\end{align}
 Suppose that the following {\bf{\emph{smallness conditions}}} hold:
\begin{align}
\label{smallness1-intro}\Omega_0^2(u_0)\delta|u_1|^{-1}\mathscr{E}^2\mathscr{W}&\lesssim 1,\\
\label{smallness-intro} C^2\delta|u_1|^{-1}\mathscr{F}\mathscr{W}\mathcal{A}&\le 1,
\end{align}
and the following {\bf \emph{auxiliary condition}} holds:
\begin{equation}\label{auxiliary-intro}
\begin{split}
\Omega_0^2(u_0)\Omega_0^{-2}(u_1)\delta|u_1|^{-1}\mathscr{F}\mathcal{A}&\lesssim 1.
\end{split}
\end{equation}
Then the smooth solution of the Einstein-scalar field equations exists in the region $0\le\ub\le\delta$, $u_0\le u\le u_1$.
\end{theorem}

The proof is by establishing the a priori estimates. One of the main difficulties of the proof is that we need to solve the equations deep near the singular vertex, approaching which the geometric quantities may growth. This has been handled first by \cite{An-Luk}, i.e., Theorem \ref{An-Luk} where $\Cb_0$ is assumed to be Minkowskian and the vertex is of course regular. On the other hand, the scalar field and the singularity bring us new difficulties. In the following, we will briefly describe the a priori estimates we want to prove. We will also explain the quantities and the conditions introduced in the above theorem. We will encounter several difficulties and mention how to overcome them.

\subsubsection{Quantitative behavior of the connection coefficients and the scalar field}
 We begin by analyzing the expected behaviors of the connection coefficients and the scalar field, which are similar to those in \cite{An-Luk}. We should keep in mind that we hope to solve the solution up to $u=u_1$ for $0\le\ub\le\delta$. By scaling consideration, we will expect that, for $\ub\in[0,\delta], u\in[u_0,u_1]$, 
\begin{align}\label{estimate-Theta-intro}
\Theta\lesssim|u|^{-1}\mathscr{F}\mathcal{A}\ \text{where}\ \Theta\in\{\Omega\chih, \Omega\tr\chi, L\phi\}.
\end{align}
$\mathscr{F}$ is chosen according to the problem at hand and $\mathcal{A}$ defined in \eqref{def-A-intro} is the scale-invariant bound for the initial data. We usually choose $\mathscr{F}\ge\displaystyle\sup_{u_0\le u\le u_1}\left(|\varphi(u)|\right)$ and hence $\mathscr{F}\mathcal{A}$ is essentially the scale-invariant bound for the initial data on $C_{u_0}$. From the Raychaudhuri equation
\begin{align*}
D(\tr\chi')=\Omega^{-2}\Theta\cdot\Theta
\end{align*}
and from that $\Omega\approx\Omega_0$ which is reasonable, $\Omega\tr\chi$ has another, in fact better estimate:
\begin{align*}
\Omega_0^2(\tr\chi'-\frac{2}{|u|})\lesssim\delta|u|^{-2}\mathscr{F}^2\mathcal{A}^2.
\end{align*}

Now we turn to the following quantites
\begin{align*}
 \Thetas\in\{\eta,\etab,\nablas\phi\},\ \underline{\Theta}\in\{\Omega\chibh,\Omega\tr\chib+\frac{2}{|u|}\},\ \Lb\phi-\frac{\psi}{|u|}.
\end{align*}
They satisfy the following equations in a schematic form:
\begin{equation}\label{equs1-intro}
\begin{split}
D(\Omega\chibh\ \text{or}\ \Omega\tr\chib+\frac{2}{|u|})&=\Omega^2(\nablas\Thetas+\Thetas\cdot\Thetas)+\text{\fbox{$\Omega\tr\chib\Theta$}}+\Theta\cdot\underline{\Theta}+\text{curvature},\\
D(\Lb\phi-\frac{\psi}{|u|})&=\Omega^2(\nablas\Thetas+\Thetas\cdot\Thetas)+\Big[\Omega\tr\chi\Lb\phi\Big]+\text{\fbox{$\Omega\tr\chib\Theta$}},\\
D(\eta\ \text{or}\ \nablas\phi)&=\text{\fbox{$\nablas\Theta$}}+\Theta\cdot\Thetas,\\
\Db\etab&=\underline{\Theta}\cdot\Thetas+\text{\fbox{$\Omega\tr\chib\eta$}}+\Big[\Lb\phi\nablas\phi\Big]+\text{curvature}.
\end{split}
\end{equation}
The boxed and square bracketed terms, which are the so-called \emph{borderline terms}, determine the behavior of the corresponding quantities respectively. First of all, it is reasonable to expect that an angular derivative $\nablas$ costs an $|u|^{-1}$, and $\Lb\phi$, $\Omega\tr\chib$, which are non-vanishing on $\Cb_0$, remain close to their initial values:
\begin{align*}
\Lb\phi\approx\frac{\psi}{|u|},\ \Omega\tr\chib\approx-\frac{2}{|u|}.
\end{align*}
Then after integration, the first three boxed terms in \eqref{equs1-intro} contribute an estimate like $\delta|u|^{-2}\mathscr{F}\mathcal{A}$ and this suggests that
\begin{align}\label{estimate-Thetas-intro}
\Thetas,\underline{\Theta}\ \text{except}\ \etab\lesssim\delta|u|^{-2}\mathscr{F}\mathcal{A}.
\end{align}
The last boxed term is then also contributes an estimate like $\delta|u|^{-2}\mathscr{F}\mathcal{A}$ after integration. Now we should pay more attentions to the square bracketed terms, containing a factor $\Lb\phi$, which behaves only like its initial value, worse than that suggested by signature consideration\footnote{The signature is a number assigned to every geometric quantity according to how many $e_3=\Omega^{-1}\Lb$ and how many $e_4=\Omega^{-1}L$ are used in its definition. See for example \cite{Ch-K} or \cite{K-R} for a detailed description. Here $\Lb\phi$ and $\Omega\chibh$ have the same signature but they have totally different behaviors.}. Moreover, its initial value $\frac{\psi}{|u|}$ is not assumed to have any a priori bounds. The resolution of this difficulty is to use the equality \eqref{Lbphi} to estimate $\Lb\phi$ in terms of the monotonically decreasing function $\Omega_0$, the restriction of $\Omega$ on $\Cb_0$. The price is that $\Omega_0$ then appears as a weight function in our estimates, and we will find out in the proof that this is not a big problem. Precisely, we will have
\begin{align*}
\int_{u_0}^u|u'||\Lb\phi|^2\D u'\lesssim\mathscr{W},\ \text{where $\mathscr{W}=\mathscr{W}(u_0,u_1)$ is defined in \eqref{def-W-intro}}.
\end{align*}
Then the contribution from the square bracketed term of the last equation is
\begin{align*}
|u||\etab|\lesssim\cdots+\int_{u_0}^u|u'||\Lb\phi||\nablas\phi|\D u'\lesssim\delta\mathscr{F}\mathcal{A}\int_{u_0}^u|u'|^{-1}|\Lb\phi|\D u'\lesssim\delta|u|^{-1}\mathscr{F}\mathscr{W}^{\frac{1}{2}}\mathcal{A},
\end{align*}
suggesting that
\begin{align}\label{estimate-etab-intro}
\etab\lesssim\delta|u|^{-2}\mathscr{F}\mathscr{W}^{\frac{1}{2}}\mathcal{A}.
\end{align}
Finally we move to the second equation. By this equation, we have
\begin{align*}
\Lb\phi-\frac{\psi}{|u|}\lesssim\delta|u|^{-2}\mathscr{F}\mathcal{A}(|\psi|+1).
\end{align*}
To eliminate $\psi$ we should integrate this estimate over $u$. For any small enough $\kappa>0$, 
\begin{align}\label{estimate-Lbphi-intro}
\left(|u|^\kappa\int_{u_0}^u|u'|^{3-\kappa}\left|\Lb\phi-\frac{\psi}{|u'|}\right|^2\D u'\right)^{\frac{1}{2}}\lesssim_\kappa\delta\mathscr{F}\mathscr{W}^{\frac{1}{2}}\mathcal{A}.
\end{align}
This is the desired estimate\footnote{A slightly better bound for $\Lb\phi-\frac{\psi}{|u|}$, with $\mathscr{W}^{\frac{1}{2}}$ dropped, is obtained in the actual proof.} for $\Lb\phi-\frac{\psi}{|u|}$. It is crucial that $\kappa=0$ leads to a logarithmic loss, which we will discuss later. 

\subsubsection{Reductive structure and the smallness condition}

The presense of the borderline terms in the equations \eqref{equs1-intro} requires that the estimates for $\Theta$ should be obtained first without knowing any informations of $\Thetas$ and $\underline{\Theta}$. For example, the equation for $\Omega\chih$ reads
\begin{align}\label{equs2-intro}
\Db(\Omega\chih)-\frac{1}{2}\Omega\tr\chib\Omega\chih=\Omega^2(\nablas\Thetas+\Thetas\cdot\Thetas)+\Theta\cdot\underline{\Theta}
\end{align}
If we assume that the estimates \eqref{estimate-Theta-intro}, \eqref{estimate-Thetas-intro} and \eqref{estimate-etab-intro} hold with the right hand side multiplied by a large constant $C^{\frac{1}{4}}$, then the above equation implies that
\begin{align*}
|u||\Omega\chih|\lesssim|u_0||\Omega\chih|\Big|_{C_{u_0}}+\Omega_0^2(u_0)( C^{\frac{1}{4}}\delta|u|^{-1}\mathscr{F}\mathscr{W}^{\frac{1}{2}}\mathcal{A}+ C^{\frac{1}{2}}\delta^2|u|^{-2}\mathscr{F}^2\mathscr{W}\mathcal{A}^2)+ C^{\frac{1}{2}}\delta|u|^{-1}\mathscr{F}^2\mathcal{A}^2.
\end{align*}
We then need the \emph{smallness condition} \eqref{smallness-intro} to obtain the desired estimate \eqref{estimate-Theta-intro} for $|u|\ge|u_1|$. Indeed, by plugging in \eqref{smallness-intro}, the contribution from the right hand side becomes
\begin{align*}
C^{-1}\mathscr{F}\mathcal{A}
\end{align*}
which is $C^{-1}$ smaller than its expected estimate \eqref{estimate-Theta-intro}. We therefore find that the right hand side of the equation \eqref{equs2-intro} contains no borderline terms, i.e., can be absorbed under the smallness condition \eqref{smallness-intro}. Having the estimates for $\Theta$, we are able to estimate $\Thetas$, $\underline{\Theta}$ and $\Lb\phi-\frac{\psi}{|u|}$ using equations \eqref{equs1-intro}. Note that we also have many terms on the right hand sides of \eqref{equs1-intro} which are not borderline and can be handled as above. Finally we will find that the estimates \eqref{estimate-Theta-intro}, \eqref{estimate-Thetas-intro} and \eqref{estimate-etab-intro} hold without $C$ and hence by a bootstrap argument, they really hold for $\ub\in[0,\delta], u\in[u_0,u_1]$.

The above argument shows that the Einstein equations have certain \emph{reductive structure}. It means that the estimates for all related quantities can be divided into several steps. Although the error terms of the equations are nonlinear and highly coupled, either they can be absorbed by a suitable smallness condition, in which case they are non-borderline, or they only involve factors which are already estimated in previous steps.

\subsubsection{Energy estimates for the curvature and the renormalization of the null Bianchi equaitons}
The estimates for the connection coefficients should be coupled with the energy estimates for the curvature components. This is because the null structure equations which we use to estimate the connection coefficients contain curvature terms viewed as sources. The energy estimates for the curvature are based on the contracted Bianchi equations. Instead of introducing the Bel-Robinson tensor as in \cite{Chr, Ch-K}, we will simply do integration by parts based on the so-called null Bianchi equations which are equations about the null components of the Weyl tensor and obtained by decompose the contracted Bianchi equations along null directions. The advantage of this procedure over introducing the Bel-Robinson tensor is that, in the case when some of the null components of the Weyl curvature cannot be controlled due to the nature of the problem, we can  still do renormalization to the null Bianchi equations such that we are able to estimate some quantities instead of the Weyl curvature. This technique was first introduced by Luk-Rodnianski in \cite{L-R1} and subsequently developed or applied in \cite{An-Luk, L-Z1, L-R2, Luk}. 

The problem we study shares many common features with the work \cite{An-Luk}. Instead of estimating the full set of the null components of the Weyl curvature, they estimated $\beta, K-|u|^{-2}, \check{\sigma}=\sigma-\frac{1}{2}\chih\wedge\chibh$ and $\betab$ where $K$ is the Gauss curvature of $S_{\ub,u}$. Recalling we couple the Einstein equations with the scalar field, the null Codazzi equation for $\divs(\Omega\chih)$ reads
\begin{align*}
\divs(\Omega\chih)=\cdots-(\Omega\beta-L\phi\nablas\phi)
\end{align*}
Then $\Omega\beta$ is expected to behave like $|u|^{-2}\mathscr{F}\mathcal{A}$. In the procedure of the energy estimates, the estimates of $\beta$ (and its derivatives) along $C_u$ should be done together with the estimates of $K-|u|^{-2}$ and $\check{\sigma}$ along $\Cb_{\ub}$. The corresponding group of the renormalized null Bianchi equations are
\begin{align*}
\Db(\Omega\beta)+\frac{1}{2}\Omega\tr\chib\Omega\beta=&-\Omega^2(\nablas K-{}^*\nablas\check{\sigma})+\cdots,\\
D\left(K-\frac{1}{|u|^{-2}}\right)=&-\divs(\Omega\beta)+\cdots,\\
D\check{\sigma}=&-\curls(\Omega\beta)+\cdots.
\end{align*}
The first equation should be written in the form $\Db(\Omega\beta)$ but not $\Db\beta$ because the latter one will involve $\omegab$, the same as the equation for $\Db(\Omega\chih)$ which is mentioned in Remark \ref{aboutomegab-intro}. We multiply the first by $|u|^2\Omega\beta$, the second by $|u|^2\Omega^2(K-|u|^{-2})$ and third by $|u|^2\Omega^2\check{\sigma}$, and sum them up. Integrating the resulting equation in the spacetime, we can then expect
\begin{align*}
\int_0^\delta\int_{S_{\ub,u}}|u|^2|\Omega\beta|^2\D\mu_{\gs}\D\ub+\int_{u_0}^u\Omega_0^2(u')|u'|^2\int_{S_{\ub,u'}}\left(\left|K-\frac{1}{|u'|^2}\right|^2+|\check{\sigma}|^2\right)\D\mu_{\gs}\D u'\lesssim\delta\mathscr{F}^2\mathcal{A}^2.
\end{align*}
Here $\Omega_0$ appears again as a weight function. The other group of the renormalized equations are the following:
\begin{align*}
\Db\left(K-\frac{1}{|u|^{2}}\right)+\frac{3}{2}\Omega\tr\chib\left(K-\frac{1}{|u|^{2}}\right)=&\divs(\Omega\betab)+\cdots,\\
\Db\check{\sigma}+\frac{3}{2}\Omega\tr\chib\check{\sigma}=&-\curls(\Omega\betab)+\cdots,\\
D(\Omega\betab)=&\Omega^2(\nablas K+{}^*\nablas\check{\sigma})+\Lb\phi\nablas L\phi+\cdots.
\end{align*}
The rule is again the $\Db$ equations do not involve $\omegab$. According to the coefficient $\frac{3}{2}$ of the second terms of the first\footnote{It is one of the most important observations in \cite{An-Luk} that this first equation can be written in the form that the coefficient of the linear part is $\frac{3}{2}$ and the terms denoted by $\cdots$ are under control. The coefficient of the linear part of the original equation is $1$.}  and second equations, we multiply the first equation by $|u|^4(K-|u|^{-2})$, the second by $|u|^4\check{\sigma}$ and the third by $|u|^4\Omega^{-2}(\Omega\betab)$, sum them up and then integrate the resulting equation. Note that we have a large factor $\Omega^{-2}$ when estimating the error terms coming from the third equation, we cannot expect that $K-|u|^{-2}$ and $\check{\sigma}$ can be estimated without $\Omega_0$ weighted. Instead, we can only hope to estimate the following
\begin{align*}
\Omega_0^2(u)\int_0^\delta\int_{S_{\ub,u}}|u|^4\left(\left|K-\frac{1}{|u'|^2}\right|^2+|\check{\sigma}|^2\right)\D\mu_{\gs}\D\ub
\end{align*}
We will see in the proof that it is bounded by $\delta^3\mathscr{F}^3\mathcal{A}^3$. On the other hand, the term $\Lb\phi\nablas L\phi$ in the third equation is a deadly term. This is because only the angular derivatives of $\Lb\phi$ have good enough behavior but $\Lb\phi$ itself behaves only like its initial value $\frac{\psi}{|u|}$ on $\Cb_0$. Therefore, we should write
\begin{align*}
\Lb\phi\nablas L\phi=D(\Lb\phi\nablas\phi)-D\Lb\phi\nablas\phi.
\end{align*}
Then the first term can be absorbed by considering the null Bianchi equation for $D(\Omega\betab+\Lb\phi\nablas\phi)$ instead of $D(\Omega\betab)$ and the second term can be controlled using the wave equation. As a consequence, only $\Omega\betab+\Lb\phi\nablas\phi$ but not $\Omega\betab$ can be controlled in the norm
\begin{align*}
\Omega_0^2(u)\int_{u_0}^u\Omega_0^{-2}(u')|u'|^4\int_{S_{\ub,u'}}|\Omega\betab+\Lb\phi\nablas\phi|^2\D\mu_{\gs}\D u'
\end{align*}
which is bounded again by $\delta^3\mathscr{F}^3\mathcal{A}^3$. We should remark here that this estimate will cause a minor difficulty. Consider the equation
\begin{align*}
\Db\etab=\cdots+(\Omega\betab+\Lb\phi\nablas\phi)-2\Lb\phi\nablas\phi
\end{align*}
where the last term contributes to the desired estimate for $\etab$ and the terms denoted by $\cdots$ are also under control. Integrating the above equation we will have
\begin{align*}
|u||\etab|\lesssim|u_0||\etab|\big|_{C_{u_0}}+\int_{u_0}^u|u'||\Omega\betab+\Lb\phi\nablas\phi|\D u'+\delta|u|^{-1}\mathscr{F}\mathscr{W}^{\frac{1}{2}}\mathcal{A}.
\end{align*}
The second term on the right hand side should be estimated as, by H\"older inequality,
\begin{align*}
\int_{u_0}^u|u'||\Omega\betab+\Lb\phi\nablas\phi|\D u'\lesssim\left(\int_{u_0}^u\Omega_0^2(u')|u'|^{-4}\D u'\right)^{\frac{1}{2}}\left(\int_{u_0}^u\Omega_0^{-2}(u')|u'|^6|\Omega\betab+\Lb\phi\nablas\phi|^2\D u'\right)^{\frac{1}{2}}.
\end{align*}
Now recalling $\Omega_0$ is monotonically decreasing, the first factor on the right is bounded by $\Omega_0(u_0)|u|^{-\frac{3}{2}}$, and the second factor is bounded by $\Omega_0^{-1}(u)\delta^{\frac{3}{2}}\mathcal{F}^{\frac{3}{2}}\mathcal{A}^{\frac{3}{2}}$ mentioned above. Then, the contribution of the term $\Omega\betab+\Lb\phi\nablas\phi$ is
\begin{align*}
\Omega_0(u_0)\Omega_0^{-1}(u)(\delta|u|^{-1}\mathscr{F}\mathcal{A})^{\frac{3}{2}}.
\end{align*}
We can see that only the condition \eqref{smallness-intro} is not enough to make sure that $\etab$ obeys the expected estimate. Therefore the \emph{auxiliary condition} \eqref{auxiliary-intro} is needed to ensure that the a priori estimates can be done.

\subsubsection{Energy estimates for the scalar field and a logarithmic loss on the top order derivatives}
Recalling we are coupling the Einstein equations with the massless scalar field $\phi$ which satisfying the wave equations, through which we can obtain estimates for the scalar field. We will not introduce the energy momentum tensor but simply do integration by parts using the wave equations written in the double null coordinate system:
\begin{align*}
\Db L\phi+\frac{1}{2}\Omega\tr\chib L\phi&=\Omega^2\Deltas\phi+2\Omega^2(\eta,\ds\phi)-\frac{1}{2}\Omega\tr\chi\Lb\phi,\\
D\ds\phi&=\nablas L\phi,\\
D\Lb\phi+\frac{1}{2}\Omega\tr\chi\Lb\phi&=\Omega^2\Deltas\phi+2\Omega^2(\etab,\ds\phi)-\frac{1}{2}\Omega\tr\chib L\phi,\\
\Db\ds\phi&=\nablas\Lb\phi.
\end{align*}
There are four equations here but the first and the third are exactly the same equation, which is the wave equation. The second and the fourth equations are simply geometric identities. Written in this form, the first two equations can be considered as a group of equations for $L\phi$, $\nablas\phi$ and the last two can be considered as a group of equations for $\nablas\phi$, $\Lb\phi$. They have the same structure as the null Bianchi equations and then the energy estimates can be done in the same way.

Now we want to discuss what will happen for the estimates for the top order derivatives of the geometric quantities. It is not surprising that the top order estimates are worse than the lower order estimates. We may think this is due to the nonlinear nature of the equations, because on a technical level, we use the transport equations to estimate the lower order derivatives, and use however the Hodge-transport coupled system to estimate the top order derivatives. And sometimes different approaches give arise to different estimates. For example, similar to \cite{An-Luk}, the estimate along $\Cb_{\ub}$ for the top order derivative of $\etab$ is worse than $\eta$. Besides this, we will have a logarithmic loss due to the presence of the scalar field. Let us look at the estimate for $\Omega\tr\chib$. Its lower order derivatives can be estimated easily by the equation
\begin{align*}
D(\Omega\tr\chib)=\Omega^2\nablas\eta+\cdots.
\end{align*}
However, integrating this equation will cause a loss of derivative and therefore we cannot use this equation to estimate the top order derivative of $\Omega\tr\chib$, which should be estimated using the equation for $\Db(\Omega\tr\chib)$ coupled with the null Codazzi equation for $\divs(\Omega\chibh)$. Because we only hope to describe here the main ideas, for convenience we suppose that we only estimate $\nablas(\Omega\tr\chib)$ itself using these equations, although we know that the lower order estimates can be obtained much more easily. Now the equation for $\Db(\Omega\tr\chib)$ reads
\begin{equation*}
\Db(\Omega\tr\chib)=-\frac{1}{2}(\Omega\tr\chib)^2+2\omegab\Omega\tr\chib-|\Omega\chibh|^2-2(\Lb\phi)^2
\end{equation*}
and applying an angular derivative $\nablas$, we have
\begin{equation}\label{equ-Dbtrchib-intro}
\begin{split}
&\Db(\Omega^{-2}\nablas(\Omega\tr\chib))+\Omega\tr\chib(\Omega^{-2}\nablas(\Omega\tr\chib))\\
&=2\Omega^{-2}\nablas\omegab(\Omega\tr\chib)-2\Omega^{-2}(\Omega\chibh)\cdot\nablas(\Omega\chibh)
-4\Omega^{-2}\Lb\phi\nablas\Lb\phi.\\
\end{split}
\end{equation}
Note that the equation is written in this form in order that it does not involve $\omegab$ itself (but involves the angular derivatives of $\omegab$). The coefficient $1$ of the second term on the left hand side suggests that the right hand side should be estimated in
\begin{align*}
\int_{u_0}^u|u'|^{3}|\cdot|\D u'.
\end{align*}
From the equation\footnote{The top order estimate for $\omegab$ should be derived using the transport-Hodge coupled system instead of this transport equation, but the term causing trouble is essentially the same.}
\begin{align*}
D\nablas\omegab&=-\Lb\phi\nablas L\phi+\cdots,
\end{align*}
the first term on the right hand side of \eqref{equ-Dbtrchib-intro} is estimated by
\begin{align*}
&\int_{u_0}^u|u'|^{3}|\Omega^{-2}\nablas\omegab(\Omega\tr\chib)|\D u'\\\lesssim&\Omega_0^{-2}(u)\delta\left(\int_{u_0}^u\frac{|\psi|^2}{|u'|}\D u'\right)^{\frac{1}{2}}\sup_{0\le\ub\le\delta}\left(\int_{u_0}^u|u'|^3\|\nablas L\phi\|_{\L^\infty(\ub,u')}^2\D u'\right)^{\frac{1}{2}}+\cdots.
\end{align*}
From the equation
\begin{align*}
D\nablas\Lb\phi&=-\frac{1}{2}\Omega\tr\chib\nablas L\phi+\cdots.
\end{align*}
The last term of \eqref{equ-Dbtrchib-intro} can be estimated similarly by
\begin{align*}
&\int_{u_0}^u|u'|^{3}|\Omega^{-2}\Lb\phi\nablas\Lb\phi|\D u'\\
\lesssim&\Omega_0^{-2}(u)\delta\left(\int_{u_0}^u\frac{|\psi|^2}{|u'|}\D u'\right)^{\frac{1}{2}}\sup_{0\le\ub\le\delta}\left(\int_{u_0}^u|u'|^3\|\nablas L\phi\|_{\L^\infty(\ub,u')}^2\D u'\right)^{\frac{1}{2}}+\cdots.
\end{align*}
We can see both estimates cannot be controlled without a logarithmic loss\footnote{Observe however that the term $\Omega\tr\chib \Lb\phi\nablas L\phi$, which we are dealing with, disappears in the equation for  $D(\nablas\omegab(\Omega\tr\chib)-2\Lb\phi\nablas\Lb\phi)$. We could have killed the logarithmic loss by using this cancellation, which unfortunately requires a loss of derivative and only works for the lower order estimates.} because $\nablas L\phi$ can only be bounded by $|u|^{-2}\mathscr{F}\mathcal{A}$. The worst thing that could happen is that this loss would accumulate and the bootstrap argument cannot be closed. Fortunately, $\nablas L\phi$ satisfies a better estimate:
\begin{align*}
|u||\nablas L\phi|\lesssim|u_0||\nablas L\phi|\big|_{C_{u_0}}+\delta|u|^{-1}\mathscr{F}^2\mathscr{W}^{\frac{1}{2}}\mathcal{A}^2.
\end{align*}
The second term on the right hand side will not suffer the logarithmic loss by first integrating along $\Cb_{\ub}$ and then plugging in the smallness condition \eqref{smallness-intro}. The loss now only depends on the initial value of $\nablas L\phi$. So we introduce the quantity $\mathscr{E}$ defined in \eqref{def-E-intro} and we will find that the top order estimates are expected to suffer a loss $\mathscr{E}$. Moreover, the next to top order derivatives of $\eta,\etab$, $\nablas\phi$ and all curvature components also suffer such a loss. The reason is that, for example, the estimate for the next to top order estimates of $\eta$, $\etab$ and $\nablas\phi$ are derived in terms of the top order derivatives of $\beta$, $\betab+\Lb\phi\nablas\phi$ and $L\phi$ respectively. Recalling in the statement of Theorem \ref{existencerough} that we will estimate up to the \engordnumber{4} order derivatives of the curvature and up to the \engordnumber{5} order derivatives of the connection coefficients and the derivatives of the scalar field function $\phi$. Then it is only the \engordnumber{3} and lower order estimates do not suffer this loss. We remark that in proving Theorem \ref{formationrough}, the formation of trapped surface, $L^\infty$ estimate without logarithmic loss for the $\nablas\eta$ is needed, therefore the estimates for up to the \engordnumber{5} order derivatives of the connection coefficients are needed.

Although only the top and next to top estimates have the loss, we find that it is convenient to only derive estimates with loss for all derivatives of $\eta,\etab,\nablas\phi$ in the proof of existence. The expected estimate obeyed by $\Thetas\in\{\eta,\etab,\nablas\phi\}$ then becomes
\begin{align*}
\Thetas\lesssim\delta|u|^{-1}\mathscr{F}\mathscr{E}\mathscr{W}^{\frac{1}{2}}\mathcal{A}.
\end{align*}
Because of this loss, an addition smallness condition should be introduced. Consider for example the equation for $D(\Omega\chibh)$ in the form
\begin{align*}
D(\Omega\chibh)=\Omega^2(\nablas\etab+\etab\cdot\etab)+\cdots
\end{align*}
The top order derivative for $\etab$ can only be bounded in a scale invariant version of the norm $L^2(C_u)$ by $\delta^{\frac{1}{2}}|u|^{-\frac{3}{2}}\mathscr{F}\mathscr{E}\mathcal{A}$, with a loss $(\delta|u|^{-1})^{\frac{1}{2}}$ as compared to its lower order estimates. Therefore, by integrating the above equation, we have,
\begin{align*}
|\Omega\chibh|\lesssim\Omega_0^2(u)(\delta^{\frac{3}{2}}|u|^{-\frac{5}{2}}\mathscr{F}\mathscr{E}\mathscr{W}^{\frac{1}{2}}\mathcal{A}+\delta^3|u|^{-4}\mathscr{F}^2\mathscr{E}^2\mathscr{W}\mathcal{A}^2)+\cdots.
\end{align*}
The desired estimate can then be obtained by introducing another \emph{smallness condition} \eqref{smallness1-intro}.
\begin{remark}
Strictly speaking, the condition \eqref{smallness1-intro} has nothing to do with the \emph{smallness} because we do not need to require that the quantity on the left hand side is very small. But we still name it the \emph{smallness contidion} because it is used very frequently as the real \emph{smallness condition} \eqref{smallness-intro}.
An unsatisfactory consequence is that when considering the next to top order estimate for $\Omega\chibh$, the first term $\Omega^2\nablas\etab$ on the right hand side of the equation for $D(\Omega\chibh)$ becomes \emph{borderline}. Nevertheless we can also close the estimates because the top order estimate for $\etab$ do not contain borderline terms involving $\Omega\chibh$. \end{remark}

We then close the discussion on the proof of the existence theorem.

\subsection{Outline of the paper} The remainder of this paper is organized as follows. In Section \ref{preliminary}, we introduce the double null coordinate system, and the notations and equations adapted to this coordinate system we will use in the proof. In Section \ref{statementexistencetheorem}, we will state the existence theorem, whose proof is then given in Section \ref{APrioriEstimate}. The formation of trapped surfaces theorems are proven in Section \ref{FoTS} and the instability theorems are proven in Section \ref{Secinstability}. In the last section, Section \ref{proofofmaintheorem}, we will state and prove the precise version of the Theorem \ref{main1}, the main result of this paper.

\subsection{Acknowledgement} Both authors are supported by NSFC 11501582, 11521101. The first author would like to thank Jonathan Luk for valuable discussions on this topic and suggestions on the first version of the manuscript. He would also like to thank Xinliang An for showing his interest in this work. Both authors would like to thank Xi-Ping Zhu for his supports, discussions and continuous encouragements.

\section{Double null coordinate system and equations}\label{preliminary}

\subsection{Double null coordinate system}\label{doublenull}

We first introduce the geometric setup. We use $(M,g)$ to denote the underlying space-time (which will be the solution) with the Lorentzian metric and use $\nabla$ to denote the Levi-Civita connection of the metric $g$. Let $u$ and $\ub$ be two optical functions on $M$, that is $$g(\nabla u,\nabla u)=g(\nabla\ub,\nabla\ub)=0.$$
 The space-time $M$ is foliated by the level sets of $\ub$ and $u$ respectively, and the functions $u$ and $\ub$ are required to increase towards the future. We use $C_u$ to denote the outgoing null hypersurfaces which are the level sets of $u$ and use ${\Cb}_{\ub}$ to denote the incoming null hypersurfaces which are the level sets of $\ub$. We denote the intersection $S_{\ub,u}=\Cb_{\ub} \cap C_u$, which is a  space-like two-sphere.

We define a positive function $\Omega$ by the formula
$$ \Omega^{-2}=-2g(\nabla\ub,\nabla u).$$
  We then define the normalized null pair $(\Lbh, \Lh)$ by
  $$\Lbh=-2\Omega\nabla\ub,\ \Lh=-2\Omega\nabla u,$$ which also be denoted by $e_3,e_4$ respectively, and define one another null pair
  $$\Lb=\Omega \Lbh,\ L=\Omega \Lh.$$
 The flows generated by $\Lb$ and $L$ preserve the double null foliation. On a given two sphere $S_{\ub, u}$ we choose a local frame $
   {e_1,e_2}$. We call ${e_1, e_2, \Lbh, \Lh}$ a \emph{null frame}.  As a convention, throughout the paper, we use capital  Latin letters $A, B, C, \cdots$ to denote an index from $1$ to $2$ and Greek letters $\alpha,\beta,\cdots$ to denote an index from $1$ to $4$, e.g. $e_A$ denotes either $e_1$ or $e_2$.

We define $\psi$ to be a tangential tensorfield if $\psi$ is \textit{a priori} a tensorfield defined on the space-time $M$ and all the possible contractions of $\psi$ with either $\Lbh$ or $\Lh$ are zeros. We use $D\psi$ and $\Db\psi$ to denote the projection to $S_{\ub,u}$ of usual Lie derivatives $\mathcal{L}_L\psi$ and $\mathcal{L}_{\Lb}\psi$. The space-time metric $g$ induces a Riemannian metric $\gs$ on $S_{\ub,u}$ and $\epsilons$ is the volume form of $\gs$ on $S_{\ub,u}$. We use $\ds$ and $\nablas$ to denote the exterior differential and covariant derivative (with respect to $\gs$) on $S_{\ub,u}$.

Using these two optical functions, we can introduce a local coordinate system $(\ub,u,\vartheta^A)$ on $M$. In such a coordinate system, the Lorentzian metric $g$ takes the following form
\begin{align*}
g=-2\Omega^2(\D\ub\otimes\D u+\D u\otimes\D \ub)+\gs_{AB}(\D\vartheta^A-b^A\D\ub)\otimes(\D\vartheta^B-b^B\D\ub),
\end{align*}
such that $\Lb$ and $L$ can be computed as $\Lb=\partial_u+b^A\partial_{\vartheta^A}$ and $L=\partial_{\ub}$. We can also require that $b^A$ vanishes on some specific null cone.

We recall the null decomposition of the connection coefficients using the null frame $(e_1,e_2,\Lbh,\Lh)$ as follows:
\begin{align*}
\chi_{AB}&=g(\nabla_A\Lh,e_B),\quad \eta_A=-\frac{1}{2}g(\nabla_{\Lbh}e_A,\Lh),\quad \omega=\frac{1}{2}\Omega g(\nabla_{\Lh}\Lbh,\Lh),\\
\chib_{AB}&=g(\nabla_A\Lbh,e_B), \quad\etab_A=-\frac{1}{2}g(\nabla_{\Lh}e_A,\Lbh), \quad\omegab=\frac{1}{2}\Omega g(\nabla_{\Lbh}\Lh,\Lbh).
\end{align*}
They are all tangential tensorfields. We also define the following normalized quantities:
$$\chi'=\Omega^{-1}\chi,\ \chib'=\Omega^{-1}\chi,\ \zeta=\frac{1}{2}(\eta-\etab).$$
 The trace of $\chi$ and $\chib$ are denoted by
 $$\tr\chi = \gs^{AB}\chi_{AB},\ \tr\chib = \gs^{AB}\chib_{AB}.$$
  By definition, we can check directly the following useful identities :
  $$\ds\log\Omega=\frac{1}{2}(\eta+\etab),\ D\log\Omega=\omega,\ \Db\log\Omega=\omegab.$$

We can also define the null components of the Weyl curvature tensor
{\bf W}:
\begin{align*}
\alpha_{AB}&=\mathbf{W}(e_A,\Lh,e_B,\Lh),\quad\beta_A=\frac{1}{2}\mathbf{W}(e_A,\Lh,\Lbh,\Lh),\quad\rho=\frac{1}{4}\mathbf{W}(\Lbh,\Lh,\Lbh,\Lh),\\
\alphab_{AB}&=\mathbf{W}(e_A,\Lbh,e_B,\Lbh),\quad\betab_A=\frac{1}{2}\mathbf{W}(e_A,\Lbh,\Lbh,\Lh),\quad\sigma=\frac{1}{4}\mathbf{W}(\Lbh,\Lh,e_A,e_B)\epsilons^{AB}.
\end{align*}

\subsection{Equations}
Before we write down the equations we will use, we first define several kinds of contraction of the tangential tensorfields. For a  symmetric tangential 2-tensorfield $\theta$, we use $\widehat{\theta}$ and $\tr\theta$ to denote the trace-free part and trace of $\theta$ (with respect to $\gs$). If $\theta$ is trace-free, $\Dh\theta$ and $\Dbh\theta$ refer to the trace-free part of $D\theta$ and $\Db\theta$. Let $\xi$ be a tangential $1$-form. We define some products and operators for later use. For the products, we define 
$$(\theta_1,\theta_2)=\gs^{AC}\gs^{BD}(\theta_1)_{AB}(\theta_2)_{CD},\ \ (\xi_1,\xi_2)=\gs^{AB}(\xi_1)_A(\xi_2)_B.$$ This also leads to the following norms 
$$|\theta|^2=(\theta,\theta),\ |\xi|^2=(\xi,\xi).$$ We then define the contractions 
\begin{align*}(\theta\cdot\xi)_A=\theta_A{}^B\xi_B&,\ (\theta_1\cdot \theta_2)_{AB}=(\theta_1)_A{}^C(\theta_2)_{CB},\\
 \theta_1 \wedge\theta_2=\epsilons^{AC}\gs^{BD} (\theta_1)_{AB}(\theta_2)_{CD}&,\ \xi_1\tensor \xi_2=\xi_1\otimes\xi_2+\xi_2\otimes\xi_1-(\xi_1,\xi_2)\gs.
\end{align*} The Hodge dual for $\xi$ is defined by $\prescript{*}{}\xi_A=\epsilons_A{}^C\xi_C$. For the operators, we define 
$$\divs\xi=\nablas^A\xi_A,\ \curls\xi_A=\epsilons^{AB}\nablas_A\xi_B,\ (\divs\theta)_A=\nablas^B\theta_{AB}.$$ We finally define a traceless operator $$(\nablas\tensor\xi)_{AB}=(\nablas\xi)_{AB}+(\nablas\xi)_{BA}-\divs\xi \,\gs_{AB}.$$


The followings are the \textit{null structure equations} (where  $K$ is the Gauss curvature of $S_{\ub,u}$):\footnote{See Chapter 1 of \cite{Chr} for the derivation of these equations in the vacuum case.}
\begin{align*}
\Db(\Omega\tr\chib)&=-\frac{1}{2}(\Omega\tr\chib)^2+2\omegab\Omega\tr\chib-|\Omega\chibh|^2-2(\Lb\phi)^2,\\
D\tr\chi'&=-\frac{1}{2}(\Omega\tr\chi')^2-|\chih|^2-2(\Lh\phi)^2,\\
\Dbh(\Omega\chih)&=\Omega^2(\nablas \tensor \eta + \eta \tensor \eta +\frac{1}{2}\tr\chib\chih-\frac{1}{2}\tr\chi \chibh+\ds\phi\tensor\ds\phi),\\
\Dh(\Omega\chibh)&=\Omega^2(\nablas \tensor \etab + \etab \tensor \etab +\frac{1}{2}\tr\chi\chibh-\frac{1}{2}\tr\chib \chih+\ds\phi\tensor\ds\phi),\\
\Db(\Omega\tr\chi)&=\Omega^2(2\divs\eta+2|\eta|^2-\tr\chi\tr\chib-2K+2|\ds\phi|^2),\\
D(\Omega\tr\chib)&=\Omega^2(2\divs\etab+2|\etab|^2-\tr\chi\tr\chib-2K+2|\ds\phi|^2),\\
D\eta &= (\Omega\chi)\cdot\etab-(\Omega\beta+L\phi\ds\phi),\\
\Db\etab &= (\Omega\chib) \cdot\eta+(\Omega\betab-\Lb\phi\ds\phi),\\
D  \omegab &=\Omega^2(2(\eta,\etab)-|\eta|^2-(\rho+\frac{1}{6}\mathbf{R}+\Lh\phi\Lbh\phi)),\\
\Db  \omega &=\Omega^2(2(\eta,\etab)-|\etab|^2-(\rho+\frac{1}{6}\mathbf{R}+\Lh\phi\Lbh\phi)),\\
K&=-\frac{1}{4}\tr \chi\tr\chib+\frac{1}{2}(\chih,\chibh)-(\rho+\frac{1}{6}\mathbf{R})+|\ds\phi|^2,\\
\divs(\Omega\chih)&=\frac{1}{2}\Omega^2\ds\tr \chi'+\Omega\chih\cdot\etab+\frac{1}{2}\Omega\tr \chi\eta-(\Omega\beta-L\phi\ds\phi),\\
\divs(\Omega\chibh)&=\frac{1}{2}\Omega^2\ds\tr \chib'+\Omega\chibh\cdot\eta+\frac{1}{2}\Omega\tr \chib\etab+(\Omega\betab+\Lb\phi\ds\phi).\\
\end{align*}

Recall the Einstein-scalar field system reads
\begin{equation*}
\begin{cases}\mathbf{R}_{\alpha\beta}=2\partial_{\alpha}\phi\partial_{\beta}\phi,\\
g^{\alpha\beta}\nabla_\alpha\nabla_\beta\phi=0\end{cases}.
\end{equation*}
Consequently, the spacetime scalar curvature $\mathbf{R}=2\partial^{\alpha}\phi\partial_{\alpha}\phi$.
The contracted second Bianchi identity will accordingly be the following inhomogeneous equation:
\begin{equation*}
\nabla^{\alpha}\mathbf{W}_{\alpha\beta\gamma\delta}=\nabla_{[\gamma}\mathbf{R}_{\delta]\beta}+\frac{1}{6}g_{\beta[\gamma}\nabla_{\delta]}\mathbf{R}.
\end{equation*}
This equation can be decomposed using the null frame $(e_1,e_2,e_3,e_4)$ into components, which we call \textit{null Bianchi equations}:\footnote{See Proposition 1.2 of \cite{Chr} in the vacuum case.}
\begin{align*}
&\Dbh\alpha-\frac{1}{2}\Omega\tr\chib \alpha+2\omegab\alpha+\Omega\{-\nablas\tensor\beta -(4\eta+\zeta)\tensor \beta+3\chih (\rho+\frac{1}{6}\mathbf{R})+3{}^*\chih \sigma\}
\\=&-\Omega\left(\nablas\Lh\phi\tensor\ds\phi-\nablas\tensor\ds\phi\Lh\phi\right.\\
&\left.-\frac{1}{2}\tr\chi\ds\phi\tensor\ds\phi-\chih\cdot\ds\phi\tensor\ds\phi+\frac{3}{2}\chih\Lh\phi\Lbh\phi-\chih|\ds\phi|^2+\frac{1}{2}\chibh(\Lh\phi)^2+\zeta\tensor\ds\phi\Lh\phi\right),\\
&D\beta+\frac{3}{2}\Omega\tr\chi\beta-\Omega\chih\cdot\beta-\omega\beta-\Omega\{\divs\alpha+(\etab+2\zeta)\cdot\alpha\}\\
=&\Omega\left(\ds\phi\Lh\Lh\phi-\Lh\phi\nablas\Lh\phi+\frac{1}{2}\tr\chi\ds\phi\Lh\phi+\chih\cdot\ds\phi\Lh\phi-\Omega^{-1}\omega\Lh\phi\ds\phi-(\Lh\phi)^2\zeta\right),\\
&\Db\beta+\frac{1}{2}\Omega\tr\chib\beta-\Omega\chibh \cdot \beta+\omegab \beta-\Omega\{\ds (\rho+\frac{1}{6}\mathbf{R})+{}^*\ds \sigma+3\eta(\rho+\frac{1}{6}\mathbf{R})+3{}^*\eta\sigma+2\chih\cdot\betab\}\\
=&-\Omega\left(\ds\phi\Deltas\phi-\nablas\Lbh\phi\Lh\phi-\frac{1}{2}\tr\chi\Lbh\phi\ds\phi+\chibh\cdot\ds\phi\Lh\phi-(\eta-\zeta)\Lbh\phi\Lh\phi+\eta|\ds\phi|^2\right),\\
&D(\rho+\frac{1}{6}\mathbf{R})+\frac{3}{2}\Omega\tr\chi (\rho+\frac{1}{6}\mathbf{R})-\Omega\{\divs \beta+(2\etab+\zeta,\beta)-\frac{1}{2}(\chibh,\alpha)\}\\=&-\Omega\left(\Lh\phi\Delta\phi-\nablas\Lh\phi\cdot\ds\phi+\chih\cdot\ds\phi\cdot\ds\phi-\frac{1}{2}\tr\chib(\Lh\phi)^2-\zeta\cdot\ds\phi\Lh\phi\right),\\
&D\sigma+\frac{3}{2}\Omega\tr\chi\sigma+\Omega\{\curls\beta+(2\etab+\zeta,{}^*\beta)-\frac{1}{2}\chibh\wedge\alpha\}\\
=&-\Omega\left(\nablas\Lh\phi\wedge\ds\phi-\chih\cdot\ds\phi\wedge\ds\phi+\zeta\wedge\ds\phi\Lh\phi\right),\\
&D\betab+\frac{1}{2}\Omega\tr\chi\betab-\Omega\chih \cdot \betab+\omega \betab+\Omega\{\ds (\rho+\frac{1}{6}\mathbf{R})-{}^*\ds \sigma+3\etab(\rho+\frac{1}{6}\mathbf{R})-3{}^*\etab\sigma-2\chibh\cdot\beta\}\\
=&\Omega\left(\ds\phi\Deltas\phi-\nablas\Lh\phi\Lbh\phi-\frac{1}{2}\tr\chib\Lh\phi\ds\phi+\chih\cdot\ds\phi\Lbh\phi-(\etab+\zeta)\Lh\phi\Lbh\phi+\etab|\ds\phi|^2\right),\\
&\Db(\rho+\frac{1}{6}\mathbf{R})+\frac{3}{2}\Omega\tr\chib (\rho+\frac{1}{6}\mathbf{R})+\Omega\{\divs \betab+(2\eta-\zeta,\betab)+\frac{1}{2}(\chih,\alphab)\}\\
=&-\Omega\left(\Lbh\phi\Delta\phi-\nablas\Lbh\phi\cdot\ds\phi+\chibh\cdot\ds\phi\cdot\ds\phi-\frac{1}{2}\tr\chi(\Lbh\phi)^2+\zeta\cdot\ds\phi\Lbh\phi\right),\\
&\Db\sigma+\frac{3}{2}\Omega\tr\chib\sigma+\Omega\{\curls\betab+(2\eta-\zeta,{}^*\betab)+\frac{1}{2}\chih\wedge\alphab\}\\
=&\Omega\left(\nablas\Lbh\phi\wedge\ds\phi-\chibh\cdot\ds\phi\wedge\ds\phi-\zeta\wedge\ds\phi\Lbh\phi\right),\\
&\Db\betab+\frac{3}{2}\Omega\tr\chib\betab-\Omega\chibh\cdot\betab-\omegab\betab+\Omega\{\divs\alphab+(\eta-2\zeta)\cdot\alphab\}\\
=&-\Omega\left(\ds\phi\Lbh\Lbh\phi-\Lbh\phi\nablas\Lbh\phi+\frac{1}{2}\tr\chib\ds\phi\Lbh\phi+\chibh\cdot\ds\phi\Lbh\phi-\Omega^{-1}\omegab\Lbh\phi\ds\phi+(\Lbh\phi)^2\zeta\right),\\
&\Dh\alphab-\frac{1}{2}\Omega\tr\chi \alphab+2\omega\alphab+\Omega\{\nablas\tensor\betab +(4\etab-\zeta)\tensor \betab+3\chibh (\rho+\frac{1}{6}\mathbf{R})-3{}^*\chibh \sigma\}\\
=&-\Omega\left(\nablas\Lbh\phi\tensor\ds\phi-\nablas\tensor\ds\phi\Lbh\phi\right.\\
&\left.-\frac{1}{2}\tr\chib\ds\phi\tensor\ds\phi-\chibh\cdot\ds\phi\tensor\ds\phi+\frac{3}{2}\chibh\Lbh\phi\Lh\phi-\chibh|\ds\phi|^2+\frac{1}{2}\chih(\Lbh\phi)^2-\zeta\tensor\ds\phi\Lbh\phi\right).\\
\end{align*}
As mentioned in the Introduction, we consider instead the following \textit{renormalized null Bianchi equations}:
\begin{align*}
DK&+\Omega\tr\chi K+\divs(\Omega\beta-L\phi\nablas\phi)\\
&-\Omega\chih\cdot\nablas\etab+\frac{1}{2}\Omega\tr\chi\divs\etab+(\Omega\beta-L\phi\nablas\phi)\cdot\etab-\Omega\chih\cdot\etab\cdot\etab+\frac{1}{2}\Omega\tr\chi|\etab|^2=0\\
D\sigmac&+\frac{3}{2}\Omega\tr\chi\sigmac+\curls(\Omega\beta-L\phi\nablas\phi)+\frac{1}{2}\Omega\chih\wedge(\etab\tensor\etab+\nablas\tensor\etab)\\
&+\etab\wedge(\Omega\beta-L\phi\nablas\phi)+2\nablas L\phi\wedge\nablas\phi=0\\
\Db(\Omega\beta-L\phi\nablas\phi)&+\frac{1}{2}\Omega\tr\chib(\Omega\beta-L\phi\nablas\phi)-\Omega\chibh\cdot(\Omega\beta-L\phi\nablas\phi)+\Omega^2\ds K-\Omega^2{}^*\ds\sigmac\\
&+3\Omega^2(\eta K-{}^*\eta\sigmac)-\frac{1}{2}\Omega^2(\ds(\chih,\chibh)+{}^*\ds(\chih\wedge\chibh))-\frac{3}{2}\Omega^2(\eta(\chih,\chibh)+{}^*\eta(\chih\wedge\chibh))\\
&+\frac{1}{4}\Omega^2\ds(\tr\chi\tr\chib)+\frac{3}{4}\Omega^2\tr\chi\tr\chib\eta-2\Omega\chih\cdot(\Omega\betab+\Lb\phi\nablas\phi)\\
=&-2\Omega^2\Deltas\phi\nablas\phi+\Omega^2\ds|\ds\phi|^2-2\Omega\chih\cdot\nablas\phi\Lb\phi+\Omega\tr\chi\Lb\phi\nablas\phi\\
&-2\Omega^2\eta\cdot\nablas\phi\nablas\phi+2\Omega^2\eta|\ds\phi|^2,\\
\Db(K-\frac{1}{|u|^2})&+\frac{3}{2}\Omega\tr\chib (K-\frac{1}{|u|^2})+(\Omega\tr\chib+\frac{2}{|u|})\frac{1}{|u|^2}\\
=&\divs(\Omega\betab+\Lb\phi\nablas\phi)+\Omega\chibh\cdot\nablas\eta+\frac{1}{2}\Omega\tr\chib\mu\\
&+(\Omega\betab+\Lb\phi\nablas\phi)\cdot\eta+\Omega\chibh\cdot\eta\cdot\eta-\frac{1}{2}\Omega\tr\chib|\eta|^2,\\
\Db\sigmac&+\frac{3}{2}\Omega\tr\chi\sigmac+\curls(\Omega\betab+\Lb\phi\nablas\phi)\\
&-\frac{1}{2}\Omega\chibh\wedge(\eta\tensor\eta+\nablas\tensor\eta)+\eta\wedge(\Omega\betab+\Lb\phi\nablas\phi)-2\nablas \Lb\phi\wedge\nablas\phi=0\\
D(\Omega\betab+\Lb\phi\nablas\phi)&+\frac{1}{2}\Omega\tr\chi(\Omega\betab+\Lb\phi\nablas\phi)-\Omega\chih\cdot(\Omega\betab+\Lb\phi\nablas\phi)-\Omega^2\ds K-\Omega^2{}^*\ds\sigmac\\
&-3\Omega^2(\etab K+{}^*\etab\sigmac)+\frac{1}{2}\Omega^2(\ds(\chih,\chibh)-{}^*\ds(\chih\wedge\chibh))+\frac{3}{2}\Omega^2(\etab(\chih,\chibh)+{}^*\etab(\chih\wedge\chibh))\\
&-\frac{1}{4}\Omega^2\ds(\tr\chi\tr\chib)-\frac{3}{4}\Omega^2\tr\chi\tr\chib\etab-2\Omega\chibh\cdot(\Omega\beta-L\phi\nablas\phi)\\
=&2\Omega^2\Deltas\phi\nablas\phi-\Omega^2\ds|\ds\phi|^2+2\Omega\chibh\cdot\nablas\phi L\phi-\Omega\tr\chib L\phi\nablas\phi\\
&+2\Omega^2\etab\cdot\nablas\phi\nablas\phi-2\Omega^2\etab|\ds\phi|^2.
\end{align*}

In the equation for $\Db (K-\frac{1}{|u|^2})$ above, $\mu$ is defined through
$$\divs\eta=K-\frac{1}{|u|^2}-\mu.$$

We also rewrite the wave equation $g^{\alpha\beta}\nabla_\alpha\nabla_\beta\phi=0$ in the following form:
\begin{align*}
\Db L\phi+\frac{1}{2}\Omega\tr\chib L\phi&=\Omega^2\Deltas\phi+2\Omega^2(\eta,\ds\phi)-\frac{1}{2}\Omega\tr\chi\Lb\phi,\\
D\ds\phi&=\nablas L\phi,\\
D\Lb\phi+\frac{1}{2}\Omega\tr\chi\Lb\phi&=\Omega^2\Deltas\phi+2\Omega^2(\etab,\ds\phi)-\frac{1}{2}\Omega\tr\chib L\phi,\\
\Db\ds\phi&=\nablas\Lb\phi.
\end{align*}

Finally, the following elliptic-transport coupled systems are also needed:


\begin{align*}
&\begin{dcases}\divs\eta=K-\frac{1}{|u|^2}-\mu,\\
\curls\eta=\sigma-\frac{1}{2}\chih\wedge\chibh,\\
D\mu+\Omega\tr\chi\mu=-\Omega\tr\chi\frac{1}{|u|^2}+\divs(2\Omega\chih\cdot\eta-\Omega\tr\chi\etab)+2\nablas L\phi\cdot\ds\phi+2L\phi\Deltas\phi\end{dcases};\\
&\begin{dcases}
\divs\etab=K-\frac{1}{|u|^2}-\mub,\\
\curls\etab=-\sigma+\frac{1}{2}\chih\wedge\chibh,\\
\Db\mub+\Omega\tr\chib\mub=-(\Omega\tr\chib+\frac{2}{|u|})\frac{1}{|u|^2}+\divs(2\Omega\chibh\cdot\etab-\Omega\tr\chib\eta)+2\nablas \Lb\phi\cdot\ds\phi+2\Lb\phi\Deltas\phi\end{dcases};\\
&\begin{dcases}\Deltas\omegab&=\omegabs+\divs(\Omega\betab+\Lb\phi\nablas\phi),\\
D\omegabs&+\Omega\tr\chi\omegabs+2\Omega\chih\cdot\nablas\nablas\omegab+2\divs(\Omega\chih)\cdot\nablas\omegab-\frac{1}{2}\divs(\Omega\tr\chi(\Omega\betab+\Lb\phi\nablas\phi))\\
&+\nablas(\Omega^2)\cdot(\ds(\rho+\frac{1}{6}\mathbf{R})+{}^*\ds\sigma)+\Deltas(\Omega^2)(\rho+\frac{1}{6}\mathbf{R})-\Deltas(\Omega^2(2\eta\cdot\etab-|\eta|^2))\\
&-\divs(\Omega\chih\cdot(\Omega\betab+\Lb\phi\nablas\phi)-2\Omega\chibh\cdot(\Omega\beta-L\phi\nablas\phi)+3\Omega^2\etab(\rho+\frac{1}{6}\mathbf{R})-3\Omega^2{}^*\etab\sigma)\\
\phantom{\Delta}=&-\divs\{2\Omega^2\ds\phi\Deltas\phi+\Lb\phi\nablas L\phi+L\phi\nablas\Lb\phi-\Omega\tr\chib L\phi\ds\phi\\
&+2\Omega\chibh\cdot\ds\phi L\phi+2\Omega^2\etab\cdot\nablas\phi\ds\phi+\Omega^2\etab|\ds\phi|^2)\}\end{dcases}.
\end{align*}
\begin{remark}
The elliptic-transport system for $\omega$ is not needed because we renormalize the equations such that $\omega$ does not appear. However, the system for $\omegab$ is needed because it is crucial in estimating the top order derivatives of $\Omega\tr\chib$.
\end{remark}

Before the end of this section, we list the commutation formulas which are used for the estimates of derivatives\footnote{See Chapter 4 of \cite{Chr} for the first group. The second group can be derived directly by the definition of curvature.}.
\begin{lemma}\label{commutator}
Given integer $i$ and tangential tensorfield $\phi$. we have
\begin{align*}
[D,\nablas^i]\phi&=\sum_{j=1}^i\nablas^j(\Omega\chi)\cdot\nablas^{i-j}\phi,\\
[\Db,\nablas^i]\phi&=\sum_{j=1}^i\nablas^j(\Omega\chib)\cdot\nablas^{i-j}\phi,
\end{align*}
and
\begin{align*}
[\mathcal{D},\nablas^i]\phi&=\sum_{j=1}^i\nablas^{j-1}K\cdot\nablas^{i-j}\phi,\\
[^*\mathcal{D},\nablas^i]\phi&=\sum_{j=1}^i\nablas^{j-1}K\cdot\nablas^{i-j}\phi.
\end{align*}
Here we use ``$\cdot$'' to represent an arbitrary contraction with the coefficients by $\gs$ or $\epsilons$. In addition, if $\phi$ is a function, then when $i=1$, all commutators above are zero; when $i\ge2$, all $i$'s are replaced by $i-1$'s in above formulas.
\end{lemma}

\section{Statement of the existence theorem}\label{statementexistencetheorem}

\subsection{Formulation of the problem} We formula the problem we are going to consider again. We consider a double characteristic initial value problem of the Einstein-scalar field equations, where the initial data is given on two null cone, $C_{u_0}$ which is outgoing, and $\Cb_0$ which is incoming, intersecting on a sphere $S_{0,u_0}$. The data on $\Cb_0$ is spherically symmetric. The restriction of the optical function $u$ on $\Cb_0$ is chosen such that the level sets of $u$ on $\Cb_0$ are the spherical sections of the symmetry, and in addition $u=-r$ where $r$ is the area radius of the corresponding section.

The data on $\Cb_0$ then consists of $\phi$ up to a constant (or $\Lb\phi$) and $\Omega$, where $\phi$ is the scalar field and $\Lb=-\partial_u$ on $\Cb_0$ ($\chibh$ vanishes because of spherical symmetry).  Now denote
$$\psi=\psi(u)=|u|\Lb\phi\Big|_{\Cb_0},\ h=h(u)=\frac{|u|\tr\chi'}{2}\Big|_{\Cb_0},\ \Omega_0=\Omega_0(u)=\Omega\Big|_{\Cb_0}.$$
We require that these three functions are smooth for $u\in[u_0,0)$  and $0<h\le1$. In fact, from the Raychaudhuri equation along $\Cb_0$, and $\Omega\tr\chib=-\frac{2}{|u|}$ which follows by setting $u=-r$, $\Omega_0$ is determined by $\psi$ through
\begin{align}\label{Omega_0}\Db\log\Omega_0=\left(\omegab\Big|_{\Cb_0}=\right)-\frac{1}{2}\frac{\psi^2}{|u|},\end{align}
and from the null structure equation for $\Db(\Omega\tr\chi)$, $h$ is determined by $\Omega_0$ through
\begin{align}\label{Omega_0^2h}
\Db(\Omega_0^2h)=\frac{1}{|u|}\Omega_0^2(h-1).
\end{align}

The function $\Omega_0$ will play a very important role in the whole paper. It has two important properties stated in Lemma \ref{Omega_0to0} which will be used frequently: {\bf $\Omega_0$ is monotonically decreasing, and if the vertex of $\Cb_0$ is singular, then $\Omega_0\to0$ as $u\to0^-$}.

\subsection{Norms} We first introduce the following scale invariant norms relative to the double null foliation (for $2\le p\le\infty$ and $q=1,2$):
\begin{align*}
\|\xi\|_{\mathbb{L}^p(\ub,u)}=&\left(\int_{S_{\ub,u}}|u|^{-2}|\xi|^p\D\mu_{\gs}\right)^{\frac{1}{p}},\\
\|\xi\|_{\mathbb{H}^n(\ub,u)}=&\sum_{i=0}^n\|(|u|\nablas)^i\xi\|_{\L^2(\ub,u)},\\
\|\xi\|_{\mathbb{L}^q_{[u_1,u_2]}\mathbb{H}^n(\ub)}=& \left(\int_{u_1}^{u_2}|u|^{-1}\|\xi\|_{\H^n(\ub,u')}^q \D u'\right)^{\frac{1}{q}},\\
\|\xi\|_{\L^q_{\ub}\mathbb{H}^n(u)}=&\left(\int_{0}^{\delta}\delta^{-1}\|\xi\|_{\H^n(\ub',u)}^q\D \ub'\right)^{\frac{1}{q}}.\end{align*}
In addition, we define
\begin{align*}
\|\xi\|_{\L^\infty_{\ub}\L^q_{[u_1,u_2]}\mathbb{H}^n}=&\sup_{\ub'\in[0,\delta]}
\left(\int_{u_1}^{u_2}|u|^{-1}\|\xi\|_{\H^n(\ub',u')}^q\D u'\right)^{\frac{1}{q}},\\
\|\xi\|_{\L^\infty_{[u_1,u_2]}\L^q_{\ub}\mathbb{H}^n}=&\sup_{u'\in[u_1,u_2]}
\left(\int_{0}^{\delta}\|\xi\|_{\H^n(\ub',u')}^q\D \ub'\right)^{\frac{1}{q}}.
\end{align*}

We are going to define the norms of various geometric quantities which will be used in the proof of the existence theorem. Set a function
\begin{align*}
\mathscr{F}=\mathscr{F}(\delta,u_0,u_1)\ge1,
\end{align*}
for some $\delta>0$ and $u_1\in(u_0,0)$ and define
\begin{align*}
\mathscr{E}=\mathscr{E}(\delta,u_0,u_1):=&\max\left\{1, \mathscr{F}^{-1}\mathcal{A}^{-1}\||u_0|(|u_0|\nablas)L\phi\|_{\L^2_{[0,\delta]}\H^4(u_0)}\left|\log\frac{|u_1|}{|u_0|}\right|^{\frac{1}{2}}\right\}\ge1,\\
\mathscr{W}=\mathscr{W}(u_0,u_1):=&\max\left\{1,\logOmega\right\}\ge1
\end{align*}
where
\begin{align*}
\mathcal{A}=\mathcal{A}(\delta,u_0,u_1)
\end{align*}
is the bound of the initial data which will be defined precisely in the statement of Theorem \ref{existencetheorem}. Now we define
    \begin{align*}
    &\mathcal{O}(\ub,u;u_0,u_1)\\
    =&\left[\delta^{-1}|u|^2\Omega_0^2\mathscr{F}^{-2}\left\|\tr\chi'-\frac{2h}{|u|}\right\|_{\H^4(\ub,u)}\right]^{\frac{1}{2}}\\
    &+\delta^{-1}|u|^2\mathscr{F}^{-1}\left\|\Omega\tr\chib+\frac{2}{|u|},\Omega\chibh\right\|_{\H^4(\ub,u)}\\
        &+|u|\mathscr{F}^{-1}\left(\|\Omega\chih\|_{\H^4(\ub,u)}+\mathscr{W}^{-\frac{1}{2}}\|\omega\|_{\H^4(\ub,u)}\right)\\&+\delta^{-1}|u|^2\mathscr{F}^{-1}\mathscr{E}^{-1}\left(\|\eta\|_{\H^4(\ub,u)}+\mathscr{W}^{-\frac{1}{2}}\|\etab\|_{\H^4(\ub,u)}\right).
    \end{align*}

    \begin{align*}
    \mathcal{E}(\ub,u;u_0,u_1)=
     &\left[\delta^{-1}\mathscr{F}^{-1}|u|^{\frac{1}{4}}\||u|^{\frac{3}{4}}(|u|\Lb\phi-\psi)\|_{\L_{[u_0,u_1]}^2\H^4(\ub)}\right.\\
        &+\mathscr{F}^{-1}\||u|L\phi-\varphi\|_{\H^4(\ub,u)}\\&\left.+\delta^{-1}|u|^2\mathscr{F}^{-1}\mathscr{E}^{-1}\|\nablas\phi\|_{\H^4(\ub,u)}\right].
    \end{align*}

 \begin{align*}
   & \widetilde{\mathcal{O}}(\ub,u;u_0;u_1)=\\
    &    \left[\delta^{-1}|u|^2\Omega_0^2\mathscr{F}^{-2}\mathscr{E}^{-1}\|(|u|\nablas)\tr\chi'\|_{\H^4(\ub,u)}\right]^{\frac{1}{2}}\\&+\delta^{-1}|u|^2\mathscr{F}^{-1}\mathscr{E}^{-1}\mathscr{W}^{-\frac{1}{2}}\|(|u|\nablas)(\Omega\tr\chib)\|_{\H^4(\ub,u)}\\
   &+|u|\mathscr{F}^{-1}\mathscr{E}^{-1}\|\Omega\chih\|_{\L^2_{\ub}\H^5(u)}\\
   &+\delta^{-\frac{1}{2}}|u|^{\frac{3}{2}}\Omega_0\mathscr{F}^{-1}\mathscr{E}^{-1}\mathscr{W}^{-\frac{1}{2}}\|\etab\|_{\L^2_{\ub}\H^5(u)}\\
   &+\delta^{-\frac{1}{2}}|u|^{\frac{3}{2}}\Omega_0\mathscr{F}^{-1}\mathscr{E}^{-1}\|\eta\|_{\L_{\ub}^{2}\H^5(u)}\\
   &+\delta^{-\frac{1}{2}}\mathscr{F}^{-1}\mathscr{E}^{-1}\|\Omega_0|u|^{\frac{3}{2}}\eta\|_{\L_{[u_0,u]}^2\H^5(\ub)}\\
    &+\delta^{-1}|u|^{\frac{1}{2}}\Omega_0\mathscr{F}^{-1}\mathscr{E}^{-1}\mathscr{W}^{-\frac{1}{2}}\|\Omega_0^{-1}|u|^{\frac{3}{2}}(\Omega\chibh)\|_{\L_{[u_0,u]}^2\H^5(\ub)}\},
    \end{align*}

\begin{align*}
   & \widetilde{\mathcal{E}}(u;u_0;u_1)\\
    =&\mathscr{F}^{-1}\mathscr{E}^{-1}\left(\||u|(|u|\nablas)L\phi\|_{\L_{\ub}^2\H^4(u)}
    +\delta^{-\frac{1}{2}}\||u|^{\frac{3}{2}}\Omega_0(|u|\nablas)\nablas\phi\|_{\L^\infty_{\ub}\L_{[u_0,u]}^2\H^4}\right)\\
    &+\Omega_0\delta^{-\frac{1}{2}}\mathscr{F}^{-1}\mathscr{E}^{-1}\left(\||u|^{\frac{3}{2}}(|u|\nablas)\nablas\phi\|_{\L_{\ub}^2\H^4(u)}
    +\delta^{-\frac{1}{2}}\|\Omega_0^{-1}|u|^{2}(|u|\nablas)\Lb\phi\|_{\L^\infty_{\ub}\L_{[u_0,u]}^2\H^4(\delta)}\right),
 \end{align*}

\begin{align*}
    &\mathcal{R}(u;u_0;u_1)\\
    = &  \mathscr{F}^{-1}\mathscr{E}^{-1}\left(\||u|^2(\Omega\beta-L\phi\ds\phi)\|_{\L_{\ub}^2\H^4(u)} 
    +\delta^{-\frac{1}{2}}\left\||u|^{\frac{5}{2}}\Omega_0(K-\frac{1}{|u|^2},\sigmac)\right\|_{\L^\infty_{\ub}\L_{[u_0,u]}^2\H^4}\right)\\
    &+\left[\Omega_0\delta^{-1}\mathscr{F}^{-\frac{3}{2}}\mathscr{E}^{-1}\right.\\
    &\ \ \ \left.\times\left(\||u|^{3}(K-\frac{1}{|u|^2},\sigmac)\|_{\L_{\ub}^2\H^4(u)}    +\delta^{-\frac{1}{2}}\|\Omega_0^{-1}|u|^{\frac{7}{2}}(\Omega\betab+\Lb\phi\ds\phi)\|_{\L^\infty_{\ub}\L_{[u_0,u]}^2\H^4}\right)\right]^{\frac{2}{3}}.\\
 \end{align*}
 \begin{remark}
 Some remarks should be made about these norms. First, in the definitions above, $\Omega_0$ appearing outside the norm symbol ``$\|\cdot\|$'' means its value at $u$, i.e., $\Omega_0=\Omega_0(u)$. {\bf In the rest of the paper, $\Omega_0$ outside ``$\|\cdot\|$'' should be understood in the same way. On the other hand, if $\Omega_0$ or $\Omega_0^{-1}$ appears inside the norm $\|\cdot\|_{\L^2_{[u_0,u]}}$, then it is a factor of the integrand.} All factors inside the norm $\|\cdot\|_{\L^2_{[u_0,u]}}$ cannot be taken out directly. Second, the norm
 \begin{align*}
 \|(|u|\nablas)(\Omega\tr\chi)\|_{\H^4(\ub,u)}
 \end{align*}
 is a norm about up to the \engordnumber{5} order derivatives of $\Omega\tr\chib$ except itself. For some of the geometric quantities, the derivatives obey better estimates as compared to the quantities themselves.
 \end{remark}
 Finally, we define
\begin{align*}\mathcal{O}=\mathcal{O}(u_0,u_1)&=\sup_{0\le\ub\le\delta\atop u_0\le u\le u_1}\mathcal{O}(\ub,u;u_0,u_1),\\
\mathcal{E}=\mathcal{E}(u_0,u_1)&=\sup_{0\le\ub\le\delta\atop u_0\le u\le u_1}\mathcal{E}(\ub,u;u_0,u_1),\\
\widetilde{\mathcal{O}}=\widetilde{\mathcal{O}}(u_0,u_1)&=\sup_{0\le\ub\le\delta\atop u_0\le u\le u_1}\widetilde{\mathcal{O}}(\ub,u;u_0,u_1),\\
\widetilde{\mathcal{E}}=\widetilde{\mathcal{E}}(u_0,u_1)&=\sup_{u_0\le u\le u_1}\widetilde{\mathcal{E}}(u;u_0,u_1),\\
\mathcal{R}=\mathcal{R}(u_0,u_1)&=\sup_{u_0\le u\le u_1}\mathcal{R}(u;u_0,u_1).
\end{align*}

 \begin{theorem}[Existence Theorem]\label{existencetheorem}
There exists a universal constant $C_0\ge1$ such that the following statement is true. Let $C\ge C_0$, $\delta>0$, $u_0$ and $u_1$ be four numbers such that $u_0<u_1<0$. Suppose that the smooth initial data given on $C_{u_0}\bigcup \Cb_0$ is as described above. The data on $\Cb_0$ is spherically symmetric and $\Omega_0(u_0)\le1$ and the data on $C_{u_0}$ satisfies
 \begin{equation}\label{def-A}
 \begin{split}
 \mathcal{A}:=&\max\left\{1,\sup_{u_0\le u\le u_1}\left(\mathscr{F}^{-1}|\varphi(u)|\right),\right.\\
 &\left.\mathscr{F}^{-1}|u_0|\sup_{0\le\ub\le\delta}\left(\|\Omega\chih\|_{\H^7(\ub,u_0)}+\|\omega\|_{\H^5(\ub,u_0)}+\|L\phi\|_{\H^5(\ub,u_0)}\right)\right\}<+\infty.
 \end{split}
 \end{equation}
 Then if the following {\bf{\emph{smallness conditions}}} hold:
\begin{equation}\label{smallness}
\begin{split}
\Omega_0^2(u_0)\delta|u_1|^{-1}\mathscr{E}^2\mathscr{W}&\le 1,\\
 C^2\delta|u_1|^{-1}\mathscr{F}\mathscr{W}\mathcal{A}&\le 1,
\end{split}
\end{equation}
and the following {\bf \emph{auxiliary condition}} holds:
\begin{equation}\label{auxiliary}
\begin{split}
\Omega_0^2(u_0)\Omega_0^{-2}(u_1)\delta|u_1|^{-1}\mathscr{F}\mathcal{A}&\le 1.
\end{split}
\end{equation}
Then the smooth solution of the Einstein-scalar field equations exists in the region $0\le\ub\le\delta$, $u_0\le u\le u_1$, and the following estimates hold:
$$\mathcal{O}, \mathcal{E}, \widetilde{\mathcal{O}}, \widetilde{\mathcal{E}}, \mathcal{R}\lesssim \mathcal{A}.$$
 \end{theorem}
 \begin{remark}
{\bf Here and in the rest of the paper, the notation $A\lesssim B$ refers to $A\le cB$ where $c$ is some universal constant. }
\end{remark}
\begin{remark}
Using the argument in Chapter 2 of \cite{Chr}, under the assumptions of Theorem \ref{existencetheorem}, we can obtain the initial bound
\begin{align*}
\sup_{0\le\ub\le\delta}&(\mathcal{O}(\ub,u_0;u_0,u_1)+\mathcal{E}(\ub,u_0;u_0,u_1)+\widetilde{\mathcal{O}}(\ub,u_0;u_0,u_1))\\
&+\widetilde{\mathcal{E}}(u_0;u_0,u_1)+\mathcal{R}(u_0;u_0,u_1)\lesssim\mathcal{A}(\delta,u_0,u_1).
\end{align*}
We will omit the proof of this. This is however the bound we actually use in the proof of Theorem \ref{existencetheorem}.
\end{remark}

\section{Proof of the existence theorem --- Theorem \ref{existencetheorem}}\label{APrioriEstimate}

\subsection{Bootstrap assumptions}

As in \cite{Chr}, the main step to prove the existence is to establish a priori estimates. Once we have the a priori estimates, the construction of the solution in double null foliation then follows using a bootstrap argument. We should remark that although we are dealing with the Einstein-scalar field equations but not the vacuum Einstein equations, the argument in \cite{Chr} can also be modified in the current case.

We begin the proof of the a priori estimates by assuming the smooth solution exists for $0\le\ub\le\delta, u_0\le u\le u_1$ and it holds for some $C\ge C_0$ the following {\bf{\emph{bootstrap assumptions}}}:
\begin{equation}\label{bootstrap}\mathcal{O}, \mathcal{E}, \widetilde{\mathcal{O}}, \widetilde{\mathcal{E}}, \mathcal{R}\lesssim C^{\frac{1}{4}}\mathcal{A}.\end{equation} 
Then it suffices to prove that under this bootstrap assumptions \eqref{bootstrap}, the smallness conditions \eqref{smallness} and the auxiliary condition \eqref{auxiliary}, it holds
$$\mathcal{O}, \mathcal{E}, \widetilde{\mathcal{O}}, \widetilde{\mathcal{E}}, \mathcal{R}\lesssim \mathcal{A}.$$
If $C_0$ is sufficiently large, then this is an improvement over the bootstrap assumptions \eqref{bootstrap}. We will prove this in the rest of this section.

\subsection{Preliminary lemmas}
We start by proving some geometric lemmas which will be frequently used in the course of the proof. First of all, we assume
\begin{align}\label{bootstrap-geometric}
\left|\Omega\tr\chi-\frac{2\Omega^2h}{|u|}\right|\lesssim C^{\frac{3}{4}}\delta|u|^{-2}\mathscr{F}^2\mathcal{A}^2,\ |\Omega\chih|\lesssim C^{\frac{1}{2}}|u|^{-1}\mathscr{F}\mathcal{A},\ \frac{1}{4}\Omega_0\le\Omega\le 4\Omega_0.
\end{align}

Let $\Lambda(\ub;u)$ and $\lambda(\ub;u)$ be the larger and smaller eigenvalues of the metric $\gs|_{S_{\ub,u}}$ with respect to $\gs|_{S_{0,u}}$, which is the standard metric on the sphere with radius $|u|$. Define
\begin{align*}
\mu(\ub;u)=\sqrt{\lambda(\ub;u)\Lambda(\ub;u)},\ \nu(\ub;u)=\sqrt{\frac{\Lambda(\ub;u)}{\lambda(\ub;u)}}.
\end{align*}
From the proof of Lemma 5.3 in \cite{Chr}, we have
\begin{align*}
\mu(\ub;u)=\exp\left(\int_0^{\ub}\Omega\tr\chi(\ub',u)\D\ub'\right),\ \nu(\ub;u)\le\exp\left(2\int_0^{\ub}|\Omega\chih(\ub',u)|\D\ub'\right).
\end{align*}
From \eqref{bootstrap-geometric} and the smallness conditions \eqref{smallness}, we have 
\begin{align}\label{estimate-Omegatrchi-geometric}|\Omega\tr\chi|\lesssim|u|^{-1}+C^{-1}|u|^{-1}\mathscr{F}\mathcal{A}\lesssim|u|^{-1}\mathscr{F}\mathcal{A}.
\end{align}
 Therefore, we have, using \eqref{bootstrap-geometric} and \eqref{smallness} again, 
\begin{align*}
|\mu(\ub;u)-1|\lesssim\delta|u|^{-1}\mathscr{F}\mathcal{A}\lesssim C^{-1},\ 1\le\nu(\ub;u)\lesssim C^{\frac{1}{2}}|u|^{-1}\mathscr{F}\mathcal{A} \lesssim C^{-1},
\end{align*}
which implies
\begin{align}\label{eigenvalues}
1-cC_0^{-1}\le\lambda(\ub;u)\le\Lambda(\ub;u)\le 1+cC_0^{-1}.
\end{align}
By a similar argument of the proof of Lemma 5.4 in \cite{Chr}, we have
\begin{align}\label{isoperimetric}
I(\ub,u)\le (1+c C_0^{-1}) I(0,u)\le\frac{1}{2\pi}(1+cC_0^{-1})\lesssim 1
\end{align}
where $I(\ub,u)$ is the isoperimetric constant of the sphere $(S_{\ub,u},\gs)$ and $C_0$ is sufficiently large. Also from \eqref{eigenvalues}, we have
\begin{align*}
1-cC_0^{-1}\le\frac{\mathrm{Area}(S_{\ub,u})}{\mathrm{Area}(S_{0,u})}\le 1+cC_0^{-1},
\end{align*}
which implies if $C_0$ is sufficiently large, 
\begin{align}\label{area}
\frac{1}{2}|u|^2\le\frac{1}{4\pi}\mathrm{Area}(S_{\ub,u})\le 2|u|^2.
\end{align}
Now by Lemma 5.1 in \cite{Chr}, the Sobolev inequalities, we have the following form of the Sobolev inequalities:
\begin{lemma}\label{Sobolev}
Given a tangential tensorfield $\theta$, we have for $q\in(2,+\infty)$,
\begin{align*}
\|\theta\|_{\L^q(\ub,u)}&\lesssim_q \|\theta\|_{\H^1(\ub,u)},\\
\|\theta\|_{\L^\infty(\ub,u)}&\lesssim \|(|u|\nablas)\theta\|_{\L^4(\ub,u)}+ \|\theta\|_{\L^4(\ub,u)}\lesssim\|\theta\|_{\H^2(\ub,u)}.
\end{align*}

\end{lemma}

Here $A\lesssim_q B$ means $A\le c_qB$ where $c_q$ is a constant depending only on $q$.

By the classical H\"older inequality and the Sobolev inequalities, we have
\begin{lemma}\label{Holder}
Given tangential tensorfields $\theta_1,\cdots,\theta_n$, for $i\ge2$, we have
\begin{align*}
\|\theta_1\cdots\theta_n\|_{\H^i(\ub,u)}\lesssim_n\|\theta_1\|_{\H^i(\ub,u)}\cdots\|\theta_n\|_{\H^i(\ub,u)}.
\end{align*}
\end{lemma}

We then introduce the following Gronwall type estimates which are need to estimate using the transport equations:
\begin{lemma}\label{Gronwall}
For an $s$-covariant tengential tensorfield $\theta$, and any real number $\nu$, we have
\begin{align*}
\|\theta\|_{\L^2(\ub,u)}
\lesssim&_s \left(\|\theta|_{\L^2(0,u)}+\delta\|D\theta\|_{\L_{\ub}^1\L^2(u)}\right).\\
\||u|^{s+\nu}\theta\|_{\L^q(\ub,u)}\lesssim&_{s,\nu}\left(\||u|^{s+\nu}\theta\|_{\L^q(\ub,u_0)}+\||u|^{s+\nu+1}(\Db\theta+\frac{\nu}{2}\Omega\tr\chib\theta)\|_{\L^1_{[u_0,u]}\L^2(\ub)}\right),
\end{align*}\end{lemma}
\begin{proof}
The first inequality is similar to those in Lemma 4.4 and 4.6 in \cite{Chr}. Using the argument in \cite{Chr}, the first inequality holds if for any $\ub',\ub\in[0,\delta]$,
\begin{align*}
\int_{\ub'}^{\ub}\Omega\tr\chi\D \ub'', \int_{\ub'}^{\ub}|\Omega\chih|\D \ub''\lesssim1.
\end{align*}
This is true because both of them are bounded by, 
\begin{align*}
\int_0^\delta C^{\frac{1}{2}}|u|^{-1}\mathscr{F}\mathcal{A}\D\ub'\lesssim C^{\frac{1}{2}}\delta|u|^{-1}\mathscr{F}\mathcal{A}\lesssim C^{-1},
\end{align*}
where we use \eqref{bootstrap-geometric} and the smallness conditions \eqref{smallness}.

The second inequality is similar to those in Lemma 4.5 and 4.7 in \cite{Chr}. It holds if for any $u',u\in[u_0,u_1]$,
\begin{align*}
\int_{u'}^u\left(\Omega\tr\chib+\frac{2}{|u|}\right)\D u'', \int_{u'}^u|\Omega\chibh|\D u''\lesssim1.
\end{align*}
This is true because, from the Sobolev ienqualities, Lemma \ref{Sobolev},
\begin{align*}
\left|\Omega\tr\chib+\frac{2}{|u|}\right|, |\Omega\chibh|\lesssim C^{\frac{1}{4}}\delta|u|^{-2}\mathscr{F}\mathcal{A},
\end{align*}
and then both of the above integrals are bounded by
\begin{align*}
\int_{u'}^{u}C^{\frac{1}{4}}\delta|u|^{-2}\mathscr{F}\mathcal{A}\D u''\lesssim C^{\frac{1}{4}}\delta|u|^{-1}\mathscr{F}\mathcal{A}\lesssim C^{-1}.
\end{align*}
\end{proof}
In addition, we have the Gronwall type estimates for the augular derivatives.
\begin{lemma}\label{Gronwallderivative}
For an $s$-covariant tengential tensorfield $\theta$, and any real number $\nu$, we have, for $n\le4$ (if $s=0$, then $n\le5$)
\begin{align*}
\|\theta\|_{\H^n(\ub,u)}
\lesssim&_{s} \|\theta\|_{\H^n(0,u)}+\delta\|D\theta\|_{\L_{\ub}^1\H^n(u)}.\\
\||u|^{s+\nu}\theta\|_{\H^n(\ub,u)}\lesssim&_{s,\nu}\left(\||u|^{s+\nu}\theta\|_{\H^n(\ub,u_0)}+\||u|^{s+\nu+1}(\Db\theta+\frac{\nu}{2}\Omega\tr\chib\theta)\|_{\L^1_{[u_0,u]}\H^n(\ub)}\right),
\end{align*}
\end{lemma}

\begin{proof}
For the first inequality, we apply the above Gronwall estimate, Lemma \ref{Gronwall} to the following equation:
\begin{align*}
D\nablas^i\theta=\nablas^iD\theta+[D,\nablas^i]\theta=\nablas^iD\theta+\sum_{j=1}^i\nablas^j(\Omega\chi)\nablas^{i-j}\theta.
\end{align*}
By H\"older inequality and Sobolev inequalies,
\begin{align*}
\|\theta\|_{\H^n(\ub,u)}
\lesssim \|\theta\|_{\H^n(0,u)}+\delta\|D\theta\|_{\L_{\ub}^1\H^n(u)}+\delta\|\Omega\chi\|_{\L^\infty_{\ub}\H^n(\ub)}\|\theta\|_{\L^\infty_{\ub}\H^n(u)},
\end{align*}
where we have used the following H\"older type inequality:
\begin{align}\label{Holder-ub}
\|\cdot\|_{\L^1_{\ub}}\lesssim\|\cdot\|_{\L^2_{\ub}}\lesssim\|\cdot\|_{\L^\infty_{\ub}}.
\end{align}
By \eqref{estimate-Omegatrchi-geometric}, the bootstrap assumptions \eqref{bootstrap} on $\Omega\chih$ and the smallness conditions \eqref{smallness}, we have, for $n\le 4$,
\begin{align*}
\|\theta\|_{\H^n(\ub,u)}
\lesssim& \|\theta\|_{\H^n(0,u)}+\delta\|D\theta\|_{\L_{\ub}^1\H^n(u)}+C^{\frac{1}{4}}\delta|u|^{-1}\mathscr{F}\mathcal{A}\|\theta\|_{\L^\infty_{\ub}\H^n(u)}\\
\lesssim& \|\theta\|_{\H^n(0,u)}+\delta\|D\theta\|_{\L_{\ub}^1\H^n(u)}+C^{-1}\|\theta\|_{\L^\infty_{\ub}\H^n(u)},
\end{align*}
Then if $C_0$ is sufficiently large, the last term on the right hand side can be absorbed by the left hand side after taking supremum on $\ub$. We then obtain the desired inequality. If $s=0$, that is, when $\theta$ is a function, then $D\nablas\theta=\nablas D\theta$. Therefore $[D,\nablas^5]\theta$ contains only four derivatives of $\Omega\chi$.

The second inequality can be obtained similarly.
\end{proof}

We then introduce an estimate for $\Omega$ and its derivatives:

\begin{lemma}
\begin{equation}\label{estimate-Omega}
\frac{1}{2}\Omega_0\leq\Omega\leq2\Omega_0,
\end{equation}
and
\begin{equation*}
\|\Omega^s\|_{\H^4(\ub,u)}\lesssim_s \Omega_0^s.
\end{equation*}
Moreover, we have
\begin{equation}\label{estimate-Omega-improve}
\|\Omega^s-\Omega_0^s\|_{\H^4(\ub,u)}\lesssim_s C^{\frac{1}{4}}\Omega_0^s\delta|u|^{-1}\mathscr{F}\mathscr{W}^{\frac{1}{2}}\mathcal{A}.
\end{equation}
\end{lemma}

\begin{proof}
Since $D\log{\Omega}=\omega$, we then write $D(\log\Omega-\log\Omega_0)=\omega$. Applying Lemma \ref{Gronwall}, 
\begin{equation*}
\|\log{\Omega}-\log{\Omega_0}\|_{\H^4(\ub,u)}\lesssim\delta\|\omega\|_{\L_{\ub}^1\H^4(u)}
\lesssim C^{\frac{1}{4}} \delta|u|^{-1}\mathscr{F}\mathscr{W}^{\frac{1}{2}}\mathcal{A}\lesssim C^{-1}.
\end{equation*}
In particular, if $C_0$ is sufficiently large, we have, by the Sobolev inequalities,
\begin{equation*}
\left|\log\frac{\Omega}{\Omega_0}\right|\le\log 2
\end{equation*}
which implies the first desired inequality. For the second, if $s\ne0$, then
\begin{align*}
&\|\Omega^s\|_{\H^4(\ub,u)}
\lesssim_s\|\Omega^s\|_{\L^{\infty}(\ub,u)}(1+\|\log{\Omega}-\log\Omega_0\|_{\H^4(\ub,u)}^4)\lesssim \Omega_0^s.
\end{align*}
The last estimate \eqref{estimate-Omega-improve} follows similarly from the equation $D\Omega^s=s\Omega^{s}\omega$ and the bootstrap assumptions \eqref{bootstrap}.
\end{proof}
This lemma implies that for a tangential tensorfield $\theta$, 
\begin{align}\label{estimate-Omegaequ}
\|\Omega^s\theta\|_{\H^n(\ub,u)}\lesssim \Omega_0^s\|\theta\|_{\H^n(\ub,u)}
\end{align}
for $n\le 4$. This inequality says that we do not need to worry about the appearance of $\Omega$ if we take less than four derivatives.

This lemma also improves the bound for $\Omega$ in \eqref{bootstrap-geometric}. The bounds for $\Omega\tr\chi$ and $\Omega\chih$ can also be improved by using the bootstrap assumptions \eqref{bootstrap} on $\mathcal{O}$ and the Sobolev inequalities, Lemma \ref{Sobolev}. The constants $C^{\frac{3}{4}}, C^{\frac{1}{2}}$ in \eqref{bootstrap-geometric} can be improved to $C^{\frac{1}{2}}, C^{\frac{1}{4}}$ respectively if $C_0$ is sufficiently large. This implies that all conclusions in this subsection hold without assuming \eqref{bootstrap-geometric}.

Finally, we introduce the elliptic estimates.
\begin{lemma}\label{elliptic}
Suppose that
\begin{align*}\|K\|_{\H^2(\ub,u)}\lesssim|u|^{-2}.\end{align*}
Now assume that $\theta$ is a tangential symmetric trace-free $(0,2)$ type tensorfield with
$$\divs\theta=f,$$ where $f$ is a tangential one-form. Then
\begin{align*}
\|\theta\|_{\H^5(\ub,u)}\lesssim\||u|f\|_{\H^4(\ub,u)}+(1+\||u|^2K\|_{\H^3(\ub,u)})\|\theta\|_{\H^2(\ub,u)}.
\end{align*}
Assume that $\xi$ is a tangential $(0,1)$ type tensorfield with
\begin{align*}
\divs\xi=f,
\curls\xi=g
\end{align*} where $f,g$ are functions. Then
\begin{align*}
\|\xi\|_{\H^5(\ub,u)}\lesssim\||u|(f,g)\|_{\H^4(\ub,u)}+(1+\||u|^2K\|_{\H^3(\ub,u)})\|\xi\|_{\H^4(\ub,u)}.
\end{align*}
Assume that $\phi$ is a function with
\begin{equation*}
\Deltas\phi=f.
\end{equation*}
We have
\begin{align*}
\|(|u|\nablas)\phi\|_{\H^4(\ub,u)} \lesssim\||u|^2f\|_{\H^3(\ub,u)}.
\end{align*}
\end{lemma}
\begin{proof}
The first and second conclusions are obtained by repeatedly use the argument of Section 7.3 in \cite{Chr}. For the last conclusion, we apply the second conclusion for the one form $\nablas\phi$ for up to the fourth order derivatives. Therefore $\|K\|_{\H^3(\ub,u)}$ does not come in and
\begin{align*}
\|(|u|\nablas)\phi\|_{\H^4(\ub,u)} \lesssim\||u|^2f\|_{\H^3(\ub,u)}+\|(|u|\nablas)\phi\|_{\L^2(\ub,u)}\end{align*}
Integration by parts gives
$$\int_{S_{\ub,u}}|\nablas\phi|^2\D\mu_{\gs}\lesssim \int_{S_{\ub,u}}|f|^2\D\mu_{\gs}.$$
These two inequalities give the desired conclusion.
\end{proof}

We end this section by the following estimates for $\Omega\tr\chi$, $\Omega\tr\chib$ and $\Lb\phi$. The zeroth order bounds of these quantities are worse than their derivatives.
\begin{lemma}
Under the assumptions of the Theorem \ref{existencetheorem} and the bootstrap assumptions \eqref{bootstrap}, we have
\begin{align*}
|u|\|\Omega\tr\chi\|_{\H^4(\ub,u)}\lesssim\mathscr{F}\mathcal{A},\ |u|\|\Omega\tr\chib\|_{\H^4(\ub,u)}\lesssim1,\ \||u|\Lb\phi\|_{\L^2_{[u_0,u]}\H^4(\ub)}\lesssim\mathscr{W}^{\frac{1}{2}}.
\end{align*}\end{lemma}
\begin{proof}
From the bootstrap assumptions \eqref{bootstrap} and the smallness conditions \eqref{smallness}:
\begin{align}\label{estimate-trchi1}
|u|\|\Omega\tr\chi\|_{\H^4(\ub,u)}\lesssim&|u|\|\Omega_0^2\tr\chi'\|_{\H^4(\ub,u)}\lesssim1+C^{\frac{1}{2}}\delta|u|^{-1}\mathscr{F}^2\mathcal{A}^2\lesssim(1+C^{-1})\mathscr{F}\mathcal{A}\lesssim\mathscr{F}\mathcal{A},\\
\label{estimate-trchib1}|u|\|\Omega\tr\chib\|_{\H^4(\ub,u)}\lesssim& 1+C^{\frac{1}{4}}\delta|u|^{-1}\mathscr{F}\mathcal{A}\lesssim(1+C^{-1})\lesssim 1,\\
\begin{split}\label{estimate-Lbphi1}\||u|\Lb\phi\|_{\L^2_{[u_0,u]}\H^4(\ub)}\lesssim&\left(\int_{u_0}^u\frac{|\psi|^2}{|u'|}\D u'\right)^{\frac{1}{2}}+|u|^{-\frac{1}{4}}\||u|^{\frac{3}{4}}(|u|\Lb\phi-\psi)\|_{\L^2_{[u_0,u]}\H^4(\ub)}\\
\lesssim&\logOmega^{\frac{1}{2}}+C^{\frac{1}{4}}\delta|u|^{-1}\mathscr{F}\mathcal{A}\lesssim\mathscr{W}^{\frac{1}{2}}.
\end{split}
\end{align}

\end{proof}

We have finished stating and proving the preliminary lemmas and begin to estimate the scalar field, the connection coefficients and the curvature components. {\bf From now on, we will frequently use the bootstrap assumptions \eqref{bootstrap} and the smallness conditions \eqref{smallness} and we may not point this out in the text.}

\subsection{Estimates for $\mathcal{E}$}
We first estimate the lower order derivatives of the derivative of the wave function $\phi$.
\begin{proposition}\label{estimate-E}
Under the assumptions of Theorem \ref{existencetheorem} and the bootstrap assumptions \eqref{bootstrap}, we have $$\mathcal{E}\lesssim \mathcal{A}+\widetilde{\mathcal{E}}.$$
\end{proposition}

\begin{proof}  {\bf \underline{Estimate for $L\phi$}:} The estimate for $L\phi$ makes use of the following equation, which is rewritten in terms of $L\phi-\varphi/|u|$:
\begin{equation}\label{equ-DbLphi}
\begin{split}
&\Db (L\phi-\varphi/|u|)+\frac{1}{2}\Omega\tr\chib (L\phi-\varphi/|u|)=\Omega^2\Deltas\phi+2\Omega^2(\eta,\ds\phi)\\
&-\frac{1}{2}(\Omega\tr\chib+\frac{2}{|u|})\frac{\varphi}{|u|}-\frac{1}{2}\Omega\tr\chi(\Lb\phi-\psi/|u|)
-\frac{1}{2}(\Omega\tr\chi-\frac{2\Omega^2h}{|u|})\frac{\psi}{|u|}-(\Omega^2-\Omega_0^2)\frac{h\psi}{|u|^2},\\
\end{split}
\end{equation}
To estimate the $\H^4(\ub,u)$ norm of $L\phi-\varphi/|u|$, we need to estimate the right hand side in $\||u|^2\cdot\|_{\L^1_{[u_0,u]}\H^4(\ub)}$. We estimate them term by term.

 The \engordnumber{1} term on the right hand side is estimated by
\begin{align*}\lesssim \Omega_0(u_0)|u|^{-\frac{1}{2}}\|\Omega_0|u|^{\frac{3}{2}}(|u|\nablas)\nablas\phi\|_{\L^2_{[u_0,u]}\H^4(\ub)}\lesssim\mathscr{F}\widetilde{\mathcal{E}},\end{align*}
where we have used a frequently used the following H\"older type inequality for $s>0$:
\begin{align}\label{Holder-u}
\||u|^{-s}\cdot\|_{\L^1_{[u_0,u]}}\lesssim\||u|^{-s}\|_{\L^2_{[u_0,u]}}\|\cdot\|_{\L^2_{[u_0,u]}}\lesssim|u|^{-s}\|\cdot\|_{\L^2_{[u_0,u]}}.
\end{align}
 The \engordnumber{2} term on the right hand side is estimated by
\begin{align*}\lesssim\Omega_0^2(u_0)|u|^{-2}\||u|^2\eta\|_{\L^\infty_{[u_0,u]}\H^4(\ub)}\||u|^2\nablas\phi\|_{\L^\infty_{[u_0,u]}\H^4(\ub)}
\lesssim\Omega_0^2(u_0)C^{\frac{1}{2}}\delta^2|u|^{-2}\mathscr{F}^2\mathscr{E}^2\mathcal{A}^2\lesssim C^{-1}\mathscr{F}\mathcal{A}.\end{align*}
 The \engordnumber{3} term on the right hand side is estimated by
\begin{align*}\lesssim|u|^{-1}\left(\sup_{u_0\le u'\le u_1}|\varphi(u')|\right)\left\||u|^2\left(\Omega\tr\chib+\frac{2}{|u|}\right)\right\|_{\L^\infty_{[u_0,u]}\H^4(\ub)}
\lesssim C^{\frac{1}{4}}\delta|u|^{-1}\mathscr{F}^2\mathcal{A}\lesssim C^{-1}\mathscr{F}\mathcal{A}.\end{align*}
 The \engordnumber{4} term on the right hand side is estimated by (using \eqref{estimate-trchi1})
\begin{align*}\lesssim&|u|^{-1}\||u|\Omega\tr\chi\|_{\L^\infty_{[u_0,u]}\H^4(\ub)}|u|^{\frac{1}{4}}\||u|^{\frac{3}{4}}(|u|\Lb\phi-\psi)\|_{\L^2_{[u_0,u]}\H^4(\ub)}\\
\lesssim&C^{\frac{1}{4}}\delta|u|^{-1}\mathscr{F}\mathcal{A}\cdot\mathscr{F}\mathcal{A}\lesssim C^{-1}\mathscr{F}\mathcal{A}.\end{align*}
The \engordnumber{5} term on the right hand side is estimated by
\begin{align*}\lesssim&|u|^{-1}\left\||u|^2\Omega_0^2\left(\tr\chi'-\frac{2h}{|u|^2}\right)\right\|_{\L^\infty_{[u_0,u]}\H^4(\ub)}\left(\int_{u_0}^u\frac{|\psi|^2}{|u'|}\D u'\right)^{\frac{1}{2}}\\
\lesssim&C^{\frac{1}{2}}\delta|u|^{-1}\mathscr{F}^2\left|\log\frac{\Omega(u_1)}{\Omega(u_0)}\right|^{\frac{1}{2}}\mathcal{A}^2\lesssim C^{-1}\mathscr{F}\mathcal{A}.\end{align*}
At last, the \engordnumber{6} term on the right hand side is estimated by, using \eqref{estimate-Omega-improve},
\begin{align*}\lesssim&|u|^{-1}\left\||u|(\Omega^2-\Omega_0^2)\right\|_{\L^\infty_{[u_0,u]}\H^4(\ub)}\left(\int_{u_0}^u\frac{|\psi|^2}{|u'|}\D u'\right)^{\frac{1}{2}}\\
\lesssim&C^{\frac{1}{4}}\delta|u|^{-1}\mathscr{F}\left|\log\frac{\Omega(u_1)}{\Omega(u_0)}\right|^{\frac{1}{2}}\mathscr{W}^{\frac{1}{2}}\mathcal{A}\lesssim C^{-1}\mathscr{F}\mathcal{A}.\end{align*}

 Summing up the above all estimates gives
\begin{equation*}\||u|L\phi-\varphi\|_{\H^4(\ub,u)}\lesssim \||u_0|L\phi-\varphi\|_{\H^4(\ub,u_0)}+\mathscr{F}(\mathcal{A}+\widetilde{\mathcal{E}})\lesssim \mathscr{F}(\mathcal{A}+\widetilde{\mathcal{E}}).\end{equation*}
This also implies
\begin{equation}\label{estimate-Lphi}\||u|L\phi\|_{\H^4(\ub,u)}\lesssim \mathscr{F}(\mathcal{A}+\widetilde{\mathcal{E}}).\end{equation}
\begin{remark} $\widetilde{\mathcal{E}}$ appears in the above inequality because we have not done anything to estimate the \engordnumber{1} term on the right hand side but simply use the definition of $\widetilde{\mathcal{E}}$. This term is also a \emph{borderline term} although it is linear. But we will see from the proof that $\widetilde{\mathcal{E}}$ can be controlled without knowing \eqref{estimate-Lphi}. 
\end{remark}
{\bf \underline{Estimate for $\nablas\phi$}:} We use the equation \begin{equation}\label{equ-Dnablasphi}
D\nablas\phi=\nablas L\phi.
\end{equation} 
From this equation, we have
\begin{equation}\label{estimate-nablasphi}
\begin{split}
\|\nablas\phi\|_{\H^4(\ub,u)}&\lesssim\delta \|\nablas L\phi\|_{\L^1_{\ub}\H^4(u)}\lesssim  \delta|u|^{-2}\||u|(|u|\nablas)L\phi\|_{\L_{\ub}^2\H^4(u)}\lesssim\delta|u|^{-2}\mathscr{F}\mathscr{E}\widetilde{\mathcal{E}}.
\end{split}
\end{equation}

{\bf \underline{Estimate for $\Lb\phi$}:} To estimate $\Lb\phi$, we use the equation for $D\Lb\phi$:
\begin{equation}\label{equ-Lbphi}D(\Lb\phi-\psi/|u|)+\frac{1}{2}\Omega\tr\chi\Lb\phi=\Omega^2\Deltas\phi+2\Omega^2(\etab,\nablas\phi)-\frac{1}{2}\Omega\tr\chib L\phi.\end{equation}
To estimate $|u|\Lb\phi-\psi$ in $\L^2_{[u_0,u]}\H^4(\ub)$, we first estimate the right hand side in $\delta\|\cdot\|_{\L^1_{\ub}\H^4(u)}$ and then integrate over $[u_0,u]$. We remark that $\Lb\phi$ cannot be estimated in $\L^\infty_{u}$ because $\psi$ can only be estimated in $\L^2_u$. Now the right hand side can be estimated by, with the last term being estimated using \eqref{estimate-trchib1} and \eqref{estimate-Lphi},
\begin{align*}
\delta\|\text{RHS}\|_{\L^1_{\ub}{\H^4(u)}}&\lesssim \Omega_0\delta^{\frac{3}{2}}|u|^{-\frac{5}{2}}\mathscr{F}\mathscr{E}\widetilde{\mathcal{E}}+\Omega_0^2C^{\frac{1}{2}}\delta^3|u|^{-4}\mathscr{F}^2\mathscr{E}^2\mathscr{W}^{\frac{1}{2}}\mathcal{A}^2+\delta|u|^{-2}\mathscr{F}(\mathcal{A}+\widetilde{\mathcal{E}})\\
&\lesssim \delta|u|^{-2}\mathscr{F}(\mathcal{A}+\widetilde{\mathcal{E}}).
\end{align*}
 Multiplying the square of the above estimate by $|u|^{\frac{5}{2}}$ and then integrating over $[u_0,u]$, we have
\begin{equation}\label{RHS-Lbphi}
\left(\int_{u_0}^u|u'|^{\frac{5}{2}}\delta\|\text{RHS}\|^2_{\L^1_{\ub}{\H^4(u')}}\D u'\right)^{\frac{1}{2}}\lesssim \delta|u|^{-\frac{1}{4}}\mathscr{F}(\mathcal{A}+\widetilde{\mathcal{E}}).
\end{equation}
\begin{remark}
The last term on the right hand side is again a \emph{borderline term}, the factor $L\phi$ in which cannot be estimated using the bootstrap assumptions \eqref{bootstrap}. Instead, we should estimate $L\phi$ by \eqref{estimate-Lphi} which we derived above. {\bf From now on, we will point it out when we are going to use an estimate which is derived previously. Otherwise, we will only use the bootstrap assumptions \eqref{bootstrap}, or the estimates \eqref{estimate-Omegaequ}, \eqref{estimate-trchi1}, \eqref{estimate-trchib1} and \eqref{estimate-Lbphi1}, or simply use the definitions of $\widetilde{\mathcal{E}}$, $\widetilde{\mathcal{O}}$ and $\mathcal{R}$, even if we have derived the corresponding estimate previously.}
\end{remark}

The second term on the left hand side of \eqref{equ-Lbphi} can be estimated by
\begin{align*}
\delta\|\Omega\tr\chi\Lb\phi\|_{\L^1_{\ub}{\H^4(u)}}\lesssim&\delta\|\Omega\tr\chi\|_{\L^\infty_{\ub}\H^4(u)}\left(\frac{|\psi|}{|u|}+\|\Lb\phi-\psi/|u|\|_{\L^2_{\ub}\H^4(u)}\right)\\
\lesssim&\delta|u|^{-1}(2\Omega_0^2h+C^{\frac{1}{2}}\delta|u|^{-1}\mathscr{F}^2\mathcal{A}^2)\left(\frac{|\psi|}{|u|}+\|\Lb\phi-\psi/|u|\|_{\L^2_{\ub}\H^4(u)}\right).
\end{align*}
Multiply the square of the above estimate by $|u|^{\frac{5}{2}}$ and then integrate over $[u_0,u]$, we have
\begin{equation}\label{LHS-Lbphi}
\begin{split}
&\left(\int_{u_0}^u|u'|^{\frac{5}{2}}\delta^2\|\Omega\tr\chi\Lb\phi\|^2_{\L^1_{\ub}{\H^4(u')}}\D u'\right)^{\frac{1}{2}}\\
\lesssim&\delta|u|^{-\frac{1}{4}}\left(\int_{u_0}^u\frac{\Omega_0^4|\psi|^2}{|u'|}\D u'\right)^{\frac{1}{2}}+C^{\frac{1}{2}}\delta^2|u|^{-\frac{5}{4}}\mathscr{F}^2\mathcal{A}^2\left(\int_{u_0}^u\frac{|\psi|^2}{|u'|}\D u'\right)^{\frac{1}{2}}\\
&+(1+C^{-1}\mathscr{F}\mathcal{A})\delta|u|^{-1}\||u|^{\frac{3}{4}}(|u|\Lb\phi-\psi)\|_{\L^\infty_{\ub}\L^2_{[u_0,u]}\H^4}\\
\lesssim&\delta|u|^{-\frac{1}{4}}\mathscr{F}\mathcal{A}
\end{split}
\end{equation}
Combining \eqref{RHS-Lbphi} and \eqref{LHS-Lbphi}, we have
\begin{equation}\label{estimate-Lbphi}
\delta^{-1}|u|^{\frac{1}{4}}\||u|^{\frac{3}{4}}(|u|\Lb \phi-\psi)\|_{\L^2_{[u_0,u]}\H^4(\ub)}\lesssim \mathscr{F}\mathcal{A}.
\end{equation}
The estimates \eqref{estimate-Lphi}, \eqref{estimate-nablasphi}, \eqref{estimate-Lbphi} give
\begin{align*}
\mathcal{E}\lesssim\mathcal{A}+\widetilde{\mathcal{E}}.
\end{align*}

\end{proof}

\subsection{Estimates for $\mathcal{O}$}
The next proposition is about the estimates for $\mathcal{O}$, which is about the lower order derivatives of the connection coefficients:
\begin{proposition}\label{estimate-O}
Under the assumptions of Theorem \ref{existencetheorem} and the bootstrap assumptions \eqref{bootstrap}, we have $$\mathcal{O}\lesssim \mathcal{A}+\widetilde{\mathcal{O}}[\eta,\etab]+\widetilde{\mathcal{E}}+\mathcal{R}+\mathcal{A}^{-\frac{1}{2}}\mathcal{R}^{\frac{3}{2}}.$$
\end{proposition}
Here $\widetilde{\mathcal{O}}[\eta,\eta]$ means the norms of $\eta,\etab$ in the definition of $\widetilde{\mathcal{O}}$.
\begin{proof}
{\bf \underline{Estimate for $\Omega\chih$}:} We consider the structure equation for $\Db(\Omega\chih)$:
\begin{equation}\label{equ-chih}\Dbh(\Omega\chih)-\frac{1}{2}\Omega\tr\chib\Omega\chih=\Omega^2(\nablas \tensor \eta + \eta \tensor \eta -\frac{1}{2}\tr\chi \chibh+\nablas\phi\tensor\nablas\phi).\end{equation}
The right hand side should be estimated in $\||u|^2\cdot\|_{\L^1_{[u_0,u]}\H^4(\ub)}$. The \engordnumber{1}, \engordnumber{2} and \engordnumber{4} term (denoted by $\uppercase\expandafter{\romannumeral1}$) on the right hand side can be estimated by
\begin{align*}
\||u|^2\uppercase\expandafter{\romannumeral1}\|_{\L^1_{[u_0,u]}\H^4(\ub)}\lesssim&\Omega_0(u_0)|u|^{-\frac{1}{2}}\||u|^{\frac{3}{2}}\Omega_0\eta\|_{\L^2_{[u_0,u]}\H^5(\ub)}+\||u|^{-2}\|_{\L^1_{[u_0,u]}}\||u|^2(\eta,\nablas\phi)\|_{\L^\infty_{[u_0,u]}\H^4(\ub)}^2\\
\lesssim& \Omega_0(u_0)\delta^{\frac{1}{2}}|u|^{-\frac{1}{2}}\mathscr{F}\mathscr{E}(\mathcal{A}+\widetilde{\mathcal{O}}[\eta])+\Omega_0^2(u_0)C^{\frac{1}{2}}\delta^2|u|^{-2}\mathscr{F}^2\mathscr{E}^2\mathcal{A}^2\\
\lesssim& \mathscr{F}(\mathcal{A}+\widetilde{\mathcal{O}}[\eta]).
\end{align*}
The \engordnumber{3} term on the right hand side is estimated by
\begin{align*}
\||u|^2\Omega\tr\chi\Omega\chibh\|\|_{\L^1_{[u_0,u]}\H^4(\ub)}\lesssim |u|^{-1}\mathscr{F}\mathcal{A}\cdot C^{\frac{1}{4}}\delta\mathscr{F}\mathcal{A}\lesssim C^{-1}\mathscr{F}\mathcal{A}.
\end{align*}
Therefore we have
\begin{align}\label{estimate-chih}
|u|\|\Omega\chih\|_{\H^4(\ub,u)}\lesssim|u_0|\|\Omega\chih\|_{\H^4(\ub,u_0)}+\mathscr{F}(\mathcal{A}+\widetilde{\mathcal{O}}[\eta])\lesssim\mathscr{F}(\mathcal{A}+\widetilde{\mathcal{O}}[\eta]).\end{align}

{\bf \underline{Estimate for $\tr\chi'$}:} We estimate $\tr\chi'$ by the equation for $D\tr\chi'$ written in form:
\begin{equation}\label{equ-Dtrchi}D\left(\tr\chi'-\frac{2h}{|u|}\right)=\Omega^{-2}\left(-\frac{1}{2}(\Omega^2\tr\chi')^2-|\Omega\chih|^2-2(L\phi)^2\right).\end{equation}
We then have  (using \eqref{estimate-trchi1}, \eqref{estimate-Lphi} and \eqref{estimate-chih})
\begin{align}\label{estimate-trchi}
\left\|\tr\chi'-\frac{2h}{|u|}\right\|_{\H^4(\ub,u)}\lesssim&\Omega_0^{-2}\delta|u|^{-2}\mathscr{F}^2(\mathcal{A}^2+\widetilde{\mathcal{E}}^2+\widetilde{\mathcal{O}}[\eta]^2).
\end{align}

{\bf \underline{Estimate for $\eta$}:} We write the equation for $D\eta$ in the following form:
\begin{align}\label{equ-Deta}
D\eta = (\Omega\chi)\cdot\etab-2L\phi\nablas\phi-(\Omega\beta-L\phi\nablas\phi).\end{align}
The first two terms on the right hand side can be estimated in $\delta\|\cdot\|_{\L^1_{\ub}\H^4(u)}$ by
\begin{align*}
\lesssim \delta|u|^{-1}\cdot C^{\frac{1}{4}}\mathscr{F}\mathcal{A}\cdot C^{\frac{1}{4}}\delta|u|^{-2}\mathscr{F}\mathscr{E}\mathscr{W}^{\frac{1}{2}}\mathcal{A}\lesssim C^{-1}\delta|u|^{-2}\mathscr{F}\mathscr{E}\mathcal{A}.
\end{align*}
The last term on the right hand side can be estimated in $\delta\|\cdot\|_{\L^1_{\ub}\H^4(u)}$ by
\begin{align*}
\lesssim\delta\|\Omega\beta-L\phi\nablas\phi\|_{\L^2_{\ub}\H^4(u)}\lesssim\delta|u|^{-2}\mathscr{F}\mathscr{E}\mathcal{R}.
\end{align*}
We have
\begin{equation}\label{estimate-eta}
\|\eta\|_{\H^4(\ub,u)}\lesssim\delta|u|^{-2}\mathscr{F}\mathscr{E}(C^{-1}\mathcal{A}+\mathcal{R})\lesssim\delta|u|^{-2}\mathscr{F}\mathscr{E}(\mathcal{A}+\mathcal{R}).
\end{equation}

{\bf \underline{Estimate for $\etab$}:} We then write the equation for $\Db\etab$ in the following form:
\begin{equation}\label{equ-Dbetab}
\Db\etab = (\Omega\chib) \cdot\eta-2\Lb\phi\nablas\phi+(\Omega\betab+\Lb\phi\nablas\phi).
\end{equation}
The right hand side should be estimated in $\||u|^2\cdot\|_{\L^1_{[u_0,u]}\H^4(\ub)}$. Note that
\begin{align*}
\|\Omega\chib\|_{\H^4(\ub,u)}\lesssim\|\Omega\tr\chib\|_{\H^4(\ub,u)}+\|\Omega\chibh\|_{\H^4(\ub,u)}\lesssim|u|^{-1}+C^{\frac{1}{4}}\delta|u|^{-2}\mathscr{F}\mathcal{A}\lesssim|u|^{-1}.
\end{align*}
Using \eqref{estimate-eta}, the \engordnumber{1} term on the right hand side is estimated by
\begin{align*}
\||u|^2\Omega\chib\eta\|_{\L^1_{[u_0,u]}\H^4(\ub)}\lesssim\delta|u|^{-1}\mathscr{F}\mathscr{E}(\mathcal{A}+\mathcal{R}).
\end{align*}
Using \eqref{estimate-Lbphi1} and \eqref{estimate-nablasphi}, the \engordnumber{2} term is estimated by
\begin{align*}
\||u|^2\Lb\phi\nablas\phi\|_{\L^1_{[u_0,u]}\H^4(\ub)}\lesssim&\||u|\Lb\phi\|_{\L^2_{[u_0,u]}\H^4(\ub)}\cdot\||u|\nablas\phi\|_{\L^2_{[u_0,u]}\H^4(\ub)}\\
\lesssim&\delta|u|^{-1}\mathscr{F}\mathscr{E}\mathscr{W}^{\frac{1}{2}}\widetilde{\mathcal{E}}.
\end{align*}
The last term is estimated by
\begin{align*}
\||u|^2(\Omega\betab+\Lb\phi\nablas\phi)\|_{\L^1_{[u_0,u]}\H^4(\ub)}\lesssim&\|\Omega_0|u|^{-\frac{3}{2}}\|_{\L^2_{[u_0,u]}\H^4(\ub)}\cdot\|\Omega_0^{-1}|u|^{\frac{7}{2}}(\Omega\betab+\Lb\phi\nablas\phi)\|_{\L^2_{[u_0,u]}\H^4(\ub)}\\
\lesssim&\Omega_0(u_0)\Omega_0^{-1}\delta^{\frac{3}{2}}|u|^{-\frac{3}{2}}\mathscr{F}^{\frac{3}{2}}\mathscr{E}\mathcal{R}^{\frac{3}{2}}\lesssim\delta|u|^{-1}\mathscr{F}\mathscr{E}\mathcal{A}^{-\frac{1}{2}}\mathcal{R}^{\frac{3}{2}}.
\end{align*}
where we have used the auxiliary condition \eqref{auxiliary}. Combining all estimates above, we have
\begin{equation}\label{estimate-etab}
\begin{split}
\||u|\etab\|_{\H^4(\ub,u)}\lesssim&\||u_0|\etab\|_{\H^4(\ub,u_0)}+\delta|u|^{-1}\mathscr{F}\mathscr{E}\mathscr{W}^{\frac{1}{2}}(\mathcal{A}+\mathcal{A}^{-\frac{1}{2}}\mathcal{R}^{\frac{3}{2}}+\widetilde{\mathcal{E}})\\\lesssim&\delta|u|^{-1}\mathscr{F}\mathscr{E}\mathscr{W}^{\frac{1}{2}}(\mathcal{A}+\mathcal{A}^{-\frac{1}{2}}\mathcal{R}^{\frac{3}{2}}+\widetilde{\mathcal{E}}).
\end{split}
\end{equation}

{\bf \underline{Estimate for $\Omega\chibh$}:} Remember the equation for $D(\Omega\chibh)$:
\begin{align}\label{equ-Dchibh}\Dh(\Omega\chibh)=\Omega^2(\nablas \tensor \etab + \etab \tensor \etab +\ds\phi\tensor\ds\phi)+\frac{1}{2}\Omega\tr\chi\Omega\chibh-\frac{1}{2}\Omega\tr\chib \Omega\chih.\end{align}
The first three terms on the right hand side are estimated in $\delta\|\cdot\|_{\L_{\ub}^1\H^4(u)}$ by
\begin{align}\label{estimate-Dchi123}
\lesssim \Omega_0(u_0)\delta^{\frac{3}{2}}|u|^{-\frac{5}{2}}\mathscr{F}\mathscr{E}\mathscr{W}^{\frac{1}{2}}\widetilde{\mathcal{O}}[\etab]+\Omega_0^2(u_0)C^{\frac{1}{2}}\delta^3|u|^{-4}\mathscr{F}^2\mathscr{E}^2\mathscr{W}\mathcal{A}^2\lesssim \delta|u|^{-2}\mathscr{F}(\mathcal{A}+\widetilde{\mathcal{O}}[\etab]).
\end{align}
 In the same norm, the \engordnumber{4} term is estimated by $\lesssim\delta\cdot|u|^{-1}\mathscr{F}\mathcal{A}\cdot C^{\frac{1}{4}}\delta|u|^{-2}\mathscr{F}\mathcal{A}\lesssim C^{-1}\delta|u|^{-2}\mathscr{F}\mathcal{A}$. The \engordnumber{5} term is estimated by (using \eqref{estimate-trchib1} and \eqref{estimate-chih}) $\lesssim\delta|u|^{-2}\mathscr{F}(\mathcal{A}+\widetilde{\mathcal{O}}[\eta])$.
Therefore, we have
\begin{align}\label{estimate-chibh}
\|\Omega\chibh\|_{\H^4(\ub,u)}\lesssim \delta|u|^{-2}\mathscr{F}(\mathcal{A}+\widetilde{\mathcal{O}}[\eta,\etab]).
\end{align}

{\bf \underline{Estimate for $\Omega\tr\chib$}:} Write the equation for $D(\Omega\tr\chib)$:
\begin{equation}\label{equ-Dtrchib}D\left(\Omega\tr\chib+\frac{2}{|u|}\right)=\Omega^2(2\divs\etab+2|\etab|^2+2|\ds\phi|^2-2K)-\Omega\tr\chi\Omega\tr\chib.\end{equation}
The right hand side is estimated in $\delta\|\cdot\|_{\L_{\ub}^1\H^4(u)}$. The estimates for the first three terms are the same with \eqref{estimate-Dchi123}. The last term is estimated by $\lesssim\delta|u|^{-2}\mathscr{F}\mathcal{A}$. The Gauss curvature term $\Omega^2K$ is estimated by
\begin{equation}\label{estimate-K1}
\begin{split}
\delta\|\Omega^2K\|_{\L^1_{\ub}\H^4(u)}\lesssim&\delta\Omega_0^2\|K-|u|^{-2}\|_{\L^2_{\ub}\H^4(u)}+\delta|u|^{-2}\\\lesssim& \Omega_0C^{\frac{3}{8}}\delta^2|u|^{-3}\mathscr{F}^{\frac{3}{2}}\mathscr{E}\mathcal{A}^{\frac{3}{2}}+\delta|u|^{-2}\lesssim\delta|u|^{-2}\mathscr{F}\mathcal{A}.
\end{split}
\end{equation}
Combining the estimates above we have
\begin{equation}\label{estimate-trchib}
\|\Omega\tr\chib+2|u|^{-1}\|_{\H^4(\ub,u)}\lesssim\delta|u|^{-2}\mathscr{F}(\mathcal{A}+\widetilde{\mathcal{O}}[\eta]).\end{equation}

{\bf \underline{Improved estimate for the derivatives of $\Omega\tr\chib$}:}
We should mention that the derivatives of $\Omega\tr\chib$ behave better from the perspective of the power of $\delta|u|^{-1}$.  This is because the derivatives of $\Omega\tr\chi, \Omega\tr\chib$ and $K$ behave better than themselves. Note that
\begin{align*}
&\delta\|(|u|\nablas)(-2\Omega^2K-\Omega\tr\chi\Omega\tr\chib)\|_{\L^1_{\ub}\H^3(u)}\\
\lesssim&\delta\Omega_0^2\|K-|u|^{-2}\|_{\L^2_{\ub}\H^4(u)}+\delta\Omega_0^2\|\tr\chi'-2h|u|^{-1}\|_{\L^\infty_{\ub}\H^4(\ub,u)}\|\Omega\tr\chib+2|u|^{-1}\|_{\L^\infty_{\ub}\H^4(\ub,u)}\\
&+\delta\Omega_0^2\left(\|\Omega\tr\chib\|_{\L^\infty_{\ub}\L^\infty(u)}\|\tr\chi'-2h|u|^{-1}\|_{\L^\infty_{\ub}\H^4(\ub,u)}+\|\tr\chi'\|_{\L^\infty_{\ub}\L^\infty(u)}\|\Omega\tr\chib+2|u|^{-1}\|_{\L^\infty_{\ub}\H^4(\ub,u)}\right)\\
\lesssim&\Omega_0C^{\frac{3}{8}}\delta^2|u|^{-3}\mathscr{F}^{\frac{3}{2}}\mathscr{E}\mathcal{A}^{\frac{3}{2}}+C^{\frac{3}{4}}\delta^3|u|^{-4}\mathscr{F}^3\mathcal{A}^3+C^{\frac{1}{4}}\delta^2|u|^{-3}\mathscr{F}^2\mathcal{A}^2
\end{align*}
Combining the above estimates with the estimates for the first three terms, which are bounded by $C^{\frac{1}{4}}\Omega_0\delta^{\frac{3}{2}}|u|^{-\frac{5}{2}}\mathscr{F}\mathscr{E}\mathscr{W}^{\frac{1}{2}}\mathcal{A}$ in $\delta\|\cdot\|_{\L^1_{\ub}\H^4(u)}$, gives
\begin{align}\label{estimate-detrchib}
\|(|u|\nablas)(\Omega\tr\chib)\|_{\H^3(\ub,u)}\lesssim C^{\frac{1}{4}}\delta^{\frac{3}{2}}|u|^{-\frac{5}{2}}\mathscr{F}\mathscr{E}\mathscr{W}^{\frac{1}{2}}\mathcal{A}.
\end{align}

{\bf \underline{Estimate for $\omega$}:}  Consider the equation 
\begin{equation}\label{equ-Dbomega}
\Db  \omega =\Omega^2(2(\eta,\etab)-|\etab|^2-|\ds\phi|^2+K)+\frac{1}{4}\Omega\tr \chi\Omega\tr\chib-\frac{1}{2}(\Omega\chih,\Omega\chibh)+L\phi\Lb\phi,
\end{equation}
The right hand side should be estimated in $\||u|\cdot\|_{\L^1_{[u_0,u]}\H^4(\ub)}$. The first three terms are estimated by
\begin{align*}
\lesssim \Omega_0^2(u_0)C^{\frac{1}{2}}\delta^2|u|^{-3}\mathscr{F}^2\mathscr{E}^2\mathscr{W}\mathcal{A}^2\lesssim C^{-2} |u|^{-1}\mathscr{F}\mathcal{A}.
\end{align*}
The last three terms are estimated by, using \eqref{estimate-Lphi},
\begin{align*}
\lesssim|u|^{-1}\mathscr{F}\mathscr{W}^{\frac{1}{2}}(\mathcal{A}+\widetilde{\mathcal{E}}).
\end{align*}
The Gauss curvature term $\Omega^2K$ is estimated by
\begin{align*}
\||u|^2\Omega^2K\|_{\L^1_{[u_0,u]}\H^4(\ub)}\lesssim&|u|^{-1}+\Omega_0|u|^{-\frac{3}{2}}\|\Omega_0|u|^{\frac{5}{2}}(K-|u|^{-2})\|_{\L^2_{[u_0,u]}\H^4(\ub)}\\\lesssim&|u|^{-1}+\Omega_0\delta^{\frac{1}{2}}|u|^{-\frac{3}{2}}\mathscr{F}\mathscr{E}(\mathcal{A}+\mathcal{R})\\
\lesssim&|u|^{-1}\mathscr{F}(\mathcal{A}+\mathcal{R}).
\end{align*}
Combining the above estimates gives
\begin{align}\label{estimate-omega}
\|\omega\|_{\H^4(\ub,u)}\lesssim|u|^{-1}\mathscr{F}\mathscr{W}^{\frac{1}{2}}(\mathcal{A}+\widetilde{\mathcal{E}}+\mathcal{R}).
\end{align}

The estimates \eqref{estimate-chih}, \eqref{estimate-trchi}, \eqref{estimate-eta}, \eqref{estimate-etab}, \eqref{estimate-chibh}, \eqref{estimate-trchib}, \eqref{estimate-omega} are the desired estimates and the proof is completed.\end{proof}

\subsection{Estimates for $\widetilde{\mathcal{E}}$} We then turn to $\widetilde{\mathcal{E}}$, which is about the derivatives of the derivatives of the wave function $\phi$.
\begin{proposition}\label{estimate-tildeE}
Under the assumptions of Theorem \ref{existencetheorem} and the bootstrap assumptions \eqref{bootstrap}, we have
$$\widetilde{\mathcal{E}}\lesssim\mathcal{A}.$$
\end{proposition}
\begin{proof}
Recall the following equations which are essentially the wave equation:
\begin{equation}\label{equ-wave}
\begin{split}
\Db L\phi+\frac{1}{2}\Omega\tr\chib L\phi&=\Omega^2\divs\nablas\phi+2\Omega^2(\eta,\nablas\phi)-\frac{1}{2}\Omega\tr\chi\Lb\phi,\\
D\nablas\phi&=\nablas L\phi,\\
D\Lb\phi+\frac{1}{2}\Omega\tr\chi\Lb\phi&=\Omega^2\divs\nablas\phi+2\Omega^2(\etab,\nablas\phi)-\frac{1}{2}\Omega\tr\chib L\phi,\\
\Db\nablas\phi&=\nablas\Lb\phi.
\end{split}
\end{equation}
We compute
\begin{align*}
&\Db(|(|u|\nablas)^iL\phi|^2\D\mu_{\gs})+D(\Omega^2|(|u|\nablas)^i\nablas\phi|^2\D\mu_{\gs})\\
=&|u|^{2i}\nablas^A(\Omega^2\nablas^i_{B_1\cdots B_5}\nablas_A\phi\cdot\nablas^{i,B_1\cdots B_5}L\phi)+\tau_1
\end{align*}
for $1\le i\le 5$, where $\tau_1$ contains no $(i+1)^{\text{st}}$ order derivatives of the derivative of $\phi$. By divergence theorem, we obtain
\begin{align*}
&\delta\||u|(|u|\nablas)L\phi\|_{\L^2_{\ub}\H^4(u)}^2+\||u|^{\frac{3}{2}}\Omega(|u|\nablas)\nablas\phi\|_{\L^2_{[u_0,u]}\H^4(\ub)}^2\\
\lesssim&\delta\||u|(|u|\nablas)L\phi\|_{\L^2_{\ub}\H^4(u_0)}^2+\||u|^{\frac{3}{2}}\Omega(|u|\nablas)\nablas\phi\|_{\L^2_{[u_0,u]}\H^4(0)}^2+\int_0^{\delta}\D\ub'\int_{u_0}^u\D u'\int_{S_{\ub',u'}}|\tau_1|\D\mu_{\gs}
\end{align*}
We remark that the power $1$ of the multiple $|u|$ before $(|u|\nablas)L\phi$ above comes in because the coefficient of $\Omega\tr\chib L\phi$ on the left hand side is $\frac{1}{2}$, and then $\Omega\tr\chib$ itself will not appear in the following estimates (instead, only the normalized quantity $\Omega\tr\chib+2|u|^{-1}$ comes in). And by direct computation, 
\begin{align*}
\int_{S_{\ub',u'}}|\tau_1|\D\mu_{\gs}\lesssim\int_{S_{\ub',u'}}|\tau_{1,1}|\D\mu_{\gs}+\int_{S_{\ub',u'}}|\tau_{1,2}|\D\mu_{\gs}
\end{align*}
where (the norm $\|\cdot\|_i$ refers to $\|\cdot\|_{\H^i(\ub',u')}$)
\begin{equation}\label{tau11}
\begin{split}
&\int_{S_{\ub',u'}}|\tau_{1,1}|\D\mu_{\gs}\\
\lesssim&|u'|^2\|(|u'|\nablas)L\phi\|_4(\|\Omega\tr\chib+2|u'|^{-2}\|_5\|L\phi\|_4+\|\Omega\tr\chib+2|u'|^{-2}\|_4\|(|u'|\nablas)L\phi\|_4)\\
&+|u'|^2\|(|u'|\nablas)\Omega^2\|_4(|u'|^{-1}\|(|u|\nablas)\nablas\phi\|_4)\|(|u'|\nablas)L\phi\|_4\\&+|u'|^2\|(|u'|\nablas)L\phi\|_4\left(\|\Omega^2\|_5\|\eta\|_4\|\nablas\phi\|_4+\|\Omega^2\|_4\|\eta\|_5\|\nablas\phi\|_4+\|\Omega^2\|_4\|\eta\|_4\|\nablas\phi\|_5\right)\\
&+|u'|^2\|(|u'|\nablas)L\phi\|_4\left(\|\Omega\tr\chi\|_4\|(|u'|\nablas)\Lb\phi\|_4+\|(|u'|\nablas)(\Omega\tr\chi)\|_4\|\Lb\phi\|_4\right)\\
&+|u'|^2\Omega_0^2(u')|\omega|\|(|u'|\nablas)\nablas\phi\|_4^2
\end{split}
\end{equation}
which we call the multiplier terms, and
\begin{equation}\label{tau12}
\begin{split}
\int_{S_{\ub',u'}}|\tau_{1,2}|\D\mu_{\gs}\lesssim&|u'|^2\|(|u'|\nablas)L\phi\|_4\left(\|(|u'|\nablas)(\Omega\chib)\|_3\|L\phi\|_4+\Omega_0^2(u')|u'|\|K\|_4\|\nablas\phi\|_4\right)\\
&+|u'|^2\Omega_0^2(u')\|(|u'|\nablas)\nablas\phi\|_4\left(\|(|u'|\nablas)(\Omega\chi)\|_4\|\nablas\phi\|_4+|u'|\|K\|_3\|L\phi\|_4\right)
\end{split}
\end{equation}
which we call the commutation terms. We will estimate these terms in $L^1_{\ub}L^1_{u}$. The \engordnumber{1} term of the \engordnumber{1} line on the right hand side of \eqref{tau11} is estimated by
\begin{align*}
\lesssim& \delta|u|^{-1}\||u|(|u|\nablas)L\phi\|_{\L^2_{\ub}\H^4(u)}\||u|^2(\Omega\tr\chib+2|u|^{-1})\|_{\L^\infty_{\ub}\L^\infty_{[u_0,u]}\H^5}\||u|L\phi\|_{\L^\infty_{\ub}\L^\infty_{[u_0,u]}\H^4}\\
\lesssim&\delta|u|^{-1}\cdot C^{\frac{1}{4}}\mathscr{F}\mathscr{E}\mathcal{A}\cdot C^{\frac{1}{4}}\delta\mathscr{F}\mathscr{E}\mathscr{W}^{\frac{1}{2}}\mathcal{A}\cdot C^{\frac{1}{4}}\mathscr{F}\mathcal{A}
\lesssim C^{-\frac{1}{4}}\delta\mathscr{F}^2\mathscr{E}^2\mathcal{A}^2\cdot(\delta|u|^{-1}\mathscr{F}\mathcal{A})^{\frac{1}{2}}.
\end{align*}
The \engordnumber{2} term of the \engordnumber{1} line of \eqref{tau11} is estimated by
\begin{align*}
\lesssim& \delta|u|^{-1}\||u|(|u|\nablas)L\phi\|^2_{\L^2_{\ub}\H^4(u)}\||u|^2(\Omega\tr\chib+2|u|^{-1})\|_{\L^\infty_{\ub}\L^\infty_{[u_0,u]}\H^4}\\
\lesssim&\delta|u|^{-1}\cdot C^{\frac{1}{2}}\mathscr{F}^2\mathscr{E}^2\mathcal{A}^2\cdot C^{\frac{1}{4}}\delta\mathscr{F}\mathcal{A}
\lesssim C^{-\frac{1}{4}}\delta\mathscr{F}^2\mathscr{E}^2\mathcal{A}^2\cdot(\delta|u|^{-1}\mathscr{F}\mathcal{A})^{\frac{1}{2}}.
\end{align*}
Noting that
\begin{equation}\label{estimate-Omega5}\|(|u'|\nablas)\Omega^2\|_4\lesssim\Omega_0^2(|u'|\|\eta\|_4+|u'|\|\etab\|_4)\lesssim\Omega_0^2 C^{\frac{1}{4}}\delta|u|^{-1}\mathscr{F}\mathscr{E}\mathscr{W}^{\frac{1}{2}}\mathcal{A}\lesssim C^{-\frac{3}{2}} \Omega_0^2\mathscr{E},\end{equation}
the \engordnumber{2} line of \eqref{tau11} is estimated by,
\begin{align*}
\lesssim&\delta|u|^{-\frac{1}{2}}\|\Omega_0^{-1}(|u|\nablas)\Omega^2\|_{\L^\infty_{\ub}\L^\infty_{[u_0,u]}\H^4}\|\Omega_0|u|^{\frac{3}{2}}\nablas\phi\|_{\L^\infty_{\ub}\L^2_{[u_0,u]}\H^5}\||u|(|u|\nablas)L\phi\|_{\L^\infty_{[u_0,u]}\L^2_{\ub}\H^4}\\
\lesssim&\delta|u|^{-\frac{1}{2}}\cdot C^{-\frac{3}{2}}\Omega_0(u_0)\mathscr{E}\cdot C^{\frac{1}{4}}\delta^{\frac{1}{2}}\mathscr{F}\mathscr{E}\mathcal{A}\cdot C^{\frac{1}{4}}\mathscr{F}\mathscr{E}\mathcal{A}
\lesssim C^{-1}\delta\mathscr{F}^2\mathscr{E}^2\mathcal{A}^2\cdot(\Omega^2(u_0)\delta|u|^{-1}\mathscr{E}^2)^{\frac{1}{2}}.
\end{align*}
The last two terms of the \engordnumber{3} line is estimated the same as above, and the \engordnumber{1} term is estimated by
\begin{align*}
\lesssim&\delta|u|^{-2}\|\Omega^2\|_{\L^\infty_{\ub}\L^\infty_{[u_0,u]}\H^5}\||u|^2\eta\|_{\L^\infty_{\ub}\L^\infty_{[u_0,u]}\H^4}\||u|^2\nablas\phi\|_{\L^\infty_{\ub}\L^\infty_{[u_0,u]}\H^4}\||u|(|u|\nablas)L\phi\|_{\L^\infty_{[u_0,u]}\L^2_{\ub}\H^4}\\
\lesssim&\delta|u|^{-2}\cdot \Omega_0^2(u_0)\mathscr{E}\cdot C^{\frac{1}{4}}\delta\mathscr{F}\mathscr{E}\mathcal{A}\cdot C^{\frac{1}{4}}\delta\mathscr{F}\mathscr{E}\mathcal{A}\cdot C^{\frac{1}{4}}\mathscr{F}\mathscr{E}\mathcal{A}
\lesssim C^{-1}\delta\mathscr{F}^2\mathscr{E}^2\mathcal{A}^2\cdot(\Omega_0^2(u_0)\delta|u|^{-1}\mathscr{E}^2)
\end{align*}
The \engordnumber{4} line is obtained by taking derivatives on $\Omega\tr\chi\Lb\phi$. It is important that at least one derivative applies on $\Omega\tr\chi$ or $\Lb\phi$, both of which have worst estimates at zeroth order. The \engordnumber{1} term is estimated by
\begin{align*}
\lesssim&\delta|u|^{-\frac{3}{4}}\||u|\Omega\tr\chi\|_{\L^\infty_{\ub}\L^\infty_{[u_0,u]}\H^4}\cdot\Omega_0(u_0)\|\Omega_0^{-1}|u|^{\frac{3}{4}}(|u|\Lb\phi-\psi)\|_{\L^\infty_{\ub}\L^2_{[u_0,u]}\H^5}\||u|(|u|\nablas)L\phi\|_{\L^\infty_{[u_0,u]}\L^2_{\ub}\H^4}\\
\lesssim&\delta|u|^{-\frac{3}{4}}\cdot \mathscr{F}\mathcal{A}\cdot C^{\frac{1}{4}}\delta|u|^{-\frac{1}{4}}\Omega_0(u_0)\Omega_0^{-1}\mathscr{F}\mathscr{E}\mathcal{A}\cdot C^{\frac{1}{4}}\mathscr{F}\mathscr{E}\mathcal{A}
\lesssim \delta\mathscr{F}^2\mathscr{E}^2\mathcal{A}^2\cdot(\delta|u|^{-1}\mathscr{F}\mathcal{A})^{\frac{1}{4}},
\end{align*}
where we have used both the smallness condition (the second one) and the auxiliary condition \eqref{auxiliary}. The \engordnumber{2} term is estimated by
\begin{align*}
\lesssim&\delta|u|^{-1}\||u|^2\Omega^2(\tr\chi'-2h/|u|)\|_{\L^\infty_{\ub}\L^\infty_{[u_0,u]}\H^5}\||u|\Lb\phi\|_{\L^\infty_{\ub}\L^\infty_{[u_0,u]}\H^4}\||u|(|u|\nablas)L\phi\|_{\L^\infty_{[u_0,u]}\L^2_{\ub}\H^4}\\
\lesssim&\delta|u|^{-1}\cdot C^{\frac{1}{2}}\delta\mathscr{F}^2\mathscr{E}\mathcal{A}^2\cdot \mathscr{W}^{\frac{1}{2}}\cdot C^{\frac{1}{4}}\mathscr{F}\mathscr{E}\mathcal{A}\\
\lesssim& C^{-\frac{1}{4}}\delta\mathscr{F}^2\mathscr{E}^2\mathcal{A}^2\cdot\left(\delta|u|^{-1}\mathscr{F}\mathscr{W}^{\frac{1}{2}}\mathcal{A}\right)^{\frac{1}{2}}.
\end{align*}
\begin{remark}
Here we have used the estimate of $\|(|u|\nablas)^5(\Omega\tr\chi)\|_{\L^2(\ub,u)}$  which is (the lower order estimates are given in \eqref{estimate-trchi1})
\begin{equation}\label{estimate-nablas5Omegatrchi-bootstrap}
\begin{split}
\|(|u|\nablas)^5(\Omega\tr\chi)\|_{\L^2(\ub,u)}\lesssim&\|(|u|\nablas)\Omega^2\|_{\H^4(\ub,u)}\|\tr\chi'\|_{\H^4(\ub,u)}+\Omega_0^2\|(|u|\nablas)\tr\chi'\|_{\H^4(\ub,u)}\\
\lesssim&C^{\frac{1}{4}}\delta|u|^{-1}\mathscr{F}\mathscr{E}\mathscr{W}^{\frac{1}{2}}\mathcal{A}\cdot|u|^{-1}\mathscr{F}\mathcal{A}+ C^{\frac{1}{2}}\delta|u|^{-2}\mathscr{F}^2\mathscr{E}\mathcal{A}^2\\
\lesssim& C^{\frac{1}{2}}\delta|u|^{-2}\mathscr{F}^2\mathscr{E}\mathscr{W}^{\frac{1}{2}}\mathcal{A}^2.
\end{split}
\end{equation}
\end{remark}
The last term of \eqref{tau11} is estimated by
\begin{align*}
\lesssim&\delta|u|^{-1}\||u|\omega\|_{\L^\infty_{\ub}\L^\infty_{[u_0,u]}\H^4}\|\Omega_0|u|^{\frac{3}{2}}\nablas\phi\|_{\L^\infty_{\ub}\L^2_{[u_0,u]}\H^5}^2\\
\lesssim&\delta|u|^{-1}\cdot C^{\frac{1}{4}}\mathscr{F}\mathscr{W}^{\frac{1}{2}}\mathcal{A}\cdot C^{\frac{1}{2}}\delta\mathscr{F}^2\mathscr{E}^2\mathcal{A}^2\\
\lesssim& C^{-\frac{1}{4}}\delta\mathscr{F}^2\mathscr{E}^2\mathcal{A}^2\cdot\left(\delta|u|^{-1}\mathscr{F}\mathscr{W}^{\frac{1}{2}}\mathcal{A}\right)^{\frac{1}{2}}.
\end{align*}

Then we turn to the estimates of the right hand side of \eqref{tau12} in $L_{\ub}^1L_u^1$. The \engordnumber{1} term of the \engordnumber{1} line is estimated by
\begin{align*}
\lesssim&\delta|u|^{-1}\||u|^2(|u|\nablas)(\Omega\chib)\|_{\L^\infty_{\ub}\L^\infty_{[u_0,u]}\H^3}\||u|L\phi\|_{\L^\infty_{\ub}\L^\infty_{[u_0,u]}\H^4}\||u|(|u|\nablas)L\phi\|_{\L^\infty_{[u_0,u]}\L^2_{\ub}\H^4}\\
\lesssim&\delta|u|^{-1}\cdot C^{\frac{1}{4}}\delta\mathscr{F}\mathcal{A}\cdot C^{\frac{1}{4}}\mathscr{F}\mathcal{A}\cdot C^{\frac{1}{4}}\mathscr{F}\mathscr{E}\mathcal{A}\lesssim C^{-\frac{1}{4}}\delta\mathscr{F}^2\mathscr{E}\mathcal{A}^2\cdot(\delta|u|^{-1}\mathscr{F}\mathcal{A})^{\frac{1}{2}}.
\end{align*}
The \engordnumber{1} term of the \engordnumber{2} line is estimated by (using both \eqref{estimate-nablas5Omegatrchi-bootstrap})
\begin{align*}
\lesssim&\delta|u|^{-\frac{3}{2}}\Omega_0(u_0)\||u|(|u|\nablas)(\Omega\chi)\|_{\L^\infty_{[u_0,u]}\L^2_{\ub}\H^4}\||u|^2\nablas\phi\|_{\L^\infty_{\ub}\L^\infty_{[u_0,u]}\H^4}\|\Omega_0|u|^{\frac{3}{2}}(|u|\nablas)\nablas\phi\|_{\L^\infty_{\ub}\L^2_{[u_0,u]}\H^5}\\
\lesssim&\delta|u|^{-\frac{3}{2}}\Omega_0(u_0)\cdot C^{\frac{1}{4}}\mathscr{F}\mathscr{E}\mathcal{A}\cdot C^{\frac{1}{4}}\delta\mathscr{F}\mathscr{E}\mathcal{A}\cdot C^{\frac{1}{4}}\delta^{\frac{1}{2}}\mathscr{F}\mathscr{E}\mathcal{A}\lesssim C^{-1}\delta\mathscr{F}^2\mathscr{E}^2\mathcal{A}^2\cdot(\Omega_0^2(u_0)\delta|u|^{-1}\mathscr{E}^2)^{\frac{1}{2}}.
\end{align*}
The estimates of the remaining terms require an estimate of $K$. Similar to \eqref{estimate-K1}, we can deduce that
\begin{equation}\label{estimate-K2}
\Omega_0\|K\|_{\L^2_{\ub}\H^4(\ub)}\lesssim|u|^{-2}+C^\frac{3}{8}\delta|u|^{-3}\mathscr{F}^{\frac{3}{2}}\mathscr{E}\mathcal{A}^{\frac{3}{2}}\lesssim|u|^{-2}\mathscr{F}^{\frac{1}{2}}\mathscr{E}\mathcal{A}^{\frac{1}{2}}.
\end{equation}
Then the \engordnumber{2} term of the \engordnumber{1} line of \eqref{tau12} is estimated by
\begin{align*}
\lesssim&\delta|u|^{-1}\Omega_0(u_0)\|\Omega_0|u|^2K\|_{\L^\infty_{[u_0,u]}\L^2_{\ub}\H^4}\||u|^2\nablas\phi\|_{\L^\infty_{\ub}\L^\infty_{[u_0,u]}\H^4}\||u|(|u|\nablas)L\phi\|_{\L^\infty_{[u_0,u]}\L^2_{\ub}\H^4}\\
\lesssim&\delta|u|^{-1}\cdot \mathscr{F}^{\frac{1}{2}}\mathscr{E}\mathcal{A}^{\frac{1}{2}}\cdot C^{\frac{1}{4}}\delta\mathscr{F}\mathcal{A}\cdot C^{\frac{1}{4}}\mathscr{F}\mathscr{E}\mathcal{A}\lesssim C^{-\frac{1}{2}}\delta\mathscr{F}^2\mathscr{E}^2\mathcal{A}^2\cdot(\delta|u|^{-1}\mathscr{F}\mathcal{A})^{\frac{1}{2}}.
\end{align*}
And the last term of \eqref{tau12} is estimated by
\begin{align*}
\lesssim&\delta|u|^{-\frac{1}{2}}\Omega_0(u_0)\|\Omega_0|u|^2K\|_{\L^\infty_{[u_0,u]}\L^2_{\ub}\H^3}\||u|L\phi\|_{\L^\infty_{\ub}\L^\infty_{[u_0,u]}\H^4}\||u|^\frac{3}{2}(|u|\nablas)\nablas\phi\|_{\L^\infty_{\ub}\L^2_{[u_0,u]}\H^4}\\
\lesssim&\delta|u|^{-\frac{1}{2}}\Omega_0(u_0)\cdot \mathscr{F}^{\frac{1}{2}}\mathscr{E}\mathcal{A}^{\frac{1}{2}}\cdot C^{\frac{1}{4}}\mathscr{F}\mathcal{A}\cdot C^{\frac{1}{4}}\delta^{\frac{1}{2}}\mathscr{F}\mathscr{E}\mathcal{A}\lesssim \delta\mathscr{F}^2\mathscr{E}^2\mathcal{A}^2\cdot(\delta|u|^{-1}\mathscr{F}\mathcal{A})^{\frac{1}{4}}.
\end{align*}

Now combining all the estimates of \eqref{tau11} and \eqref{tau12} above, we achieve that (note that $\nablas\phi(\ub=0)=0$)
\begin{equation}\label{estimate-nablas5Lphi-improve}
\begin{split}
&\||u|(|u|\nablas)L\phi\|_{\L^2_{\ub}\H^4(u)}^2+\delta^{-1}\||u|^{\frac{3}{2}}\Omega_0(|u|\nablas)\nablas\phi\|_{\L^2_{[u_0,u]}\H^4(\ub)}^2\\
\lesssim&\||u|(|u|\nablas)L\phi\|_{\L^2_{\ub}\H^4(u_0)}^2+\mathscr{F}^2\mathscr{E}^2\mathcal{A}^2\cdot\left(\delta|u|^{-1}\mathscr{W}(\mathscr{F}\mathcal{A}+\Omega_0^2(u_0)\mathscr{E}^2)\right)^{\frac{1}{4}}.
\end{split}
\end{equation}
In particular, this implies
\begin{equation}\label{estimate-nablas5Lphi}
\||u|(|u|\nablas)L\phi\|_{\L^2_{\ub}\H^4(u)}+\delta^{-\frac{1}{2}}\||u|^{\frac{3}{2}}\Omega_0(|u|\nablas)\nablas\phi\|_{\L^2_{[u_0,u]}\H^4(\ub)}\lesssim\mathscr{F}\mathscr{E}\mathcal{A}.
\end{equation}
\begin{remark}
The improved estimate \eqref{estimate-nablas5Lphi-improve} is crucial in the following parts of the proof.
\end{remark}

We consider the last two equations of \eqref{equ-wave}. We compute
\begin{align*}
&\Db(|u||(|u|\nablas)^i\nablas\phi|^2\D\mu_{\gs})+D(\Omega^{-2}|u||(|u|\nablas)^i\Lb\phi|^2\D\mu_{\gs})\\
=&|u|^{2i+1}\nablas^A(\nablas^i_{B_1\cdots B_i}\Lb\phi\cdot\nablas^{i,B_1\cdots B_i}\nablas_A\phi)+|u|\tau_2
\end{align*}
for $1\le i\le 5$, where $\tau_2$ contains no sixth order derivatives of the derivative of $\phi$. By divergence theorem, we have
\begin{align*}
&\delta\Omega_0^2(u)\||u|^{\frac{3}{2}}(|u|\nablas)\nablas\phi\|_{\L^2_{\ub}\H^4(u)}^2+\Omega_0^2(u)\||u|^{2}\Omega^{-1}(|u|\nablas)\Lb\phi\|_{\L^2_{[u_0,u]}\H^4(\ub)}^2\\
\lesssim&\delta\Omega_0^2(u)\||u|^{\frac{3}{2}}(|u|\nablas)\nablas\phi\|_{\L^2_{\ub}\H^4(u_0)}^2+\Omega_0^2(u)\||u|^{2}\Omega^{-1}(|u|\nablas)\Lb\phi\|_{\L^2_{[u_0,u]}\H^4(0)}^2\\
&+\int_0^{\delta}\D\ub'\int_{u_0}^u\D u'\int_{S_{\ub',u'}}\Omega_0^2(u)|u'||\tau_2|\D\mu_{\gs}
\end{align*}
where $\tau_2$ contains no sixth order derivatives of the derivative of $\phi$. By direct computation, 
\begin{align*}
\int_{S_{\ub',u'}}\Omega_0^2(u)|u'||\tau_2|\D\mu_{\gs}\lesssim\int_{S_{\ub',u'}}\Omega_0^2(u)|u'||\tau_{2,1}|\D\mu_{\gs}+\int_{S_{\ub',u'}}\Omega_0^2(u)|u'||\tau_{2,2}|\D\mu_{\gs}
\end{align*}
where the multiplier terms are
\begin{equation}\label{tau21}
\begin{split}
&\int_{S_{\ub',u'}}\Omega_0^2(u)|u'||\tau_{2,1}|\D\mu_{\gs}\\
\lesssim
&\Omega_0^2(u)\Omega_0^{-2}(u')\left[|u'|^3\|(|u'|\nablas)\Lb\phi\|_4\left(|\Omega\tr\chi|\|(|u'|\nablas)\Lb\phi\|_4+\|(|u'|\nablas)(\Omega\tr\chi)\|_4\|\Lb\phi\|_4\right)\right.\\
&+|u'|^3|\omega|\|(|u'|\nablas)\Lb\phi\|_4^2\\
&+|u'|^3\|(|u'|\nablas)\Omega^2\|_4(|u'|^{-1}\|(|u|\nablas)\nablas\phi\|_4)\|(|u'|\nablas)\Lb\phi\|_4\\&+|u'|^3\|(|u'|\nablas)\Lb\phi\|_4\left(\|\Omega^2\|_5\|\etab\|_4\|\nablas\phi\|_4+\|\Omega^2\|_4\|\etab\|_5\|\nablas\phi\|_4+\|\Omega^2\|_4\|\etab\|_4\|\nablas\phi\|_5\right)\\
&\left.+|u'|^3\|(|u'|\nablas)\Lb\phi\|_4\left(|\Omega\tr\chib|\|(|u'|\nablas)L\phi\|_4+\|(|u'|\nablas)(\Omega\tr\chib)\|_4\|L\phi\|_4\right)\right]\\
&+\Omega_0^2(u)|u'|^3\|(|u'|\nablas)\nablas\phi\|_4(|\Omega\tr\chib|\|(|u|\nablas)\nablas\phi\|_4+\|\Omega\tr\chib+2|u'|^{-2}\|_5\|\nablas\phi\|_4)\
\end{split}
\end{equation}
and the commutation terms are
\begin{equation}\label{tau22}
\begin{split}
&\int_{S_{\ub',u'}}|u'||\tau_{2,2}|\D\mu_{\gs}\\
\lesssim&\Omega_0^2(u)\Omega_0^{-2}(u')|u'|^3\|(|u'|\nablas)\Lb\phi\|_4\left(\|(|u'|\nablas)(\Omega\chi)\|_3\|\Lb\phi\|_4+\Omega_0^2(u')|u'|\|K\|_4\|\nablas\phi\|_4\right)\\
&+\Omega_0^2(u)|u'|^3\|(|u'|\nablas)\nablas\phi\|_4\left(\|(|u'|\nablas)(\Omega\chib)\|_4\|\nablas\phi\|_4+|u'|\|K\|_3\|(|u|\nablas)\Lb\phi\|_3\right).
\end{split}
\end{equation}
The right hand side of \eqref{tau21} and \eqref{tau22} should be estimated in $L_{\ub}^1L^1_u$. The \engordnumber{1} terms of the first two lines of \eqref{tau21} is estimated by
\begin{align*}
\lesssim&\delta|u|^{-1}\||u|\Omega\tr\chi, |u|\omega\|_{\L^\infty_{\ub}\L^\infty_{[u_0,u]}\H^4}\cdot\Omega_0\|\Omega_0^{-1}|u|^2(|u|\nablas)\Lb\phi\|^2_{\L^\infty_{\ub}\L^2_{[u_0,u]}\H^4}\\
\lesssim& \delta|u|^{-1}\cdot C^{\frac{1}{4}}\mathscr{F}\mathscr{W}^{\frac{1}{2}}\mathcal{A}\cdot C^{\frac{1}{2}}\delta^2\mathscr{F}^2\mathscr{E}^2\mathcal{A}^2\lesssim C^{-1}\delta^2\mathscr{F}^2\mathscr{E}^2\mathcal{A}^2.
\end{align*}
The \engordnumber{2} term of the \engordnumber{1} line is estimated by (using \eqref{estimate-nablas5Omegatrchi-bootstrap})
\begin{align*}
\lesssim&\delta|u|^{-1}\||u|^2\Omega^2(\tr\chi'-2h/|u|)\|_{\L^\infty_{\ub}\L^\infty_{[u_0,u]}\H^5}\||u|\Lb\phi\|_{\L^\infty_{\ub}\L^2_{[u_0,u]}\H^4}\cdot\Omega_0\|\Omega_0^{-1}|u|^2(|u|\nablas)\Lb\phi\|_{\L^\infty_{\ub}\L^2_{[u_0,u]}\H^4}\\
\lesssim& \delta|u|^{-1}\cdot C^{\frac{1}{2}}\delta\mathscr{F}^2\mathscr{E}\mathscr{W}^{\frac{1}{2}}\mathcal{A}^2\cdot \mathscr{W}^{\frac{1}{2}}\cdot C^{\frac{1}{4}}\delta\mathscr{F}\mathscr{E}\mathcal{A}\lesssim C^{-1}\delta^2\mathscr{F}^2\mathscr{E}^2\mathcal{A}^2.
\end{align*}
In view of \eqref{estimate-Omega5}, all terms of the \engordnumber{3} and \engordnumber{4} lines of \eqref{tau21} are estimated in the same manner. They are estimated by
\begin{align*}
\lesssim&\delta|u|^{-\frac{3}{2}}\Omega_0\||u|^2(\nablas\phi,\eta, \etab)\|_{\L^\infty_{\ub}\L^\infty_{[u_0,u]}\H^4}\cdot\Omega_0\||u|^\frac{3}{2}(\nablas\phi,\etab)\|_{\L^\infty_{[u_0,u]}\L^2_{\ub}\H^5}\cdot\Omega_0\|\Omega_0^{-1}|u|^2(|u|\nablas)\Lb\phi\|_{\L^\infty_{\ub}\L^2_{[u_0,u]}\H^4}\\
\lesssim&\delta|u|^{-\frac{3}{2}}\Omega_0\cdot C^{\frac{1}{4}}\delta\mathscr{F}\mathscr{E}\mathscr{W}^{\frac{1}{2}}\mathcal{A}\cdot C^{\frac{1}{4}}\delta^{\frac{1}{2}}\mathscr{F}\mathscr{E}\mathscr{W}^{\frac{1}{2}}\mathcal{A}\cdot C^{\frac{1}{4}}\delta\mathscr{F}\mathscr{E}\mathcal{A}\lesssim C^{-1}\delta^2\mathscr{F}^2\mathscr{E}^2\mathcal{A}^2.
\end{align*}
The \engordnumber{2} term of the \engordnumber{5} line is estimated by
\begin{align*}
\lesssim&\delta|u|^{-1}\||u|^2(\Omega\tr\chib+2|u|^{-2})\|_{\L^\infty_{\ub}\L^\infty_{[u_0,u]}\H^5}\||u|L\phi\|_{\L^\infty_{\ub}\L^\infty_{[u_0,u]}\H^4}\cdot\Omega_0\|\Omega_0^{-1}|u|^2(|u|\nablas)\Lb\phi\|_{\L^\infty_{\ub}\L^2_{[u_0,u]}\H^4}\\
\lesssim& \delta|u|^{-1}\cdot C^{\frac{1}{4}}\delta\mathscr{F}\mathscr{E}\mathscr{W}^{\frac{1}{2}}\mathcal{A}\cdot C^{\frac{1}{4}}\mathscr{F}\mathcal{A}\cdot C^{\frac{1}{4}}\delta\mathscr{F}\mathscr{E}\mathcal{A}\lesssim C^{-1}\delta^2\mathscr{F}^2\mathscr{E}^2\mathcal{A}^2.
\end{align*}
The \engordnumber{2} term of the last line is estimated by
\begin{align*}
\lesssim&\delta|u|^{-\frac{3}{2}}\Omega_0\||u|^2(\Omega\tr\chib+2|u|^{-2})\|_{\L^\infty_{\ub}\L^\infty_{[u_0,u]}\H^5}\||u|^2\nablas\phi\|_{\L^\infty_{\ub}\L^\infty_{[u_0,u]}\H^4}\cdot \Omega_0\||u|^\frac{3}{2}(|u|\nablas)\nablas\phi\|_{\L^\infty_{[u_0,u]}\L^2_{\ub}\H^4}\\
\lesssim& \delta|u|^{-\frac{3}{2}}\Omega_0\cdot C^{\frac{1}{4}}\delta\mathscr{F}\mathscr{E}\mathscr{W}^{\frac{1}{2}}\mathcal{A}\cdot C^{\frac{1}{4}}\delta\mathscr{F}\mathscr{E}\mathcal{A}\cdot C^{\frac{1}{4}}\delta^{\frac{1}{2}}\mathscr{F}\mathscr{E}\mathcal{A}\lesssim C^{-1}\delta^2\mathscr{F}^2\mathscr{E}^2\mathcal{A}^2.
\end{align*}

Now we are going to the crucial terms, the \engordnumber{1} terms of the \engordnumber{5} line and the last line. Recalling the definition of $\mathscr{E}$, using \eqref{estimate-nablas5Lphi-improve}, the \engordnumber{1} term of the \engordnumber{5} line is estimated by
\begin{align*}
\lesssim&\delta\left(\int_{u_0}^u|u'|^{-1}||u'|\Omega\tr\chib|^2\||u|(|u|\nablas)L\phi\|_{\L^2_{\ub}\H^4(u')}^2\D u'\right)^{\frac{1}{2}}\cdot\Omega_0\|\Omega_0^{-1}|u|^2(|u|\nablas)\Lb\phi\|_{\L^\infty_{\ub}\L^2_{[u_0,u]}\H^4}\\
\lesssim&\delta\left(\int_{u_0}^u|u'|^{-1}\left(\||u|(|u|\nablas)L\phi\|_{\L^2_{\ub}\H^4(u_0)}^2+\mathscr{F}^2\mathscr{E}^2\mathcal{A}^2\cdot\left(\delta|u|^{-1}\mathscr{W}(\mathscr{F}\mathcal{A}+\mathscr{E}^2)\right)^{\frac{1}{4}}\right)\D u'\right)^{\frac{1}{2}}\\
&\times\Omega_0\|\Omega_0^{-1}|u|^2(|u|\nablas)\Lb\phi\|_{\L^\infty_{\ub}\L^2_{[u_0,u]}\H^4}\\
\lesssim&\delta\left(\left|\log\frac{|u_1|}{|u_0|}\right|^{\frac{1}{2}}\||u|(|u|\nablas)L\phi\|_{\L^2_{\ub}\H^4(u_0)}+\mathscr{F}\mathscr{E}\mathcal{A}\cdot\left(\delta|u|^{-1}\mathscr{W}(\mathscr{F}\mathcal{A}+\mathscr{E}^2)\right)^{\frac{1}{8}}\right)\\
&\times\Omega_0\|\Omega_0^{-1}|u|^2(|u|\nablas)\Lb\phi\|_{\L^\infty_{\ub}\L^2_{[u_0,u]}\H^4}\\
\lesssim&\delta\mathscr{F}\mathscr{E}\mathcal{A}\cdot\Omega_0\|\Omega_0^{-1}|u|^2(|u|\nablas)\Lb\phi\|_{\L^\infty_{\ub}\L^2_{[u_0,u]}\H^4}.
\end{align*}
Using \eqref{estimate-nablas5Lphi}, the \engordnumber{1} term of the last line is estimated by
\begin{align*}
\lesssim\delta\||u|\Omega\tr\chi\|_{\L^\infty_{\ub}\L^\infty_{[u_0,u]}\H^4}\|\Omega_0|u|^{\frac{3}{2}}(|u|\nablas)\nablas\phi\|^2_{\L^\infty_{\ub}\L^2_{[u_0,u]}\H^4}\lesssim\delta^2\mathscr{F}^2\mathscr{E}^2\mathcal{A}^2.
\end{align*}

Finally, we turn to the estimates of \eqref{tau22} in $L^1_{\ub}L^1_u$. The estimates of the \engordnumber{1} terms of both lines are similar to the \engordnumber{2} terms of the last two lines of \eqref{tau22}, and the \engordnumber{2} term of the \engordnumber{1} line is estimated by (using \eqref{estimate-K2})
\begin{align*}
\lesssim&\delta|u|^{-1}\Omega_0\|\Omega_0|u|^2K\|_{\L^\infty_{[u_0,u]}\L^2_{\ub}\H^4}\||u|^2\nablas\phi\|_{\L^\infty_{\ub}\L^\infty_{[u_0,u]}\H^4}\cdot\Omega_0\|\Omega_0^{-1}|u|^2(|u|\nablas)\Lb\phi\|_{\L^\infty_{\ub}\L^2_{[u_0,u]}\H^4}\\
\lesssim&\delta|u|^{-1}\Omega_0\cdot \mathscr{F}^{\frac{1}{2}}\mathscr{E}\mathcal{A}^{\frac{1}{2}}\cdot C^{\frac{1}{4}}\delta\mathscr{F}\mathscr{E}\mathcal{A}\cdot C^{\frac{1}{4}}\delta\mathscr{F}\mathscr{E}\mathcal{A}\lesssim C^{-1}\delta^2\mathscr{F}^2\mathscr{E}^2\mathcal{A}^2.
\end{align*}
And the last term of \eqref{tau22} is estimated by
\begin{align*}
\lesssim&\delta|u|^{-\frac{1}{4}}\Omega_0\||u|^2K\|_{\L^\infty_{[u_0,u]}\L^2_{\ub}\H^3}\||u|^\frac{7}{4}(\Lb\phi-\psi/|u|)\|_{\L^\infty_{\ub}\L^\infty_{[u_0,u]}\H^4}\cdot\Omega_0\||u|^{\frac{3}{2}}(|u|\nablas)\nablas\phi\|_{\L^\infty_{[u_0,u]}\L^2_{\ub}\H^4}\\
\lesssim&\delta|u|^{-\frac{1}{4}}\cdot\mathscr{F}^{\frac{1}{2}}\mathscr{E}\mathcal{A}^{\frac{1}{2}}\cdot C^{\frac{1}{4}}\delta|u|^{-\frac{1}{4}}\mathscr{F}\mathcal{A}\cdot C^{\frac{1}{4}}\delta^{\frac{1}{2}}\mathscr{F}\mathscr{E}\mathcal{A}\lesssim C^{-\frac{1}{2}}\delta^2\mathscr{F}^2\mathscr{E}^2\mathcal{A}^2.
\end{align*}

Combining all estimates above together, noting that $\nablas\Lb\phi(\ub=0)=0$, we have
\begin{equation*}
\begin{split}
&\delta\Omega_0^2(u)\||u|^{\frac{3}{2}}(|u|\nablas)\nablas\phi\|_{\L^2_{\ub}\H^4(u)}^2+\Omega_0^2(u)\||u|^{2}\Omega_0^{-1}(|u|\nablas)\Lb\phi\|_{\L^2_{[u_0,u]}\H^4(\ub)}^2\\
\lesssim&\delta\||u|^{\frac{3}{2}}(|u|\nablas)\nablas\phi\|_{\L^2_{\ub}\H^4(u_0)}^2+\delta^2\mathscr{F}^2\mathscr{E}^2\mathcal{A}^2+\delta\mathscr{F}\mathscr{E}\mathcal{A}\cdot\Omega_0(u)\|\Omega_0^{-1}|u|^2(|u|\nablas)\Lb\phi\|_{\L^\infty_{\ub}\L^2_{[u_0,u]}\H^4},
\end{split}
\end{equation*}
which implies
\begin{equation}\label{estimate-nablas5nablasphi}
\begin{split}
\Omega_0(u)\||u|^{\frac{3}{2}}(|u|\nablas)\nablas\phi\|_{\L^2_{\ub}\H^4(u)}+\delta^{-\frac{1}{2}}\Omega_0(u)\||u|^{2}\Omega_0^{-1}(|u|\nablas)\Lb\phi\|_{\L^2_{[u_0,u]}\H^4(\ub)}\lesssim\delta^{\frac{1}{2}}\mathscr{F}\mathscr{E}\mathcal{A}
\end{split}
\end{equation}

The proof of Proposition \ref{estimate-tildeE} is then complete.\end{proof}

From Proposition \ref{estimate-E} and \ref{estimate-tildeE}, we conclude
\begin{align*}
\mathcal{E}\lesssim\mathcal{A}.
\end{align*}
In particular, $\widetilde{\mathcal{E}}$ in all estimates in Proposition \ref{estimate-O} can be dropped.

\subsection{Estimates for $\widetilde{\mathcal{O}}$} We then use the elliptic systems to estimate $\widetilde{\mathcal{O}}$, which involves the top order derivatives of the connection coefficients. 
\begin{proposition}\label{estimate-tildeO}
Under the assumptions of Theorem \ref{existencetheorem} and the bootstrap assumptions \eqref{bootstrap}, we have
$$\widetilde{\mathcal{O}}\lesssim\mathcal{A}+\mathcal{R}.$$
\end{proposition}

\begin{proof} 

{\bf \underline{Estimate for $K$}:} In order to apply elliptic estimate, we need an appropriate estimate for $K$. Consider the equation for $DK$:
\begin{equation}\label{equ-DK}D(K-|u|^{-2})+\Omega\tr\chi K=\divs\divs(\Omega\chih)-\frac{1}{2}\Deltas(\Omega\tr\chi).\end{equation}
Integrate the equation \eqref{equ-DK}, for up to second order derivatives of $K$, using \eqref{estimate-chih} and \eqref{estimate-trchi}, we have
\begin{align}\label{estimate-K-lower}
\|K-|u|^{-2}\|_{\H^2(\ub,u)}\lesssim&\delta|u|^{-3}\mathscr{F}\mathcal{A}+\delta^2|u|^{-4}\mathscr{F}^2\mathcal{A}^2\lesssim\delta|u|^{-3}\mathscr{F}\mathcal{A}.
\end{align}
Using \eqref{estimate-nablas5Omegatrchi-bootstrap}, we have
\begin{align*}
\|K-|u|^{-2}\|_{\H^3(\ub,u)}\lesssim C^{\frac{1}{4}}\delta|u|^{-3}\mathscr{F}\mathscr{E}\mathcal{A}+C^{\frac{1}{2}}\delta^2|u|^{-4}\delta\mathscr{F}^2\mathscr{E}\mathcal{A}^2\lesssim C^{\frac{1}{4}}\delta|u|^{-3}\mathscr{F}\mathscr{E}\mathcal{A}.
\end{align*}
This then in particular implies
\begin{align}\label{estimate-K3}
\|K\|_{\H^2(\ub,u)}\lesssim|u|^{-2}, \ \|K\|_{\H^3(\ub,u)}\lesssim|u|^{-2}\mathscr{E}.
\end{align}
\begin{remark}
We have a loss $\mathscr{E}$ on the third order derivative of $K$.
\end{remark}


{\bf \underline{Estimates for $\mu$ and $\eta$ involving the top order derivatives}:} Consider first the equation for $D\mu$,
\begin{equation}\label{equ-Dmu}D\mu+\Omega\tr\chi\mu=-\Omega\tr\chi\frac{1}{|u|^2}+\divs(2\Omega\chih\cdot\eta-\Omega\tr\chi\etab)+2\nablas L\phi\cdot\ds\phi+2L\phi\Deltas\phi.\end{equation}
The right hand side should estimated in $\delta\|\cdot\|_{\L^1_{\ub}\H^4(u)}$. Noting that $\Omega\chih,\Omega\tr\chi, L\phi$ share a similar bound and $\eta, \etab,\nablas\phi$ share a similar bound, the last three terms are estimated by
\begin{align*}
\lesssim&\delta|u|^{-1}\|\Omega\chih,\Omega\tr\chi,L\phi\|_{\L^2_{\ub}\H^5(u)}\|\eta,\etab,\nablas\phi\|_{\L^\infty_{\ub}\H^4(u)}+\delta|u|^{-1}\|\chih,\tr\chi,L\phi\|_{\L^\infty_{\ub}\H^4(u)}\|\eta,\etab,\nablas\phi\|_{\L^2_{\ub}\H^5(u)}\\
\lesssim&\delta|u|^{-1}\cdot C^{\frac{1}{4}}|u|^{-1}\mathscr{F}\mathscr{E}\mathcal{A}\cdot C^{\frac{1}{4}}\delta|u|^{-2}\mathscr{F}\mathscr{E}\mathscr{W}^{\frac{1}{2}}\mathcal{A}+\delta|u|^{-1}\cdot C^{\frac{1}{4}}|u|^{-1}\mathscr{F}\mathcal{A}\cdot C^{\frac{1}{4}}\Omega_0^{-1}\delta^{\frac{1}{2}}|u|^{-\frac{3}{2}}\mathscr{F}\mathscr{E}\mathscr{W}^{\frac{1}{2}}\mathcal{A}\\
\lesssim& C^{\frac{1}{2}}\Omega_0^{-1}\delta^{\frac{3}{2}}|u|^{-\frac{7}{2}}\mathscr{F}^2\mathscr{E}\mathscr{W}^{\frac{1}{2}}\mathcal{A}^2.
\end{align*}
Therefore, with $\||u|^{-2}\Omega\tr\chi\|_{\H^4(\ub,u)}\lesssim|u|^{-3}\mathscr{F}\mathcal{A}$, 
\begin{align*}
\|\mu\|_{\H^4(\ub,u)}\lesssim\delta|u|^{-1}\mathscr{F}\mathcal{A}\|\mu\|_{\L^\infty_{\ub}\H^4(\ub,u)}+\Omega_0^{-1}\delta|u|^{-3}\mathscr{F}^{\frac{3}{2}}\mathscr{E}\mathcal{A}^{\frac{3}{2}}.\end{align*}
Because $\delta|u|^{-1}\mathscr{F}\mathcal{A}\le C^{-2}\le C_0^{-2}$, if $C_0$ is sufficiently large, the first term above can be absorbed by the left hand side, and then we have
\begin{align}\label{estimate-mu}
\|\mu\|_{\H^4(\ub,u)}\lesssim\Omega_0^{-1}\delta|u|^{-3}\mathscr{F}^{\frac{3}{2}}\mathscr{E}\mathcal{A}^{\frac{3}{2}}.\end{align}

Then from the elliptic system for $\eta$
\begin{align}\label{equ-ellipticeta}
\begin{cases}\divs\eta=K-\frac{1}{|u|^2}-\mu\\
\curls\eta=\check{\sigma}\end{cases},
\end{align}
we have, by elliptic estimate and \eqref{estimate-eta},
\begin{align*}
\|\eta\|_{\H^5(\ub,u)}\lesssim&|u|\|K-|u|^{-2},\check{\sigma}\|_{\H^4(\ub,u)}+|u|\|\mu\|_{\H^4(\ub,u)}+\mathscr{E}\|\eta\|_{\H^4(\ub,u)}\\
\lesssim&|u|\|K-|u|^{-2},\check{\sigma}\|_{\H^4(\ub,u)}+\Omega_0^{-1}\delta|u|^{-2}\mathscr{F}\mathscr{E}(\Omega_0\mathscr{E}(\mathcal{A}+\mathcal{R})+\mathscr{F}^{\frac{1}{2}}\mathcal{A}^{\frac{3}{2}}).
\end{align*}
Then we have,
\begin{equation}\label{estimate-nablas5etaub}
\begin{split}
\|\eta\|_{\L^2_{\ub}\H^5(u)}\lesssim&|u|\|K-|u|^{-2},\check{\sigma}\|_{\L^2_{\ub}\H^4(u)}+\Omega_0^{-1}\delta|u|^{-2}\mathscr{F}\mathscr{E}(C^{\frac{1}{4}}\Omega_0\mathscr{E}\mathcal{A}+\mathscr{F}^{\frac{1}{2}}\mathcal{A}^{\frac{3}{2}})\\
\lesssim&C^{\frac{3}{8}}\Omega_0^{-1}\delta|u|^{-2}\mathscr{F}^{\frac{3}{2}}\mathscr{E}\mathcal{A}^{\frac{3}{2}}+\delta|u|^{-2}\mathscr{F}\mathscr{E}^2(\mathcal{A}+\mathcal{R})\lesssim \Omega_0^{-1}\delta^{\frac{1}{2}}|u|^{-\frac{3}{2}}\mathscr{F}\mathscr{E}(\mathcal{A}+\mathcal{R}),
\end{split}
\end{equation}
and 
\begin{equation}\label{estimate-nablas5etau}
\begin{split}
\|\Omega_0|u|^{\frac{3}{2}}\eta\|_{\L^2_{[u_0,u]}\H^5(\ub)}\lesssim&\|\Omega_0|u|^{\frac{5}{2}}(K-|u|^{-2},\check{\sigma})\|_{\L^2_{[u_0,u]}\H^4(\ub)}+\delta|u|^{-\frac{1}{2}}\mathscr{F}\mathscr{E}(\Omega_0\mathscr{E}(\mathcal{A}+\mathcal{R})+\mathscr{F}^{\frac{1}{2}}\mathcal{A}^{\frac{3}{2}})\\
\lesssim&\delta^{\frac{1}{2}}\mathscr{F}\mathscr{E}(\mathcal{A}+\mathcal{R}),
\end{split}
\end{equation}

{\bf \underline{Improved estimate for the derivatives of $\mu$}:} We can improve the estimates for the derivatives of $\mu$. This is because the first term $-|u|^{-2}\Omega\tr\chi$ has better estimates for its derivatives. We estimate
\begin{align*}
-|u|\nablas(\Omega\tr\chi\mu)-|u|\nablas(\Omega\tr\chi)\frac{1}{|u|^2}
\end{align*}
in $\delta\|\cdot\|_{\L^1_{\ub}\H^3(u)}$ by, using \eqref{estimate-mu},
\begin{align*}
\lesssim\delta|u|^{-1}\mathscr{F}\mathcal{A}\|(|u|\nablas)\mu\|_{\L^\infty_{\ub}\H^3(u)}+\delta\cdot C^{\frac{1}{2}}\delta|u|^{-2}\mathscr{F}^2\mathcal{A}^2\cdot\Omega_0^{-1}|u|^{-2}\mathscr{F}^{\frac{1}{2}}\mathscr{E}\mathcal{A}^{\frac{1}{2}}.
\end{align*}
Therefore, if $C_0$ is sufficiently large such that the first term can be absorbed, we have, combining with the estimates for the last three terms of \eqref{equ-Dmu},
\begin{equation}\label{estimate-nablasmu}
\|(|u|\nablas)\mu\|_{\H^3(\ub,u)}\lesssim C^{\frac{1}{2}}\Omega_0^{-1}\delta^{\frac{3}{2}}|u|^{-\frac{7}{2}}\mathscr{F}^2\mathscr{E}\mathscr{W}^{\frac{1}{2}}\mathcal{A}^2.
\end{equation}

{\bf \underline{Estimates for $\mub$ and $\etab$ involving the top order derivatives}:}  We consider the equation for $\Db\mub$:
\begin{equation}\label{equ-Dbmub}
\Db\mub+\Omega\tr\chib\mub=-(\Omega\tr\chib+\frac{2}{|u|})\frac{1}{|u|^2}+\divs(2\Omega\chibh\cdot\etab-\Omega\tr\chib\eta)+2\nablas \Lb\phi\cdot\ds\phi+2\Lb\phi\Deltas\phi
\end{equation}

We compute the initial value of $\mub$ on $C_{u_0}$ first. Because $\divs\etab=K-\frac{1}{|u|^2}-\mub$, then
\begin{align}\label{estimate-mub0}
\|\mub\|_{\L^2_{\ub}\H^4(u_0)}\lesssim|u|^{-1}\|\etab\|_{\L^2_{\ub}\H^5(u_0)}+\|K-|u|^{-2}\|_{\L^2_{\ub}\H^4(u_0)}\lesssim\Omega_0^{-1}(u_0)\delta^{\frac{1}{2}}|u|^{-\frac{5}{2}}\mathscr{F}\mathscr{E}\mathscr{W}^{\frac{1}{2}}\mathcal{A}.
\end{align}
So in the following, we estimate $\mub$ in $\L^2_{\ub}\H^4(u)$. The right hand side of \eqref{equ-Dbmub} should be estimated in $\||u|^3\cdot\|_{\L^1_{[u_0,u]}\H^4(\ub)}$. The first term is easily estimated by $\lesssim C^{\frac{1}{4}}\delta|u|^{-1}\mathscr{F}\mathcal{A}$. The term $\divs(2\Omega\chibh\cdot\etab)$ is estimated by
\begin{align*}
&|u|^{-\frac{3}{2}}\Omega_0(u_0)\|\Omega_0^{-1}|u|^{\frac{3}{2}}\Omega\chibh\|_{\L^2_{[u_0,u]}\H^5(\ub)}\||u|^{2}\etab\|_{\L^\infty_{[u_0,u]}\H^4(\ub)}\\
&+\Omega_0^{-1}|u|^{-1}\||u|^{2}\Omega\chibh\|_{\L^\infty_{[u_0,u]}\H^4(\ub)}\|\Omega_0|u|\etab\|_{\L^2_{[u_0,u]}\H^5(\ub)}\\
\lesssim&|u|^{-2}\Omega_0(u_0)\cdot \Omega_0^{-1}C^{\frac{1}{4}}\delta\mathscr{F}\mathscr{E}\mathscr{W}^{\frac{1}{2}}\mathcal{A}\cdot C^{\frac{1}{4}}\delta\mathscr{F}\mathscr{E}\mathscr{W}^{\frac{1}{2}}\mathcal{A}+\Omega_0^{-1}C^{\frac{1}{4}}|u|^{-1}\delta\mathscr{F}\mathcal{A}\|\Omega_0|u|\etab\|_{\L^2_{[u_0,u]}\H^5(\ub)}.
\end{align*}
The term $\divs(\Omega\tr\chib\eta)$ is estimated by, using \eqref{estimate-nablas5etau},
\begin{align*}
&|u|^{-2}\||u|^{2}(|u|\nablas)(\Omega\tr\chib)\|_{\L^\infty_{[u_0,u]}\H^4(\ub)}\||u|^{2}\etab\|_{\L^\infty_{[u_0,u]}\H^4(\ub)}\\
&+\Omega_0^{-1}|u|^{-\frac{1}{2}}\||u|\Omega\tr\chib\|_{\L^\infty_{[u_0,u]}\H^4(\ub)}\|\Omega_0|u|^{\frac{3}{2}}\eta\|_{\L^2_{[u_0,u]}\H^5(\ub)}\\
\lesssim&|u|^{-2}\Omega^{-1}_0\Omega_0\cdot C^\frac{1}{4}\delta\mathscr{F}\mathscr{E}\mathscr{W}^{\frac{1}{2}}\mathcal{A}\cdot C^{\frac{1}{4}}\delta\mathscr{F}\mathscr{E}\mathscr{W}^{\frac{1}{2}}\mathcal{A}+\Omega_0^{-1}|u|^{-\frac{1}{2}}\cdot1\cdot \delta^\frac{1}{2}\mathscr{F}\mathscr{E}(\mathcal{A}+\mathcal{R}).
\end{align*}
The last two terms are estimated by, using \eqref{estimate-nablas5Lphi},
\begin{align*}
&|u|^{-2}\Omega_0(u_0)\|\Omega_0^{-1}|u|^{-2}(|u|\nablas)\Lb\phi\|_{\L^2_{[u_0,u]}\H^4(\ub)}\||u|^{2}\nablas\phi\|_{\L^\infty_{[u_0,u]}\H^4(\ub)}\\
&+\Omega_0^{-1}|u|^{-\frac{1}{2}}\||u|\Lb\phi\|_{\L^\infty_{[u_0,u]}\H^4(\ub)}\|\Omega_0|u|^{\frac{3}{2}}\nablas\phi\|_{\L^2_{[u_0,u]}\H^5(\ub)}\\
\lesssim&\Omega_0(u_0)|u|^{-2}\cdot C^{\frac{1}{4}}\Omega_0^{-1}\delta\mathscr{F}\mathscr{E}\mathcal{A}\cdot C^{\frac{1}{4}}\delta\mathscr{F}\mathscr{E}\mathcal{A}+\Omega_0^{-1}|u|^{-\frac{1}{2}}\cdot\mathscr{W}^{\frac{1}{2}}\cdot\delta^{\frac{1}{2}}\mathscr{F}\mathscr{E}\mathcal{A}.
\end{align*}
Combining the estimates above, we have, using \eqref{estimate-mub0},
\begin{align*}
\||u|^2\mub\|_{\H^4(\ub,u)}\lesssim&\||u|^2\mub\|_{\H^4(\ub,u_0)}+\Omega_0^{-1}\delta^{\frac{1}{2}}|u|^{-\frac{1}{2}}\mathscr{F}\mathscr{E}\mathscr{W}^{\frac{1}{2}}(\mathcal{A}+\mathcal{R})+C^{-1}\Omega_0^{-1}\|\Omega_0|u|\etab\|_{\L^2_{[u_0,u]}\H^5(\ub)}\\
\lesssim&\Omega_0^{-1}\delta^{\frac{1}{2}}|u|^{-\frac{1}{2}}\mathscr{F}\mathscr{E}\mathscr{W}^{\frac{1}{2}}(\mathcal{A}+\mathcal{R})+C^{-1}\Omega_0^{-1}\|\Omega_0|u|\etab\|_{\L^2_{[u_0,u]}\H^5(\ub)}.
\end{align*}
Therefore, taking $\|\cdot\|_{\L^2_{\ub}}$ of the above inequality, we have
\begin{equation}\label{estimate-mub}
\begin{split}
\||u|^2\mub\|_{\L^2_{\ub}\H^4(u)}\lesssim&\Omega_0^{-1}\delta^{\frac{1}{2}}|u|^{-\frac{1}{2}}\mathscr{F}\mathscr{E}\mathscr{W}^{\frac{1}{2}}(\mathcal{A}+\mathcal{R})+C^{-1}\Omega_0^{-1}\|\Omega_0|u|\etab\|_{\L^2_{[u_0,u]}\L^2_{\ub}\H^5}\\
\lesssim&\Omega_0^{-1}\delta^{\frac{1}{2}}|u|^{-\frac{1}{2}}\mathscr{F}\mathscr{E}\mathscr{W}^{\frac{1}{2}}(\mathcal{A}+\mathcal{R})+C^{-1}\Omega_0^{-1}\cdot\|C^{\frac{1}{4}}\delta|u|^{-\frac{1}{2}}\mathscr{F}\mathscr{E}\mathscr{W}^{\frac{1}{2}}\mathcal{A}\|_{\L^2_{[u_0,u]}}\\
\lesssim&\Omega_0^{-1}\delta^{\frac{1}{2}}|u|^{-\frac{1}{2}}\mathscr{F}\mathscr{E}\mathscr{W}^{\frac{1}{2}}(\mathcal{A}+\mathcal{R}).
\end{split}
\end{equation}

Then from the elliptic system for $\etab$:
\begin{equation}\label{equ-ellipticetab}
\begin{cases}\divs\etab=K-\frac{1}{|u|^2}-\mub\\
\curls\etab=-\check{\sigma}\end{cases},
\end{equation}
we have
\begin{align*}
\|\etab\|_{\H^5(\ub,u)}\lesssim&|u|\|K-|u|^{-2},\check{\sigma}\|_{\H^4(\ub,u)}+|u|\|\mub\|_{\H^4(\ub,u)}+\mathscr{E}\|\etab\|_{\H^4(\ub,u)}.
\end{align*}
We then have, using \eqref{estimate-etab}, \begin{equation}\label{estimate-nablas5etab}
\begin{split}
\|\etab\|_{\L^2_{\ub}\H^5(u)}\lesssim&|u|\|K-|u|^{-2},\check{\sigma}\|_{\L^2_{\ub}\H^4(u)}+|u|\|\mub\|_{\L^2_{\ub}\H^4(u)}+\delta|u|^{-2}\mathscr{F}\mathscr{E}^2\mathscr{W}^{\frac{1}{2}}(\mathcal{A}+\mathcal{A}^{-\frac{1}{2}}\mathcal{R}^{\frac{3}{2}})\\
\lesssim&\Omega_0^{-1}\delta^{\frac{1}{2}}|u|^{-\frac{3}{2}}\mathscr{F}\mathscr{E}\mathscr{W}^{\frac{1}{2}}(\mathcal{A}+\mathcal{R}+\mathcal{A}^{-\frac{1}{2}}\mathcal{R}^{\frac{3}{2}}).
\end{split}
\end{equation}

From \eqref{estimate-nablas5etaub}, \eqref{estimate-nablas5etau} and \eqref{estimate-nablas5etab}, the dependence on $\widetilde{\mathcal{O}}$ of the estimates \eqref{estimate-chih}, \eqref{estimate-trchi}, \eqref{estimate-chibh} and \eqref{estimate-trchib}, and of Proposition \ref{estimate-O} can now be replaced by $\mathcal{R}+\mathcal{A}^{-\frac{1}{2}}\mathcal{R}^{\frac{3}{2}}$.

{\bf \underline{Estimates for $\Omega\chih$ involving the top order derivatives}:}  Now we consider the equation for $\divs(\Omega\chih)$:
\begin{align}\label{equ-divschih}
\divs(\Omega\chih)=\frac{1}{2}\Omega^2\nablas\tr \chi'+\Omega\chih\cdot\etab+\frac{1}{2}\Omega\tr \chi\eta-(\Omega\beta-L\phi\nablas\phi).\end{align}
The terms on the  right hand side are estimated in $|u|\|\cdot\|_{\H^4(\ub,u)}$ by
\begin{align*}
&\Omega_0^2\||u|\nablas\tr\chi'\|_{\H^4(\ub,u)}+|u|(\|\Omega\chih\etab\|_{\H^4(\ub,u)}+\|\Omega_0^2\tr\chi'\eta\|_{\H^4(\ub,u)}+\|\Omega\beta-L\phi\nablas\phi\|_{\H^4(\ub,u)})\\
\lesssim&C^{\frac{1}{4}}\delta|u|^{-2}\mathscr{F}^2\mathscr{E}\mathcal{A}^2+C^{\frac{1}{4}}|u|^{-1}\mathscr{F}\mathcal{A}\cdot C^{\frac{1}{4}}\delta|u|^{-1}\mathscr{F}\mathscr{E}\mathscr{W}^{\frac{1}{2}}\mathcal{A}+\|\Omega\beta-L\phi\nablas\phi\|_{\H^4(\ub,u)}\\
\lesssim&C^{-1}|u|^{-1}\mathscr{F}\mathscr{E}\mathcal{A}+|u|\|\Omega\beta-L\phi\nablas\phi\|_{\H^4(\ub,u)}.
\end{align*}
Then by elliptic estimate (using \eqref{estimate-chih}, \eqref{estimate-K3}),
\begin{align*}
\|\Omega\chih\|_{\H^5(\ub,u)}\lesssim&|u|\|\divs(\Omega\chih)\|_{\H^4(\ub,u)}+\mathscr{E}\|\Omega\chih\|_{\H^4(\ub,u)}\\
\lesssim&|u|^{-1}\mathscr{F}\mathscr{E}(\mathcal{A}+\mathcal{R}+\mathcal{A}^{-\frac{1}{2}}\mathcal{R}^{\frac{3}{2}})+|u|\|\Omega\beta-L\phi\nablas\phi\|_{\H^4(\ub,u)},
\end{align*}
whose $\L^2_{\ub}$ norm is then bounded by
\begin{align}\label{estimate-nablas5Omegachih}
\|\Omega\chih\|_{\L^2_{\ub}\H^5(u)}\lesssim|u|^{-1}\mathscr{F}\mathscr{E}(\mathcal{A}+\mathcal{R}+\mathcal{A}^{-\frac{1}{2}}\mathcal{R}^{\frac{3}{2}}).
\end{align}

{\bf \underline{Estimates for $\tr\chi'$ involving the top order derivatives}:} Now we consider the equation for $D\nablas\tr\chi'$, written in the following form:
\begin{equation}\label{equ-Dnablastrchi}
\begin{split}
D\nablas\tr\chi'=&\Omega^{-2}(\eta+\etab)\left(\frac{1}{2}(\Omega\tr\chi)^2+|\Omega\chih|^2+2(L\phi)^2\right)\\
&-(\Omega\tr\chi)\nablas\tr\chi'-2\Omega^{-2}\nablas(\Omega\chih)\cdot\Omega\chih-4\Omega^{-2}\nablas L\phi L\phi.
\end{split}
\end{equation}
which is derived by taking $\nablas$ to \eqref{equ-Dtrchi}. Here we use $2\nablas\log\Omega=\eta+\etab$. Now consider the fourth order derivative $\nablas^4$ of the above equation. To estimate $\nablas\tr\chi'$ in $\H^4$, the right hand side should be estimated in $\delta\|\cdot\|_{\L^1_{\ub}\H^4(u)}\lesssim\delta\|\cdot\|_{\L^2_{\ub}\H^4(u)}$. The first line is estimated by
\begin{align*}
\lesssim&\delta\Omega_0^{-2}\|\eta,\etab\|_{\L^\infty_{\ub}\H^4(u)}\|\Omega\tr\chi,\Omega\chih,L\phi\|_{\L^\infty_{\ub}\H^4(u)}^2\\
\lesssim&\delta\Omega_0^{-2}\cdot C^{\frac{1}{4}}\delta|u|^{-2}\mathscr{F}\mathscr{E}\mathscr{W}^{\frac{1}{2}}\mathcal{A}\cdot C^{\frac{1}{2}}|u|^{-2}\mathscr{F}^2\mathcal{A}^2\\
\lesssim&C^{-1}\Omega_0^{-2}\delta|u|^{-3}\mathscr{F}^2\mathscr{E}\mathcal{A}^2.
\end{align*}
The \engordnumber{2} line involves fifth order derivative of the connection coefficients and $L\phi$, which should be estimated using \eqref{estimate-Lphi}, \eqref{estimate-chih}, \eqref{estimate-nablas5Lphi} and \eqref{estimate-nablas5Omegachih}:
\begin{align*}
&\delta\Omega_0^{-2}\|\Omega_0^2\nablas\tr\chi'\|_{\L^2_{\ub}\H^4(u)}\|\Omega\tr\chi\|_{\L^\infty_{\ub}\H^4(u)}+\delta\Omega_0^{-2}\|\nablas(\Omega\chih),\nablas L\phi\|_{\L^2_{\ub}\H^4(u)}\|\Omega\chih,L\phi\|_{\L^\infty_{\ub}\H^4(u)}\\
\lesssim&\delta\Omega_0^{-2}(C^{\frac{1}{4}}\delta|u|^{-3}\mathscr{F}^2\mathscr{E}\mathcal{A}^2\cdot|u|^{-1}\mathscr{F}\mathcal{A}+|u|^{-2}\mathscr{F}\mathscr{E}(\mathcal{A}+\mathcal{R}+\mathcal{A}^{-\frac{1}{2}}\mathcal{R}^{\frac{3}{2}})\cdot|u|^{-1}\mathscr{F}(\mathcal{A}+\mathcal{R}+\mathcal{A}^{-\frac{1}{2}}\mathcal{R}^{\frac{3}{2}})\\\lesssim&\Omega_0^{-2}\delta|u|^{-3}\mathscr{F}^2\mathscr{E}(\mathcal{A}+\mathcal{R}+\mathcal{A}^{-\frac{1}{2}}\mathcal{R}^{\frac{3}{2}})^2.
\end{align*}
Combining the above estimates, we have
\begin{align}\label{estimate-nablas5trchi}
\|(|u|\nablas)\tr\chi'\|_{\H^4(\ub,u)}\lesssim\Omega_0^{-2}\delta|u|^{-2}\mathscr{F}^2\mathscr{E}(\mathcal{A}+\mathcal{R}+\mathcal{A}^{-\frac{1}{2}}\mathcal{R}^{\frac{3}{2}})^2.
\end{align}

{\bf \underline{Estimates for $\Omega\tr\chib$ involving the top order derivatives}:} The next step is to estimate $\nablas^5(\Omega\tr\chib)$. This is the most subtle estimate in the whole argument. Before this, we should first obtain an estimate of $\nablas^i\omegab$.
Recall the elliptic-transport system for $\omegab$:
\begin{align*}
&\begin{cases}\Deltas\omegab&=\omegabs+\divs(\Omega\betab+\Lb\phi\nablas\phi),\\
D\omegabs&+\Omega\tr\chi\omegabs+2\Omega\chih\cdot\nablas\nablas\omegab+2\divs(\Omega\chih)\cdot\nablas\omegab-\frac{1}{2}\divs(\Omega\tr\chi(\Omega\betab+\Lb\phi\nablas\phi))\\
&+\nablas(\Omega^2)\cdot(\ds(\rho+\frac{1}{6}\mathbf{R})+{}^*\ds\sigma)+\Deltas(\Omega^2)(\rho+\frac{1}{6}\mathbf{R})-\Deltas(\Omega^2(2\eta\cdot\etab-|\eta|^2))\\
&-\divs(\Omega\chih\cdot(\Omega\betab+\Lb\phi\nablas\phi)-2\Omega\chibh\cdot(\Omega\beta-L\phi\nablas\phi)+3\Omega^2\etab(\rho+\frac{1}{6}\mathbf{R})-3\Omega^2{}^*\etab\sigma)\\
\phantom{\Delta}=&-\divs\{2\Omega^2\ds\phi\Deltas\phi+\Lb\phi\nablas L\phi+L\phi\nablas\Lb\phi-\Omega\tr\chib L\phi\ds\phi\\
&+2\Omega\chibh\cdot\ds\phi L\phi+2\Omega^2\etab\cdot\nablas\phi\ds\phi+\Omega^2\etab|\ds\phi|^2)\}\end{cases}.
\end{align*}

Because $\omegab$ is a function and satisfies a Poisson equation, the elliptic estimate for $\nablas\omegab$ does not depend on $\omegab$ itself. Therefore no bounds for the lower derivatives of $\omegab$ are needed. We have
\begin{align}\label{elliptic-omegab}
\|(|u|\nablas)\omegab\|_{\H^4(\ub,u)}\lesssim|u|^2\|\omegabs\|_{\H^3(\ub,u)}+|u|\|\Omega\betab+\Lb\phi\nablas\phi\|_{\H^4(\ub,u)}.
\end{align}

The second equation above looks complicated. We can estimate the right hand side term by term, but it is not hard to see the equation can be written in the following schematic form: (we will write $\nablas\log\Omega\sim|u|^{-1}$, $\nablas^2\log\Omega\sim\nablas(\eta+\etab)$ and use $K=-\frac{1}{4}\tr \chi\tr\chib+\frac{1}{2}(\chih,\chibh)-(\rho+\frac{1}{6}\mathbf{R})+|\ds\phi|^2$, $\check{\sigma}=\sigma-\frac{1}{2}\chih\wedge\chibh$)
\begin{equation}\label{equ-Domegabs-schematic}
\begin{split}
D\omegabs=&\Omega\tr\chi\nablas^2\omegab+\sum_{i+j=1}\nablas^i(\Omega\chih)\nablas^j\nablas\omegab\\
&+\sum_{i+j=1}\nablas^i(\Omega\chih\ \text{or}\ \Omega\tr\chi)\nablas^j(\Omega\betab+\Lb\phi\nablas\phi)\\
&+\Omega^2\sum_{i+j=1}\nablas^i(\eta\ \text{or}\ \etab)\nablas^j(K\ \text{or}\ \check{\sigma})\\
&+\sum_{i+j+k=1}\nablas^i(\eta\ \text{or}\ \etab)\nablas^j(\Omega\chih\ \text{or}\ \Omega\tr\chi)\nablas^k(\Omega\chibh\ \text{or}\ \Omega\tr\chib)+\Omega^2\nablas(\eta\ \text{or}\ \etab)\cdot(\eta\ \text{or}\ \etab)^2\\
&+\Omega^2\sum_{i+j=2}\nablas^i(\eta\ \text{or}\ \etab)\nablas^j(\eta\ \text{or}\ \etab)\\
&+\sum_{i+j=1}\nablas^i(\Omega\chibh)\nablas^j(\Omega\beta-L\phi\nablas\phi)\\
&+\Omega^2\sum_{i+j=1}\nablas^i\nablas\phi\nablas^j\nablas^2\phi+\sum_{i+j=2}\nablas^i\Lb\phi\nablas^jL\phi\\
&+\Omega^2\sum_{i+j+k=1}\nablas^i(\eta\ \text{or}\ \etab)\nablas^j\nablas\phi\nablas^k\nablas\phi+\sum_{i+j+k=1}\nablas^i(\Omega\tr\chib\ \text{or}\ \Omega\chibh)\nablas^jL\phi\nablas^k\nablas\phi
\end{split}
\end{equation}
We want to estimate $\omegabs$ in $\H^3$ so the right hand side should be estimated in $\delta\|\cdot\|_{\L^1_{\ub}\H^3(u)}\lesssim\delta\|\cdot\|_{\L^2_{\ub}\H^3(u)}$. The first two lines, placing $\Omega\tr\chi$ and $\Omega\chih$ in $\L^\infty_{\ub}\H^4$,  are estimated by
\begin{align*}
\lesssim\delta|u|^{-2}\cdot C^{\frac{1}{4}}|u|^{-1}\mathscr{F}\mathcal{A}(\|(|u|\nablas)\omegab\|_{\L^1_{\ub}\H^4(u)}+|u|\|\Omega\betab+\Lb\phi\nablas\phi\|_{\L^1_{\ub}\H^4(u)}).
\end{align*}
We have also used the relation $\nablas\sim|u|^{-1}$ here and will use this in the following steps. The \engordnumber{3} line, placing $\eta$ or $\etab$ in $\L^\infty_{\ub}\H^4$, is estimated by
\begin{align*}
\lesssim&\delta|u|^{-1}\Omega_0\cdot C^{\frac{1}{4}}\delta|u|^{-2}\mathscr{F}\mathscr{E}\mathscr{W}^{\frac{1}{2}}\mathcal{A}\cdot\Omega_0\|K,\check{\sigma}\|_{\L^2_{\ub}\H^4(u)}\lesssim\delta|u|^{-1}\Omega_0\cdot C^{\frac{1}{4}}\delta|u|^{-2}\mathscr{F}\mathscr{E}\mathscr{W}^{\frac{1}{2}}\mathcal{A}\cdot|u|^{-2}\mathscr{F}^{\frac{1}{2}}\mathscr{E}\mathcal{A}^{\frac{1}{2}}\\
\lesssim &C^{-\frac{1}{2}}\delta|u|^{-4}\mathscr{F}\mathscr{E}\mathcal{A}\cdot(\Omega_0^2\delta|u|^{-1}\mathscr{E}^2)^{\frac{1}{2}}.
\end{align*}
The \engordnumber{4} line consists of terms with three factors, all of which are placed in $\H^4$. The \engordnumber{1} term of this line is estimated by
\begin{align*}
\lesssim\delta|u|^{-1}\cdot C^{\frac{1}{4}}\delta|u|^{-2}\mathscr{F}\mathscr{E}\mathscr{W}^{\frac{1}{2}}\mathcal{A}\cdot C^{\frac{1}{4}}|u|^{-1}\mathscr{F}\mathcal{A}\cdot|u|^{-1}\lesssim C^{-\frac{1}{2}} \delta|u|^{-4}\mathscr{F}\mathscr{E}\mathcal{A}\cdot(\delta|u|^{-1}\mathscr{F}\mathcal{A})^{\frac{1}{2}}
\end{align*}
and the \engordnumber{2} term of this line is estimated by
\begin{align*}
\lesssim\Omega_0^2\delta|u|^{-1}\cdot(C^{\frac{1}{4}}\delta|u|^{-2}\mathscr{F}\mathscr{E}\mathscr{W}^{\frac{1}{2}}\mathcal{A})^3\lesssim C^{-1}\delta|u|^{-4}\mathscr{F}\mathscr{E}\mathcal{A}\cdot(\delta|u|^{-1}\mathscr{F}\mathcal{A}).
\end{align*}
Note that $\nablas\phi$ shares the same estimate with $\eta,\etab$ and $L\phi$ shares the same estimate with $\Omega\chih$ for lower order (up to fourth order derivatives), so the last line can be estimated in the same manner as above. The \engordnumber{6} line, placing $\Omega\chibh$ in $\L^\infty_{\ub}\H^4$, is estimated by
\begin{align*}
\lesssim&\delta|u|^{-1}\cdot C^{\frac{1}{4}}\delta|u|^{-2}\mathscr{F}\mathcal{A}\cdot\|\Omega\beta-L\phi\nablas\phi\|_{\L^2_{\ub}\H^4(u)}\lesssim\delta|u|^{-1}\cdot C^{\frac{1}{4}}\delta|u|^{-2}\mathscr{F}\mathcal{A}\cdot C^{\frac{1}{4}}|u|^{-2}\mathscr{F}\mathscr{E}\mathcal{A}\\
\lesssim&C^{-\frac{1}{2}} \delta|u|^{-4}\mathscr{F}\mathscr{E}\mathcal{A}\cdot(\delta|u|^{-1}\mathscr{F}\mathcal{A})^{\frac{1}{2}}.\end{align*}

It remains to estimate the \engordnumber{5} and the \engordnumber{7} lines. They contain terms with factors which are second order derivative of $\eta,\etab,\nablas\phi,L\phi,\Lb\phi$, which are of the highest order. The \engordnumber{5} line and the \engordnumber{1} term of the \engordnumber{7} line, placing the highest order $\eta,\etab$ and $\nablas\phi$ in $\L^2_{\ub}\H^5(u)$ and the other factor in $\L^\infty_{\ub}\H^4$, are estimated by
\begin{align*}
\Omega_0^2\delta|u|^{-2}\cdot C^{\frac{1}{4}}\delta|u|^{-2}\mathscr{F}\mathscr{E}\mathscr{W}^{\frac{1}{2}}\mathcal{A}\cdot \Omega_0^{-1}C^{\frac{1}{4}}\delta^{\frac{1}{2}}|u|^{-\frac{3}{2}}\mathscr{F}\mathscr{E}\mathscr{W}^{\frac{1}{2}}\mathcal{A}\lesssim C^{-\frac{1}{2}}\delta|u|^{-4}\mathscr{F}\mathscr{E}\mathcal{A}\cdot(\delta|u|^{-1}\mathscr{F}\mathcal{A}).
\end{align*}

Finally, we turn the \engordnumber{2} term of the \engordnumber{7} line, which is the most subtle term. If at least one derivative has applied to $\Lb\phi$, then $L\phi$ can be placed in $\L^\infty_{\ub}\H^4$ (using \eqref{estimate-Lphi} with $\widetilde{\mathcal{E}}$ dropped) and this term is estimated by
\begin{align*}
\lesssim\delta|u|^{-2}\cdot\|(|u|\nablas)\Lb\phi\|_{\L^1_{\ub}\H^4(u)}\cdot|u|^{-1}\mathscr{F}\mathcal{A}=\delta|u|^{-3}\mathscr{F}\mathcal{A}\|(|u|\nablas)\Lb\phi\|_{\L^1_{\ub}\H^4(u)}.
\end{align*}
If no derivatives have applied to $\Lb\phi$, then the term, using \eqref{estimate-nablas5Lphi-improve}, is estimated by
\begin{align*}
\lesssim\delta|u|^{-3}\|\Lb\phi\|_{\L^\infty_{\ub}\H^4(u)}\left(\||u|(|u|\nablas)L\phi\|_{\L^2_{\ub}\H^4(u_0)}+\mathscr{F}\mathscr{E}\mathcal{A}\cdot\left(\delta|u|^{-1}\mathscr{W}(\mathscr{F}\mathcal{A}+\Omega_0^2(u_0)\mathscr{E}^2)\right)^{\frac{1}{8}}\right).
\end{align*}

Combining all estimates above, and using \eqref{elliptic-omegab}, we have
\begin{equation}\label{estimate-omegab-step1}
\begin{split}
&\|(|u|\nablas)\omegab\|_{\H^4(\ub,u)}\\
\lesssim&C^{\frac{1}{4}}\delta|u|^{-1}\mathscr{F}\mathcal{A}\|(|u|\nablas)\omegab\|_{\L^1_{\ub}\H^4(u)}+C^{\frac{1}{4}}\delta\mathscr{F}\mathcal{A}\|\Omega\betab+\Lb\phi\nablas\phi\|_{\L^1_{\ub}\H^4(u)}\\
&+|u|\|\Omega\betab+\Lb\phi\nablas\phi\|_{\H^4(\ub,u)}+\delta|u|^{-1}\mathscr{F}\mathcal{A}\|(|u|\nablas)\Lb\phi\|_{\L^1_{\ub}\H^4(u)}\\
&+\delta|u|^{-1}\|\Lb\phi\|_{\L^1_{\ub}\H^4(u)}\left(\||u|(|u|\nablas)L\phi\|_{\L^2_{\ub}\H^4(u_0)}+\mathscr{F}\mathscr{E}\mathcal{A}\cdot\left(\delta|u|^{-1}\mathscr{W}(\mathscr{F}\mathcal{A}+\Omega_0^2(u_0)\mathscr{E}^2)\right)^{\frac{1}{8}}\right)\\
&+\delta|u|^{-2}\mathscr{F}\mathscr{E}\mathcal{A}\cdot(\delta|u|^{-1}(\mathscr{F}\mathcal{A}+\Omega_0^2(u_0)\mathscr{E}^2))^{\frac{1}{2}}.
\end{split}
\end{equation}

Then we estimate $\|\Omega_0^{-1}|u|^2\cdot\|_{\L^1_{[u_0,u]}}$ of the above quantity, which means that we compute $\int_{u_0}^u\Omega_0^{-1}(u')|u'|\cdot\D u'$ of the above inequality. Note that $\L^1_{\ub}$ and $\L^1_{[u_0,u]}$ commute and then $\||u|^{-s}\cdot\|_{\L^1_{[u_0,u]}\L^1_{\ub}\H^i}=\||u|^{-s}\cdot\|_{\L^1_{\ub}\L^1_{[u_0,u]}\H^i}\lesssim|u|^{-s}\|\cdot\|_{\L^\infty_{\ub}\L^2_{[u_0,u]}\H^i}$ if $s>0$, we have
\begin{equation}\label{estimate-omegab-step2}
\begin{split}
&\|\Omega_0^{-1}|u|^2(|u|\nablas)\omegab\|_{\L^1_{[u_0,u]}\H^4(\ub)}\\\lesssim&C^{\frac{1}{4}}\delta|u|^{-1}\mathscr{F}\mathcal{A}\|\Omega_0^{-1}|u|^2(|u|\nablas)\omegab\|_{\L^\infty_{\ub}\L^1_{[u_0,u]}\H^4}+|u|^{-\frac{1}{2}}\|\Omega_0^{-1}|u|^{\frac{7}{2}}(\Omega\betab+\Lb\phi\nablas\phi)\|_{\L^\infty_{\ub}\L^2_{[u_0,u]}\H^4}\\
&+\delta|u|^{-1}\mathscr{F}\mathcal{A}\|\Omega_0^{-1}|u|^2(|u|\nablas)\Lb\phi\|_{\L^\infty_{\ub}\L^2_{[u_0,u]}\H^4}\\
&+\delta\Omega_0^{-1}(u)\|1\|_{\L^\infty_{\ub}\L^2_{[u_0,u]}\H^4}\||u|\Lb\phi\|_{\L^\infty_{\ub}\L^2_{[u_0,u]}\H^4}\||u|(|u|\nablas)L\phi\|_{\L^2_{\ub}\H^4(u_0)}\\
&+\delta\Omega_0^{-1}(u)\||u|\Lb\phi\|_{\L^\infty_{\ub}\L^2_{[u_0,u]}\H^4}\mathscr{F}\mathscr{E}\mathcal{A}\cdot\left(\delta|u|^{-1}\mathscr{W}(\mathscr{F}\mathcal{A}+\Omega_0^2(u_0)\mathscr{E}^2)\right)^{\frac{1}{8}}\\
&+\delta\Omega_0^{-1}(u)\mathscr{F}\mathscr{E}\mathcal{A}\cdot(\delta|u|^{-1}(\mathscr{F}\mathcal{A}+\Omega_0^2(u_0)\mathscr{E}^2))^{\frac{1}{2}}.
\end{split}
\end{equation}
Here $\|1\|_{\L^\infty_{\ub}\L^2_{[u_0,u]}\H^4}\lesssim\left|\log\frac{|u_1|}{|u_0|}\right|^{\frac{1}{2}}$ will contribute for the logarithmic loss. We assume
\begin{equation}\label{bootstrap-omegab}
\|\Omega_0^{-1}|u|^2(|u|\nablas)\omegab\|_{\L^\infty_{\ub}\L^1_{[u_0,u]}\H^4}\lesssim C^{\frac{1}{4}}\Omega_0^{-1}(u)\delta\mathscr{F}\mathscr{E}\mathscr{W}^{\frac{1}{2}}\mathcal{A}.
\end{equation}
Then \eqref{estimate-omegab-step2} implies, recalling the definition of $\mathscr{E}$,
\begin{equation}\label{estimate-omegab-step3}
\begin{split}
&\|\Omega_0^{-1}|u|^2(|u|\nablas)\omegab\|_{\L^1_{[u_0,u]}\H^4(\ub)}\\\lesssim&\delta|u|^{-1}\mathscr{F}\mathcal{A}\cdot C^{\frac{1}{4}}\Omega_0^{-1}(u)\delta\mathscr{F}\mathscr{E}\mathscr{W}^{\frac{1}{2}}\mathcal{A}+|u|^{-\frac{1}{2}}\cdot C^{\frac{3}{8}}\Omega_0^{-1}(u)\delta^{\frac{3}{2}}\mathscr{F}^{\frac{3}{2}}\mathscr{E}\mathcal{A}^{\frac{3}{2}}\\
&+\delta|u|^{-1}\mathscr{F}\mathcal{A}\cdot\Omega_0^{-1}(u)\delta\mathscr{F}\mathscr{E}\mathcal{A}\\
&+\delta\Omega_0^{-1}(u)\mathscr{F}\mathscr{E}\mathscr{W}^{\frac{1}{2}}\mathcal{A}\\
&+\delta\Omega_0^{-1}(u)\mathscr{F}\mathscr{E}\mathscr{W}^{\frac{1}{2}}\mathcal{A}\cdot\left(\delta|u|^{-1}\mathscr{W}(\mathscr{F}\mathcal{A}+\Omega_0^2(u_0)\mathscr{E}^2)\right)^{\frac{1}{8}}\\
&+\delta\Omega_0^{-1}(u)\mathscr{F}\mathscr{E}\mathcal{A}\cdot(\delta|u|^{-1}(\mathscr{F}\mathcal{A}+\Omega_0^2(u_0)\mathscr{E}^2))^{\frac{1}{2}}\\
\lesssim&\delta\Omega_0^{-1}(u)\mathscr{F}\mathscr{E}\mathscr{W}^{\frac{1}{2}}\mathcal{A}.
\end{split}
\end{equation}
If $C_0$ is sufficiently large, this estimate improves \eqref{bootstrap-omegab} and this implies that \eqref{estimate-omegab-step3} holds without assuming \eqref{bootstrap-omegab}.

Now we are ready to the estimate for $\nablas^5(\Omega\tr\chib)$. Recall the equation for $\Db(\Omega\tr\chib)$:
\begin{equation}\label{equ-Dbtrchib}
\Db(\Omega\tr\chib)=-\frac{1}{2}(\Omega\tr\chib)^2+2\omegab\Omega\tr\chib-|\Omega\chibh|^2-2(\Lb\phi)^2.
\end{equation}
Apply $\nablas$ to the above equation, and write it in the following form:
\begin{equation}\label{equ-Dbnablastrchib}
\begin{split}
&\Db(\Omega^{-2}\nablas(\Omega\tr\chib))+\Omega\tr\chib(\Omega^{-2}\nablas(\Omega\tr\chib))\\
&=2\Omega^{-2}\nablas\omegab(\Omega\tr\chib)-2\Omega^{-2}(\Omega\chibh)\cdot\nablas(\Omega\chibh)
-4\Omega^{-2}\Lb\phi\nablas\Lb\phi.\\
\end{split}
\end{equation}
The right hand side should be estimated in $\||u|^4\cdot\|_{\L^1_{[u_0,u]}\H^4(\ub)}$. The \engordnumber{1} term is estimated by, using $\|\Omega\tr\chib\|_{\H^4(\ub,u)}\lesssim|u|^{-1}$, 
\begin{align*}
\|\Omega^{-2}|u|^3(|u|\nablas)\omegab\Omega\tr\chib\|_{\L^1_{[u_0,u]}\H^4(\ub)}\lesssim\Omega_0^{-1}\|\Omega_0^{-1}|u|^2(|u|\nablas)\omegab\|_{\L^1_{[u_0,u]}\H^4(\ub)}\lesssim\Omega_0^{-2}\delta\mathscr{F}\mathscr{E}\mathscr{W}^{\frac{1}{2}}\mathcal{A}.
\end{align*}
The \engordnumber{2} term is estimated by, \begin{align*}
&\|\Omega^{-2}|u|^3(|u|\nablas)(\Omega\chibh)\cdot(\Omega\chibh)\|_{\L^1_{[u_0,u]}\H^4(\ub)}\lesssim C^{\frac{1}{4}}\delta\mathscr{F}\mathcal{A}\|\Omega_0^{-2}|u|\Omega\chibh\|_{\L^1_{[u_0,u]}\H^5(\ub)}\\
\lesssim&\Omega_0^{-1}C^{\frac{1}{4}}\delta\mathscr{F}\mathcal{A}\||u|^{-\frac{1}{2}}\|_{\L^2_{[u_0,u]}}\|\Omega_0^{-1}|u|^{\frac{3}{2}}\Omega\chibh\|_{\L^2_{[u_0,u]}\H^5(\ub)}\lesssim C^{\frac{1}{2}}\Omega_0^{-2}\delta^2|u|^{-1}\mathscr{F}^2\mathscr{E}\mathscr{W}^{\frac{1}{2}}\mathcal{A}^2\\
\lesssim& C^{-1}\Omega_0^{-2}\delta\mathscr{F}\mathscr{E}\mathscr{W}^{\frac{1}{2}}\mathcal{A}.
\end{align*}
The last term is estimated by, using \eqref{estimate-nablas5nablasphi} and \eqref{estimate-Lbphi1},
\begin{align*}
&\|\Omega^{-2}|u|^3(|u|\nablas)(\Lb\phi)\cdot\Lb\phi\|_{\L^1_{[u_0,u]}\H^4(\ub)}\\\lesssim&\Omega_0^{-1}\|\Omega^{-1}|u|^2(|u|\nablas)(\Lb\phi)\|_{\L^2_{[u_0,u]}\H^4(\ub)}\||u|\Lb\phi\|_{\L^2_{[u_0,u]}\H^4(\ub)}\\
\lesssim&\Omega_0^{-2}\delta\mathscr{F}\mathscr{E}\mathscr{W}^{\frac{1}{2}}\mathcal{A}.
\end{align*}
Combining the above estimates,  we have
\begin{equation}\label{estimate-Omegatrchib}
\||u|^2(|u|\nablas)(\Omega\tr\chib)\|_{\H^4(\ub,u)}\lesssim\delta\mathscr{F}\mathscr{E}\mathscr{W}^{\frac{1}{2}}\mathcal{A}.
\end{equation}

{\bf \underline{Estimates for $\Omega\chibh$ involving the top order derivatives}:} Consider the equation for $\divs(\Omega\chibh)$:
\begin{equation}\label{equ-divschibh}
\divs(\Omega\chibh)=\frac{1}{2}\ds(\Omega\tr \chib)+\Omega\chibh\cdot\eta-\frac{1}{2}\Omega\tr \chib\eta+(\Omega\betab+\Lb\phi\ds\phi).\end{equation}
Using \eqref{estimate-Omegatrchib}, the right hand side is estimated in $\H^4(\ub,u)$ by, using in particular \eqref{estimate-eta},
\begin{align*}
\lesssim\delta|u|^{-3}\mathscr{F}\mathscr{E}\mathscr{W}^{\frac{1}{2}}\mathcal{A}+|u|^{-1}\cdot\delta|u|^{-3}\mathscr{F}\mathscr{E}(\mathcal{A}+\mathcal{R})+\|\Omega\betab+\Lb\phi\nablas\phi\|_{\H^4(\ub,u)}.
\end{align*}
By elliptic estimate, using \eqref{estimate-chibh} by replacing $\widetilde{\mathcal{O}}$ by $\mathcal{R}+\mathcal{A}^{-\frac{1}{2}}\mathcal{R}^{\frac{3}{2}}$,
\begin{align*}
\|\Omega\chibh\|_{\H^5(\ub,u)}\lesssim&|u|\|\divs(\Omega\chibh)\|_{\H^4(\ub,u)}+\mathscr{E}\|\Omega\chibh\|_{\H^4(\ub,u)}\\
\lesssim&\delta|u|^{-2}\mathscr{F}\mathscr{E}\mathscr{W}^{\frac{1}{2}}(\mathcal{A}+\mathcal{R}+\mathcal{A}^{-\frac{1}{2}}\mathcal{R}^{\frac{3}{2}})+|u|\|\Omega\betab+\Lb\phi\nablas\phi\|_{\H^4(\ub,u)}.
\end{align*}
Then
\begin{equation}\label{estimate-Omegachibh}
\begin{split}
&|u|^{\frac{1}{2}}\|\Omega_0^{-1}|u|^{\frac{3}{2}}\Omega\chibh\|_{\L^2_{[u_0,u]}\H^5(\ub)}\\
\lesssim&\Omega_0^{-1}\delta\mathscr{F}\mathscr{E}\mathscr{W}^{\frac{1}{2}}(\mathcal{A}+\mathcal{R}+\mathcal{A}^{-\frac{1}{2}}\mathcal{R}^{\frac{3}{2}})+|u|^{-\frac{1}{2}}\|\Omega_0^{-1}|u|^{\frac{7}{2}}(\Omega\betab+\Lb\phi\nablas\phi)\|_{_{\L^2_{[u_0,u]}\H^4(\ub)}}\\
\lesssim&\Omega_0^{-1}\delta\mathscr{F}\mathscr{E}\mathscr{W}^{\frac{1}{2}}(\mathcal{A}+\mathcal{R}+\mathcal{A}^{-\frac{1}{2}}\mathcal{R}^{\frac{3}{2}})+C^{\frac{3}{8}}\Omega_0^{-1}|u|^{-\frac{1}{2}}\delta^{\frac{3}{2}}\mathscr{F}^{\frac{3}{2}}\mathscr{E}\mathcal{A}^{\frac{3}{2}}\\
\lesssim&\Omega_0^{-1}\delta\mathscr{F}\mathscr{E}\mathscr{W}^{\frac{1}{2}}(\mathcal{A}+\mathcal{R}+\mathcal{A}^{-\frac{1}{2}}\mathcal{R}^{\frac{3}{2}}).
\end{split}
\end{equation}
This concludes the estimate for $\nablas^5(\Omega\chibh)$.
We have completed the proof of Proposition \ref{estimate-tildeO}.

\end{proof}

\subsection{Estimates for $\mathcal{R}$}

Finally we turn to the estimates for $\mathcal{R}$.
\begin{proposition}\label{estimate-R}
Under the assumptions of Theorem \ref{existencetheorem} and the bootstrap assumptions \eqref{bootstrap}, we have
$$\mathcal{R}\lesssim\mathcal{A}.$$
\end{proposition}

\begin{proof}
The proof is by using the renormalized Bianchi equations. At first, we consider the equations for $\Db(\Omega\beta-L\phi\nablas\phi)$-$DK$-$D\check{\sigma}$. 
\begin{align*}
D(K-|u|^{-2})&+\Omega\tr\chi K+\divs(\Omega\beta-L\phi\nablas\phi)\\
&-\Omega\chih\cdot\nablas\etab+\frac{1}{2}\Omega\tr\chi\divs\etab+(\Omega\beta-L\phi\nablas\phi)\cdot\etab-\Omega\chih\cdot\etab\cdot\etab+\frac{1}{2}\Omega\tr\chi|\etab|^2=0\\
D\sigmac&+\frac{3}{2}\Omega\tr\chi\sigmac+\curls(\Omega\beta-L\phi\nablas\phi)+\frac{1}{2}\Omega\chih\wedge(\etab\tensor\etab+\nablas\tensor\etab)\\
&+\etab\wedge(\Omega\beta-L\phi\nablas\phi)+2\nablas L\phi\wedge\nablas\phi=0\\
\Db(\Omega\beta-L\phi\nablas\phi)&+\frac{1}{2}\Omega\tr\chib(\Omega\beta-L\phi\nablas\phi)-\Omega\chibh\cdot(\Omega\beta-L\phi\nablas\phi)+\Omega^2\ds (K-|u|^{-2})-\Omega^2{}^*\ds\sigmac\\
&+3\Omega^2(\eta K-{}^*\eta\sigmac)-\frac{1}{2}\Omega^2(\ds(\chih,\chibh)+{}^*\ds(\chih\wedge\chibh))-\frac{3}{2}\Omega^2(\eta(\chih,\chibh)+{}^*\eta(\chih\wedge\chibh))\\
&+\frac{1}{4}\Omega^2\ds(\tr\chi\tr\chib)+\frac{3}{4}\Omega^2\tr\chi\tr\chib\eta-2\Omega\chih\cdot(\Omega\betab+\Lb\phi\nablas\phi)\\
=&-2\Omega^2\Deltas\phi\nablas\phi+\Omega^2\ds|\ds\phi|^2-2\Omega\chih\cdot\nablas\phi\Lb\phi+\Omega\tr\chi\Lb\phi\nablas\phi\\
&-2\Omega^2\eta\cdot\nablas\phi\nablas\phi+2\Omega^2\eta|\ds\phi|^2,
\end{align*}
We compute
\begin{align*}
&\Db(|u|^2|(|u|\nablas)^i(\Omega\beta-L\phi\nablas\phi)|^2\D\mu_{\gs})+D(\Omega^2|u|^2|(|(|u|\nablas)^i(K-|u|^{-2})|^2+|(|u|\nablas)^i\check{\sigma}|^2)\D\mu_{\gs})\\
=&-|u|^{2+2i}\nablas^A(\Omega^2\nablas^i_{B_1\cdots B_i}(\Omega\beta-L\phi\nablas\phi)_A\nablas^{i,B_1\cdots B_i}(K-|u|^{-2})\\&+\Omega^2\nablas^i_{B_1\cdots B_i}{}^*(\Omega\beta-L\phi\nablas\phi)_A\nablas^{i,B_1\cdots B_i}\sigmac)+|u|^2\tau_{3}
\end{align*}
for $0\le i\le4$, where $\tau_3$ contains no fifth order derivative of curvature components. The weight $|u|^2$ respects the coefficient $\frac{1}{2}$ of the term $\frac{1}{2}\Omega\tr\chib(\Omega\beta-L\phi\nablas\phi)$ of the third equation. By divergence theorem, we have
\begin{align*}
&\delta\||u|^2(\Omega\beta-L\phi\nablas\phi)\|_{\L^2_{\ub}\H^4(u)}^2+\||u|^{\frac{5}{2}}\Omega(K-|u|^{-2},\check{\sigma})\|_{\L^2_{[u_0,u]}\H^4(\ub)}^2\\
\lesssim&\delta\||u|^2(\Omega\beta-L\phi\nablas\phi)\|_{\L^2_{\ub}\H^4(u_0)}^2+\||u|^{\frac{5}{2}}\Omega(K-|u|^{-2},\check{\sigma})\|_{\L^2_{[u_0,u]}\H^4(0)}^2\\
&+\int_0^{\delta}\D\ub'\int_{u_0}^u\D u'\int_{S_{\ub',u'}}|u|^2|\tau_3|\D\mu_{\gs}
\end{align*}
where
\begin{align*}
\int_{S_{\ub',u'}}|u'|^2|\tau_3|\D\mu_{\gs}\lesssim\int_{S_{\ub',u'}}|u'|^2|\tau_{3,1}|\D\mu_{\gs}+\int_{S_{\ub',u'}}|u'|^2|\tau_{3,2}|\D\mu_{\gs}
\end{align*}
and the multiplier terms
\begin{equation}\label{tau31}
\begin{split}
&\int_{S_{\ub',u'}}|u'|^2|\tau_{3,1}|\D\mu_{\gs}\\
\lesssim&|u'|^4\|\Omega\beta-L\phi\nablas\phi\|_4^2\|\Omega\tr\chib+2|u'|^{-2},\Omega\chibh\|_4\\
&+|u'|^4\Omega_0^2\cdot|u'|^{-1}\|K-|u|^{-2},\check{\sigma}\|_4\|\Omega\beta-L\phi\nablas\phi\|_4\\
&+|u'|^4\Omega_0^2\|\Omega\beta-L\phi\nablas\phi\|_4\|\eta,\etab\|_4\|K,K-|u|^{-2},\check{\sigma}\|_4\\
&+|u'|^4\|\Omega\beta-L\phi\nablas\phi\|_4|u'|^{-1}(\|(|u'|\nablas)(\Omega\chih)\|_4\|\Omega\chibh\|_4+\|\Omega\chih\|_4(|u'|\nablas)(\Omega\chibh)\|_4)\\
&+|u'|^4\|\Omega\beta-L\phi\nablas\phi\|_4|u'|^{-1}(\|(|u'|\nablas)(\Omega\tr\chi)\|_4\|\Omega\tr\chib\|_4+\|\Omega\tr\chi\|_4\|(|u'|\nablas)(\Omega\tr\chib)\|_4)\\
&+|u'|^4\|\Omega\beta-L\phi\nablas\phi\|_4\left(\|\eta\|_4\|\Omega\chih,\Omega\tr\chi\|_4\|\Omega\chibh,\Omega\tr\chib\|_4\right)\\
&+|u'|^4\|\Omega\beta-L\phi\nablas\phi\|_4\|\Omega\chih\|_4\|\Omega\betab+\Lb\phi\nablas\phi\|_4\\
&+|u'|^4\Omega_0^2\|\Omega\beta-L\phi\nablas\phi\|_4\|\nablas\phi\|_4|u'|^{-1}\|\nablas\phi\|_5\\
&+|u'|^4\|\Omega\beta-L\phi\nablas\phi\|_4\left(\|\Omega\chih,\Omega\tr\chi\|_4\|\nablas\phi\|_4\|\Lb\phi\|_4+\Omega_0^2\|\eta\|_4\|\nablas\phi\|_4^2\right)\\
&+|u'|^4\Omega_0^2\|K-|u|^{-2},\sigmac\|_4\|\Omega\chih, \Omega\tr\chi\|_4(\|K,K-|u|^{-2},\sigmac\|_4+|u'|^{-1}\|\etab\|_5)\\
&+|u'|^4\Omega_0^2\|K-|u|^{-2},\sigmac\|_4\|\|\Omega\chih,\Omega\tr\chi\|_4\|\etab\|_4^2\\
&+|u'|^4\Omega_0^2\|\sigmac\|_4|u'|^{-1}\|L\phi\|_5\|\nablas\phi\|_4\\
&+|u'|^4\Omega_0^2|\omega|\|K-|u|^{-2},\sigmac\|_4^2
\end{split}
\end{equation}
and commutation terms
\begin{equation}\label{tau32}
\begin{split}
&\int_{S_{\ub',u'}}|u'|^2|\tau_{3,2}|\D\mu_{\gs}\\
\lesssim&|u'|^4\|\Omega\beta-L\phi\nablas\phi\|_4\left(\|(|u'|\nablas)(\Omega\chib)\|_2\|\Omega\beta-L\phi\nablas\phi\|_3+\Omega_0^2(u')|u'|\|K\|_2\|K-|u|^{-2},\sigmac\|_4\right)\\
&+|u'|^4\Omega_0^2(u')\|K-|u|^{-2},\sigmac\|_4\left(\|(|u'|\nablas)(\Omega\chi)\|_3\|K-|u|^{-2},\sigmac\|_3+|u'|\|K\|_3\|\Omega\beta-L\phi\nablas\phi\|_4\right).
\end{split}
\end{equation}

\eqref{tau31} and \eqref{tau32} should be estimated in $L_{\ub}^1L^1_u$, similar to \eqref{tau11}, \eqref{tau12} and \eqref{tau21} and \eqref{tau22}, although more terms are needed to estimate. To do the estimates, the basic rule is that, we always place both $|u'|^2\|\Omega\beta-L\phi\nablas\phi\|_4$ and $\Omega_0|u'|^3\|K-|u|^{-2},\check{\sigma}\|_4$ in $\L^\infty_{[u_0,u]}\L_{\ub}^2$. The lower order derivatives of the connection coefficients and $L\phi$, $\nablas\phi$ should be placed in $\L_{\ub}^\infty\L^\infty_{[u_0,u]}\H^4$. 

The \engordnumber{1} line of the right hand side is estimated by, 
\begin{align*}
\lesssim\delta|u|^{-1}\cdot C^{\frac{1}{2}}\mathscr{F}^2\mathscr{E}^2\mathcal{A}^2\cdot C^{\frac{1}{4}}\delta\mathscr{F}\mathcal{A}\lesssim C^{-1}\delta\mathscr{F}^2\mathscr{E}^2\mathcal{A}^2.
\end{align*}
The \engordnumber{2} line is estimated by, \begin{align*}
\lesssim\delta|u|^{-1}\cdot C^{\frac{3}{8}}\delta\mathscr{F}^{\frac{3}{2}}\mathscr{E}\mathcal{A}^{\frac{3}{2}}\cdot C^{\frac{1}{4}}\mathscr{F}\mathscr{E}\mathcal{A} \lesssim C^{-1}\delta\mathscr{F}^2\mathscr{E}\mathcal{A}^2.
\end{align*}
The \engordnumber{3} line is estimated by, using \eqref{estimate-K2} to estimate $K$,
\begin{align*}
\lesssim\Omega_0(u_0)\delta|u|^{-1}\cdot C^{\frac{1}{4}}\mathscr{F}\mathscr{E}\mathcal{A}\cdot C^{\frac{1}{4}}\delta\mathscr{F}\mathscr{E}\mathscr{W}^\frac{1}{2}\mathcal{A}\cdot\mathscr{F}^{\frac{1}{2}}\mathscr{E}\mathcal{A}^{\frac{1}{2}}\lesssim C^{-\frac{1}{2}}\delta\mathscr{F}^2\mathscr{E}^2\mathcal{A}^2.
\end{align*}
The \engordnumber{1} term of the \engordnumber{4} line is estimated by, placing $|u'|\|(|u'|\nablas)(\Omega\chih)\|_4$ in $\L^\infty_{[u_0,u]}\L_{\ub}^2$,
\begin{align*}
\lesssim\delta|u|^{-1}\cdot C^{\frac{1}{4}}\mathscr{F}\mathscr{E}\mathcal{A}\cdot  C^{\frac{1}{4}}\mathscr{F}\mathscr{E}\mathcal{A}\cdot C^{\frac{1}{4}}\delta\mathscr{F}\mathcal{A}\lesssim C^{-1}\delta\mathscr{F}^2\mathscr{E}^2\mathcal{A}^2.
\end{align*}
The \engordnumber{1} term of the \engordnumber{5} line is estimated by, placing $|u'|\|(|u'|\nablas)(\Omega\tr\chi)\|_5$ in $\L_{\ub}^\infty\L^\infty_{[u_0,u]}$
\begin{align*}
\lesssim\delta|u|^{-1}\cdot C^{\frac{1}{4}}\mathscr{F}\mathscr{E}\mathcal{A}\cdot  C^{\frac{1}{2}}\delta\mathscr{F}^2\mathscr{E}\mathcal{A}^2\cdot 1\lesssim C^{-1}\delta\mathscr{F}^2\mathscr{E}^2\mathcal{A}^2.
\end{align*}
The \engordnumber{2} terms of the \engordnumber{4} and \engordnumber{5} lines are estimated by, placing $\|\Omega_0^{-1}|u'|^{\frac{3}{2}}(|u'|\nablas)(\Omega\chibh)\|_5$ in $\L_{\ub}^\infty\L^2_{[u_0,u]}$, and $\||u'|^2(|u'|\nablas)(\Omega\tr\chib)\|_5$ in $\L_{\ub}^\infty\L^\infty_{[u_0,u]}$, and using the auxiliary condition \eqref{auxiliary},
\begin{align*}
\lesssim\delta|u|^{-1}\cdot C^{\frac{1}{4}}\mathscr{F}\mathscr{E}\mathcal{A}\cdot  C^{\frac{1}{4}}\mathscr{F}\mathcal{A}\cdot C^{\frac{1}{4}}\Omega_0(u_0)\Omega_0^{-1}\delta\mathscr{F}\mathscr{E}\mathscr{W}^{\frac{1}{2}}\mathcal{A}\lesssim C^{-\frac{1}{4}}\delta\mathscr{F}^2\mathscr{E}^2\mathcal{A}^2.
\end{align*}
The \engordnumber{6} line is estimated by
\begin{align*}
\lesssim\delta|u|^{-1}\cdot C^{\frac{1}{4}}\mathscr{F}\mathscr{E}\mathcal{A}\cdot C^{\frac{1}{4}}\delta\mathscr{F}\mathscr{E}\mathcal{A}\cdot C^{\frac{1}{4}}\mathscr{F}\mathcal{A}\cdot1\lesssim C^{-1}\delta\mathscr{F}^2\mathscr{E}^2\mathcal{A}^2.
\end{align*}
The \engordnumber{7} line is estimated by, placing $\|\Omega_0^{-1}|u|^{\frac{7}{2}}(\Omega\betab+\Lb\phi\nablas\phi)\|_4$ in $\L_{\ub}^\infty\L^2_{[u_0,u]}$,
\begin{align*}
\lesssim\delta|u|^{-\frac{3}{2}}\cdot C^{\frac{1}{4}}\mathscr{F}\mathscr{E}\mathcal{A}\cdot  C^{\frac{1}{4}}\mathscr{F}\mathcal{A}\cdot C^{\frac{1}{4}}\Omega_0(u_0)\Omega_0^{-1}\delta^{\frac{3}{2}}\mathscr{F}^\frac{3}{2}\mathscr{E}\mathcal{A}^{\frac{3}{2}}\lesssim C^{-1}\delta\mathscr{F}^2\mathscr{E}^2\mathcal{A}^2.
\end{align*}
Here we have used the auxiliary condition \eqref{auxiliary}. The \engordnumber{8} line is estimated by, placing $\Omega_0|u|^{\frac{3}{2}}\|\nablas\phi\|_5$ in $\L^\infty_{[u_0,u]}\L_{\ub}^2$,
\begin{align*}
\lesssim\Omega_0(u_0)\delta|u|^{-\frac{3}{2}}\cdot C^{\frac{1}{4}}\mathscr{F}\mathscr{E}\mathcal{A}\cdot  C^{\frac{1}{4}}\delta\mathscr{F}\mathscr{E}\mathcal{A}\cdot C^{\frac{1}{4}}\delta^{\frac{1}{2}}\mathscr{F}\mathscr{E}\mathcal{A}\lesssim C^{-1}\delta\mathscr{F}^2\mathscr{E}^2\mathcal{A}^2.
\end{align*}
The \engordnumber{1} term of the \engordnumber{9} line is estimated by, using $\||u|\Lb\phi\|_{\L^\infty_{\ub}\L^2_{[u_0,u]}\H^4}\lesssim\mathscr{W}^{\frac{1}{2}}$,
\begin{align*}
\lesssim\delta|u|^{-1}\cdot C^{\frac{1}{4}}\mathscr{F}\mathscr{E}\mathcal{A}\cdot  C^{\frac{1}{4}}\mathscr{F}\mathcal{A}\cdot C^{\frac{1}{4}}\delta\mathscr{F}\mathscr{E}\mathcal{A}\cdot\mathscr{W}^{\frac{1}{2}}\lesssim C^{-1}\delta\mathscr{F}^2\mathscr{E}^2\mathcal{A}^2.
\end{align*}
The \engordnumber{2} term of the \engordnumber{9} line is estimated by,
\begin{align*}
\lesssim\delta|u|^{-3}\cdot C^{\frac{1}{4}}\mathscr{F}\mathscr{E}\mathcal{A}\cdot  (C^{\frac{1}{4}}\delta\mathscr{F}\mathscr{E}\mathcal{A})^3\lesssim C^{-1}\delta\mathscr{F}^2\mathscr{E}^2\mathcal{A}^2.
\end{align*}
The \engordnumber{10} and \engordnumber{13} lines are estimated in the same manner. $\omega$ has the worst estimate among $\omega,\Omega\chih,\Omega\tr\chi$, and $\Omega_0K$ has the worst estimate among $\Omega_0K,\Omega_0(K-|u|^{-2}),\Omega_0\sigmac,\Omega_0|u|^{-1}\etab$ where $\Omega_0|u|^{-1}\|\etab\|_5$ is placed in $\L^\infty_{[u_0,u]}\L^2_{\ub}$. Then these two lines are estimated by
\begin{align*}
\lesssim\delta|u|^{-1}\cdot C^{\frac{3}{8}}\delta\mathscr{F}^{\frac{3}{2}}\mathscr{E}\mathcal{A}^{\frac{3}{2}}\cdot C^{\frac{1}{4}}\mathscr{F}\mathscr{W}^{\frac{1}{2}}\mathcal{A}\cdot\mathscr{F}^{\frac{1}{2}}\mathscr{E}\mathcal{A}^{\frac{1}{2}}\lesssim C^{-1}\delta\mathscr{F}^2\mathscr{E}^2\mathcal{A}^2.
\end{align*}
The \engordnumber{11} line is estimated by
\begin{align*}
\lesssim\Omega_0(u_0)\delta|u|^{-3}\cdot C^{\frac{3}{8}}\delta\mathscr{F}^{\frac{3}{2}}\mathscr{E}\mathcal{A}^{\frac{3}{2}}\cdot C^{\frac{1}{4}}\mathscr{F}\mathcal{A}\cdot (C^{\frac{1}{4}}\delta\mathscr{F}\mathscr{E}\mathcal{A})^2\lesssim C^{-1}\delta\mathscr{F}^2\mathscr{E}^2\mathcal{A}^2.
\end{align*}
The \engordnumber{12} line is estimated by, placing $|u|\|L\phi\|_5$ in $\L^\infty_{[u_0,u]}\L^2_{\ub}$,
\begin{align*}
\lesssim\Omega_0(u_0)\delta|u|^{-2}\cdot C^{\frac{3}{8}}\delta\mathscr{F}^{\frac{3}{2}}\mathscr{E}\mathcal{A}^{\frac{3}{2}}\cdot C^{\frac{1}{4}}\mathscr{F}\mathscr{E}\mathcal{A}\cdot C^{\frac{1}{4}}\delta\mathscr{F}\mathscr{E}\mathcal{A}\lesssim C^{-\frac{1}{8}}\delta\mathscr{F}^2\mathscr{E}^2\mathcal{A}^2.
\end{align*}

Using the estimates \eqref{estimate-K3} for $K$ in $\H^3(\ub,u)$, \eqref{tau32} can be estimated in a similar way, by $\lesssim\delta\mathscr{F}^2\mathscr{E}^2\mathcal{A}^2$. Then we have, for $K(\ub=0)=|u|^{-2}$,
\begin{equation*}
\begin{split}
&\delta\||u|^2(\Omega\beta-L\phi\nablas\phi)\|_{\L^2_{\ub}\H^4(u)}^2+\||u|^{\frac{5}{2}}\Omega_0(K-|u|^{-2},\check{\sigma})\|_{\L^2_{[u_0,u]}\H^4(\ub)}^2\\
\lesssim&\delta\||u|^2(\Omega\beta-L\phi\nablas\phi)\|_{\L^2_{\ub}\H^4(u_0)}^2+\delta\mathscr{F}^2\mathscr{E}^2\mathcal{A}^2,
\end{split}
\end{equation*}
that is,
\begin{equation}\label{estimate-beta}
\begin{split}
&\||u|^2(\Omega\beta-L\phi\nablas\phi)\|_{\L^2_{\ub}\H^4(u)}+\delta^{-\frac{1}{2}}\||u|^{\frac{5}{2}}\Omega_0(K-|u|^{-2},\check{\sigma})\|_{\L^2_{[u_0,u]}\H^4(\ub)}\lesssim\mathscr{F}\mathscr{E}\mathcal{A}.
\end{split}
\end{equation}


Finally, we consider the equations for $\Db K$-$\Db\sigmac$-$D(\Omega\betab+\Lb\phi\nablas\phi)$:
\begin{align*}
\Db(K-\frac{1}{|u|^2})&+\frac{3}{2}\Omega\tr\chib (K-\frac{1}{|u|^2})+(\Omega\tr\chib+\frac{2}{|u|})\frac{1}{|u|^2}\\
=&\divs(\Omega\betab+\Lb\phi\nablas\phi)+\Omega\chibh\cdot\nablas\eta+\frac{1}{2}\Omega\tr\chib\mu\\
&+(\Omega\betab+\Lb\phi\nablas\phi)\cdot\eta+\Omega\chibh\cdot\eta\cdot\eta-\frac{1}{2}\Omega\tr\chib|\eta|^2,\\
\Db\sigmac&+\frac{3}{2}\Omega\tr\chi\sigmac+\curls(\Omega\betab+\Lb\phi\nablas\phi)\\
&-\frac{1}{2}\Omega\chibh\wedge(\eta\tensor\eta+\nablas\tensor\eta)+\eta\wedge(\Omega\betab+\Lb\phi\nablas\phi)-2\nablas \Lb\phi\wedge\nablas\phi=0\\
D(\Omega\betab+\Lb\phi\nablas\phi)&+\frac{1}{2}\Omega\tr\chi(\Omega\betab+\Lb\phi\nablas\phi)-\Omega\chih\cdot(\Omega\betab+\Lb\phi\nablas\phi)-\Omega^2\ds K-\Omega^2{}^*\ds\sigmac\\
&-3\Omega^2(\etab K+{}^*\etab\sigmac)+\frac{1}{2}\Omega^2(\ds(\chih,\chibh)-{}^*\ds(\chih\wedge\chibh))+\frac{3}{2}\Omega^2(\etab(\chih,\chibh)+{}^*\etab(\chih\wedge\chibh))\\
&-\frac{1}{4}\Omega^2\ds(\tr\chi\tr\chib)-\frac{3}{4}\Omega^2\tr\chi\tr\chib\etab-2\Omega\chibh\cdot(\Omega\beta-L\phi\nablas\phi)\\
=&2\Omega^2\Deltas\phi\nablas\phi-\Omega^2\ds|\ds\phi|^2+2\Omega\chibh\cdot\nablas\phi L\phi-\Omega\tr\chib L\phi\nablas\phi\\
&+2\Omega^2\etab\cdot\nablas\phi\nablas\phi-2\Omega^2\etab|\ds\phi|^2.
\end{align*}

Recall that we write the equation $\Db K$ in the form:
$$\Db(K-\frac{1}{|u|^2})+\frac{3}{2}\Omega\tr\chib (K-\frac{1}{|u|^2})=\cdots$$
in order to obtain good enough estimates. However, we have two terms which are still not good enough:
$$(\Omega\tr\chib+\frac{2}{|u|})\frac{1}{|u|^2},\ \frac{1}{2}\Omega\tr\chib\mu.$$
Fortunately, these two terms have better bounds after taking derivatives, in view of two improved estimates \eqref{estimate-detrchib}, \eqref{estimate-nablasmu} for $\Omega\tr\chib$ and $\mu$ respectively. So we will firstly estimate the derivatives of $K-|u|^{-2}$, $\sigmac$, $\Omega\betab+\Lb\phi\nablas\phi$.

Write
\begin{align*}
&\Db(|u|^4(|(|u|\nablas)^i(K-|u|^{-2})|^2+|(|u|\nablas)^i\check{\sigma}|^2)\D\mu_{\gs})\\
&+D(\Omega^{-2}|u|^4|(|u|\nablas)^i(\Omega\betab+\Lb\phi\nablas\phi)|^2\D\mu_{\gs})\\
=&|u|^{4+2i}\nablas^A(\nablas^i_{B_1\cdots B_i}(\Omega\betab+\Lb\phi\nablas\phi)_A\nablas^{i,B_1\cdots B_i}(K-|u|^{-2})\\&+\nablas^i_{B_1\cdots B_i}{}^*(\Omega\beta-L\phi\nablas\phi)_A\nablas^{i,B_1\cdots B_i}\sigmac)+|u|^4\tau_{4}
\end{align*}
for $1\le i\le4$, where $\tau_4$ contains no fifth order derivative of curvature components. 
By divergence theorem, we have
\begin{align*}
&\Omega_0^2(u)\delta\||u|^3(|u|\nablas)^i(K-|u|^{-2},\sigmac)\|_{\L^2_{\ub}\L^2(u)}^2+\Omega_0^2(u)\|\Omega^{-1}|u|^{\frac{7}{2}}(|u|\nablas)^i(\Omega\betab+\Lb\phi\nablas\phi)\|_{\L^2_{[u_0,u]}\L^2(\ub)}^2\\
\lesssim&\Omega_0^2(u)\delta\||u|^3(|u|\nablas)^i(K-|u|^{-2},\sigmac)\|_{\L^2_{\ub}\L^2(u_0)}^2+\Omega_0^2(u)\|\Omega^{-1}|u|^{\frac{7}{2}}(|u|\nablas)^i(\Omega\betab+\Lb\phi\nablas\phi)\|_{\L^2_{[u_0,u]}\L^4(0)}^2\\
&+\int_0^{\delta}\D\ub'\int_{u_0}^u\D u'\int_{S_{\ub',u'}}\Omega_0^2(u)|u|^4|\tau_4|\D\mu_{\gs}
\end{align*}
where
\begin{align*}
\int_{S_{\ub',u'}}\Omega_0^2(u)|u'|^4|\tau_4|\D\mu_{\gs}\lesssim\int_{S_{\ub',u'}}\Omega_0^2(u)|u'|^4|\tau_{4,1}|\D\mu_{\gs}+\int_{S_{\ub',u'}}\Omega_0^2(u)|u'|^4|\tau_{4,2}|\D\mu_{\gs}
\end{align*}
and the multiplier terms
\begin{equation}\label{tau41}
\begin{split}
&\int_{S_{\ub',u'}}\Omega_0^2(u)|u'|^4|\tau_{4,1}|\D\mu_{\gs}\\
\lesssim&\Omega_0^2(u)\Omega_0^{-2}(u')\left[|u'|^6\|\Omega\betab+\Lb\phi\nablas\phi\|_4^2(|\omega|+\|\Omega\tr\chi,\Omega\chih\|_4)\right.\\
&+|u'|^6\Omega_0^2\cdot|u'|^{-1}\|K-|u|^{-2},\check{\sigma}\|_4\|\Omega\betab+\Lb\phi\nablas\phi\|_4\\
&+|u'|^6\Omega_0^2\|\Omega\betab+L\phi\nablas\phi\|_4\|\eta,\etab\|_4\|K,K-|u|^{-2},\check{\sigma}\|_4\\
&+|u'|^6\|\Omega\betab+\Lb\phi\nablas\phi\|_4|u'|^{-1}(\|(|u'|\nablas)(\Omega\chih)\|_4\|\Omega\chibh\|_4+\|\Omega\chih\|_4(|u'|\nablas)(\Omega\chibh)\|_4)\\
&+|u'|^6\|\Omega\betab+\Lb\phi\nablas\phi\|_4|u'|^{-1}(\|(|u'|\nablas)(\Omega\tr\chi)\|_4\|\Omega\tr\chib\|_4+\|\Omega\tr\chi\|_4\|(|u'|\nablas)(\Omega\tr\chib)\|_4)\\
&+|u'|^6\|\Omega\betab+\Lb\phi\nablas\phi\|_4\left(\|\etab\|_4\|\Omega\chih,\Omega\tr\chi\|_4\|\Omega\chibh,\Omega\tr\chib\|_4\right)\\
&+|u'|^6\|\Omega\betab+\Lb\phi\nablas\phi\|_4\|\Omega\chibh\|_4\|\Omega\beta-L\phi\nablas\phi\|_4\\
&+|u'|^6\Omega_0^2\|\Omega\betab+\Lb\phi\nablas\phi\|_4\|\nablas\phi\|_4|u'|^{-1}\|\nablas\phi\|_5\\
&\left.+|u'|^6\|\Omega\betab+\Lb\phi\nablas\phi\|_4\left(\|\Omega\chibh,\Omega\tr\chib\|_4\|\nablas\phi\|_4\|L\phi\|_4+\Omega_0^2\|\etab\|_4\|\nablas\phi\|_4^2\right)\right]\\
&+|u'|^6\Omega_0^2\|K-|u|^{-2},\sigmac\|_4\|\Omega\chibh, \Omega\tr\chib+2|u|^{-1}\|_4(\|K-|u|^{-2},\sigmac\|_4+|u'|^{-1}\|\eta\|_5)\\
&+|u'|^6\Omega_0^2\|K-|u|^{-2},\sigmac\|_4\|\|\Omega\chibh,\Omega\tr\chib\|_4\|\eta\|_4^2\\
&+|u'|^6\Omega_0^2\|\sigmac\|_4|u'|^{-1}\|(|u|\nablas)\Lb\phi\|_4\|\nablas\phi\|_4\\
&+|u'|^6\Omega_0^2\|K-|u|^{-2}\|_4((|u|^{-2}+\|\mu\|_3)\|(|u|\nablas)(\Omega\tr\chib)\|_3+\|\Omega\tr\chib\|_4\|(|u|\nablas)\mu\|_3)
\end{split}
\end{equation}
and commutation terms
\begin{equation}\label{tau42}
\begin{split}
&\int_{S_{\ub',u'}}\Omega_0^2(u)|u'|^4|\tau_{4,2}|\D\mu_{\gs}\\
\lesssim&\Omega_0^2(u)\Omega_0^{-2}(u')|u'|^6\|\Omega\betab+\Lb\phi\nablas\phi\|_4\left(\|(|u'|\nablas)(\Omega\chi)\|_3\|\Omega\betab+\Lb\phi\nablas\phi\|_3+\Omega_0^2(u')|u'|\|K\|_2\|K-|u|^{-2},\sigmac\|_4\right)\\
&+|u'|^6\Omega_0^2(u')\|K-|u|^{-2},\sigmac\|_4\left(\|(|u'|\nablas)(\Omega\chib)\|_2\|K-|u|^{-2},\sigmac\|_3+|u'|\|K\|_3\|\Omega\betab+\Lb\phi\nablas\phi\|_4\right).
\end{split}
\end{equation}

\eqref{tau41} and \eqref{tau42} should be estimated in $\L_{\ub}^1L^1_u$. In the following estimate, $\|\Omega_0^{-1}|u|^{\frac{7}{2}}(\Omega\betab+\Lb\phi\nablas\phi)\|_4$ should be placed in $\L^\infty_{\ub}\L^2_{[u_0,u]}$. We remark that all terms of the last line of \eqref{tau41} contain a factor that is either the derivatives of $\Omega\tr\chib$ or the derivatives of $\mu$, which have better estimates. This is because we are considering the derivatives of $K-|u|^{-2}$, $\sigmac$ and $\Omega\betab+\Lb\phi\nablas\phi$ but not themselves.

We begin the estimates. The \engordnumber{1} line of the right hand side of \eqref{tau42} is estimated by, 
\begin{align*}
\lesssim\delta|u|^{-1}\cdot C^{\frac{3}{4}}\delta^3\mathscr{F}^3\mathscr{E}^2\mathcal{A}^3\cdot C^{\frac{1}{4}}\mathscr{F}\mathscr{W}^{\frac{1}{2}}\mathcal{A}\lesssim C^{-1}\delta^3\mathscr{F}^3\mathscr{E}^2\mathcal{A}^3.
\end{align*}
The \engordnumber{2} line is estimated by, \begin{align*}
\lesssim\delta|u|^{-\frac{1}{2}}\cdot C^{\frac{3}{8}}\delta^{\frac{3}{2}}\mathscr{F}^{\frac{3}{2}}\mathscr{E}\mathcal{A}^{\frac{3}{2}}\cdot C^{\frac{3}{8}}\delta^{\frac{3}{2}}\mathscr{F}^{\frac{3}{2}}\mathscr{E}\mathcal{A}^{\frac{3}{2}} \lesssim C^{-\frac{1}{4}}\delta^3\mathscr{F}^3\mathscr{E}^2\mathcal{A}^3.
\end{align*}
To estimate the \engordnumber{3} line, we need a refined version of the estimate \eqref{estimate-K2} for $K$:
\begin{align*}
\Omega_0\|K\|_{\L^2_{\ub}\H^4(\ub)}\lesssim\Omega_0|u|^{-2}+C^\frac{3}{8}\delta|u|^{-3}\mathscr{F}^{\frac{3}{2}}\mathscr{E}\mathcal{A}^{\frac{3}{2}}\lesssim\Omega_0|u|^{-2}+C^{-\frac{3}{2}}|u|^{-2}\mathscr{F}^{\frac{1}{2}}\mathscr{E}\mathcal{A}^{\frac{1}{2}}.
\end{align*}
Then the \engordnumber{3} line is estimated by
\begin{align*}
\lesssim&\Omega_0\delta|u|^{-\frac{1}{2}}\cdot C^{\frac{3}{8}}\delta^{\frac{3}{2}}\mathscr{F}^{\frac{3}{2}}\mathscr{E}\mathcal{A}^{\frac{3}{2}}\cdot C^{\frac{1}{4}}\delta\mathscr{F}\mathscr{E}\mathcal{A}\cdot[\Omega_0+C^{-\frac{3}{2}}\mathscr{F}^{\frac{1}{2}}\mathscr{E}\mathcal{A}^{\frac{1}{2}}]\\
\lesssim&C^{-\frac{1}{8}}\delta^3\mathscr{F}^{\frac{5}{2}}\mathscr{E}^2\mathcal{A}^{\frac{5}{2}}+ C^{-\frac{3}{4}}\delta^3\mathscr{F}^3\mathscr{E}^2\mathcal{A}^3\lesssim C^{-\frac{1}{8}}\delta^3\mathscr{F}^3\mathscr{E}^2\mathcal{A}^3.
\end{align*}
The \engordnumber{1} term of the \engordnumber{4} line is estimated by, placing $|u'|\|(|u'|\nablas)(\Omega\chih)\|_4$ in $\L^\infty_{[u_0,u]}\L_{\ub}^2$,
\begin{align*}
\lesssim\delta|u|^{-\frac{1}{2}}\cdot C^{\frac{3}{8}}\delta^{\frac{3}{2}}\mathscr{F}^{\frac{3}{2}}\mathscr{E}\mathcal{A}^{\frac{3}{2}}\cdot  C^{\frac{1}{4}}\mathscr{F}\mathscr{E}\mathcal{A}\cdot C^{\frac{1}{4}}\delta\mathscr{F}\mathcal{A}\lesssim C^{-\frac{1}{8}}\delta^3\mathscr{F}^3\mathscr{E}^2\mathcal{A}^3.
\end{align*}
The \engordnumber{1} term of the \engordnumber{5} line is estimated by, placing $|u'|\|(|u'|\nablas)(\Omega\tr\chi)\|_4$ in $\L_{\ub}^\infty\L^\infty_{[u_0,u]}$
\begin{align*}
\lesssim\delta|u|^{-\frac{1}{2}}\cdot C^{\frac{3}{8}}\delta^{\frac{3}{2}}\mathscr{F}^{\frac{3}{2}}\mathscr{E}\mathcal{A}^{\frac{3}{2}}\cdot  C^{\frac{1}{2}}\delta\mathscr{F}^2\mathscr{E}\mathcal{A}^2\cdot 1\lesssim C^{-\frac{1}{8}}\delta^3\mathscr{F}^3\mathscr{E}^2\mathcal{A}^3.
\end{align*}
The \engordnumber{2} terms of the \engordnumber{4} and \engordnumber{5} lines are estimated by, placing $\|\Omega_0^{-1}|u'|^{\frac{3}{2}}(|u'|\nablas)(\Omega\chibh)\|_5$ in $\L_{\ub}^\infty\L^2_{[u_0,u]}$, and $\||u'|^2(|u'|\nablas)(\Omega\tr\chib)\|_5$ in $\L_{\ub}^\infty\L^\infty_{[u_0,u]}$, 
\begin{align*}
\lesssim\delta|u|^{-\frac{1}{2}}\cdot C^{\frac{3}{8}}\delta^{\frac{3}{2}}\mathscr{F}^{\frac{3}{2}}\mathscr{E}\mathcal{A}^{\frac{3}{2}}\cdot  C^{\frac{1}{4}}\mathscr{F}\mathcal{A}\cdot C^{\frac{1}{4}}\delta\mathscr{F}\mathscr{E}\mathscr{W}^{\frac{1}{2}}\mathcal{A}\lesssim C^{-\frac{1}{8}}\delta^3\mathscr{F}^3\mathscr{E}^2\mathcal{A}^3.
\end{align*}
The \engordnumber{6} line is estimated by
\begin{align*}
\lesssim\delta|u|^{-\frac{1}{2}}\cdot C^{\frac{3}{8}}\delta^{\frac{3}{2}}\mathscr{F}^{\frac{3}{2}}\mathscr{E}\mathcal{A}^{\frac{3}{2}}\cdot C^{\frac{1}{4}}\delta\mathscr{F}\mathscr{E}\mathscr{W}^{\frac{1}{2}}\mathcal{A}\cdot C^{\frac{1}{4}}\mathscr{F}\mathcal{A}\cdot1\lesssim C^{-\frac{1}{8}}\delta^3\mathscr{F}^3\mathscr{E}^2\mathcal{A}^3.
\end{align*}
The \engordnumber{7} line is estimated by,
\begin{align*}
\lesssim\delta|u|^{-\frac{1}{2}}\cdot C^{\frac{3}{8}}\delta^{\frac{3}{2}}\mathscr{F}^{\frac{3}{2}}\mathscr{E}\mathcal{A}^{\frac{3}{2}}\cdot  C^{\frac{1}{4}}\delta\mathscr{F}\mathcal{A}\cdot C^{\frac{1}{4}}\mathscr{F}\mathscr{E}\mathcal{A}\lesssim C^{-\frac{1}{8}}\delta^3\mathscr{F}^3\mathscr{E}^2\mathcal{A}^3.
\end{align*}
The \engordnumber{8} line is estimated by, placing $\Omega_0|u|^{\frac{3}{2}}\|\nablas\phi\|_5$ in $\L^\infty_{[u_0,u]}\L_{\ub}^2$,
\begin{align*}
\lesssim\Omega_0\delta|u|^{-1}\cdot C^{\frac{3}{8}}\delta^{\frac{3}{2}}\mathscr{F}^{\frac{3}{2}}\mathscr{E}\mathcal{A}^{\frac{3}{2}}\cdot  C^{\frac{1}{4}}\delta\mathscr{F}\mathscr{E}\mathcal{A}\cdot C^{\frac{1}{4}}\delta^{\frac{1}{2}}\mathscr{F}\mathscr{E}\mathcal{A}\lesssim C^{-\frac{1}{8}}\delta^3\mathscr{F}^3\mathscr{E}^2\mathcal{A}^3.
\end{align*}
The \engordnumber{1} term of the \engordnumber{9} line is estimated by, 
\begin{align*}
\lesssim\delta|u|^{-\frac{1}{2}}\cdot C^{\frac{3}{8}}\delta^{\frac{3}{2}}\mathscr{F}^{\frac{3}{2}}\mathscr{E}\mathcal{A}^{\frac{3}{2}}\cdot  1\cdot C^{\frac{1}{4}}\delta\mathscr{F}\mathscr{E}\mathcal{A}\cdot C^{\frac{1}{4}}\mathscr{F}\mathcal{A}\lesssim C^{-\frac{1}{8}}\delta^3\mathscr{F}^3\mathscr{E}^2\mathcal{A}^3.
\end{align*}
The \engordnumber{2} term of the \engordnumber{9} line is estimated by,
\begin{align*}
\lesssim\Omega_0^2\delta|u|^{-\frac{5}{2}}\cdot C^{\frac{3}{8}}\delta^{\frac{3}{2}}\mathscr{F}^{\frac{3}{2}}\mathscr{E}\mathcal{A}^{\frac{3}{2}}\cdot  (C^{\frac{1}{4}}\delta\mathscr{F}\mathscr{E}\mathscr{W}^{\frac{1}{2}}\mathcal{A})^3\lesssim C^{-1}\delta^3\mathscr{F}^3\mathscr{E}^2\mathcal{A}^3.
\end{align*}
The \engordnumber{10} line is estimated by, placing $\Omega_0|u|^{-1}\|\eta\|_5$ in $\L^\infty_{[u_0,u]}\L^2_{\ub}$,
\begin{align*}
\lesssim\delta|u|^{-\frac{1}{2}}\cdot C^{\frac{3}{8}}\delta\mathscr{F}^{\frac{3}{2}}\mathscr{E}\mathcal{A}^{\frac{3}{2}}\cdot C^{\frac{1}{4}}\delta\mathscr{F}\mathscr{W}^{\frac{1}{2}}\mathcal{A}\cdot C^{\frac{1}{4}}\delta^{\frac{1}{2}}\mathscr{F}\mathscr{E}\mathcal{A}\lesssim C^{-\frac{1}{8}}\delta^3\mathscr{F}^3\mathscr{E}^2\mathcal{A}^3.
\end{align*}
The \engordnumber{11} line is estimated by
\begin{align*}
\lesssim\Omega_0\delta|u|^{-1}\cdot C^{\frac{3}{8}}\delta\mathscr{F}^{\frac{3}{2}}\mathscr{E}\mathcal{A}^{\frac{3}{2}}\cdot 1\cdot  (C^{\frac{1}{4}}\delta\mathscr{F}\mathscr{E}\mathcal{A})^2\lesssim C^{-\frac{1}{8}}\delta^3\mathscr{F}^3\mathscr{E}^2\mathcal{A}^3.
\end{align*}
The \engordnumber{12} line is estimated by, placing $\|\Omega_0^{-1}|u|^{\frac{3}{4}}\Lb\phi\|_5$ in $\L^\infty_{\ub}\L^2_{[u_0,u]}$,
\begin{align*}
\lesssim\Omega_0\delta|u|^{-1}\cdot C^{\frac{3}{8}}\delta\mathscr{F}^{\frac{3}{2}}\mathscr{E}\mathcal{A}^{\frac{3}{2}}\cdot C^{\frac{1}{4}}\delta\mathscr{F}\mathscr{E}\mathcal{A}\cdot C^{\frac{1}{4}}\delta\mathscr{F}\mathscr{E}\mathcal{A}\lesssim C^{-1}\delta^3\mathscr{F}^3\mathscr{E}^2\mathcal{A}^3.
\end{align*}
To estimate the last line, we should use two improved estimates \eqref{estimate-detrchib} and \eqref{estimate-nablasmu}. From \eqref{estimate-mu}, we know $\|\mu\|_3\lesssim C^{-2}\Omega_0^{-1}|u|^{-2}\mathscr{F}^{\frac{1}{2}}\mathscr{E}\mathcal{A}^{\frac{1}{2}}$, then the last line is then estimated by
\begin{align*}
\lesssim\delta|u|^{-\frac{1}{2}}\cdot C^{\frac{3}{8}}\delta\mathscr{F}^\frac{3}{2}\mathscr{E}\mathcal{A}^\frac{3}{2}\cdot &(C^{-2}\mathscr{F}^{\frac{1}{2}}\mathscr{E}\mathcal{A}^{\frac{1}{2}}\cdot C^{\frac{1}{4}}\delta^{\frac{3}{2}}\mathscr{F}\mathscr{E}\mathscr{W}^{\frac{1}{2}}\mathcal{A}\\
&+1\cdot C^{\frac{1}{2}}\delta^{\frac{3}{2}}\mathscr{F}^2\mathscr{E}\mathscr{W}^{\frac{1}{2}}\mathcal{A}^2)\lesssim C^{-\frac{1}{8}}\delta^3\mathscr{F}^3\mathscr{E}^2\mathcal{A}^3.
\end{align*}

The commutation terms \eqref{tau42} can be estimated similarly to the estimate of \eqref{tau32}, using \eqref{estimate-K3} for the Gauss curvature $K$. It is easily find that
\begin{align*}
\int_0^\delta\D\ub'\int_{u_0}^u\D u'\int_{S_{\ub',u'}}\Omega_0^2(u)|u'|^4|\tau_{4,2}|\D\mu_{\gs}\lesssim \delta^3\mathscr{F}^3\mathscr{E}^2\mathcal{A}^3.
\end{align*}
Therefore, since $\betab(\ub=0)=0$, for $1\le i\le4$,
\begin{align*}
&\Omega_0^2(u)\delta\||u|^3(|u|\nablas)^i(K-|u|^{-2},\sigmac)\|_{\L^2_{\ub}\L^2(u)}^2+\Omega_0^2(u)\|\Omega_0^{-2}|u|^{\frac{7}{2}}(|u|\nablas)^i(\Omega\betab+\Lb\phi\nablas\phi)\|_{\L^2_{[u_0,u]}\L^2(\ub)}^2\\
\lesssim&\Omega_0^2(u)\delta\||u|^3(|u|\nablas)^i(K-|u|^{-2},\sigmac)\|_{\L^2_{\ub}\L^2(u_0)}^2+\delta^3\mathscr{F}^3\mathscr{E}^2\mathcal{A}^3\end{align*}
which means
\begin{equation}\label{estimate-K-sigmac}
\begin{split}
&\Omega_0\||u|^3(|u|\nablas)^i(K-|u|^{-2},\sigmac)\|_{\L^2_{\ub}\L^2(u)}\\
&+\Omega_0\delta^{-\frac{1}{2}}\|\Omega_0^{-1}|u|^{\frac{7}{2}}(|u|\nablas)^i(\Omega\betab+\Lb\phi\nablas\phi)\|_{\L^2_{[u_0,u]}\L^2(\ub)}\lesssim\delta\mathscr{F}^\frac{3}{2}\mathscr{E}\mathcal{A}^\frac{3}{2}.\end{split}\end{equation}

It remains to estimate the zeroth order bound of $K$ and $\Omega\betab+\Lb\phi\nablas\phi$, that is, to show
\begin{align*}
\Omega_0\||u|^3(K-|u|^{-2},\sigmac)\|_{\L^2_{\ub}\L^2(u)}+\Omega_0\delta^{-\frac{1}{2}}\|\Omega_0^{-1}|u|^{\frac{7}{2}}(\Omega\betab+\Lb\phi\nablas\phi)\|_{\L^2_{[u_0,u]}\L^2(\ub)}\lesssim\delta\mathscr{F}^\frac{3}{2}\mathscr{E}\mathcal{A}^\frac{3}{2}.
\end{align*}

For $K$, this is followed from \eqref{estimate-K-lower} and $\Omega_0\le1$. For $\sigmac$ and $\Omega\betab+\Lb\phi\nablas\phi$, this is followed from directly integrating the Bianchi equation for $\Db\sigmac$ and $D(\Omega\betab+\Lb\phi\nablas\phi)$. The reason for why this works is simple. Remember that the zeroth order quantity cannot be estimated because we have two bad terms $|u|^{-2}(\Omega\tr\chib+2|u|^{-1})$ and $\frac{1}{2}\Omega\tr\chib\mu$ whose estimates are good enough only after we take derivatives on them. However, both of these two terms come from the equation for $\Db(K-|u|^{-2})$. This means when we directly integrate the equations for $\Db\sigmac$ and $D(\Omega\betab+\Lb\phi\nablas\phi)$, we will not encounter these two terms. The price is that we will lose derivative because we will encounter $\curls(\Omega\betab+\Lb\phi\nablas\phi)$ and $\nablas K, \nablas\sigmac$. But this still works because we are only dealing with the zeroth order. 

Finally, we have completed the proof of Proposition \ref{estimate-R}.

\end{proof}

Combining all the propositions proved above then concludes the proof of Theorem \ref{existencetheorem}.

\section{Formation of trapped surfaces}\label{FoTS}

In this section, we prove a formation of trapped surfaces theorem. The theorem we will prove is the following:
\begin{theorem}[Formation of trapped surfaces]\label{formationoftrappedsurfaces}
There exists a universal constant $C_1\ge C_0$ such that the following statement is true. Consider the same characteristic initial value problem as in Theorem \ref{existencetheorem} and let $\mathscr{F}$, $\mathcal{A}$, $\mathscr{E}$ be defined as in Theorem \ref{existencetheorem} with three numbers $\delta$, $u_0<u_1<0$. Suppose that the initial data given on $\Cb_0$ is spherically symmetric with $\Omega(u_0)\le1$, and\begin{align}\label{condition-CA}
\mathcal{A}\le a
\end{align}
for some constant $a\ge1$. Suppose also that $C\ge C_1$ and
\begin{align}\label{conditionexistence}\Omega_0^2(u_0)\mathscr{F}a\ge C^{-2}\Omega_0^4(u_0)\mathscr{E}^2.\end{align}
Then the smooth solution of the Einstein-scalar field equations exists for $0\le \ub\le\delta$, $u_0\le u\le u_1$, where $u_1$ is defined by
\begin{equation}\label{def-u1}
\Omega^2_0(u_1)|u_1|=C^2\Omega_0^2(u_0)\delta\mathscr{F}a.
\end{equation}
If in addition 
\begin{equation}\label{conditiontrapped}
\inf_{\vartheta\in S^2}\int_0^{\delta}|u_0|^2(|\Omega\chih|^2+2|L\phi|^2)(\ub',u_0,\vartheta)\D\ub'\ge  17C^2\Omega_0^2(u_0)\delta\mathscr{F}a,
\end{equation}
together with
\begin{align}\label{conditiontrapped1}
\Omega_0^2(u_0)\mathscr{F}a\ge 16C^{-2}\left(\int_{u_0}^{u_1}\frac{\Omega_0^2h|\psi|}{|u'|}\D u'\right)^{2},
\end{align}
hold, then the sphere $S_{\delta,u_1}$ is a closed trapped surface.
\end{theorem}
\begin{proof}
First of all, we shall prove the solution exists in the region $0\le\ub\le\delta$, $u_0\le u\le u_1$ by applying Theorem \ref{existencetheorem}. We should verify $\delta$ and $u_1$ defined by \eqref{def-u1} satisfies the smallness and auxiliary conditions \eqref{smallness} and \eqref{auxiliary}. By \eqref{condition-CA}, \eqref{conditionexistence} and \eqref{def-u1}, we have
\begin{align*}
&\max\left\{\Omega_0^2(u_0)\delta|u_1|^{-1}\mathscr{E}^2, C^2\delta|u_1|^{-1}\mathscr{F}\mathcal{A}\right\}\le C^2\delta|u_1|^{-1}\mathscr{F}a\le\frac{\Omega_0^2(u_1)}{\Omega_0^2(u_0)}\le\mathscr{W}^{-1},
\end{align*}
where we use the fact that $x|\log x|\le 1$ for $x\in(0,1]$. This verifies both the smallness conditions and the auxiliary condition, if $C\ge C_1\ge C_0$. As a consequence, the estimates stated in Theorem \ref{existencetheorem} also hold.

The next step is to improve the estimates for $\eta$ and $\nablas\phi$ for up to the third order derivatives. We will prove
\begin{proposition}
Suppose that the assumptions and conclusions in Theorem \ref{existencetheorem} hold, with $C\ge C_1\ge C_0$. Then if $C_1$ is sufficiently large, 
\begin{align*}
\|\nablas\phi\|_{\H^3(\ub,u)},\|\eta\|_{\H^3(\ub,u)},\mathscr{W}^{-\frac{1}{2}}\|\etab\|_{\H^3(\ub,u)}\lesssim \delta|u|^{-2}\mathscr{F}\mathcal{A}.
\end{align*}
\end{proposition}
\begin{proof}
From the equation \eqref{equ-Dnablasphi}, we have
\begin{align}\label{estimate-nablasphi-improve}
\|\nablas\phi\|_{\H^3(\ub,u)}\lesssim\delta|u|^{-1}\|L\phi\|_{\L^\infty_{\ub}\H^4(u)}\lesssim\delta|u|^{-2}\mathscr{F}\mathcal{A}.
\end{align}

We plug two Codazzi equations \eqref{equ-divschih} and \eqref{equ-divschibh} in two transport equations \eqref{equ-Deta} and \eqref{equ-Dbetab}. We have two equations
\begin{align}\label{equ-Deta-improve}
D\eta = \Omega\tr\chi\etab-2L\phi\nablas\phi+\divs(\Omega\chih)-\frac{1}{2}\nablas(\Omega\tr \chi).\\
\label{equ-Dbetab-improve}\Db\etab = \Omega\tr\chib\eta-2\Lb\phi\nablas\phi+\divs(\Omega\chibh)-\frac{1}{2}\nablas(\Omega\tr \chib).
\end{align}
We assume 
\begin{align}\label{bootstrap-etab-improve}
\|\etab\|_{\H^3(\ub,u)}\lesssim C^{\frac{1}{4}}\delta|u|^{-2}\mathscr{F}\mathscr{W}^{\frac{1}{2}}\mathcal{A}.
\end{align}
Then from equation \eqref{equ-Deta-improve}, we have
\begin{equation}\label{estimate-eta-improve}
\begin{split}
\|\eta\|_{\H^3(\ub,u)}\lesssim&\delta\|\Omega\tr\chi\etab,2L\phi\nablas\phi,\nablas(\Omega\tr\chi)\|_{\L^\infty_{\ub}\H^3(u)}+\delta|u|^{-1}\|\Omega\chih\|_{\L^\infty_{\ub}\H^4(u)}\\
\lesssim&\delta|u|^{-1}\mathscr{F}\mathcal{A}\cdot C^{\frac{1}{4}}\delta|u|^{-2}\mathscr{F}\mathscr{W}^{\frac{1}{2}}\mathcal{A}+\delta\cdot|u|^{-2}\mathscr{F}\mathcal{A}\lesssim\delta|u|^{-2}\mathscr{F}\mathcal{A}.
\end{split}
\end{equation}
Then from the equation \eqref{equ-Dbetab-improve} and the above estimate \eqref{estimate-eta-improve}, we have
\begin{equation}\label{estimate-etab-improve}
\begin{split}
\||u|\etab\|_{\H^3(\ub,u)}\lesssim&\||u|^2(\Omega\tr\chib\eta,\Lb\phi\nablas\phi,\nablas(\Omega\tr\chib))\|_{\L^1_{[u_0,u]}\H^3(\ub)}+\||u|\Omega\chibh\|_{\L^1_{[u_0,u]}\H^4(\ub)}\\
\lesssim&\delta|u|^{-1}\mathscr{F}\mathscr{W}^{\frac{1}{2}}\mathcal{A}.
\end{split}
\end{equation}
Then this improves \eqref{bootstrap-etab-improve} when $C_1$ is sufficiently large, which implies that the estimates \eqref{estimate-eta-improve} and \eqref{estimate-etab-improve} hold without assuming \eqref{bootstrap-etab-improve}. We have completed the improved estimates.

\end{proof}
Then we begin to prove that at $S_{\delta,u_1}$, $\tr\chi', \Omega\tr\chib<0$. That $\Omega\tr\chib<0$ is followed easily by, for any $\ub\in[0,\delta], u\in[u_0,u_1]$, $\vartheta\in S^2$,
\begin{align*}
\left|\Omega\tr\chib(\ub,u,\vartheta)+\frac{2}{|u|}\right|\le c\delta|u|^{-2}\mathscr{F}\mathcal{A}\le\frac{cC^{-1}}{|u|}\le\frac{1}{|u|}
\end{align*}
if $C_1$ is sufficiently large.

{\bf In the rest of the proof,  the estimates are derived for any particular $\vartheta\in S^2$} and we only need to prove $\tr\chi'(\delta,u_1,\vartheta)<0$. We integrate the equation \eqref{equ-Dtrchi} along the $u$ curve at any $u\in[u_0,u_1]$, we have
\begin{align}\label{estimate-trchi-formation-step1}
\left|\tr\chi'-\frac{2h}{|u|}\right|\lesssim\delta\sup_{\ub}|\Omega\tr\chi|\sup_{\ub}|\tr\chi'|+\Omega_0^{-2}(u)|u|^{-2}\int_0^{\delta}|u|^2(|\Omega\chih|^2+2|L\phi|^2)\D\ub.
\end{align}
Therefore, we have
\begin{align*}
|\tr\chi'|\lesssim&\frac{1}{|u|}+\delta|u|^{-1}\mathscr{F}\mathcal{A}\sup_{\ub}|\tr\chi'|+\Omega_0^{-2}(u)|u|^{-2}\int_0^{\delta}|u|^2(|\Omega\chih|^2+2|L\phi|^2)\D\ub.
\end{align*}
If $C_1$ is sufficiently large such that the second term on the right hand side can be absorbed by the left hand side, we have
\begin{align}\label{estimate-trchi-formation}
|\tr\chi'|\lesssim\frac{1}{|u|}+\Omega_0^{-2}(u)|u|^{-2}\int_0^{\delta}|u|^2(|\Omega\chih|^2+2|L\phi|^2)\D\ub.
\end{align}
Substitute this back to \eqref{estimate-trchi-formation-step1} we have
\begin{align}\label{estimate-trchi-2h/|u|-formation}
\left|\tr\chi'-\frac{2h}{|u|}\right|\lesssim\delta|u|^{-2}\mathscr{F}\mathcal{A}+\Omega_0^{-2}(u)|u|^{-2}\int_0^{\delta}|u|^2(|\Omega\chih|^2+2|L\phi|^2)\D\ub.
\end{align}

Now we turn to the equations \eqref{equ-chih} and \eqref{equ-DbLphi}, and compute
\begin{equation}\label{equ-Dbchih2}
\begin{split}
&\frac{\partial}{\partial u}(|u|^2|\Omega\chih|^2+2|u|^2|L\phi|^2)
=2|u|^2\langle\Omega^2(\nablas\tensor\eta+\eta\tensor\eta+\nablas\phi\tensor\nablas\phi)-\nablas_b(\Omega\chih),\Omega\chih\rangle\\
&+4|u|(\Omega^2\Deltas\phi+2\Omega^2\langle\eta,\ds\phi\rangle-\nablas_bL\phi)\cdot(|u|L\phi)
-\left(\Omega\tr\chib+\frac{2}{|u|}\right)(|u|^2|\Omega\chih|^2+2|u|^2|L\phi|^2)\\
&-\Omega\tr\chi|u|^2\langle\Omega\chibh,\Omega\chih\rangle-2\Omega\tr\chi|u|\Lb\phi\cdot|u|L\phi,
\end{split}
\end{equation}
where $b$ is a vector field satisfying $D b=-4\Omega^2\zeta^\sharp$, which implies, together with \eqref{estimate-eta-improve}, \eqref{estimate-etab-improve}, and $b(\ub=0)=0$,
\begin{align}\label{estimate-b}
|b|\lesssim\Omega_0^2\delta^2|u|^{-2}\mathscr{F}\mathscr{W}^{\frac{1}{2}}\mathcal{A}.
\end{align}

Using the \eqref{estimate-eta-improve}, \eqref{estimate-etab-improve}, \eqref{estimate-b}, and the Sobolev inequalities, the first two lines of the right hand side can be bounded in $L^\infty$ with improved estimates. We then integrate \eqref{equ-Dbchih2} along the $\ub$ curve (but not the integral curve of $\Db$), and the right hand side should be estimated in $\int_{u_0}^u|\cdot|\D u'=\int_{u_0}^u|\cdot|(\ub,u',\vartheta)\D u'$. The first two terms of the right hand side are estimated by
\begin{align*}
\lesssim\int_{u_0}^u\Omega^2(u')\cdot\delta|u'|^{-2}\mathscr{F}\mathcal{A}\cdot(|u||\Omega\chih|+|u||L\phi|)\D u'.
\end{align*}
The \engordnumber{3} term of the right hand side is estimated by
\begin{align*}
\lesssim\delta\mathscr{F}\mathcal{A}\int_{u_0}^u|u'|^{-2}(|u'|^2|\Omega\chih|^2+2|u'|^2|L\phi|^2)\D u'.
\end{align*}
The \engordnumber{1} term of the \engordnumber{3} line is estimated by, using \eqref{estimate-trchi-formation} to estimate $\Omega\tr\chi$,
\begin{align*}
\lesssim&\delta\mathscr{F}\mathcal{A}\int_{u_0}^u|u'|^{-1}(|u'||\Omega\chih|)\left(\frac{\Omega_0^2(u')}{|u'|}+|u'|^{-2}\int_0^{\delta}|u'|^2(|\Omega\chih|^2+2|L\phi|^2)\D\ub\right)\D u'\\
\lesssim&\delta\mathscr{F}\mathcal{A}\Omega_0^2(u_0)\int_{u_0}^u|u'|^{-2}\cdot|u'||\Omega\chih|\D u'+\delta|u|^{-2}\mathscr{F}^2\mathcal{A}^2\sup_{u_0\le u'\le u}\int_0^{\delta}|u'|^2\left(|\Omega\chih|^2+2|L\phi|^2\right)\D\ub
\end{align*}
The last term should be estimated more carefully. We estimate, using \eqref{estimate-trchi-formation} and \eqref{estimate-trchi-2h/|u|-formation},
\begin{align*}
&\int_{u_0}^u\left|\Omega\tr\chi|u'|\Lb\phi\cdot|u'|L\phi-\frac{2\Omega^2h\psi}{|u'|}|u'|L\phi\right|\D u'\\
\lesssim&\int_{u_0}^u\left|\Omega\tr\chi-\frac{2\Omega^2h}{|u'|}\right|\cdot\psi|u'||L\phi|\D u'+\int_{u_0}^u|\Omega\tr\chi|||u'|\Lb\phi-\psi||u'||L\phi|\D u'\\
\lesssim&\Omega_0^2\delta\mathscr{F}\mathcal{A}\int_{u_0}^u|u'|^{-2}|\psi|\cdot|u'||L\phi|\D u'+\mathscr{F}\mathcal{A}\int_{u_0}^u|u'|^{-2}|\psi|\D u'\sup_{u_0\le u'\le u}\int_0^{\delta}|u'|^2(|\Omega\chih|^2+2|L\phi|^2)\D\ub\\
&+\Omega_0^2\int_{u_0}^u|u'|^{-1}||u'|\Lb\phi-\psi|\cdot|u'||L\phi|\D u'\\
&+\mathscr{F}\mathcal{A}\int_{u_0}^u|u'|^{-2}||u'|\Lb\phi-\psi|\D u'\sup_{u_0\le u'\le u}\int_0^{\delta}|u'|^2(|\Omega\chih|^2+2|L\phi|^2)\D\ub\\
\lesssim&\Omega_0^2\delta\mathscr{F}\mathcal{A}\int_{u_0}^u|u'|^{-2}|\psi|\cdot|u'||L\phi|\D u'+|u|^{-1}\mathscr{F}\mathscr{W}^{\frac{1}{2}}\mathcal{A}\sup_{u_0\le u'\le u}\int_0^{\delta}|u'|^2(|\Omega\chih|^2+2|L\phi|^2)\D\ub\\
&+\Omega_0^2\int_{u_0}^u|u'|^{-1}||u'|\Lb\phi-\psi|\cdot|u'||L\phi|\D u'+\delta|u|^{-2}\mathscr{F}^2\mathcal{A}^2\sup_{u_0\le u'\le u}\int_0^{\delta}|u'|^2(|\Omega\chih|^2+2|L\phi|^2)\D\ub.
\end{align*}
Integrating \eqref{equ-Dbchih2}, using the above all estimates, for some universal $c$,
\begin{align*}
&\pm\left\{[|u|^2|\Omega\chih|^2+2|u|^2|L\phi|^2](\ub,u,\vartheta)-[|u|^2|\Omega\chih|^2+2|u|^2|L\phi|^2](\ub,u_0,\vartheta)\right\}\\
\le&c|u|^{-1}\mathscr{F}\mathcal{A}\mathscr{W}^{\frac{1}{2}}\sup_{u_0\le u'\le u}\int_0^{\delta}[|u'|^2(|\Omega\chih|^2+2|L\phi|^2)](\ub',u,\vartheta)\D\ub'\\
&+c\delta\mathscr{F}\mathcal{A}\int_{u_0}^u|u'|^{-2}(|u'|^2|\Omega\chih|^2+2|u'|^2|L\phi|^2)(\ub,u',\vartheta)\D u'+\int_{u_0}^u\frac{4\Omega_0^2h|\psi|}{|u'|}\cdot|u'||L\phi|(\ub,u',\vartheta)\D u'\\
&+c\Omega_0^2(u_0)\delta\mathscr{F}\mathcal{A}\int_{u_0}^u|u'|^{-2}(1+|\psi|)\cdot(|u'||L\phi|+|u'||\Omega\chih|)(\ub,u',\vartheta)\D u'\\
&+c\Omega_0^2(u_0)\int_{u_0}^u|u'|^{-1}||u'|\Lb\phi(\ub,u',\vartheta)-\psi(u')|\cdot|u'||L\phi|(\ub,u',\vartheta)\D u'.
\end{align*}
Now integrate the above inequality over $\ub$, and assume, for all $u\in[u_0,u_1]$,
\begin{align}\label{bootstrap-energy}
\int_0^{\delta}[|u|^2(|\Omega\chih|^2+2|L\phi|^2)](\ub',u,\vartheta)\D\ub'\le2\int_0^{\delta}[|u'|^2(|\Omega\chih|^2+2|L\phi|^2)](\ub',u_0,\vartheta)\D\ub'\triangleq 2I_0(\vartheta),
\end{align}
we have, for $C_1$ sufficiently large, using \eqref{conditiontrapped} and \eqref{conditiontrapped1},
\begin{align*}
&\left|\int_0^\delta[|u|^2|\Omega\chih|^2+2|u|^2|L\phi|^2](\ub',u,\vartheta)\D\ub'-I_0(\vartheta)\right|\\
\le& cC^{-1}\cdot (I_0(\vartheta)+\Omega_0^2(u_0)\delta^{\frac{1}{2}}I_0^{\frac{1}{2}}(\vartheta))+4\delta^{\frac{1}{2}}I_0^{\frac{1}{2}}(\vartheta)\int_{u_0}^u\frac{\Omega_0^2h|\psi|}{|u'|}\D u'\le\frac{1}{2}I_0(\vartheta).\end{align*}
This improves \eqref{bootstrap-energy} and therefore holds without assuming \eqref{bootstrap-energy}. On the other hand, we derive that for any $u\in[u_0,u_1]$,
\begin{align*}
\int_0^\delta[|u|^2|\Omega\chih|^2+2|u|^2|L\phi|^2](\ub',u,\vartheta)\D\ub'\ge\frac{1}{2}I_0(\vartheta)
\end{align*}
From this inequality, \eqref{conditiontrapped} and \eqref{def-u1}, we have, in particular,
\begin{align}\label{conditiontrapped3}
\int_0^\delta[|u_1|^2|\Omega\chih|^2+2|u_1|^2|L\phi|^2](\ub',u_1,\vartheta)\D\ub'\ge\frac{17}{2}\Omega_0^2(u_1)|u_1|.
\end{align}

Now we integrate again the equation \eqref{equ-Dtrchi} along $\ub$ curve at $u=u_1$, using the estimate \eqref{estimate-Omega}, we have, for $C_1$ sufficiently large,
\begin{align*}
\tr\chi'(\delta,u_1,\vartheta)-\frac{2h(u_1)}{|u_1|}\le&-\int_0^{\delta}\Omega^{-2}(|\Omega\chih|^2+2|L\phi|^2)(\ub',u_1,\vartheta)\D\ub'\\
\le&-\frac{1}{4}\Omega_0^{-2}(u_1)\int_0^{\delta}(|\Omega\chih|^2+2|L\phi|^2)(\ub',u_1,\vartheta)\D\ub'\\
\le&-\frac{17}{16}\cdot\frac{2}{|u_1|}\end{align*}
which implies $\tr\chi'(\delta,u_1,\vartheta)<0$.

\end{proof}

As a special case, we have
\begin{theorem}\label{formationoftrappedsurfaces1}
Consider the characteristic initial value problem as in Theorem \ref{existencetheorem}. Suppose that the initial data on $\Cb_0$ is spherically symmetric with $\Omega_0(u_0)\le1$ and the data on $C_{u_0}$ obeys
 \begin{equation*}
\max\left\{\sup_{u_0\le u\le u_1}|\varphi(u)|,|u_0|\sup_{0\le\ub\le\delta}\left(\|\Omega\chih\|_{\H^7(\ub,u_0)}+\|\omega, L\phi\|_{\H^5(\ub,u_0)}\right)\right\}\le\Omega_0^2(u_0)a\left|\log\frac{|u_1|}{|u_0|}\right|
 \end{equation*}
for some $a\ge 1$. 
Then the smooth solution of the Einstein-scalar field equations exists for $0\le\ub\le\delta$, $u_0\le u\le u_1$ where $u_1$ is defined by
\begin{equation}\label{def-u1'}
\Omega^2_0(u_1)|u_1|=C^2\Omega_0^4(u_0)\delta a\left|\log\frac{|u_1|}{|u_0|}\right|
\end{equation}
for some $C\ge C_1$ where $C_1$ is given in Theorem \ref{formationoftrappedsurfaces}. If in addition
\begin{equation}\label{conditiontrapped'}
\inf_{\vartheta\in S^2}\int_0^{\delta}|u_0|^2(|\Omega\chih|^2+2|L\phi|^2)(\ub',u_0,\vartheta)\D\ub'\ge  17C^2\Omega_0^4(u_0)\delta a\left|\log\frac{|u_1|}{|u_0|}\right|,
\end{equation}
and $C\ge \max\{C_1,4\}$, then the sphere $S_{\delta,u_1}$ is a closed trapped surface.
\end{theorem}
\begin{proof}
To conclude that $S_{\delta,u_1}$ is a closed trapped surface, we only need to verify \eqref{conditionexistence} and \eqref{conditiontrapped1} in Theorem \ref{formationoftrappedsurfaces} for $\mathscr{F}(u_0,u_1)=\Omega_0^2(u_0)\left|\log\frac{|u_1|}{|u_0|}\right|$. Indeed, by the definition of $\mathcal{A}$, we have
\begin{align*}
\||u_0|(|u_0|\nablas)L\phi\|_{\L^2_{[0,\delta]}\H^4(u_0)}\le \mathscr{F}\mathcal{A},
\end{align*}
and therefore
\begin{align*}
\mathscr{E}^2\le\left|\log\frac{|u_1|}{|u_0|}\right|\le C^2\Omega_0^{-2}(u_0)\mathscr{F}\mathcal{A}.
\end{align*}
This verifies \eqref{conditionexistence}. On the other hand, if $C\ge4$,
\begin{align*}
\int_{u_0}^{u_1}\frac{\Omega_0^2h|\psi|}{|u'|}\D u'\le\Omega_0^2(u_0)\left|\log\frac{|u_1|}{|u_0|}\right|^{\frac{1}{2}}\le\frac{C}{4}\Omega_0(u_0)\sqrt{\mathscr{F}\mathcal{A}}.
\end{align*}
This verifies \eqref{conditiontrapped1}.
\end{proof}

\section{Instability theorems}\label{Secinstability}

In this section, we will prove the following instability theorem.
\begin{theorem}\label{instabilitytheorem}
Suppose that the smooth initial data given on $C_{u_0}\bigcup \Cb_0$ satisfies the following properties.  The data on $\Cb_0$ is spherically symmetric, and $\frac{m}{r}$ does not tend to zero when $r\to0^+$. The data on $C_{u_0}$, given on a fixed intervel, say $\ub\in[0,1]$, is arbitrary. Then the following statements are true:
\\
\indent \underline{Case 1}: If $\varphi(u)$ is not bounded as $u\to 0^-$, then there exists two sequences $\delta_n\to0^+$, $u_{0,n}\to0^-$ such that the smooth solution of the Einstein-scalar field equaitons exists for $0\le\ub\le \delta_n$, $u_0\le u\le u_{0,n}$ and the assumptions of Theorem \ref{formationoftrappedsurfaces1} hold for $\delta=\delta_n$, $u_0=u_{0,n}$ for all $n$. Consequently,  there exists a sequence $u_{1,n}\to0^-$ such that the solution remains smooth for $0\le\ub\le\delta_n$, $u_{0,n}\le u_0\le u_{1,n}$ and $S_{\delta_n,u_{1,n}}$ are closed trapped surfaces for all $n$.
\\
\indent\underline{Case 2}: Suppose that $\varphi(u)$ is bounded for all $u\in[u_0,0)$. Denote
\begin{align}\label{def-f}
f(\widetilde{u};\gamma)=\frac{1}{\tdelta(\tu;\gamma)}\inf_{\vartheta\in S^2}\int_0^{\tdelta(\tu;\gamma)}(|u_0|^2|\Omega\chih(\ub,u_0,\vartheta)|^2+||u_0|L\phi(\ub,u_0,\vartheta)+(\varphi(\tu)-\varphi(u_0))|^2)\D\ub
\end{align}
where $\tdelta(\tu;\gamma)$ is a function of $\tu$ defined through
\begin{align}
\label{u*tilde}\widetilde{\Omega}^{2-\gamma}=&17C^2\widetilde{\Omega}^4\left|\log\frac{|\widetilde{u}_*|}{|\widetilde{u}|}\right|,\\
\label{deltatilde}\widetilde{\Omega}_*^2|\widetilde{u}_*|=&C^2\widetilde{\Omega}^4\widetilde{\delta}\left|\log\frac{|\widetilde{u}_*|}{|\widetilde{u}|}\right|,
\end{align} 
where $\widetilde{\Omega}=\Omega_0(\tu)$ and $\widetilde{\Omega}_*=\Omega_0(\tu_*)$, $\gamma\in(0,2)$ and $C\ge \max\{C_1,4\}$ where $C_1$ is given in Theorem \ref{formationoftrappedsurfaces}. Then there exists some $\varepsilon_0>0$ such that if we can find some $\tu$ such that $|\tu|<\varepsilon_0$ and
\begin{align}\label{instabilitycondition}
\Omega_0^{\gamma-2}(\widetilde{u})f(\tu;\gamma)\ge2
\end{align}
for some $\gamma\in(0,2)$, then the smooth solution of the Einstein-scalar field equations exists for $0\le\delta\le\tdelta$, $u_0\le u\le \tu$ and the assumptions of Theorem \ref{formationoftrappedsurfaces1} hold for $\delta=\tdelta$, $u_0=\tu$. Consequently, there exists some $\tu_*$ such that the solution remains smooth for $0\le\ub\le\tdelta$, $\tu\le u\le \tu_*$ and $S_{\tdelta,\tu_*}$ is a closed trapped surface.
\end{theorem}
We have the following direct corollary.
\begin{corollary}\label{instabilitycorollary}
If in addition to the assumptions of \ref{instabilitytheorem} in Case 2, we have
\begin{align*}
\limsup_{\tu\to0^-}\Omega_0^{\gamma-2}(\widetilde{u})f(\tu;\gamma)>2
\end{align*}
for some $\gamma\in(0,2)$, then there exist two sequences $\tdelta_n\to0^+$, $\tu_{0,n}\to0^-$, such that the smooth solution of the Einstein-scalar field equations exists for $0\le\delta\le\tdelta_n$, $u_0\le u\le \tu_{0,n}$ and the assumptions of Theorem \ref{formationoftrappedsurfaces1} hold for $\delta=\tdelta_n$, $u_0=\tu_{0,n}$ for all $n$. Consequently, there exists a sequence $\tu_{1,n}\to0^-$ such that the solution remains smooth for $0\le\ub\le\tdelta_n$, $\tu_{0,n}\le u\le \tu_{1,n}$ and $S_{\tdelta,\tu_{1,n}}$ are closed trapped surfaces for all $n$. 

\end{corollary}


From the smoothness assumption on the initial data on $C_{u_0}$ in Theorem \ref{instabilitytheorem}, there always exist some $A\ge1$ such that
\begin{align}\label{conditionA}
|u_0|\sup_{\ub\in[0,1]}\left(\|\Omega\chih\|_{\H^{10}(\ub,u_0)}+\|\omega\|_{\H^8(\ub,u_0)}+\|L\phi\|_{\H^8(\ub,u_0)}\right)\le A.
\end{align}
Additionally, there also exists some $E'$ such that
\begin{align}\label{conditionE'}
|u_0|\|L\phi-\overline{L\phi}\|_{H^1_{\ub}\H^8(u_0)}=|u_0|\left(\int_{0}^{\delta}\|D(L\phi-\overline{L\phi})\|_{\H^8(\ub,u_0)}\D\ub\right)^{\frac{1}{2}}\le E',
\end{align}
which implies that for some $E$,
\begin{align}\label{conditionE}
\sup_{\delta\in(0,1]}\||u_0|(|u_0|\nablas)L\phi\|^2_{\L^2_{[0,\delta]}\H^7(u_0)}\left|\log\frac{C^2\delta}{|u_0|}\right|\le E
\end{align}
The smoothness is needed to guarantee the existence of the solution but the quantitative behaviors of the solution we need in the proof will only depend on $A$ and $E$.

\begin{proof}[Proof of Case 1] 

In this case, $\varphi(u)$ is not bounded as $u\to0^-$. There exists a sequence of $u_{0,n}\to0$ such that $\varphi_n=\varphi(u_{0,n})\to\infty$ and $$|\varphi_n|=\displaystyle\sup_{u_0\le u'\le u_{0,n}}|\varphi(u')|.$$

Define $\delta_n$, $u_{1,n}$ in terms of $u_{0,n}$ by
\begin{align}
\label{u*}\frac{\varphi_n^2}{2}=&17C^2\Omega_n^4\left|\log\frac{|u_{1,n}|}{|u_{0,n}|}\right|,\\
\label{delta}\Omega_{1,n}^2|u_{1,n}|=&C^2\Omega_n^4\delta_n\left|\log\frac{|u_{1,n}|}{|u_{0,n}|}\right|
\end{align}
where $\Omega_n=\Omega_0(u_{0,n})$,  $\Omega_{1,n}=\Omega_0(u_{1,n})$. And we may choose some $N_1$ such that for $n>N_1$,
\begin{align}\label{F*n>=1}
\Omega_n^2\left|\log\frac{|u_{1,n}|}{|u_{0,n}|}\right|=C^{-2}\Omega_n^{-2}\frac{\varphi_n^2}{34} \ge1,
\end{align}
and this implies, together with \eqref{delta},
\begin{align}\label{estimate-deltan}
 C^2\delta_n\le|u_{1,n}|.
\end{align}

First of all, we shall prove
\begin{proposition}\label{u0toun}Setting the function $\mathscr{F}$ to be
 $$\displaystyle\mathscr{F}_n=\mathscr{F}(\delta_n,u_0,u_{0,n})=\sup_{u_0\le u'\le u_{0,n}}|\varphi(u')|=|\varphi_n|.$$
Then there exists some $N_2$ such that for $n>N_2$, $|\varphi_n|\ge1$ and the solution of the Einstein equations exists in the region $0\le\ub\le\delta_n$, $u_0\le u\le u_{0,n}$, and the estimates stated in Theorem \ref{existencetheorem} hold.
\end{proposition}

 \begin{proof} We only need to verify the smallness conditions \eqref{smallness} and the auxiliary condition \eqref{auxiliary} for $\mathscr{F}_n$. From the \eqref{conditionA} and $\varphi_n\to\infty$, there exists some $n_1$ such that the corresponding $\mathcal{A}_n=\mathcal{A}(\delta_n,u_0,u_{0,n})$ is uniformly bounded by $1$ for $n>n_1$. We also denote the corresponding functions $\mathscr{E},\mathscr{W}$ by $\mathscr{E}_{n},\mathscr{W}_{n}$.

 Substitute \eqref{estimate-deltan} to \eqref{u*}, we have
\begin{align}\label{estimate-deltanun-1}
\delta_n\le C^{-2}|u_{1,n}|\le C^{-2}|u_{0,n}|\mathrm{e}^{-C^{-2}\Omega_n^{-4}\frac{\varphi_n^2}{34}},
\end{align}
then
\begin{align*}
C^2\delta_n|u_{0,n}|^{-1}\mathscr{F}_n\mathscr{W}_n\mathcal{A}_n\le \mathrm{e}^{-C^{-2}\Omega_n^{-4}\frac{\varphi_n^2}{34}}\cdot|\varphi_n|\cdot\left|\log\frac{\Omega_n}{\Omega_0(u_0)}\right|.
\end{align*}
The right hand side tends to zero as $n\to\infty$. Because $C\ge C_1\ge C_0$ where $C_0$ is given in Theorem \ref{existencetheorem}, there exists some $n_2$ such that for some $n>n_2$, the second one of the smallness conditions \eqref{smallness} holds. From \eqref{conditionE} and \eqref{estimate-deltan} we have $\mathscr{E}_n^2\le E$ for sufficiently large $n$ and
hence
\begin{align*}
&\Omega_0^2(u_0)\delta_n|u_{0,n}|^{-1}\mathscr{E}_n^2\mathscr{W}_n\le C^{-2}\left|\log\frac{\Omega_n}{\Omega_0(u_0)}\right|\mathrm{e}^{-C^{-2}\Omega_n^{-4}\frac{\varphi_n^2}{34}}E
\end{align*}
where the right hand side also tends to zero. This verifies the first one of the smallness conditions \eqref{smallness} for $n>n_3$ where $n_3$ is some positive integer. We also compute
\begin{align*}
\delta_n|u_{0,n}|^{-1}\mathscr{F}_n\mathscr{W}_n\mathcal{A}_n\Omega_0^2(u_0)\Omega_n^{-2}\le C^{-2} \mathrm{e}^{-C^{-2}\Omega_n^{-4}\frac{\varphi_n^2}{34}}\cdot|\varphi_n|\cdot\left|\log\frac{\Omega_n}{\Omega_0(u_0)}\right|\cdot\frac{\Omega_0^2(u_0)}{\Omega_n^2}.
\end{align*}
The right hand side also tends to zero as $n\to\infty$. Therefore the auxiliary condition \eqref{auxiliary} follows if $n$ is larger than some sufficiently large $n_4$. Finally, we can take $N_2=\max\{n_1,n_2,n_3,n_4\}$.
\end{proof}

We will then solve the solution from $u_{0,n}$ to $u_{1,n}$. The function $\mathscr{F}$ we use is
\begin{align}\label{def-F*n}
\mathscr{F}_{1,n}=\mathscr{F}(\delta_n,u_{0,n},u_{1,n})=\Omega_n^2\left|\log\frac{|u_{1,n}|}{|u_{0,n}|}\right|\ge 1.
\end{align}
Now we are going to derive an estimate for $\mathcal{A}_{1,n}$:
\begin{proposition} There exists some $N_3$ such that, for $n>N_3$,
\begin{align}\label{estimate-A*n}
\mathcal{A}_{1,n}=1.
\end{align}
\end{proposition}

\begin{proof}
From Proposition \ref{u0toun}, or the conclusion of Theorem \ref{existencetheorem}, we have
\begin{align*}
|u_{0,n}|\|\Omega\chih, \omega, L\phi\|_{\L^\infty_{\ub}\H^4(u_{0,n})}\lesssim\mathscr{F}_n\mathscr{W}_n^{\frac{1}{2}}\mathcal{A}_n\lesssim \mathscr{F}_n\mathscr{W}_n^{\frac{1}{2}}.
\end{align*}
Because now the assumptions in Theorem \ref{instabilitytheorem} requires that we have bounds for three order higher derivatives of the initial data than those in Theorem \ref{existencetheorem}, we can prove using similar argument to that in Section \ref{APrioriEstimate}, that
\begin{align*}
|u_{0,n}|\|\Omega\chih,\omega, L\phi\|_{\L^\infty_{\ub}\H^7(u_{0,n})}\lesssim\mathscr{F}_n\mathscr{W}_n^{\frac{1}{2}}\mathcal{A}_n\lesssim \mathscr{F}_n\mathscr{W}_n^{\frac{1}{2}}.
\end{align*}
We omit the details here. Then
\begin{align*}
\mathscr{F}_{1,n}^{-1}|u_{0,n}|\|\Omega\chih, \omega, L\phi\|_{\L^\infty_{\ub}\H^7(u_{0,n})}\lesssim \mathscr{F}_{1,n}^{-1}\mathscr{F}_n\mathscr{W}_n^{\frac{1}{2}}\le|\varphi_n|C^2\Omega_n^2\frac{34}{\varphi_n^2} \left|\log\frac{\Omega_n}{\Omega_0(u_0)}\right|^{\frac{1}{2}}
\end{align*}
which tends to zero as $n\to\infty$. We should also compute $\displaystyle\sup_{u_{0,n}\le u'\le u_{1,n}}(\mathscr{F}_{1,n}^{-1}|\varphi(u')|)$. We compute
\begin{align*}
\sup_{u_{0,n}\le u'\le u_{1,n}}|\varphi(u')|\le&|\varphi_n|+\int_{u_{0,n}}^{u_{1,n}}\frac{\Omega_0^2(u')h(u')|\psi(u')|}{|u'|}\D u'\\
\le&|\varphi_n|+\Omega_n^2\left|\log\frac{|u_{1,n}|}{|u_{0,n}|}\right|^{\frac{1}{2}}\\
\le&|\varphi_n|(1+(5C)^{-1}).
\end{align*}
The last inequality is by \eqref{u*}. From \eqref{def-F*n} and \eqref{F*n>=1}, we also deduce that 
\begin{align*}\displaystyle\sup_{u_{0,n}\le u'\le u_{1,n}}(\mathscr{F}_{1,n}^{-1}|\varphi(u')|)\to0
\end{align*} as $n\to\infty$. Then by the definition \eqref{def-A} of $\mathcal{A}_{1,n}$, \eqref{estimate-A*n} follows if $n$ is larger than some $N_3$.

\end{proof}

The final proposition we want to prove is
\begin{proposition} There exists some $N_4$ such that for $n>N_4$, we have
\begin{align}\label{conditiontrappedLphiunpre}
\inf_{\vartheta\in S^2}\int_0^{\delta_n}|u_{0,n}|^2[|\Omega\chih|^2+2|L\phi|^2](\ub',u_{0,n},\vartheta)\D\ub'\ge 17C^2\Omega_n^4\left|\log\frac{|u_{1,n}|}{|u_{0,n}|}\right|.
\end{align} 
\end{proposition}
\begin{proof}

Similar to deriving \eqref{estimate-Lphi} using the equaiton \eqref{equ-DbLphi}, we repeat the derivation of \eqref{estimate-Lphi}, with the estimate of the term $\nablas\phi,\eta$ being improved using \eqref{estimate-nablasphi-improve},  \eqref{estimate-eta-improve}, we achieve the following: for all $\vartheta\in S^2$,
\begin{equation}\label{estimate-uLphi-varphi-improve}
\sup_{0\le\ub\le\delta_n}||u_{0,n}|L\phi(\ub,u_{0,n},\vartheta)-\varphi(u)|\lesssim \sup_{0\le\ub\le\delta_n}||u_0|L\phi(\ub,u_0,\vartheta)-\varphi(u_0)|+\delta_n|u_{0,n}|^{-1}\mathscr{F}_n^2\mathscr{W}_n\mathcal{A}_n^2.
\end{equation}
The first term on the right hand side is bounded from \eqref{conditionA} by $2A$. The second term can also be bounded by $A$ for $n$ sufficiently large. This is because, by \eqref{estimate-deltanun-1},
\begin{align*}
\delta_n|u_{0,n}|^{-1}\mathscr{F}_n^2\mathscr{W}_n\mathcal{A}_n^2\le C^{-2}\mathrm{e}^{-C^{-2}\Omega_n^{-4}\frac{\varphi_n^2}{34}}\cdot\varphi_n^2\left|\log\frac{\Omega_n}{\Omega_0(u_0)}\right|^{\frac{1}{2}}\to0
\end{align*}
as $n\to\infty$.
Then we have
\begin{align*}
\sup_{0\le\ub\le\delta_n}||u_{0,n}|L\phi(\ub,u_{0,n},\vartheta)-\varphi(u)|\le cA
\end{align*}
for some universal $c$, which implies $||u_{0,n}|L\phi(\ub,u_{0,n},\vartheta)|\ge\frac{|\varphi_n|}{2}, \forall0\le\ub\le\delta_n$ when $n$ is larger than some $N_4$.
Together with \eqref{u*} and \eqref{estimate-A*n}, we have
\begin{align*}
\int_0^{\delta_n}|u_{0,n}|^2|L\phi(\ub',u_{0,n},\vartheta)|^2\D\ub'\ge\frac{17}{2}C^2\Omega_n^4\left|\log\frac{|u_{1,n}|}{|u_{0,n}|}\right|,
\end{align*}
for all $\vartheta\in S^2$ and this implies \eqref{conditiontrappedLphiunpre}.

\end{proof}

By \eqref{delta}, \eqref{def-F*n},  \eqref{estimate-A*n} and \eqref{conditiontrappedLphiunpre}, we can apply Theorem \ref{formationoftrappedsurfaces1}, with $u_0$ in the theorem replaced by $u_{0,n}$ and $u_1$ by $u_{1,n}$, and with $a=1$, to conclude that if $n>N:=\max\{N_1,N_2,N_3,N_4\}$, the solution exists in $0\le\ub\le\delta_n$, $u_{0,n}\le u\le u_{1,n}$, and the sphere $S_{\delta_n, u_{1,n}}$ is a closed trapped surface. In addition, as $n\to\infty$, $\delta_n\to0^+$ and $u_{1,n}\to0^-$. This completes the proof of Case 1.

\end{proof}

\begin{proof}[Proof of Case 2] In this case, $\varphi(u)$ is bounded for all $u\in[u_0,0)$. 
Similar to Case 1, we would like to find some $\tu$ such that we can apply Theorem \ref{formationoftrappedsurfaces1} on $C_{\tu}$. Let $\tdelta$, $\widetilde{u}_*$ be defined in terms of $\tu$ through \eqref{u*tilde} and \eqref{deltatilde}. Then we can find some $\varepsilon_1$ such that if $|\tu|<\varepsilon_1$,
\begin{align}\label{F*>=1}
\widetilde{\Omega}^2\left|\log\frac{|\widetilde{u}_*|}{|\widetilde{u}|}\right|=\frac{1}{17}C^{-2}\widetilde{\Omega}^{-\gamma}\ge1,
\end{align}
and this implies, together with \eqref{deltatilde},
\begin{align}
\label{estimate-deltatilde}
C^2\tdelta\le |\tu_*|.
\end{align}

First of all, we will prove
\begin{proposition}\label{u0toutilde}Setting the function $\mathscr{F}$ to be
 $$\widetilde{\mathscr{F}}=\mathscr{F}(\tdelta,u_0,\widetilde{u})=\max\left\{A,\sup_{u_0\le u'\le\widetilde{u}}|\varphi(u')|\right\},$$
Then there exists some $\varepsilon_2$ such that for $|\tu|<\varepsilon_2$, the solution of the Einstein equations exists in the region $0\le\ub\le\tdelta$, $u_0\le u\le\tu$, and the estimates stated in Theorem \ref{existencetheorem} hold.
\end{proposition}
\begin{proof}
We only need to verify the smallness conditions and the auxiliary condition. The corresponding $\widetilde{\mathcal{A}}$ is then uniformly bounded by $1$ again. We also have $\widetilde{\mathscr{F}}\le F$ uniformly. We denote the corresponding functions $\mathscr{E},\mathscr{W}$ by $\widetilde{\mathscr{E}},\widetilde{\mathscr{W}}$.

Substitude \eqref{estimate-deltatilde} back to \eqref{u*tilde}, we have
\begin{align}\label{estimate-deltau-1tilde}
\tdelta|\tu|^{-1}\le C^{-2}\mathrm{e}^{-\frac{1}{17}C^{-2}\widetilde{\Omega}^{-2-\gamma}}.
\end{align}
We compute
 \begin{align*}
C^2\tdelta|\tu|^{-1}\widetilde{\mathscr{F}}\widetilde{\mathscr{W}}\widetilde{\mathcal{A}}\le \mathrm{e}^{-\frac{1}{17}C^{-2}\widetilde{\Omega}^{-2-\gamma}}\cdot F\cdot\left|\log\frac{\tOmega}{\Omega_0(u_0)}\right|
\end{align*}
which will tend to zero as $\tu\to0$. This verifies the second one of the smallness conditions because $C\ge C_1\ge C_0$, and the auxiliary condition is verified similarly. Similar to Case 1, using the bound \eqref{conditionE}, we will have
\begin{align*}
\Omega_0^2(u_0)\tdelta|\tu|^{-1}\widetilde{\mathscr{E}}^2\widetilde{\mathscr{W}}\le C^{-2}\mathrm{e}^{-\frac{1}{17}C^{-2}\widetilde{\Omega}^{-2-\gamma}}\left|\log\frac{\tOmega}{\Omega_0(u_0)}\right|E
\end{align*}
which will also tend to zero as $\tu\to0$. This verifies the first one of the smallness conditions.

 \end{proof}
We will then solve the solution from $\tu$ to $\tu_*$. The function $\mathscr{F}$ is set to be
\begin{align}\label{def-F*utilde}
\widetilde{\mathscr{F}}_*=\mathscr{F}(\tdelta,\widetilde{u},\widetilde{u}_*)=\widetilde{\Omega}^2\left|\log\frac{|\widetilde{u}_*|}{|\widetilde{u}|}\right|\ge1.\end{align}
Similar to Case 1, we derive an estimate for $\widetilde{\mathcal{A}}_*$:
\begin{proposition}
There exists some $\varepsilon_3$ such that for $|\tu|<\varepsilon_3$, 
\begin{align}\label{estimate-Atilde*}
\widetilde{\mathcal{A}}_*=\mathcal{A}(\tdelta,\widetilde{u},\widetilde{u}_*)=1.
\end{align}
\end{proposition}
\begin{proof} Again similar to Case 1, we compute the quantity 
\begin{align*}
 \widetilde{\mathscr{F}}_*^{-1}\widetilde{\mathscr{F}}\widetilde{\mathscr{W}}^{\frac{1}{2}}\le17C^2\tOmega^\gamma \left|\log\frac{\tOmega}{\Omega_0(u_0)}\right|^{\frac{1}{2}}F
\end{align*}
which tends to zero as $\tu\to0^-$.
The proof is then completed after we note that
\begin{align*}
\sup_{\tu\le u'\le \tu_*}(\widetilde{\mathscr{F}}^{-1}_*|\varphi(u')|)\to0
\end{align*}
as $\tu\to0$ (because of \eqref{def-F*utilde} and \eqref{F*>=1}).
\end{proof}

The final proposition we will prove is
\begin{proposition}
There exists some $\varepsilon_4$ such that if $|\tu|<\varepsilon_4$ and 
\begin{align*}
\Omega_0^{\gamma-2}(\widetilde{u})f(\tu)\ge2,
\end{align*}
then we have
 \begin{align}\label{conditiontrappedLphiutildepre}
\inf_{\vartheta\in S^2}\int_0^{\tdelta}(|\tu|^2|\Omega\chih|^2+2|\tu|^2|L\phi|^2)(\ub',\tu,\vartheta)\D\ub'\ge 17C^2\tOmega^4\tdelta\left|\log\frac{|\tu_*|}{|\tu|}\right|.
\end{align}
\end{proposition}
\begin{proof}
We integrate the equation \eqref{equ-DbLphi} from $u=u_0$ to $u=\tu$, and obtain, similar to \eqref{estimate-uLphi-varphi-improve}, for all $0\le\ub\le\tdelta$, $\vartheta\in S^2$,
\begin{align*}
|(|\widetilde{u}|L\phi(\ub,\widetilde{u},\vartheta)-\varphi(u))-(|u_0|L\phi(\ub,u_0,\vartheta)-\varphi(u_0))|\lesssim\tdelta|\widetilde{u}|^{-1}\widetilde{\mathscr{F}}^2\widetilde{\mathscr{W}}\widetilde{\mathcal{A}}^2,
\end{align*}
which implies
\begin{align}\label{estimate-utildeLphi-lower}
||\widetilde{u}|L\phi(\ub,\widetilde{u},\vartheta)|\ge||u_0|L\phi(\ub,u_0,\vartheta)+(\varphi(\widetilde{u})-\varphi(u_0))|-c\tdelta|\widetilde{u}|^{-1}\widetilde{\mathscr{F}}^2\widetilde{\mathscr{W}}\widetilde{\mathcal{A}}^2,
\end{align}
On the other hand, consider the equation \eqref{equ-chih}, written in the following form:
\begin{equation*}
\begin{split}
\frac{\partial}{\partial u}(|u|^2|\Omega\chih|^2)
=&2|u|^2\langle\Omega^2(\nablas\tensor\eta+\eta\tensor\eta+\nablas\phi\tensor\nablas\phi-\nablas_b(\Omega\chih)),\Omega\chih\rangle\\
&-\left(\Omega\tr\chib+\frac{2}{|u|}\right)|u|^2|\Omega\chih|^2-\Omega\tr\chi|u|^2\langle\Omega\chibh,\Omega\chih\rangle,
\end{split}
\end{equation*}
we will have, for all $0\le\ub\le\tdelta$, $\vartheta\in S^2$
\begin{align}\label{estimate-utildechih-lower}
|\tu|^2|\Omega\chih(\ub,\tu,\vartheta)|^2\ge|u_0|^2|\Omega\chih(\ub,u_0,\vartheta)|^2-c\tdelta|\tu|^{-1}\widetilde{\mathscr{F}}^3\widetilde{\mathcal{A}}^3.
\end{align}

Integrating the sqaure of \eqref{estimate-utildeLphi-lower} and \eqref{estimate-utildechih-lower} over $\ub$, together with \eqref{estimate-deltau-1tilde}, we have
\begin{align*}
&\frac{1}{\tdelta}\int_0^{\tdelta}(|\tu|^2|\Omega\chih|^2+2|\tu|^2|L\phi|^2)(\ub',\tu,\vartheta)\D\ub'\\
\ge&\frac{1}{\tdelta}\int_0^{\tdelta}(|u_0|^2|\Omega\chih(\ub,u_0,\vartheta)|^2+||u_0|L\phi(\ub,u_0,\vartheta)+(\varphi(\tu)-\varphi(u_0))|^2)\D\ub'-c C^{-2}\mathrm{e}^{-\frac{1}{17}C^{-2}\widetilde{\Omega}^{-2-\gamma}}F^3\\
=& f(\tu)-c C^{-2}\mathrm{e}^{-\frac{1}{17}C^{-2}\widetilde{\Omega}^{-2-\gamma}}F^3.
\end{align*}

Then there exists some $\varepsilon_4$ such that if $|\tu|<\varepsilon_4$, such that 
\begin{align*}
cC^{-2}\mathrm{e}^{-\frac{1}{18}C^{-2}\widetilde{\Omega}^{-2-\gamma}}F^3\le\frac{1}{2}\tOmega^{2-\gamma}.
\end{align*}
If $\Omega_0^{\gamma-2}(\widetilde{u})f(\tu;\gamma)\ge2$, we then have
 \begin{align*}
\frac{1}{\tdelta}\int_0^{\tdelta}(|\tu|^2|\Omega\chih|^2+2|\tu|^2|L\phi|^2)(\ub',\tu,\vartheta)\D\ub'\ge&2\tOmega^{2-\gamma}-\frac{1}{2}\tOmega^{2-\gamma}\ge\tOmega^{2-\gamma}.
\end{align*}
In view of \eqref{u*tilde} and \eqref{estimate-Atilde*},
 \begin{align*}
\frac{1}{\tdelta}\int_0^{\tdelta}(|\tu|^2|\Omega\chih|^2+2|\tu|^2|L\phi|^2)(\ub',\tu,\vartheta)\D\ub'\ge 17C^2\tOmega^4\left|\log\frac{|\tu_*|}{|\tu|}\right|,
\end{align*}
for all $\vartheta\in S^2$ which implies \eqref{conditiontrappedLphiutildepre}.
\end{proof}
By \eqref{deltatilde}, \eqref{def-F*utilde}, \eqref{estimate-Atilde*} and \eqref{conditiontrappedLphiutildepre}, we can apply Theorem \ref{formationoftrappedsurfaces1} with $a=1$, and conclude that $S_{\tdelta,\tu_*}$ is a closed trapped surface if $|\tu|<\varepsilon_0:=\min\{\varepsilon_1,\varepsilon_2,\varepsilon_3,\varepsilon_4\}$. This completes the proof of Case 2.

\end{proof}
\begin{remark}\label{remark-sufficientsingular}
The additional bound \eqref{conditionE}, or the stronger bound \eqref{conditionE'}, which is more natural, are introduced to guarantee the validity of the first one of the smallness assumptions \eqref{smallness} in both cases. We will use instead the bound \eqref{conditionE} in the actual proof in the next section. On the other hand, if we assume the rate of $\Omega_0$ tending to zero is sufficiently fast, for example,
\begin{align*}
\sup_{u_0\le u<0}\mathrm{e}^{-\frac{1}{17}C^{-2}\Omega_0^{-2}(u)}\left|\log\frac{|u|}{|u_0|}\right|<+\infty,
\end{align*}
which is clearly not optimal, then we can see from the proof that we do not need the bound \eqref{conditionE}. 
\end{remark}

\section{Gravitational perturbations}\label{proofofmaintheorem}

Recalling that the initial data on $C_{u_0}$ consists the conformal metric $\widehat{\gs}$, the lapse $\Omega$ and the scalar field function $\phi$. In this section, we should consider non-smooth initial data. 
\subsection{The instability theorems for non-smooth initial data}  First of all, we state and prove the instability theorem for non-smooth initial data.

\begin{theorem}\label{instabilitytheoremweak}
Suppose that the initial data on $\Cb_0$ is smooth, spherically symmetric, singular at the vertex, and  $\varphi(u)$ is bounded. The data on $C_{u_0}$ is given on a fixed intervel, say $\ub\in[0,1]$, not necessarily smooth, satisfies
\begin{equation}\label{condition-A-generic}
\begin{split}
|u_0|\sup_{\ub\in[0,1]}\left(\|\Omega\chih\|_{\H^{10}(\ub,u_0)}+\|\omega\|_{\H^8(\ub,u_0)}+\|L\phi\|_{\H^8(\ub,u_0)}\right)\le A
\end{split}
\end{equation}
and
\begin{equation}\label{condition-E-generic}
\begin{split}
\min\left\{\sup_{\delta\in(0,1]}\||u_0|(|u_0|\nablas)L\phi\|^2_{\L^2_{[0,\delta]}\H^7(u_0)}\left|\log\frac{C^2\delta}{|u_0|}\right|,\sup_{u_0\le u<0}\mathrm{e}^{-\frac{1}{17}C^{-2}\Omega_0^{-2}(u)}\left|\log\frac{|u|}{|u_0|}\right|\right\}\le E 
\end{split}
\end{equation}
for some $A\ge 1, E$ and $C\ge\max\{C_1,4\}$. Suppose also that for some $\gamma\in(0,2)$,
\begin{align}\label{instabilitycondition-generic}
\Omega_0^{\gamma-2}(\widetilde{u})f(\tu;\gamma)\ge32,
\end{align}
for some $|\tu|<\varepsilon_0$ where $f(\tu;\gamma)$ is the function defined in \eqref{def-f} and $\varepsilon_0$ is given in Theorem \ref{instabilitytheorem}.

Then there exists a pair of numbers $(\tdelta, \tu_*)$ such that the following conclusions hold: Let $(\widehat{\gs}_n, \Omega_n, L\phi_n)$ be a sequence of smooth initial data on $C_{u_0}$, satisfying \eqref{condition-A-generic} and \eqref{condition-E-generic} with the same constants $A$ and $E$. Suppose that
\begin{align*}
\frac{1}{2}|\xi|_{\gs}\le|\xi|_{\gs_n}\le 2|\xi|_{\gs}
\end{align*}
for all $2$-covariant tensor field $\xi$ and all $n$, and
\begin{align}
&\label{converge-chihn}\lim_{n\to\infty}\sup_{\vartheta\in S^2}\int_0^{1}|u_0|^2|\Omega_n\chih_n-\Omega\chih|_{\gs}^2(\ub,u_0,\vartheta)\D\ub=0,\\
&\label{converge-Lphin}\lim_{n\to\infty}\sup_{\vartheta\in S^2}\int_0^{1}|u_0|^2|L\phi_n(\ub,u_0,\vartheta)-L\phi(\ub,u_0,\vartheta)|^2\D\ub=0.
\end{align}
Then there exists some $N$ such that for all $n>N$, the sphere $S_{\tdelta, \tu_*}$ is a closed trapped surface in the maximal development of the initial data $(\widehat{\gs}_n, \Omega_n, L\phi_n)$. Moreover, the sphere $S_{\tdelta,\tu_*}$ is uniformly strictly trapped, in the sense that for all $n>N$,
\begin{align*}
(\tr\chi')_n|_{S_{\tdelta,\tu_*}}\le-\frac{17}{16}\cdot\frac{2}{|\tu_*|}, \ \ \Omega_n\tr\chib_n|_{S_{\tdelta,\tu}}\le-\frac{1}{|\tu_*|}.
\end{align*}
\end{theorem}
\begin{proof}
Choose numbers $\tdelta$ and $\tu_*$ defined through \eqref{u*tilde} and \eqref{deltatilde} in terms of $\tu$. The trapped surface being strictly trapped uniformly is followed directly from the proof of Theorem \ref{formationoftrappedsurfaces}. We only need to verify that for all $n>N$ for some $N$, \eqref{instabilitycondition} holds and then apply Theorem \ref{instabilitytheorem} and Remark \ref{remark-sufficientsingular}. Indeed, from the assumptions, we can find some $N_1$ such that for all $n>N_1$,
\begin{align*}
\sup_{\vartheta\in S^2}\int_0^{\tdelta(\tu;\gamma)}|u_0|^2|\Omega_n\chih_n-\Omega\chih|_{\gs}^2(\ub,u_0,\vartheta)\D \ub\le\frac{1}{8}\tdelta(\tu;\gamma)f(\tu;\gamma),
\end{align*}
then we obtain
\begin{align*}
&\frac{1}{\tdelta(\tu;\gamma)}\int_0^{\tdelta(\tu;\gamma)}|u_0|^2|\Omega_n\chih_n|_{\gs_n}^2(\ub,u_0,\vartheta)\D \ub\\
\ge&\frac{1}{4\tdelta(\tu;\gamma)}\left(\frac{1}{2}\int_0^{\tdelta(\tu;\gamma)}|u_0|^2|\Omega\chih|_{\gs}^2(\ub,u_0,\vartheta)\D \ub-\int_0^{\tdelta(\tu;\gamma)}|u_0|^2|\Omega_n\chih_n-\Omega\chih|_{\gs}^2(\ub,u_0,\vartheta)\D \ub\right)\\
\ge&\frac{1}{8\tdelta(\tu;\gamma)}\int_0^{\tdelta(\tu;\gamma)}|u_0|^2|\Omega\chih|_{\gs}^2(\ub,u_0,\vartheta)\D \ub-\frac{1}{32}\tdelta(\tu;\gamma)f(\tu;\gamma).
\end{align*}
Similar argument allows us to find some $N_2$ such that for $n>N_2$,
\begin{align*}
&\frac{1}{\tdelta(\tu;\gamma)}\int_0^{\tdelta(\tu;\gamma)}||u_0|L\phi_n(\ub,u_0,\vartheta)+(\varphi(\tu)-\varphi(u_0)|^2\D \ub\\
\ge&\frac{1}{8\tdelta(\tu;\gamma)}\int_0^{\tdelta(\tu;\gamma)}||u_0|L\phi(\ub,u_0,\vartheta)+(\varphi(\tu)-\varphi(u_0)|^2\D \ub-\frac{1}{32}\tdelta(\tu;\gamma)f(\tu;\gamma).
\end{align*}
Summing up the above two inequalities implies that \eqref{instabilitycondition} holds for $(\widehat{\gs}_n,\Omega_n,L\phi_n)$ if $n>N=\max\{N_1,N_2\}$ from the assumption \eqref{instabilitycondition-generic}. 
\end{proof}
From the conclusions of the above theorem, if the solutions of the Einstein equations with initial data $(\widehat{\gs}_n,\Omega_n,\phi_n)$ converge, in any senses such that $(\tr\chi')_n$ and $\Omega_n\tr\chib_n$ converge pointwisely on $S_{\tdelta,\tu}$, then the limiting spacetime has a closed trapped surface. Therefore we may say, that the future development of such initial data has a closed trapped surface.

Finally, we have the following.
\begin{theorem}\label{instabilitycorollaryweak}
Suppose that the initial data on $\Cb_0$ is smooth, spherically symmetric and singular at the vertex. The data $(\widehat{\gs},\Omega,\phi)$ on $C_{u_0}$ satisfies \eqref{condition-A-generic} and \eqref{condition-E-generic} for some $C\ge\max\{C_1,4\}$. If $\varphi(u)$ is bounded, we assume in addition
\begin{align}\label{instabilityconditionlimsup-generic}
\limsup_{\tu\to0^-}\Omega_0^{\gamma-2}(\tu)f(\tu;\gamma)>32
\end{align}
 for some $\gamma\in(0,2)$.
 
Then we can find two sequences $\tdelta_k\to0^+$ and $\tu_{1,k}\to0^-$, such that the following conclusions hold: Let $(\widehat{\gs}_n,\Omega_n,\phi_n)$ be a sequence of smooth initial data satisfying the same assumptions as in the statement of Theorem \ref{instabilitytheoremweak}. 
Then for every $k$, there exists some $N=N_k$, such that for all $n>N_k$, 
\begin{align*}
(\tr\chi')_n|_{S_{\tdelta_k,\tu_{1,k}}}\le-\frac{17}{16}\cdot\frac{2}{|\tu_{1,k}|}, \ \ \Omega_n\tr\chib_n|_{S_{\tdelta_k,\tu_{1,k}}}\le-\frac{1}{|\tu_{1,k}|}.
\end{align*}
In particular, $S_{\tdelta_k,\tu_{1,k}}$ is trapped in the maximal development of $(\widehat{\gs}_n,\Omega_n,L\phi_n)$ for all $n>N_k$.
\end{theorem}
\begin{proof}
When $\varphi(u)$ is bounded, this is a direct corollary of Theorem \ref{instabilitytheoremweak}. When $\varphi(u)$ is unbounded, we can see from the proof of Case 1 in Theorem \ref{instabilitytheorem} that the choice of the sequence of $u_{0,n}$, and therefore the locations of the closed trapped surfaces only depends on the initial data on $\Cb_0$, and the conclusion follows from this observation. Moreover, in this case, the sequences of $\chih_n$ and $(L\phi)_n$ does not need to satisfy \eqref{converge-chihn} and \eqref{converge-Lphin}.

\end{proof}
Using a limiting argument, we may also say, that the future development of the such initial data has a sequence of closed trapped surfaces approaching the singularity.
\begin{remark}
In Theorem \ref{instabilitytheoremweak}, if the function $f$ in the condition \eqref{instabilitycondition-generic} is replaced by another function $f'$ defined as
\begin{align}\label{def-f'}
f'(\widetilde{u};\gamma)=\frac{1}{\tdelta(\tu;\gamma)}\inf_{\vartheta\in S^2}\int_0^{\tdelta(\tu;\gamma)}|u_0|^2|\Omega\chih|^2(\ub,u_0,\vartheta)\D\ub,
\end{align}
then the converging of $L\phi$ in \eqref{converge-Lphin} is not needed. \end{remark}
\subsection{The space of the initial data sets} The final part is to investigate the space of the initial data sets. We first choose the stereographic charts $(\vartheta_1,\vartheta_2)$ (both north pole and south pole charts) on $S_{0,u_0}$ and extend them to the whole $C_{u_0}$ for $\ub\in[0,1]$ by requiring $L\theta_A=0$, $A=1,2$. Since we only consider the perturbations on the conformal metric $\widehat{\gs}$ for simplicity, we will fix the lapse $\Omega$ such that $\Omega\in C^1_{\ub}H^{8}_\vartheta$ and the scalar field function $\phi$ such that $\phi\in C^1_{\ub}H^{8}_\vartheta$. We also fix an initial data on $\Cb_0$ which is smooth, spherically symmetric and singular at the vertex. Moreover, we assume that \eqref{condition-E-generic} holds for some $E$ where the metric on $S_{\ub,u_0}$ is understood to be the standard round metric with radius $|u_0|$ because the real metric $\gs$ is not yet defined. If $\phi$ is assumed to be smooth, then \eqref{condition-E-generic} does hold for some $E$.

In such a coordinate system, the conformal metric $\widehat{\gs}$ can be written as a pair of symmetric positive definite matrices $\widehat{\gs}_{AB}$ writing in the form
\begin{align}\label{conformalmetricincoordinate}
\widehat{\gs}_{AB}(\ub,u_0,\vartheta)=\frac{|u_0|^2}{(1+\frac{1}{4}(\vartheta_1^2+\vartheta_2^2))^2}m_{AB}(\ub,\vartheta)
\end{align}
where $m_{AB}$ represents a pair of matrices with determinant $1$ and satisfying the coordinate transformation rule, see for example Chapter 2 in \cite{Chr}. For simplicity, we use a single notation $\Psi_{AB}$ to denote the pair  of matrices $\Psi_{AB}, \Psi'_{AB}$ satisfying the coordinate transformation rule. Then this pair of matrices defines a tensor $\Psi$. We then introduce the definitions of the spaces of the initial data. {\bf We remark that all the definitions, statements and proofs below depend on the initial data on $\Cb_0$ we fix.}

\begin{definition}
We define $\mathcal{I}$ to be the space of the conformal metrics $\widehat{\gs}$ defined for $\ub\in[0,1]$ and $\vartheta\in S^2$ such that $\widehat{\gs}_{AB}\in C^1_{\ub}H^{10}_\vartheta$, and $\widehat{\gs}(0,\vartheta)$ is the standard round metric with radius $|u_0|$, $\frac{\partial}{\partial\ub}\widehat{\gs}(0,\vartheta)=0$.
\end{definition}

\begin{definition} We define $\mathcal{E}\subset \mathcal{I}$ to be the  collection of $\widehat{\gs}\in\mathcal{I}$ such that the conclusion of Theorem \ref{instabilitycorollaryweak} does not hold for any sequences $\tdelta_k\to0^+$ and $\tu_{1,k}\to 0^-$.
\end{definition}
\begin{definition} We define $\mathcal{E}_\varepsilon\subset \mathcal{I}$ to be the  collection of $\widehat{\gs}\in\mathcal{I}$ such that the conclusion of Theorem \ref{instabilitytheoremweak} does not hold for any $(\tdelta, \tu_*)$ with $\tdelta<\varepsilon$.
\end{definition}

 \begin{remark}
 It is not difficult to use a similar argument in Chapter 2 in \cite{Chr} that \eqref{condition-A-generic} holds for some $A$ and \eqref{condition-E-generic} holds for some different $E$ by comparing  on $S_{\ub,u_0}$ the real metric $\gs$ and the standard round metric.
 \end{remark}
At last, we will prove the following precise forms of Theorem \ref{main1} and \ref{main2}, the main results of the present article.\begin{theorem}[Precise version of Theorem \ref{main1}]\label{main1precise} $\mathcal{E}$ is of first category in $\mathcal{I}$. This is to say, 
$\mathcal{E}^c$, the complement of $\mathcal{E}$ in $\mathcal{I}$, contains a subset that is a countably intersection of open and dense subsets in $\mathcal{I}$.
\end{theorem}
\begin{theorem}[Precise version of Theorem \ref{main2}]\label{main2precise} $\mathcal{E}_\varepsilon^c$ contains a subset that is open and dense in $\mathcal{I}$ for all $\varepsilon>0$.
\end{theorem}
\begin{proof}[Proof of Theorem \ref{main1precise}]
We begin by defining $\mathcal{N}_{2-\gamma}\subset\mathcal{I}$ such that
\begin{align*}
\mathcal{N}_{2-\gamma}^c=\{\widehat{\gs}\in\mathcal{I}|\limsup_{\tu\to0^-}\Omega_0^{\gamma-2}(\tu)f'(\tu;\gamma)>32\}
\end{align*}
where the function $f'$ is defined in \eqref{def-f'}.  From Theorem \ref{instabilitycorollaryweak}, $\mathcal{E}\subset\mathcal{N}_{2-\gamma}$ for any $\gamma\in(0,2)$. Then we only need to prove that $\mathcal{N}_{2-\gamma}$ is of first category in $\mathcal{I}$ for some $\gamma\in(0,2)$. We define also $\mathcal{N}_{2-\gamma,\varepsilon}\subset\mathcal{I}$ such that \begin{align*}
\mathcal{N}_{2-\gamma,\varepsilon}^c=\{\widehat{\gs}\in\mathcal{I}|\text{ there exists some $\tu$ with $\tdelta(\tu;\gamma)<\varepsilon$ such that }\Omega_0^{\gamma-2}(\tu)f'(\tu;\gamma)\ge33\}.
\end{align*}
It is clear that $\mathcal{N}_{2-\gamma}^c\supset\bigcap_{i}\mathcal{N}^c_{2-\gamma,\varepsilon_i}$ for any sequences $\varepsilon_i\to0$, and therefore we only need to prove that $\mathcal{N}_{2-\gamma,\varepsilon}^c$ is open and dense in $\mathcal{I}$ for all $\varepsilon>0$ and $\gamma\in(0,2)$. From now on, we fix some $\varepsilon>0$ and $\gamma\in(0,2)$ arbitrarily.

Using the argument in the proof of Theorem \ref{instabilitytheoremweak}, we know that $\mathcal{N}_{2-\gamma,\varepsilon}^c$ is in fact open in $L^\infty_\vartheta H^1_{\ub}$. Since $C^1_{\ub}H^{10}_\vartheta\hookrightarrow C^1_{\ub}C^8_\vartheta\hookrightarrow C^8_\vartheta C^1_{\ub}\hookrightarrow L^\infty_\vartheta H^1_{\ub}$, then $\mathcal{N}_{2-\gamma,\varepsilon}^c$ is also open in $C^1_{\ub}H^{10}_\vartheta$. The denseness is followed by the following proposition.
\begin{proposition}\label{dense}Given any $\widehat{\gs}\in\mathcal{I}$, there exists a one-parameter family $\widehat{\gs}_t\in\mathcal{I}$ at least for sufficiently small $t$ such that $\widehat{\gs}_t\in\mathcal{N}_{2-\gamma,\varepsilon}^c$ for $t\ne0$, and $\widehat{\gs}_t\to\widehat{\gs}$ in $C^1_{\ub}H^{10}_\vartheta$ as $t\to0$.\end{proposition}
\begin{proof}
Fix two pairs of symmetric trace-free matrices $\widetilde{\Psi}_{AB}(\vartheta)$ and $\widetilde{\Psi}'_{AB}(\vartheta)$ on $S^2=S_{0,u_0}$ such that for each $p\in S^2$, either $\widetilde{\Psi}_{AB}(p)\ne0$ or $\widetilde{\Psi}'_{AB}(p)\ne0$. Then we define a pair of symmetric trace-free matrix valued functions $\overline{\Psi}_{AB}(\ub,\vartheta)$ for $\ub\in(0,1],\vartheta\in S^2$ by
\begin{align*}
\overline{\Psi}_{AB}(\ub,\vartheta)=\begin{cases}\widetilde{\Psi}_{AB}(\vartheta),\ &2^{-4n-4}<\ub\le2^{-4n-3}\\\widetilde{\Psi}'_{AB}(\vartheta),\ &2^{-4n-2}<\ub\le2^{-4n-1},\\\text{continuously extended},\ &\text{otherwise} \end{cases}\ \text{for}\ n=0,1,2,\cdots.
\end{align*}
The magnitude  $|\overline{\Psi}|=\sqrt{\sum_{A,B=1,2}(\overline{\Psi}_{AB})^2}$ is a well-defined function because of the coordinate transformation rule. We then define the normalization $\Psi_{AB}$ of $\overline{\Psi}_{AB}$ by
\begin{align}\label{def-Psi}
\Psi_{AB}(\ub,\vartheta)=\overline{\Psi}_{AB}(\ub,\vartheta)\Bigg/\left(\frac{1}{\ub}\int_0^{\ub}|\overline{\Psi}|^2(\ub',\vartheta)\D\ub'\right)^{\frac{1}{2}},
\end{align}
which is well-defined since $\int_0^{\ub}|\overline{\Psi}|^2(\ub',\vartheta)\D\ub'$ is nowhere zero for all $\ub\in(0,1]$. $\Psi_{AB}$ is by definition also a pair of symmetric trace-free matrices satisfying the coordinate transformation rule and for all $\ub\in(0,1]$ and $\vartheta\in S^2$,
\begin{align}\label{property-PsiAB}
\frac{1}{\ub}\int_0^{\ub}|\Psi|^2(\ub',\vartheta)\D\ub'=1.
\end{align}
However, $\Psi$ is not continuous at $\ub=0$ and we will make a cut-off. For any $\delta>0$, let $j_\delta(\ub)$ be a continuous cut-off function such that $j_\delta(0)=0$, $j_\delta(\ub)=1$ for $\frac{\delta}{2}\le \ub\le1$, and $0\le j_\delta(\ub)\le 1$. Then $j_\delta\Psi_{AB}$ is defined for all $\ub\in[0,1]$ and from \eqref{property-PsiAB}, we have
\begin{align}\label{property-jdeltaPsiAB}
\frac{1}{\delta}\int_0^{\delta}|j_\delta\Psi|^2(\ub,\vartheta)\D\ub\ge\frac{1}{2}.
\end{align}
for all $\delta\in(0,1]$ and $\vartheta\in S^2$.

Now given any $t\in\mathbb{R}$ with small absolute value, we will choose a $\tu_t>-\varepsilon$ sufficiently close to $0$ according to $t$ such that $\tu_t\to0^-$ as $t\to0$. The choice of $\tu_t$ will be determined in the course of the proof. We then define
\begin{align}\label{def-mAB-cutoff}
\mathfrak{m}_{t;AB}(\ub,\vartheta)=\exp\left(t\int_0^{\ub}j_{\tdelta(\tu_t;\gamma)}(\ub')\Psi_{AB}(\ub',\vartheta)\D\ub'\right)
\end{align}
for all $t\in\mathbb{R}$. Because $\int_0^{\ub}j_{\tdelta(\tu_t;\gamma)}\Psi_{AB}(\ub',\vartheta)\D\ub'$ is also a pair of symmetric trace-free matrices satisfying the coordinate transformation rule, then $\mathfrak{m}_{t;AB}$  is then a pair of symmetric matrices with determinate $1$ satisfying the coordinate transformation rule and this defines a family of conformal metrics $\widehat{\gfs}_t$ through \eqref{conformalmetricincoordinate}. We then compute
\begin{align*}
|\Omega\chih_{\gfs_t}|^2_{\gfs_t}=\frac{1}{4}(\mathfrak{m}^{-1}_t)^{AC}(\mathfrak{m}^{-1}_t)^{BD}\frac{\partial \mathfrak{m}_{t;AB}}{\partial\ub}\frac{\partial \mathfrak{m}_{t;CD}}{\partial\ub}
\end{align*}
in both coordinate charts\footnote{See formula (2.70) in \cite{Chr}.}. From the definition \eqref{def-mAB-cutoff} of $\mathfrak{m}_{t;AB}$, we have
\begin{align*}
(\mathfrak{m}_t^{-1})^{AB}=I_{AB}+O_t(\ub)
\end{align*}
where $I_{AB}$ is the identity matrix and
\begin{align*}
\frac{\partial \mathfrak{m}_{t;AB}}{\partial\ub}=tj_{\tdelta(\tu_t;\gamma)}\Psi_{AB}+O_t(\ub).
\end{align*}
Then we have
\begin{align*}
|\Omega\chih_{\gfs_t}|^2_{\gfs_t}=\frac{t^2}{4}|j_{\tdelta(\tu_t;\gamma)}\Psi|^2+O_t(\ub).
\end{align*}
Here the constants in $O_t(\ub)$ depend on $\widetilde{\Psi}_{AB}$, $\widetilde{\Psi}'_{AB}$ and $t$, but not on the particular choice of $\tu_t$. Therefore we can choose $\tu_t$ sufficiently close to $0$ such that $|O_t(\ub)|\le\frac{t^2}{16}$ for $0\le\ub\le\tdelta(\tu_t;\gamma)$ and then from \eqref{property-jdeltaPsiAB},
\begin{align}\label{property-chihgfst}
\frac{1}{\tdelta(\tu_t;\gamma)}\int_0^{\tdelta(\tu_t;\gamma)}|\Omega\chih_{\gfs_t}|^2_{\gfs_t}(\ub,\vartheta)\D\ub\ge\frac{t^2}{16}.
\end{align}

Now we begin to construct $\widehat{\gs}_t$. For the given $\widehat{\gs}\in\mathcal{I}$, we have a pair of matrices $m_{AB}$ given by \eqref{conformalmetricincoordinate}. We then define 
\begin{align}\label{def-ghatt-cutoff}
m_{t;AB}=m_{AC}\mathfrak{m}_{t;CB}
\end{align}
which also satisfies the coordinate transformation rule and hence defines a family of the conformal metrics $\widehat{\gs}_t$ for all small $t$. In particular, $\widehat{\gs}_0\equiv\widehat{\gs}$ and $\widehat{\gs}_t\in\mathcal{I}$ and $\widehat{\gs}_t\to \widehat{\gs}$ as $t\to0$. Let $\gs_t$ be the corresponding family of full metrics and $\chih_t$ be the corresponding family of shear tensors, by direct computations, we have, 
\begin{align*}
|\Omega\chih_t|_{{\gs}_t}^2(\ub,u_0,\vartheta)\ge\frac{1}{2}|\Omega\chih_{\gfs_t}|^2_{\gfs_t}-|\Omega\chih|_{\gs}^2(\ub,u_0,\vartheta)-O_t(\ub).
\end{align*}
By choosing $\tu_t$ sufficiently close to zero such that $|\Omega\chih|^2_{\gs}(\ub,u_0,\vartheta)\le\frac{t^2}{128}$ (recalling that $\chih(0,u_0,\vartheta)=0$) and $|O_t(\ub)|\le\frac{t^2}{128}$ for all $0\le\ub\le\tdelta(\tu_t;\gamma)$, $\vartheta\in S^2$, from \eqref{property-chihgfst}, we have
\begin{align}\label{property-chiht}
\frac{1}{\tdelta(\tu_t;\gamma)}\int_0^{\tdelta(\tu_t;\gamma)}|\Omega\chih_t|^2_{\gs_t}(\ub,u_0,\vartheta)\D\ub\ge\frac{t^2}{64}.
\end{align}
Finally, we can choose $\tu_t$ sufficiently close to zero such that $\Omega_0^{2-\gamma}(\tu_t)\le\frac{t^2}{64\times 33}$, and conclude that $\widehat{\gs}_t\in\mathcal{N}^c_{2-\gamma,\varepsilon}$ for all nonzero small $t$.

\end{proof}
This proves the denseness of $\mathcal{N}^c_{2-\gamma,\varepsilon}$ and the proof of Theorem \ref{main1precise} is completed.

\end{proof}
\begin{proof}[Proof of Theorem \ref{main2precise}]
We denote 
$$B_A=\left\{\widehat{\gs}\in\mathcal{I}\Big||u_0|\sup_{\ub\in[0,1]}\|\Omega\chih\|_{\H^{10}(\ub,u_0)}< A\right\}$$
for every $A\ge 1$. Fix a sequence $A_i\to+\infty$, then $\bigcup_i B_{A_i}=\mathcal{I}$. From Theorem \ref{instabilitycorollaryweak}, for any given $i, \varepsilon>0,\gamma\in(0,2)$, there exists some $\varepsilon_{0,i}<\varepsilon$ such that
$$\mathcal{N}^c_{2-\gamma,\varepsilon_{0,i}}\cap B_{A_i}\subset\mathcal{E}^c_{\varepsilon}.$$
Therefore
$$\mathcal{E}^c_\varepsilon\supset\bigcup_i(\mathcal{N}^c_{2-\gamma,\varepsilon_{0,i}}\cap B_{A_i}).$$
The set on the right is obviously open since it is a union of open subsets, and we only need to show that it is dense. Given any bounded open subset $U$ of $\mathcal{I}$, $U\subset B_{A_{i_0}}$ for some $i_0$. Then since $\mathcal{N}^c_{2-\gamma,\varepsilon_{0,i_0}}$ is dense by Proposition \ref{dense},
$$U\cap \bigcup_i(\mathcal{N}^c_{2-\gamma,\varepsilon_{0,i}}\cap B_{A_i})\supset U\cap(\mathcal{N}^c_{2-\gamma,\varepsilon_{0,{i_0}}}\cap B_{A_{i_0}})=U\cap \mathcal{N}^c_{2-\gamma,\varepsilon_{0,{i_0}}}\ne\emptyset.$$
This shows that $\mathcal{E}^c_\varepsilon$ contains an open and dense subset in $\mathcal{I}$ and completes the proof of Theorem \ref{main2precise}.

\end{proof}

\end{document}